\documentclass[sigconf]{cidr-2027}

\usepackage{microtype}
\usepackage{enumitem}
\usepackage{xspace}
\usepackage{booktabs}
\usepackage{tikz}
\usepackage{colortbl}

\setlist{nosep,leftmargin=1.25em}

\newtheorem{theorem}{Theorem}
\newtheorem{proposition}{Proposition}
\newtheorem{lemma}{Lemma}
\newtheorem{corollary}{Corollary}
\newtheorem{definition}{Definition}
\theoremstyle{remark}
\newtheorem{example}{Example}
\newtheorem{remark}{Remark}

\makeatletter
\renewcommand\paragraph{\@startsection{paragraph}{4}{\z@}%
  {-.5\baselineskip \@plus -2\p@ \@minus -.2\p@}%
  {-0.5em}%
  {\normalfont\bfseries}}
\makeatother

\newcommand{\InputOrd}{\preceq_I}
\newcommand{\OutputOrd}{\preceq_O}
\newcommand{\InputGe}{\succeq_I}
\newcommand{\OutputGe}{\succeq_O}

\newcommand{\commitment}{commitment\xspace}
\newcommand{\commitments}{commitments\xspace}
\newcommand{\Commitment}{Commitment\xspace}
\newcommand{\Commitments}{Commitments\xspace}
\newcommand{\LiftedOrd}{\trianglelefteq}
\newcommand{\novote}{\texttt{\_}}  
\newcommand{\inrel}{extends\xspace}          
\newcommand{\inrelbase}{extend\xspace}       
\newcommand{\inrelpart}{extending\xspace}    
\newcommand{\ingrows}{grows\xspace}          
\newcommand{\indesc}{input extension\xspace} 
\newcommand{\indescs}{input extensions\xspace}

\newcommand{\outrel}{extends\xspace}          
\newcommand{\outrelbase}{extend\xspace}       
\newcommand{\outrelpart}{extending\xspace}    
\newcommand{\outrelpast}{extended\xspace}     
\newcommand{\outnoun}{output extension\xspace}
\newcommand{\outnouns}{output extensions\xspace}
\newcommand{\Spec}{\mathcal{S}}
\newcommand{\spectype}[1][]{I/O specification#1\xspace}
\newcommand{\Spectype}[1][]{I/O specification#1\xspace}
\newcommand{\fiber}[1]{P(#1)}
\newcommand{\dog}{\texttt{dog}}
\newcommand{\catv}{\texttt{cat}}
\newcommand{\Subsumed}{\mathit{Subsumed}}
\newcommand{\Excluded}{\mathit{Excluded}}
\newcommand{\grayrow}[1]{\textcolor{gray}{#1}}
\newcommand{\sh}[1]{\cellcolor{gray!22}#1}
\newcommand{\strikerow}{%
  \noalign{\vskip-0.5\dimexpr\ht\strutbox+\dp\strutbox\relax}%
  \noalign{\vskip-0.5\arrayrulewidth}%
  \arrayrulecolor{black!65}\hline\arrayrulecolor{black}%
  \noalign{\vskip 0.5\dimexpr\ht\strutbox+\dp\strutbox\relax}%
  \noalign{\vskip-0.5\arrayrulewidth}}

\begin{document}

\title{Complete CALM: A Unified Criterion for Coordination}

\author{Joseph M. Hellerstein}
\affiliation{%
  \institution{UC Berkeley \& Amazon Web Services}
  \city{Berkeley, CA}
  \country{USA}
}

\begin{abstract}
Coordination is often the dominant cost in concurrent systems. When is it
intrinsically needed? Too little
produces heisenbugs; too much adds latency and sacrifices
availability.  CALM, CAP, and CRDTs each answer the question for one class of
programs, consistency models, or design patterns, in formalisms that cannot be
compared or combined, which leaves developers guessing.

We provide a simple formalism that unifies the literature and
offers new results.  A \spectype (spec) is a \emph{possibility table} containing, for
each input, every output a system may produce, with partial orders on inputs and
outputs.  Its runs are \emph{consistent} when they are confluent as input grows.
\emph{Functional CALM} generalizes the CALM Theorem from Datalog to any function
between ordered domains: the runs are consistent exactly when the function is
monotone.  \emph{Complete CALM}
extends this to nondeterministic specs, where the test is
\emph{Lifted Monotonicity}; \emph{Complete CAP} follows as a corollary.
Highly Available Transactions, Invariant Confluence, and later CALM refinements
map into the same framework.

Prior CALM results stop at monotonicity: for non-monotone
specs, they say only that coordination is needed somewhere.  Our framework offers constructive results on both sides: it says exactly
what coordination a spec needs, and when.  A non-monotone
spec's outputs drive irrevocable \emph{\commitments} that are required for consistency,
enforceable by coordination mechanisms like locks, atomics, or consensus.
A composition of specs has \emph{tight coordination} when the \commitments of an upstream stage
are necessary and sufficient for later stages.
We show how to order the upstream's outputs to achieve it, explaining designs such
as consensus logs in front of a replicated store.
\end{abstract}


\maketitle

\section{Introduction}
\label{sec:intro}

Concurrent systems rely on coordination mechanisms ranging from
atomic instructions and locks to barriers and consensus protocols.  These mechanisms
become the dominant cost in many systems, and they are tricky to get right.
Too little coordination
produces heisenbugs: rare, timing-dependent anomalies that are hard to
reproduce.  Too much adds overhead and latency, and in distributed settings
sacrifices availability.  So when is coordination avoidable, and when is it
intrinsically required?  CALM~\cite{ameloot2013relational},
CAP~\cite{gilbert2002cap}, and CRDTs~\cite{shapiro2011crdt} each answer the
question for one class of programs, consistency models, or design patterns, in
formalisms that cannot be compared or combined.  This leaves developers
guessing.

CALM's answer was a crisp, bidirectional criterion: a problem has a consistent,
coordination-free distributed implementation exactly when it is
monotonic~\cite{hellerstein2020keeping,hellerstein2010declarative}.  But its
proof covered a restricted setting: Datalog programs evaluated by relational transducer
networks~\cite{ameloot2013relational}.  CAP's answer was that a networked
shared-data system cannot be both consistent and available during a
partition~\cite{brewer2000towards,brewer2012cap}.  But its proof covered another
restricted setting: linearizable registers~\cite{gilbert2002cap}.  Both are
invoked broadly: CALM as a rule that monotone programs need no coordination, CAP
as a general constraint on distributed systems.  The folklore outgrew the
proofs.

The apparent connection between CALM and CAP was never established in a common
formalism, and nearby results only complicate the picture. CRDTs use
inflationary updates over join-semilattices for coordination-free
convergence~\cite{shapiro2011crdt}; HATs identify transactional guarantees
available under partition~\cite{bailis2013hat}; invariant confluence asks when
independently reached database states merge
safely~\cite{bailis2014coordination}. The diverse formalisms leave cross-cutting
questions unclear, including whether CALM encompasses CAP, whether CRDT join
laws are necessary, why only some transactional isolation levels remain
available during partition, and whether invariant confluence differs from CALM
(Section~\ref{sec:instances}).

We put these results on common ground, in a formalism that clarifies the
literature and offers new results.  Prior CALM results stop at the monotone
boundary: for a non-monotone program, they say only that coordination is needed
somewhere, not what it must accomplish or where it belongs.  Few real systems
are monotone throughout, so that is where developers most need help.  Our
framework is constructive on both sides of the boundary.  It answers questions the literature leaves open: what coordination a spec or
composition of specs requires, and whether an upstream stage's coordination is
necessary and sufficient for the stages after it.

\subsection{The model: \Spectype[s]}
In systems, as in life, our possibilities narrow over time. Sometimes the world
moves on, and possibilities are lost to us. Sometimes possibilities disappear
because we are forced to decide among them. A decision is incidental when
every possibility it removes is subsumed by a remaining alternative. A
decision is fateful when it commits us: it rules out a possibility that no remaining alternative subsumes. When must a system make a fateful decision?

We ground the problem in a simple model of \emph{\spectype[s]}, or \emph{specs}
for short, independent of any language, protocol, or placement. The designer of
a spec provides a \emph{possibility relation} $P\subseteq I\times O$: a virtual
table that may be infinite, like a mathematical function. Each row in this table
is a possibility observable from the outside: it pairs the input received so far
with a permitted output. The design includes the input domain $I$ and output domain $O$,
with partial orders that help capture the spec's semantics.  On inputs, the order
is accumulation: $i\InputOrd j$ means that $j$ \emph{\inrel} $i$ with more input
from the environment.  On outputs, the order is subsumption: $o\OutputOrd u$
means that  $u$ \emph{\outrel} $o$ by subsuming all relevant information in
$o$. These are both application-specific notions that must be captured well for
a spec to be faithful.

As a system runs, possibilities narrow: the spec's table loses rows.  When new
input arrives, the world has moved on, and rows that cannot be reached by
extending that input are lost.  When the system decides among possibilities, it
makes a controlled decision to remove rows.  The decision is incidental if each
removed output has a remaining output that \outrel it, and fateful if
not.  Choosing between two candidates in a leader election is fateful: no remaining row subsumes the possibility of the loser having won.  Closing
the polls early is also fateful, ruling out every row with a late ballot.  We call a fateful decision a
\emph{\commitment}: an irrevocable decision, enforceable by coordination
mechanisms such as locks and consensus.

An I/O automaton projects onto an \spectype[] that captures the observable
behavior of all executions of the automaton
(Section~\ref{sec:operational-models}).  The projection drops the order in
which an execution passes through the spec's rows, so the spec may also allow
behaviors the automaton cannot produce.  Modeling the automaton itself bridges
this gap and separates two questions: does any implementation of the spec need
coordination, and does this automaton need more, because of decisions the spec
does not require?
Despite its simplicity, our model supports not only application invariants,
but also general-purpose consistency definitions (Appendix~\ref{app:registers},
Section~\ref{sec:hat-instance}, Appendix~\ref{app:cc-bases}).

\subsection{Consistency and a criterion for achieving it}
The model provides clean definitions for the two sides of CALM: Consistency
and Monotonicity (Section~\ref{sec:functional-calm}).  Our Consistency definition is general: a spec's runs
(sequences of input and output events) are \emph{consistent} when they remain
\emph{confluent} as input grows: at any input extending all of the runs, they still reach a
common state.  Applied to different specs,
this one definition captures classical consistency models, from eventual
consistency of replicated data through linearizability, serializability and
other isolation levels (Section~\ref{sec:instances}).  It also captures consistency of
application specs.

A spec with one output per input is a function, and its monotonicity is standard: $i\InputOrd j$
implies $f(i)\OutputOrd f(j)$.  Nondeterministic specs, such as leader election
and transaction scheduling, need a richer monotonicity property that ensures
not only that outputs grow as inputs grow, but also that different output
choices for the same input can grow by $\OutputOrd$ to a common permitted
output.  We call this \emph{Lifted Monotonicity}, and a spec that has
it \emph{lifted monotone} (Section~\ref{sec:lifted-monotonicity}).  It reduces to standard monotonicity for functional specs.

Confluence is an operational property of all the runs a spec allows.  Lifted Monotonicity,
by contrast, is a static property of the spec's table: it can be checked on
pairs of rows, without exploring runs.  \emph{Functional CALM}
(Section~\ref{sec:functional-calm-theorem}) and \emph{Complete CALM}
(Section~\ref{sec:complete-calm-theorem}) show that the two coincide: the runs a spec allows are consistent exactly when the spec is
lifted monotone.  Hence Consistency As Lifted Monotonicity, a new reading of the
CALM acronym.

\subsection{Coordination is for \commitment}
What of coordination?  A \commitment isolates the contract a coordination
mechanism achieves.  Our criterion identifies exactly the \commitments a
spec's outputs require, independent of how they are implemented.  An implementation is correct when its
mechanisms enforce those \commitments.
In a composition of specs, coordination is \emph{tight}
when an upstream stage makes every \commitment the composition needs (sufficient), and no others (necessary).  Guard Tightening shows how to order the upstream's output domain to achieve this.  In state-machine replication, for example, the consensus log
makes every \commitment; the replicated store needs none, so any coordination it employs is overhead (Section~\ref{sec:composition}).

Distributed histories are captured under the same consistency criteria as
sequential runs (Section~\ref{sec:distributed-execution}).  Complete CAP
(Section~\ref{sec:complete-cap}) says Lifted Monotonicity exactly characterizes specs
that can keep every permitted output both available and consistent during a
partition.  The Gilbert--Lynch CAP Theorem~\cite{gilbert2002cap} follows directly.

\subsection{Contributions}
We make contributions in four categories:

\noindent
\emph{1. A unified model.} A system's observable behavior is captured via
        \spectype[s]. Consistency is defined as
        confluence of a spec's runs as input grows, unifying general consistency models,
        application-specific requirements, and their distributed executions.

\noindent
\emph{2. A unified criterion.} For specs with one output per input,
        Functional CALM shows that the runs are consistent exactly when the spec is monotone. For nondeterministic specs,
        we prove Complete CALM: Consistency As Lifted Monotonicity. Complete CAP is a corollary:
        the CAP tradeoff holds only under failure of Lifted Monotonicity.

\noindent
\emph{3. Constructive results for coordination.} We define \commitments,
        the underlying contracts that coordination mechanisms enforce,
        and identify exactly the \commitments a non-monotone spec requires.
        Specs compose by equijoin, and Guard Tightening shows how to order an
        upstream stage's outputs to achieve tight coordination.

\noindent
\emph{4. A map of the literature.} Highly Available Transactions, Invariant
        Confluence, CRDT observation, and later refinements of CALM are all
        explained via Complete CALM (Section~\ref{sec:instances}).

\vspace{0.5em}
\noindent
The paper develops these contributions in order, then relates the criterion to
earlier results (Section~\ref{sec:instances}) and the broader CALM agenda
(Section~\ref{sec:discussion}).  The appendices provide proofs and further
applications; Appendix~\ref{app:notation} collects notation and defined terms.

\section{Functional CALM}
\label{sec:functional-calm}

We begin with the simplified case of functional specs, in which each input has a
single specified output.  We develop the framework on a small election,
introducing each definition when the example first needs it
(Section~\ref{sec:election-story}), and then prove Functional CALM
(Section~\ref{sec:functional-calm-theorem}).  This setting covers the
Datalog-centric CALM literature and CRDTs.

\subsection{The framework by example}
\label{sec:election-story}
\label{sec:functional-specifications}
\label{sec:functional-events}
\label{sec:executions}

\begin{figure*}[t]
  \centering
  \small
  \newcommand{\rowhdr}{\toprule $\mathrm{in}$ & $\mathrm{out}$ \\ \midrule}
  \setlength{\tabcolsep}{6pt}
  \begin{tabular}[t]{@{}l c@{}}
    \multicolumn{2}{@{}l}{(a) Initial state, current input $a\novote\novote$}\\
    \rowhdr
    \sh{$\mathbf{a\novote\novote}$} & \sh{$\bot$} \\
    $ab\novote$ & \sh{$\bot$} \\
    $aba$ & \sh{$a$} \\
    $abb$ & \sh{$b$} \\
    $aa\novote$ & \sh{$a$} \\
    $aaa$ & \sh{$a$} \\
    $aab$ & \sh{$a$} \\
    $a\novote a$ & \sh{$a$} \\
    $a\novote b$ & \sh{$\bot$} \\
    \bottomrule
    \multicolumn{2}{@{}l}{\footnotesize upper fiber $\{\bot,a,b\}$; undetermined}\\
    \multicolumn{2}{@{}l}{\footnotesize $\bot$ needs no \commitment}
  \end{tabular}
  \hfill
  \begin{tabular}[t]{@{}l c@{}}
    \multicolumn{2}{@{}l}{(b) After $\mathsf{input}(ab\novote)$}\\
    \rowhdr
    \grayrow{$a\novote\novote$} & \grayrow{$\bot$} \\
    \sh{$\mathbf{ab\novote}$} & \sh{$\bot$} \\
    $aba$ & \sh{$a$} \\
    $abb$ & \sh{$b$} \\
    \grayrow{$aa\novote$} & \grayrow{$a$} \\
    \grayrow{$aaa$} & \grayrow{$a$} \\
    \grayrow{$aab$} & \grayrow{$a$} \\
    \grayrow{$a\novote a$} & \grayrow{$a$} \\
    \grayrow{$a\novote b$} & \grayrow{$\bot$} \\
    \bottomrule
    \multicolumn{2}{@{}l}{\footnotesize upper fiber $\{\bot,a,b\}$; undetermined}\\
    \multicolumn{2}{@{}l}{\footnotesize $\bot$ needs no \commitment}
  \end{tabular}
  \hfill
  \begin{tabular}[t]{@{}l c@{}}
    \multicolumn{2}{@{}l}{(c) After $\mathsf{input}(aa\novote)$ instead}\\
    \rowhdr
    \grayrow{$a\novote\novote$} & \grayrow{$\bot$} \\
    \grayrow{$ab\novote$} & \grayrow{$\bot$} \\
    \grayrow{$aba$} & \grayrow{$a$} \\
    \grayrow{$abb$} & \grayrow{$b$} \\
    \sh{$\mathbf{aa\novote}$} & \sh{$a$} \\
    $aaa$ & \sh{$a$} \\
    $aab$ & \sh{$a$} \\
    \grayrow{$a\novote a$} & \grayrow{$a$} \\
    \grayrow{$a\novote b$} & \grayrow{$\bot$} \\
    \bottomrule
    \multicolumn{2}{@{}l}{\footnotesize upper fiber $\{a\}$; determined}\\
    \multicolumn{2}{@{}l}{\footnotesize $a$ needs no \commitment}
  \end{tabular}
  \hfill
  \begin{tabular}[t]{@{}l c@{}}
    \multicolumn{2}{@{}l}{(d) After (b), then polls closed}\\
    \rowhdr
    \grayrow{$a\novote\novote$} & \grayrow{$\bot$} \\
    \sh{$\mathbf{ab\novote}$} & \sh{$\bot$} \\
    \grayrow{$aba$} & \grayrow{$a$} \\ \strikerow
    \grayrow{$abb$} & \grayrow{$b$} \\ \strikerow
    \grayrow{$aa\novote$} & \grayrow{$a$} \\
    \grayrow{$aaa$} & \grayrow{$a$} \\
    \grayrow{$aab$} & \grayrow{$a$} \\
    \grayrow{$a\novote a$} & \grayrow{$a$} \\
    \grayrow{$a\novote b$} & \grayrow{$\bot$} \\
    \bottomrule
    \multicolumn{2}{@{}l}{\footnotesize upper fiber $\{\bot\}$; determined}\\
    \multicolumn{2}{@{}l}{\footnotesize $\bot$ is the result}
  \end{tabular}
  \caption{The election spec and states it reaches, limited to inputs
  \inrelpart $a\novote\novote$.  Bold text: current input; shaded: current
  input and its upper fiber.  Gray text: lost after an input event; struck
  through if removed by a \commitment.  Panels (b)
  and (c) follow different input events.  In (d) the registrar closes the polls, a \commitment that emitting
  $\bot$ does not require, and emits $\bot$.}
  \label{fig:election-frames}
\end{figure*}

Consider the election illustrated in
Figure~\ref{fig:election-frames}.  Three voters each cast at most one immutable
ballot for $a$ or $b$.  The output is whichever candidate has a majority, or
$\bot$ while there is none.  An input is a vector of the ballots received so
far, one entry per voter, with a blank for a ballot not yet received.  We write
it as a string: $ab\novote$ means $v_1$ voted $a$, $v_2$ voted $b$, and $v_3$'s
ballot has not arrived.  Each row of the figure's tables pairs the ballots
received so far with the permitted result.  For brevity, the figure shows only
inputs \inrelpart $a\novote\novote$.

\paragraph{Specs and monotonicity.}
A spec is a possibility table for a function over ordered types.  To avoid
ambiguity, we state this precisely.

\begin{definition}[Functional \spectype]
\label{def:functional-specification}
A \emph{functional \spectype} is a tuple
\[
  \Spec=(I,\InputOrd,O,\OutputOrd,P),
\]
  where $(I,\InputOrd)$ and $(O,\OutputOrd)$ are ordered input and
  output types, and $P\subseteq I\times O$ is its \emph{possibility relation}:
  a typed relation with schema $(\mathrm{in},\mathrm{out})$ in which every
  input appears in at least one row ($P$ is \emph{total} on $I$) and at most one
  (a functional dependency):
\[
  (i,o)\in P\ \land\ (i,o')\in P
  \quad\Longrightarrow\quad o=o'.
\]
\end{definition}

The spec induces the function
\[f:(I,\InputOrd)\to(O,\OutputOrd)\]
with $P=\{(i,f(i)):i\in I\}.$  Our examples use partial orders, but our results
need less.  When we say \emph{order}, we mean a \emph{preorder}: a reflexive,
transitive relation.  Antisymmetry is not required: two distinct values may
each be below the other, making them equivalent, as two paraphrases of one
answer are.

The two orders require thoughtful design.
Inputs grow as ballots arrive: $i\InputOrd j$, read ``$j$ \emph{\inrel} $i$'', means that
$j$ has everything $i$ has and possibly more. So $ab\novote$ \inrel $a\novote\novote$,
while $aa\novote$ and $ab\novote$ are incomparable. Outputs grow as announcements
say more: $o\OutputOrd u$, read ``$u$ \outrel $o$'', means that $u$ makes every
claim $o$ makes. Either winner \outrel $\bot$, which claims nothing, while $a$
and $b$ are incomparable.

A functional spec is monotone when growing input only
ever extends its output:

\begin{definition}[Functional monotonicity]
\label{def:functional-monotonicity}
A functional \spectype $\Spec$, with induced function $f$, is
\emph{monotone} if
\begin{equation}
  i\InputOrd j
  \quad\Longrightarrow\quad
  f(i)\OutputOrd f(j).
  \label{eq:functional-monotonicity}
\end{equation}
\end{definition}

The election is monotone: to give an example, the input sequence
$a\novote\novote\InputOrd ab\novote\InputOrd aba$ produces the output sequence
$\bot\OutputOrd\bot\OutputOrd a$.

\paragraph{States and input events.}
As a run proceeds, the possibilities narrow.  Each panel of
Figure~\ref{fig:election-frames} shows one moment: the ballots received so
far, and the rows that can still be reached from them.  We call such a moment
a specification state.

\begin{definition}[Specification state, current input, remainder table]
\label{def:residual-state}
Given a specification $\Spec=(I,\InputOrd,O,\OutputOrd,P)$, a \emph{specification state} is a
pair $\sigma=(i,R)$, where $i\in I$ is the \emph{current input} and
$R\subseteq P$ is the \emph{remainder table}, such that:
\begin{enumerate}
  \item $R$ contains a row with input $i$; and
  \item every row $(k,o)\in R$ satisfies $i\InputOrd k$.
\end{enumerate}
\end{definition}

We say that a row $(k,u)$ is
\emph{at input} $k$.  To write remainder tables compactly, we borrow selection
from relational algebra: a predicate on the columns $\mathrm{in}$ and
$\mathrm{out}$ in square brackets selects rows, so $P[\mathrm{in}\InputGe i]$
holds the rows of $P$ whose input \inrel $i$, and commas separate conjuncts, as
in $R[\mathrm{in}\InputGe j,\ \mathrm{out}\OutputGe o]$.  Here $\InputGe$ and
$\OutputGe$ are the converse orders.  We write $\pi_{\mathrm{out}}(R)$ for the
set of outputs in $R$.  Not every input can still arrive in a state:

\begin{definition}[Legal input]
\label{def:legal-input}
An input $k$ is \emph{legal} at a state $\sigma=(i,R)$ when $R$ contains a row at $k$.
\end{definition}

A run starts from the \emph{canonical state} at its first input $i$, which
keeps every row of $P$ whose input \inrel $i$: it is the state of a system that
has received $i$ and emitted nothing,
\begin{equation}
  \sigma_P(i)=\bigl(i,\ P[\mathrm{in}\InputGe i]\bigr).
  \label{eq:canonical-state}
\end{equation}

When $v_2$'s ballot for $b$ arrives, rows such as $aa\novote$ no longer extend
the input, and they are lost: the world has moved on, and panel~(b) results.
For a state $(i,R)$ and a legal input $j$ with $i\InputOrd j$, the \emph{input
event}
\begin{equation}
  (i,R)\xrightarrow{\mathsf{input}(j)}
  \bigl(j,\ R' = R[\mathrm{in}\InputGe j]\bigr)
  \label{eq:input-event}
\end{equation}
makes $j$ current and drops the rows whose input does not \inrelbase $j$,
which we call \emph{input loss}.  An input event names the cumulative input
$j$: everything received so far, including all of $i$.  Thus
$\mathsf{input}(j)$ then $\mathsf{input}(k)$ has the effect of
$\mathsf{input}(k)$ alone, $\mathsf{input}(i)$ at $(i,R)$ has no effect, and an
input event from a canonical state leads to another canonical state.

Two sets of outputs matter in every state: those at the current
input (its \emph{fiber}), and those at the current input or any input that
\inrel it (its \emph{upper fiber}).  In panel~(b), the
fiber at $ab\novote$ is $\{\bot\}$ and the upper fiber, shaded, is $\{\bot,a,b\}$.  Formally:

\begin{definition}[Fiber, upper fiber]
\label{def:fiber}
Let $R$ be a table and $k$ an input.  The \emph{fiber} of $R$ at $k$ is
$R(k)=\pi_{\mathrm{out}}(R[\mathrm{in}=k])$, the outputs paired with $k$.  The
\emph{upper fiber} of $R$ at $k$ is $\pi_{\mathrm{out}}(R[\mathrm{in}\InputGe
k])$, the outputs paired with $k$ or with any input that \inrel $k$.
\end{definition}

A functional fiber has at most one value.  A state is determined when only one
output remains possible, as in panel~(c), where input alone has settled the
outcome.

\begin{definition}[Determined state]
\label{def:determined}
Let $\sigma = (i, R)$ be a specification state.  The upper fiber of $\sigma$
is $\pi_{\mathrm{out}}(R[\mathrm{in}\InputGe i]) = \pi_{\mathrm{out}}(R)$.
We say $\sigma$ is \emph{determined} when $|\pi_{\mathrm{out}}(R)|=1$.
\end{definition}

\paragraph{Output events and \commitments.}
Inputs arrive from the environment, but outputs are the system's to choose, and
once emitted they may have been observed and cannot be retracted.  At a state $(i,R)$, the current input
has a unique row $(i,o)$, so $o$, the \emph{current output}, is the only output
the system may emit.  Emitting $\bot$ in panel~(b) is safe, since every output
still possible builds on it; emitting $a$ there would not be, since $b$ is
still possible.

\begin{definition}[Legal output]
\label{def:exposable}
The output $o$ at $(i,R)$ is a \emph{legal output} when every output in the
upper fiber \outrel $o$.
\end{definition}

So $\bot$ is legal in panels~(a) and~(b), and $a$ in panel~(c).  Possibilities
can also be removed deliberately.  In panel~(d), a registrar that implements
the spec closes the polls when it emits $\bot$ at $ab\novote$, removing
$(aba,a)$ and $(abb,b)$ while keeping the current input.  No remaining row
\outrel them: this is a \commitment, the fateful decision of the
introduction.  A row $(k,u)$ is \emph{\outrelpast} by another
row $(k,w)$ at the same input when $w$ \outrel $u$.

\begin{definition}[\Commitment]
\label{def:commitment}
A \emph{\commitment} is a partial function $e$ on specification states that
removes a nonempty set of rows, none of which is \outrelpast by a row that
remains at the same input.  That is, whenever $e(i,R)=(j,R')$, we have
$R'\subsetneq R$ and
\begin{equation}
  (k,u)\in R\setminus R',\ (k,w)\in R'
  \quad\Longrightarrow\quad
  u\not\OutputOrd w.
  \label{eq:commitment-unsubsumed}
\end{equation}
It leaves the current input unchanged: $j=i$.
\end{definition}

Closing the polls was not required by the spec: $\bot$ was already legal, so
it was the registrar's choice, made through a \emph{\commitment event}.
\Commitment events are optional choices made by systems, beyond what a spec
requires; this section uses only canonical \commitments, and we return to
\commitment events in Section~\ref{sec:systems}.

In a functional spec each input has one row, so the last condition always
holds.  When $o$ is not legal,
we make it legal via a \commitment that removes the rows whose outputs conflict
with $o$.  Emitting $o$ only once it is legal keeps a run from
contradicting $o$ later on, when the environment may already have acted on it.

\begin{definition}[Output event]
\label{def:output-event}
At a state $(i,R)$ with current output $o$, the \emph{output event}
$\mathsf{output}(o)$ atomically keeps only the rows whose outputs \outrelbase $o$,
and emits $o$:
\begin{equation}
  (i,R)\xrightarrow{\mathsf{output}(o)}\bigl(i,\ R' = R[\mathrm{out}\OutputGe o]\bigr).
  \label{eq:output-action}
\end{equation}
\end{definition}

\begin{definition}[Canonical \commitment]
\label{def:canonical-commitment-functional}
When $o$ is not legal at $(i,R)$, equivalently $R'\neq R$, the event
removes the rows whose outputs do not \outrelbase $o$; this is
$o$'s \emph{canonical \commitment} $x_o(i,R)=(i,R')$, making $o$ legal.  In a functional table each
input has only one row, so $x_o$ rules out those rows' inputs ($x_o$ is a \commitment).
\end{definition}

An output event thus always emits a legal output, and applies a \commitment
exactly when its output was not already legal.  The current row $(i,o)$
remains, since $o$ \outrel itself, so the event always ends at a specification
state with current input $i$.  \emph{Determination} is the process of reaching
a determined state.  An input event can achieve it through input loss, as in
panel~(c), and a \commitment can do so by ruling out inputs that would
otherwise remain possible, as in panel~(d).  An event from a determined state
always leads to another determined state with that output: events only
remove rows, and a row at the current input always remains.

\begin{figure}[t]
  \centering
  \small
  \setlength{\tabcolsep}{6pt}
  \begin{minipage}{0.86\columnwidth}
  (e) As (b), all outputs incomparable, $\bot$ emitted\par\smallskip
  \begin{minipage}[t]{0.44\linewidth}
    \begin{tabular*}{\linewidth}[t]{@{}l@{\extracolsep{\fill}}c@{}}
      \toprule
      $\mathrm{in}$ & $\mathrm{out}$ \\ \midrule
      \grayrow{$a\novote\novote$} & \grayrow{$\bot$} \\
      \sh{$\mathbf{ab\novote}$} & \sh{$\bot$} \\
      \grayrow{$aba$} & \grayrow{$a$} \\ \strikerow
      \grayrow{$abb$} & \grayrow{$b$} \\ \strikerow
      \grayrow{$aa\novote$} & \grayrow{$a$} \\
    \end{tabular*}
  \end{minipage}\hfill
  \begin{minipage}[t]{0.44\linewidth}
    \begin{tabular*}{\linewidth}[t]{@{}l@{\extracolsep{\fill}}c@{}}
      \toprule
      $\mathrm{in}$ (cont.) & $\mathrm{out}$ (cont.) \\ \midrule
      \grayrow{$aaa$} & \grayrow{$a$} \\
      \grayrow{$aab$} & \grayrow{$a$} \\
      \grayrow{$a\novote a$} & \grayrow{$a$} \\
      \grayrow{$a\novote b$} & \grayrow{$\bot$} \\
      \bottomrule
    \end{tabular*}
    {\footnotesize upper fiber $\{\bot\}$; determined}
  \end{minipage}\par\smallskip
  \end{minipage}
  \caption{Panel (b) under the equality order, after $\mathsf{output}(\bot)$, wrapped
  to preserve space.}
  \label{fig:nonmonotone-frame}
\end{figure}

Figure~\ref{fig:nonmonotone-frame} considers a variant of our example that
changes only the output order $\OutputOrd$: it defines a new spec in which distinct
outputs are incomparable.
Under this order (called the \emph{equality order}), no output \outrel both $\bot$ and a future winner, so $\bot$ is not legal at $ab\novote$; emitting it applies
$x_\bot$, which removes $(aba,a)$ and $(abb,b)$.  That change in the output
order is Functional CALM in miniature: the spec is monotone under the original
order, where later outputs \outrelbase every earlier one; under the equality order
it is non-monotone, and some outputs can be emitted only with a \commitment.

\paragraph{Runs and consistency.}
Runs are the foundation of the consistency side of CALM.  A spec's runs are
sequences of its events, like those that produce panels~(a)--(c).

\begin{definition}[Specification event]
\label{def:specification-event}
Input and output events are \emph{specification events}: they change the
specification state.  We call them simply \emph{events} when unambiguous.
\end{definition}

\begin{definition}[Transition system]
\label{def:transition-system}
The \emph{transition system} consists of specification states and
specification events.
\end{definition}

\begin{definition}[Canonical run]
\label{def:run}
A \emph{canonical run}, or just \emph{run}, is any finite transition sequence of the transition system from the
canonical state $\sigma_P(i)$ at some input $i$.
\end{definition}

Runs range over every behavior the spec permits, not only the behaviors a
particular system happens to produce.  A system that never uses some of the
spec's latitude is better described by a narrower spec.  For the spec as
written, we assume that any permitted behavior may occur.  A run is
\emph{\commitment-free} when none of its output events applies a
\commitment.  A state is \emph{reachable without \commitment} when some
\commitment-free run reaches it.

Different runs may retain different rows.  Intuitively, two runs can still
reach a common state only if some later output extends what each has emitted.
Consider the rows $(ab\novote,\bot)$ and $(aba,a)$ under the equality order.
Their inputs have the \emph{common \indesc} $aba$, an input that \inrel both.
Suppose one run emits $\bot$ at $ab\novote$, and another receives $aba$ and
emits $a$.  No output at $aba$ \outrel both $\bot$ and $a$ under the equality order, so the rows have
no \emph{common \outnoun} there, and the two runs can never reach a common
state.

This motivates our first use of confluence, for rows.  Rows with a common
\outnoun at $j$ are \emph{confluent at} $j$.  A set of rows is
\emph{row-consistent} when it is confluent at every common \indesc of its
inputs, and \emph{row-inconsistent} otherwise.  When the context is clearly
about rows, we say simply \emph{consistent} or \emph{inconsistent}.

Our example rows $(ab\novote,\bot)$ and $(aba,a)$ are therefore inconsistent under the
equality order, so a run that emits $\bot$ at $ab\novote$ must apply the canonical \commitment $x_\bot$.

Runs follow the same pattern: they are \emph{confluent at} $j$ when they can
reach a common state with current input $j$, and consistent when they are
confluent at every common \indesc.  A spec is consistent when all of its
runs are.

\begin{definition}[Confluence at an input]
\label{def:confluence}
A finite family of runs ending at states $(i_1,R_1),\ldots,(i_k,R_k)$ is \emph{confluent at} $j$, a common \indesc of their current inputs, when each can be continued to a common state with current input $j$.
\end{definition}

\begin{definition}[Consistency of a spec]
\label{def:consistency}
The runs of an \spectype are \emph{consistent} when every finite family of
them is confluent at every common \indesc of their current inputs.
\end{definition}

The two uses of confluence agree: runs are confluent at $j$ exactly when the
rows of the outputs they have emitted are confluent at $j$
(Lemma~\ref{lem:reconciliation}).

Consistency requires confluence at every common extension of the runs' current
inputs, not just at some larger common extension that arrives later.  Otherwise
the runs could emit conflicting outputs before that larger input arrives, and
waiting cannot retract either output from the environment.  Even when later
input could bring the runs to a common state, the environment may already have
acted on the outputs through channels the system does not see, as in Lamport's
``anomalous behavior''~\cite{lamport1978time}.  And even when such a larger
input exists, it might never arrive.  Appendix~\ref{app:notation} collects the
notation for reference.

The election shows the whole arc.  Input narrows the possibilities, and
an output that everything remaining builds on is legal and needs nothing more.
An output that is not legal needs a \commitment; otherwise, runs that emit it
would diverge for good.  Next, we establish when that happens: when the spec is
not monotone.

\subsection{The Functional CALM theorem}
\label{sec:functional-calm-theorem}

We now have both sides that the introduction promised.  Monotonicity is a
static property of the spec's table, and consistency is an operational property
of the runs it allows.  Functional CALM bridges the two: the runs are
consistent exactly when the spec is monotone.  The proof reduces consistency to
a property of remainder tables, then to legality of outputs, and finally to
monotonicity.

Events are relational selections, which commute, so a run's remainder table
depends only on its current input and prior outputs:

\begin{lemma}[Rows remaining after a run]
\label{lem:run-residual}
If a run ends at current input $i$ after emitting
$o_1,\ldots,o_m$, its remainder table is
\begin{equation}
  P[\mathrm{in}\InputGe i,\ \mathrm{out}\OutputGe o_1,\ \ldots,\
    \mathrm{out}\OutputGe o_m],
  \label{eq:run-residual}
\end{equation}
the rows whose input \inrel $i$ and whose output \outrel each $o_r$.
\end{lemma}

\begin{proof}[Proof sketch]
Input and output events are both relational selections, and selections
commute; successive input selections reduce to the last, since each input
extends the one before.  Full proof in Appendix~\ref{app:proofs-sec2}.
\end{proof}

For example, in the equality-order variant of
Figure~\ref{fig:nonmonotone-frame}, a run that emits $\bot$ at $ab\novote$ keeps
$P[\mathrm{in}\InputGe ab\novote,\ \mathrm{out}\OutputGe\bot]$: the single row
$(ab\novote,\bot)$, as the figure shows.

A set of runs can be continued to a common state exactly when their
current remainder tables share a row at the target input:

\begin{proposition}[Survivor test]
\label{prop:confluence}
Runs ending at states $(i_r,R_r)$, $r=1,\ldots,k$, are confluent at a common
\indesc $j$ of their current inputs if and only if some row $(j,u)$ is in every
$R_r$.  For any such $u$, $\mathsf{input}(j)$ followed by $\mathsf{output}(u)$
brings every run to the same state,
$\bigl(j,\ P[\mathrm{in}\InputGe j,\ \mathrm{out}\OutputGe u]\bigr)$.
\end{proposition}

\begin{proof}[Proof sketch]
First assume the runs are confluent at $j$.  Rows are never added, so a row of
the common state at $j$ was already in every $R_r$.  Now assume $(j,u)$ is in
every $R_r$.  Then $\mathsf{input}(j)$ followed by $\mathsf{output}(u)$ brings
every run to the same state, so the runs are confluent at $j$.  Full proof in
Appendix~\ref{app:proofs-sec2}.
\end{proof}

In a functional spec the only row at $j$ is $(j,f(j))$, so the test
asks whether $(j,f(j))$ survives in every run.

Consistency thus ensures that runs never diverge irreparably: there is no
``split brain.''  Whatever events a finite collection of runs contains, at any common extension of their current inputs, some state is reachable by every run
and contradicts none of their outputs; Proposition~\ref{prop:confluence}
identifies it.  Consistency is a safety property of the spec's runs: it guarantees that a common
state exists, not that a system reaches it
(Section~\ref{sec:commitment-conflict}).  Reaching it takes an algorithm,
message delivery, and a termination argument, which belong to the operational
model in which a spec is realized.  With the survivor test, consistency comes down to legality.
An output that is not legal needs a \commitment, which removes a row $(k,y)$
that another run may keep; the two runs then share no row at $k$ and cannot be
reconciled.

\begin{lemma}[Confluence and legality]
\label{lem:confluence-legality}
The runs of an \spectype are consistent if and only if no output event of any
run applies a \commitment.
\end{lemma}

\begin{proof}[Proof sketch]
First assume the runs are consistent, and suppose some run's output event
$\mathsf{output}(o)$ applies $x_o$.  It removes some row $(k,y)$ that is not confluent with $(i,o)$ at $k$.  Compare that run with one from the state
before that event that instead takes $\mathsf{input}(k)$ and then emits $y$:
the two share no row at $k$, so by the Survivor test they are not confluent at
$k$, a contradiction.  Now assume no output event applies a \commitment, and
take a family of runs and a common extension $j$ of their current inputs.  Each run keeps a row
$(i_r,u_r)$ at its current input whose output \outrel everything it emitted
(Lemma~\ref{lem:run-residual}), and $u_r$ is legal at $\sigma_P(i_r)$.
Applying that legality one run at a time shows that these rows are confluent at $j$, and the Survivor test gives confluence.  Full proof in
Appendix~\ref{app:proofs-sec2}.
\end{proof}

In the election under the equality order, emitting $\bot$ at $ab\novote$ applies
a \commitment, since $\bot$ is not legal there, and the two runs of
Section~\ref{sec:election-story} share no row at $aba$.

In the functional case, monotonicity has a strong consequence: emitting outputs
never changes the state, which depends only on the input.

\begin{corollary}[State is fixed by input]
\label{cor:state-by-input}
In a monotone functional specification, every run ending at current
input $i$ is in the canonical state $\sigma_P(i)$, whatever outputs it emitted
along the way.
\end{corollary}

\begin{proof}[Proof sketch]
By monotonicity, the outputs a run emits remove no row at an extension of $i$ (Lemma~\ref{lem:run-residual}).  Full proof in Appendix~\ref{app:proofs-sec2}.
\end{proof}

To be confluent at a common extension $j$, each run therefore needs only the
single event $\mathsf{input}(j)$, which takes it to $\sigma_P(j)$, the state of a
run that starts with all of the input.  Knowing the input is enough to know the state, and with it the one
output $f(j)$.

Monotonicity also governs legality:

\begin{proposition}[No-\commitment test]
\label{prop:functional-exposure-boundary}
A functional \spectype is monotone if and only if every output at
every state reachable without \commitment is legal.
\end{proposition}

\begin{proof}[Proof sketch]
First assume $f$ is monotone.  By Corollary~\ref{cor:state-by-input}, runs
reach only canonical states, and at $\sigma_P(i)$ the output $f(i)$ is legal
because $f$ is monotone at $i$.  Now assume every output at every state
reachable without \commitment is legal.  In particular $f(i)$ is legal at
$\sigma_P(i)$, which input alone reaches, and that holds exactly when $f$ is
monotone at $i$.  Full proof in Appendix~\ref{app:proofs-sec2}.
\end{proof}

\begin{theorem}[Functional CALM]
\label{thm:functional-calm}
The runs of a functional \spectype are consistent if and only if the spec is
monotone.
\end{theorem}

\begin{proof}[Proof sketch]
By Lemma~\ref{lem:confluence-legality}, the runs are consistent exactly when
no output event of a run applies a \commitment.  First assume they are
consistent, and suppose $f(i)\not\OutputOrd f(k)$ for some $i\InputOrd k$.  At
$\sigma_P(i)$ the event $\mathsf{output}(f(i))$ then applies $x_{f(i)}$ and
rules out input $k$, a contradiction, so the spec is monotone.  Now assume the
spec is monotone.  It can emit every current output without \commitment
(Proposition~\ref{prop:functional-exposure-boundary}), so by induction on its
length every run is \commitment-free, and consistency follows.  Full proof in
Appendix~\ref{app:proofs-sec2}.
\end{proof}

Functional CALM restates the familiar CALM dichotomy at the level of the spec
rather than inside a particular programming language.
Classical CALM fits, with inputs $J$ as sets of facts ordered by
inclusion and output views $Q(J)$ ordered by inclusion.  Functional
monotonicity becomes
        \begin{equation}
  J\subseteq J'
  \quad\Longrightarrow\quad
  Q(J)\subseteq Q(J'),
  \label{eq:query-monotonicity}
        \end{equation}
the standard definition of a monotone query~\cite{abiteboul1995foundations};
positive Datalog satisfies it, and negation may not.  Functional CALM
generalizes the statement to arbitrary input and
output orders, whatever formalism expresses the spec.  CRDT state specs have the same functional shape, with received payloads ordered by inclusion and states by the semilattice order~\cite{shapiro2011crdt} (Sections~\ref{sec:query-instance}
and~\ref{sec:crdt-instance}).

\subsection{Discussion: fixed inputs}
\label{sec:fixed-inputs}

A fixed problem might seem to call for a spec with a single input: the
complete one, with no extensions.  Such a spec is trivially monotone and needs
no \commitment.  What it leaves out are the incomplete inputs at which a system
may answer along the way, such as $ab\novote$ in the election, whose universe of
inputs is itself small and finite.  Declaring an input complete, as closing
the polls does, removes the rows at inputs that have yet to arrive, and with
them outputs that no remaining row subsumes.  That removal is a \commitment
whether or not the spec is monotone and whether or not its domain is finite.

\section{Complete CALM}
\label{sec:complete-calm}

We now extend CALM to nondeterministic specs, whose fibers may contain several outputs: we define Lifted Monotonicity (\S\ref{sec:relational-specifications}), lift output events (\S\ref{sec:lifted-output-events}), and prove Complete CALM (\S\ref{sec:complete-calm-theorem}).

\subsection{Nondeterministic specs and Lifted Monotonicity}
\label{sec:relational-specifications}

We begin with the full relational definition of an \spectype.
\begin{definition}[\Spectype]
\label{def:relational-specification}
An \emph{\spectype} is
        \begin{equation}
  \Spec=(I,\InputOrd,O,\OutputOrd,P),
  \qquad P\subseteq I\times O,
  \label{eq:relational-specification}
        \end{equation}
  where $(I,\InputOrd)$ and $(O,\OutputOrd)$ are ordered sets and
  $P\subseteq I\times O$ is its \emph{possibility relation}.  We depict $P$ as
a table and call its members rows; it need not be finite.  The relation $P$
  is total on $I$: every input appears in at least one row.  A row $(i,o)\in P$
  is a possible input/output pair: output $o$ is permitted at input $i$.  As in
  Section~\ref{sec:functional-specifications}, the fiber at $i$ is
        \begin{equation}
  \fiber{i}=\{o\in O:(i,o)\in P\},
  \label{eq:fiber}
        \end{equation}
  and it may contain several outputs.  A functional \spectype is the special
  case $\fiber{i}=\{f(i)\}$.
\end{definition}

\begin{definition}[Nondeterministic \spectype]
\label{def:nondeterministic}
An \spectype is \emph{nondeterministic} when some fiber
contains more than one output.  Specification nondeterminism describes permitted results, not how a system chooses among them, and does not
by itself imply a \commitment.
\end{definition}

Specification states and input events
(Section~\ref{sec:election-story})
apply unchanged to relational tables.

Figure~\ref{fig:ambiguity-frames} contrasts two specs with several outputs in
one fiber: serial schedules for two transactions and route sets observed by a
traceroute-like tool, whose input is the set of sources traced so far.  For the discussion here, transaction outputs are simply
serial schedules, to which executions correspond via equivalence;
Section~\ref{sec:hat-instance} and Appendix~\ref{app:cc-bases} study richer transactional specs.

\label{sec:ambiguity-story}
\begin{figure*}[t]
  \centering
  \small
  \setlength{\tabcolsep}{4pt}
  \begin{minipage}[t]{0.15\textwidth}
    \raggedright
    (f) After $[A,B]$\par\smallskip
    \begin{tabular*}{\linewidth}[t]{@{}l@{\extracolsep{\fill}}l@{}}
      \toprule
      $\mathrm{in}$ & $\mathrm{out}$ \\ \midrule
      $\{A,B\}$ & \sh{$[A,B]$} \\
      \grayrow{$\{A,B\}$} & \grayrow{$[B,A]$} \\ \strikerow
      \bottomrule
    \end{tabular*}\par\smallskip
    {\footnotesize no remaining row \outrel the removed one: a \commitment}
  \end{minipage}\hfill
  \begin{minipage}[t]{0.405\textwidth}
    \raggedright
    (g) After $\{\texttt{s}{\to}\texttt{s1}\}$ at input $\{\texttt{s}\}$\par\smallskip
    \begin{tabular*}{\linewidth}[t]{@{}l@{\extracolsep{\fill}}l@{}}
      \toprule
      $\mathrm{in}$ & $\mathrm{out}$ \\ \midrule
      \sh{$\{\texttt{\textbf{s}}\}$} & \sh{$\{\texttt{s}{\to}\texttt{s1}\}$} \\
      \grayrow{$\{\texttt{s}\}$} & \grayrow{$\{\texttt{s}{\to}\texttt{h}\}$} \\
      \grayrow{$\{\texttt{s}\}$} & \grayrow{$\{\texttt{s}{\to}\texttt{h},\ \texttt{s}{\to}\texttt{h}{\to}\texttt{s2}\}$} \\
      \sh{$\{\texttt{\textbf{s}}\}$} & \sh{$\{\texttt{s}{\to}\texttt{s1},\ \texttt{s}{\to}\texttt{h}\}$} \\
      \sh{$\{\texttt{\textbf{s}}\}$} & \sh{$\{\texttt{s}{\to}\texttt{s1},\ \texttt{s}{\to}\texttt{h},\ \texttt{s}{\to}\texttt{h}{\to}\texttt{s2}\}$} \\
      $\{\texttt{s},\texttt{t}\}$ & \sh{$\{\texttt{s}{\to}\texttt{s1}\}$} \\
      \grayrow{$\{\texttt{s},\texttt{t}\}$} & \grayrow{$\{\texttt{s}{\to}\texttt{h}\}$} \\
      \grayrow{$\{\texttt{s},\texttt{t}\}$} & \grayrow{$\{\texttt{s}{\to}\texttt{h},\ \texttt{s}{\to}\texttt{h}{\to}\texttt{s2}\}$} \\
      $\{\texttt{s},\texttt{t}\}$ & \sh{$\{\texttt{s}{\to}\texttt{s1},\ \texttt{s}{\to}\texttt{h}\}$} \\
      $\{\texttt{s},\texttt{t}\}$ & \sh{$\{\texttt{s}{\to}\texttt{s1},\ \texttt{s}{\to}\texttt{h},\ \texttt{s}{\to}\texttt{h}{\to}\texttt{s2}\}$} \\
      \grayrow{$\{\texttt{s},\texttt{t}\}$} & \grayrow{$\{\texttt{t}{\to}\texttt{s2}\}$} \\
      $\{\texttt{s},\texttt{t}\}$ & \sh{$\{\texttt{s}{\to}\texttt{s1},\ \texttt{t}{\to}\texttt{s2}\}$} \\
      \grayrow{$\{\texttt{s},\texttt{t}\}$} & \grayrow{$\{\texttt{s}{\to}\texttt{h},\ \texttt{t}{\to}\texttt{s2}\}$} \\
      \grayrow{$\{\texttt{s},\texttt{t}\}$} & \grayrow{$\{\texttt{s}{\to}\texttt{h},\ \texttt{s}{\to}\texttt{h}{\to}\texttt{s2},\ \texttt{t}{\to}\texttt{s2}\}$} \\
      $\{\texttt{s},\texttt{t}\}$ & \sh{$\{\texttt{s}{\to}\texttt{s1},\ \texttt{s}{\to}\texttt{h},\ \texttt{t}{\to}\texttt{s2}\}$} \\
      $\{\texttt{s},\texttt{t}\}$ & \sh{$\{\texttt{s}{\to}\texttt{s1},\ \texttt{s}{\to}\texttt{h},\ \texttt{s}{\to}\texttt{h}{\to}\texttt{s2},\ \texttt{t}{\to}\texttt{s2}\}$} \\
      \bottomrule
    \end{tabular*}\par\smallskip
    {\footnotesize each removed row is \outrelpast by a remaining one at its input: an output subsumption}
  \end{minipage}\hfill
  \begin{minipage}[t]{0.405\textwidth}
    \raggedright
    (h) After $\mathsf{input}(\{\texttt{s},\texttt{t}\})$\par\smallskip
    \begin{tabular*}{\linewidth}[t]{@{}l@{\extracolsep{\fill}}l@{}}
      \toprule
      $\mathrm{in}$ & $\mathrm{out}$ \\ \midrule
      \grayrow{$\{\texttt{s}\}$} & \grayrow{$\{\texttt{s}{\to}\texttt{s1}\}$} \\
      \grayrow{$\{\texttt{s}\}$} & \grayrow{$\{\texttt{s}{\to}\texttt{h}\}$} \\
      \grayrow{$\{\texttt{s}\}$} & \grayrow{$\{\texttt{s}{\to}\texttt{h},\ \texttt{s}{\to}\texttt{h}{\to}\texttt{s2}\}$} \\
      \grayrow{$\{\texttt{s}\}$} & \grayrow{$\{\texttt{s}{\to}\texttt{s1},\ \texttt{s}{\to}\texttt{h}\}$} \\
      \grayrow{$\{\texttt{s}\}$} & \grayrow{$\{\texttt{s}{\to}\texttt{s1},\ \texttt{s}{\to}\texttt{h},\ \texttt{s}{\to}\texttt{h}{\to}\texttt{s2}\}$} \\
      \sh{$\{\texttt{\textbf{s}},\texttt{\textbf{t}}\}$} & \sh{$\{\texttt{s}{\to}\texttt{s1}\}$} \\
      \grayrow{$\{\texttt{s},\texttt{t}\}$} & \grayrow{$\{\texttt{s}{\to}\texttt{h}\}$} \\
      \grayrow{$\{\texttt{s},\texttt{t}\}$} & \grayrow{$\{\texttt{s}{\to}\texttt{h},\ \texttt{s}{\to}\texttt{h}{\to}\texttt{s2}\}$} \\
      \sh{$\{\texttt{\textbf{s}},\texttt{\textbf{t}}\}$} & \sh{$\{\texttt{s}{\to}\texttt{s1},\ \texttt{s}{\to}\texttt{h}\}$} \\
      \sh{$\{\texttt{\textbf{s}},\texttt{\textbf{t}}\}$} & \sh{$\{\texttt{s}{\to}\texttt{s1},\ \texttt{s}{\to}\texttt{h},\ \texttt{s}{\to}\texttt{h}{\to}\texttt{s2}\}$} \\
      \grayrow{$\{\texttt{s},\texttt{t}\}$} & \grayrow{$\{\texttt{t}{\to}\texttt{s2}\}$} \\
      \sh{$\{\texttt{\textbf{s}},\texttt{\textbf{t}}\}$} & \sh{$\{\texttt{s}{\to}\texttt{s1},\ \texttt{t}{\to}\texttt{s2}\}$} \\
      \grayrow{$\{\texttt{s},\texttt{t}\}$} & \grayrow{$\{\texttt{s}{\to}\texttt{h},\ \texttt{t}{\to}\texttt{s2}\}$} \\
      \grayrow{$\{\texttt{s},\texttt{t}\}$} & \grayrow{$\{\texttt{s}{\to}\texttt{h},\ \texttt{s}{\to}\texttt{h}{\to}\texttt{s2},\ \texttt{t}{\to}\texttt{s2}\}$} \\
      \sh{$\{\texttt{\textbf{s}},\texttt{\textbf{t}}\}$} & \sh{$\{\texttt{s}{\to}\texttt{s1},\ \texttt{s}{\to}\texttt{h},\ \texttt{t}{\to}\texttt{s2}\}$} \\
      \sh{$\{\texttt{\textbf{s}},\texttt{\textbf{t}}\}$} & \sh{$\{\texttt{s}{\to}\texttt{s1},\ \texttt{s}{\to}\texttt{h},\ \texttt{s}{\to}\texttt{h}{\to}\texttt{s2},\ \texttt{t}{\to}\texttt{s2}\}$} \\
      \bottomrule
    \end{tabular*}\par\smallskip
    {\footnotesize the emitted route set $\{\texttt{s}{\to}\texttt{s1}\}$ keeps \outnouns at the new input}
  \end{minipage}
  \caption{Two specs with several rows at one input.  (f) Serial orders of two
  concurrent transactions, under the equality order, after the output event
  $[A,B]$.  (g) and (h) Route sets traced to a service on \texttt{s1} and
  \texttt{s2}.  The input is the set of sources traced so far; inputs and
  outputs are both ordered by inclusion, and the rows shown are those whose
  input \inrel $\{\texttt{s}\}$.  Source \texttt{s} reaches \texttt{s1} directly and
  \texttt{s2} through hop \texttt{h}; a second source \texttt{t} reaches
  \texttt{s2} directly.  An output is a nonempty set of routes from the traced
  sources that contains each prefix of its routes.  Panel (g) shows the state
  after the output event $\{\texttt{s}{\to}\texttt{s1}\}$ at input
  $\{\texttt{s}\}$, and (h) the state after a further input event to
  $\{\texttt{s},\texttt{t}\}$.  Legend as in Figure~\ref{fig:election-frames};
  in (g) gray alone marks output subsumption, and in (h) gray also marks the
  rows at $\{\texttt{s}\}$ lost after the input event.}
  \label{fig:ambiguity-frames}
\end{figure*}

The two specs behave differently when an output is emitted.  The transaction
schedules $[A,B]$ and $[B,A]$ have no common extension under the equality
order, so emitting either is fateful, requiring a \commitment that removes the
other.  In the traceroute spec, by contrast, each
removed row has a remaining row at its input whose route set contains both the
removed route set and the emitted observation, so the removal is an
\emph{output subsumption}: it loses nothing and needs no \commitment.

Lifted Monotonicity captures this difference by combining two properties:
Relational Monotonicity, a condition across multiple inputs, and directedness,
a condition within the fiber at each input.

\label{sec:lifted-monotonicity}
\begin{definition}[Relational Monotonicity]
\label{def:relational-monotonicity}
An \spectype satisfies \emph{Relational Monotonicity}
when every output keeps some \outnoun as input \ingrows:
        \begin{equation}
  i\InputOrd j,\ o\in\fiber{i}
  \quad\Longrightarrow\quad
  \exists u\in\fiber{j}:\ o\OutputOrd u.
  \label{eq:relational-monotonicity}
        \end{equation}
\end{definition}

In the traceroute spec, tracing from a second source \texttt{t} moves the
input from $\{\texttt{s}\}$ to $\{\texttt{s},\texttt{t}\}$
(Figure~\ref{fig:ambiguity-frames}(h)).  Every route set permitted at
$\{\texttt{s}\}$ is still permitted at $\{\texttt{s},\texttt{t}\}$, so
Relational Monotonicity holds.

For singleton fibers this is functional monotonicity,
Equation~\ref{eq:functional-monotonicity}.  Relational Monotonicity compares fibers across inputs, so it does not
require the outputs within one fiber to have a common extension.

\begin{definition}[Directed fiber, directed fibers]
\label{def:fiber-directedness}\label{def:directed-outputs}
A fiber $F\subseteq O$ is \emph{directed} when every nonempty finite
$X\subseteq F$ has a common extension in $F$.  A spec has \emph{directed fibers} when
every fiber is directed.
\end{definition}

A singleton fiber is trivially directed.  Each traceroute fiber of
Figure~\ref{fig:ambiguity-frames}(g) is directed too: the union of finitely
many permitted route sets is itself a permitted route set, and it contains each
of them.  The two transaction orders in panel~(f), by contrast, have no common extension.

Directedness is more permissive than a join-semilattice.  It asks only that
each finite set of outputs in a fiber have some common extension within the
fiber, while a join-semilattice requires a \emph{least} one.  For example, the fiber
$\{\{a\},\{b\},\{a,b,c\},\{a,b,d\},\{a,b,c,d\}\}$ under inclusion is directed,
since $\{a,b,c,d\}$ \outrel every member, but $\{a\}$ and $\{b\}$ have no least common extension: $\{a,b,c\}$ and $\{a,b,d\}$ are both minimal.

\paragraph{Lifted Monotonicity.}
To state both requirements at once, we lift $\OutputOrd$ from outputs to sets
of outputs.  For nonempty
$A,B\subseteq O$, define
        \begin{equation}
  A\LiftedOrd B
  \;\Longleftrightarrow\;
  \text{every nonempty finite }X\subseteq A\text{ has a common extension in }B,
  \label{eq:joint-output-order}
        \end{equation}
the \emph{joint refinement lift} of $\OutputOrd$.  Lifted Monotonicity, defined
next, asks that $\fiber{i}\LiftedOrd\fiber{j}$ whenever $i\InputOrd j$.

\begin{definition}[Lifted Monotonicity]
\label{def:lifted-monotonicity}
An \spectype satisfies \emph{Lifted Monotonicity}, and is called
\emph{lifted monotone}, or just \emph{monotone}, when it satisfies Relational
Monotonicity and has directed fibers.  Equivalently, for every $i\InputOrd j$
and every nonempty finite $X\subseteq\fiber{i}$, some $u\in\fiber{j}$
\outrel every member of $X$; that is, $i\mapsto\fiber{i}$ is monotone under the joint refinement lift:
        \begin{equation}
  i\InputOrd j
  \quad\Longrightarrow\quad
  \fiber{i}\LiftedOrd\fiber{j}.
  \label{eq:lifted-monotonicity}
        \end{equation}
\end{definition}

With singleton fibers, Lifted Monotonicity is functional monotonicity
(Definition~\ref{def:functional-monotonicity}).  The two forms of the definition agree.  Taking $X$ to be a
single output gives Relational Monotonicity, and taking $j=i$ gives directed
fibers.  Conversely, take $i\InputOrd j$ and a finite nonempty
$X\subseteq\fiber{i}$.  Relational Monotonicity gives each member of $X$ an extension in $\fiber{j}$, and since $\fiber{j}$ is directed, those extensions have a common extension there.  Lifted Monotonicity therefore constrains a spec
even at $j=i$, where it asks that each fiber be directed.  The transaction
fiber of Figure~\ref{fig:ambiguity-frames}(f) is not, so that spec fails
Lifted Monotonicity already at $j=i$.

In powerdomain terms, Relational Monotonicity is monotonicity of
$i\mapsto\fiber{i}$ under the Hoare preorder~\cite{abramsky1994domain}, which
asks only that each output in $A$ have an extension in $B$.  The joint
refinement lift asks that every finite set of them have a common one, and the
two agree when $B$ is directed.  Unlike the Hoare preorder, the joint
refinement lift is not reflexive: $A\LiftedOrd A$ holds only when $A$ is
directed.

\paragraph{The row-consistency condition.}
Lifted Monotonicity compares the outputs at one input with those at its
extensions.  But runs may emit outputs at incomparable inputs, say $o_1$ at
$i_1$ and $o_2$ at $i_2$, and such runs can be reconciled only at a common
extension $j$ of both.  So we ask directly that the rows $(i_1,o_1)$ and
$(i_2,o_2)$ be confluent at $j$.
In the traceroute spec, one run may trace from \texttt{s} and emit
$\{\texttt{s}{\to}\texttt{s1}\}$, while another traces from a second source
\texttt{t} and emits $\{\texttt{t}{\to}\texttt{s2}\}$.  At $\{\texttt{s},\texttt{t}\}$, which \inrel both inputs, the route set $\{\texttt{s}{\to}\texttt{s1},\ \texttt{t}{\to}\texttt{s2}\}$ \outrel both outputs.

\begin{definition}[Row-consistency condition]
\label{def:common-expansion-condition}
An \spectype satisfies the \emph{row-consistency condition} when every finite set of rows of $P$ is consistent: the rows have a common \outnoun at every common \indesc of their inputs.
\end{definition}

For the empty set of rows the condition asks only that $P(j)$ be nonempty, which
follows from $P$ being total on $I$.  Pairs of rows suffice: if every two rows are consistent, so is every finite
set.  Two rows of a finite set are confluent at the common \indesc $j$, and
replacing them with one row at $j$ that \outrel both leaves a smaller set,
until one row remains
(Appendix~\ref{app:pairs-suffice}).

Although it compares rows at different inputs, the condition is equivalent to
Lifted Monotonicity:

\begin{proposition}[Row-consistency form of Lifted Monotonicity]
\label{prop:common-expansion-form}
An \spectype satisfies Lifted Monotonicity if and only
if it satisfies the row-consistency condition.
\end{proposition}

\begin{proof}[Proof sketch]
First assume Lifted Monotonicity, and take finitely many rows
$(i_1,o_1),\ldots,(i_k,o_k)$ whose inputs have a common extension $j$.  Lifted Monotonicity gives Relational Monotonicity, which carries each $o_r$ to some extension in the fiber at $j$, and that fiber is directed, so those extensions have a common extension there.  Now assume the row-consistency
condition, and take $i\InputOrd j$ and a finite $X\subseteq\fiber{i}$.  The
rows $(i,x)$ for $x\in X$ are then confluent at $j$, which is what Lifted Monotonicity asks for.  Full proof in
Appendix~\ref{app:proofs-sec3}.
\end{proof}

\subsection{Lifted output events}
\label{sec:lifted-output-events}

To connect Lifted Monotonicity to consistency, we return to runs, where
nondeterminism raises a new question: when emitting one output removes others,
is that removal a \commitment?  The output event is unchanged from the
functional case: at a state $(i,R)$ with $o\in R(i)$, $\mathsf{output}(o)$ emits
$o$ and keeps exactly the rows whose outputs \outrelbase $o$, ending at
$(i,R[\mathrm{out}\OutputGe o])$ (Equation~\ref{eq:output-action}).

Removing a row $(k,y)$ does not always lose the information in its output.  If some $u\in R(k)$ \outrel both $y$ and $o$, then removing $(k,y)$ loses nothing: the row $(k,u)$ is kept, and a later output $u$ reports everything $y$
would have.  This is output subsumption, and we put $(k,y)$ in $\Subsumed_o$.
Every other removed row goes in $\Excluded_o$.  With several outputs at an
input, the canonical \commitment $x_o(i,R)=(i,R\setminus\Excluded_o)$ removes
only the rows in $\Excluded_o$; the output event still removes the rows in
$\Subsumed_o$, but that removal is not a \commitment.  This is where the last
condition of Definition~\ref{def:commitment} matters: each row in
$\Subsumed_o$ has a remaining extension at its input.  In the functional case every removed row is excluded,
so $x_o$ is as before.

The two specs of Figure~\ref{fig:ambiguity-frames} fall on opposite sides of this split.  Emitting $[A,B]$ in panel~(f) puts the row for $[B,A]$ in $\Excluded_o$, ruling out that output at $\{A,B\}$.  The route sets removed in panel~(g), at both inputs, are in $\Subsumed_o$, so that output needs no \commitment.
Appendix~\ref{app:lifted-output-events} gives the formal definitions.

$o$ is a legal output at a state when $\Excluded_o$ is empty;
$\mathsf{output}(o)$ then applies no \commitment.  Equivalently, at every legal input $j$, the row $(i,o)$ and each remaining row at $j$ are confluent at $j$ (Appendix~\ref{app:lifted-output-events}).

\subsection{The Complete CALM theorem}
\label{sec:complete-calm-theorem}

Specification states, input events, runs, \commitment-free runs, and
consistency are defined as before.  The results of Section~\ref{sec:functional-calm-theorem}
use output events only through Equation~\ref{eq:output-action}, which lifted
output events still satisfy, so the rows-remaining lemma
(Lemma~\ref{lem:run-residual}), the Survivor test
(Proposition~\ref{prop:confluence}), and
Lemma~\ref{lem:confluence-legality}, with legality as in
Section~\ref{sec:lifted-output-events}, apply unchanged.  The question is
therefore when no output event of a run applies a \commitment, and Lifted
Monotonicity answers it.

\begin{proposition}[No-\commitment test]
\label{prop:exposure-boundary}
An \spectype satisfies Lifted Monotonicity if and only
if every output at every state reachable without \commitment is legal.
\end{proposition}

\begin{proof}[Proof sketch]
First assume Lifted Monotonicity, in its row-consistency form
(Proposition~\ref{prop:common-expansion-form}).  At a state reachable without
\commitment, the row of an output $o$ and each remaining row at a legal input $j$ are confluent at $j$, through a row the state keeps
(Lemma~\ref{lem:run-residual}), so $o$ is legal.  Now assume every such output
is legal, and take rows $(i_r,o_r)$ with a common extension $j$ of their inputs.  Each $o_r$ is legal at $\sigma_P(i_r)$, so, taking the rows one at a time, they are confluent at $j$, and Lifted Monotonicity follows.  Full proof in
Appendix~\ref{app:proofs-sec3}.
\end{proof}

\begin{theorem}[Complete CALM]
\label{thm:complete-calm}
The runs of an \spectype are consistent if and only if the spec satisfies
Lifted Monotonicity.
\end{theorem}

\begin{proof}[Proof sketch]
By Lemma~\ref{lem:confluence-legality}, the runs are consistent exactly when
no output event of a run applies a \commitment.  First assume they are
consistent.  Every output at every state reachable without \commitment is then
legal, since emitting one that is not would apply a \commitment, so the spec
satisfies Lifted Monotonicity (Proposition~\ref{prop:exposure-boundary}).  Now
assume Lifted Monotonicity.  Every output at every state reachable without
\commitment is then legal (Proposition~\ref{prop:exposure-boundary}), so by
induction on its length every run is \commitment-free, and consistency
follows.  Full proof in Appendix~\ref{app:proofs-sec3}.
\end{proof}

Complete CALM also links consistency of runs to consistency of rows.  The definition of
consistency asks whether runs can be continued to a common state; by the
theorem and Proposition~\ref{prop:common-expansion-form}, that holds exactly
when every finite set of the spec's rows is row-consistent.  So consistency can
be checked on the table itself, and pairs of rows suffice.

The condition that characterizes consistency of runs is a property of the spec
alone: Lifted Monotonicity, which for a function is ordinary monotonicity.
Both theorems also hold for distributed histories
(Section~\ref{sec:distributed-consistency}).

For functional specs, Lifted Monotonicity and the row-consistency condition
are both functional monotonicity, so Complete CALM and
Proposition~\ref{prop:exposure-boundary} reduce to Functional CALM and
Proposition~\ref{prop:functional-exposure-boundary}.

\section{Distributed Execution}
\label{sec:distributed-execution}

This section models distributed execution across \emph{participants}:
computational units with private state communicating via explicit messages,
akin to processes, nodes, replicas, or actors.  We give Functional CALM and Complete CALM a distributed reading,
including a convergence guarantee for functional specs
(Corollary~\ref{cor:convergence}), then define systems, relate their runs to
canonical runs, and show when the runs of systems are consistent; Complete CAP
follows.


\subsection{Distributed histories}
\label{sec:histories}

Until now, runs have been independent: each state changes only through that
run's input and output events.  We now place one run at each participant in
a finite set $\mathcal N$ of any size.  Participant $p$ starts in the canonical state
$\sigma_P(i_p)$ and still changes its state only through its own specification
events.  The new ingredient is communication.  Besides learning input directly
from the environment, a participant can learn cumulative input from another
participant.

Participants exchange messages by $\mathsf{send}(m)$ and
$\mathsf{receive}(m)$.  A \emph{message} $m=(j,d)$ carries a cumulative input $j$, or $\varnothing$ if it
carries none, and application-specific data $d$, such as timestamps, version
vectors, or state identifiers (described below).  Receiving a message does not
change the specification state.  A participant at $(i,R)$ that receives $(j,d)$ with
$j\neq\varnothing$ may incorporate $j$ by its own $\mathsf{input}(k)$ for a
legal common \indesc $k$ of $i$ and $j$, when one exists.  The system chooses $k$; the data $d$ can influence that choice but never
changes the state directly.

Send and receive events connect the participants' locally ordered runs.  Adding
an edge from each send to its matching receive gives the standard model of a
distributed history; the transitive closure of these edges and the local orders
is happens-before.

\begin{definition}[Distributed history, history event]
\label{def:distributed-history}
Fix an \spectype, participants $\mathcal N$, and an initial input $i_p$
for each.  A \emph{distributed history} $H$ is a finite set of
\emph{history events}, each an occurrence at one participant of a
specification event (Definition~\ref{def:run}), $\mathsf{send}(m)$, or
$\mathsf{receive}(m)$, with:
\begin{enumerate}
  \item a total local order on each participant's history events;
  \item for each receive, a matching send of the same message at another
        participant, each send matching at most one receive;
  \item for each participant $p$, its specification events, in local order,
        form a defined transition sequence from $\sigma_P(i_p)$: each
        is legal at the state where it occurs; and
    \item no history event happens before itself.
\end{enumerate}
\end{definition}

\begin{definition}[Happens-before]
\label{def:happens-before}
\emph{Happens-before} is the transitive closure of the local orders of a
distributed history and of each send preceding its matching receive.
\end{definition}

\begin{definition}[Consistent cut]
\label{def:consistent-cut}
A \emph{consistent cut} $C$ is a subset of history events closed under happens-before~\cite{lamport1978time,chandy1985distributed}.  The ``consistent cut'' term is standard and unrelated to our notion of consistency.
\end{definition}

A consistent cut contains a prefix of each participant's history events, so by
item~3
it determines one run per participant.

\begin{definition}[Participant run]
\label{def:participant-run}
At a consistent cut $C$, the specification events of $C$ at participant $p$,
in local order, form its \emph{participant run}, a run
(Definition~\ref{def:run}) whose final state we write $\sigma_p(C)$.
\end{definition}

\begin{definition}[Endpoint]
\label{def:endpoint}
The \emph{endpoint} of a consistent cut $C$ is
$\operatorname{end}(H,C)=(\sigma_p(C))_{p\in\mathcal N}$.
\end{definition}

\begin{figure}[t]
  \centering
  \begin{tikzpicture}[x=.82cm,y=.58cm,>=stealth,
                      every node/.style={font=\footnotesize}]
    \node[anchor=west,font=\small] at (-1,7.2) {(a) Direct delivery $H_d$};
    \node[anchor=east] at (.35,5.5) {$p$};
    \node[anchor=east] at (.35,3.9) {$q$};
    \node[anchor=west] at (.6,5.5) {$\sigma_a$};
    \node[anchor=west] at (.6,3.9) {$\sigma_a$};
    \draw[->] (1.4,5.5) -- (6.4,5.5);
    \draw[->] (1.4,3.9) -- (6.4,3.9);
    \fill (4.15,5.5) circle (1.7pt);
    \fill (4.15,3.9) circle (1.7pt);
    \node[below,fill=white,inner sep=.5pt] at (4.15,5.5)
      {$\mathsf{input}(ab\novote)$};
    \node[above,fill=white,inner sep=.5pt] at (4.15,3.9)
      {$\mathsf{input}(ab\novote)$};
    \node[anchor=east] at (1.8,6.2) {$\mathsf{env}$};
    \node[anchor=east] at (1.8,3.2) {$\mathsf{env}$};
      \draw[->] (1.95,6.14) -- node[midway,above,sloped,fill=white,inner sep=.5pt]
        {$ab\novote$} (4.03,5.65);
      \draw[->] (1.95,3.26) -- node[midway,below,sloped,fill=white,inner sep=.5pt]
        {$ab\novote$} (4.03,3.75);
    \draw[dashed] (6.45,3.55) -- (6.45,5.85);
    \node[above] at (6.45,5.85) {$C_d$};
    \node[anchor=west] at (6.6,5.5) {$\sigma_1$};
    \node[anchor=west] at (6.6,3.9) {$\sigma_1$};

    \node[anchor=west,font=\small] at (-1,2.3) {(b) Forwarded delivery $H_f$};
    \node[anchor=east] at (.35,.65) {$p$};
    \node[anchor=east] at (.35,-.95) {$q$};
    \node[anchor=west] at (.6,.65) {$\sigma_a$};
    \node[anchor=west] at (.6,-.95) {$\sigma_a$};
    \draw[->] (1.4,.65) -- (8.15,.65);
    \draw[->] (1.4,-.95) -- (8.15,-.95);
    \node[anchor=east] at (1.75,1.35) {$\mathsf{env}$};
      \draw[->] (1.9,1.29) -- node[midway,above,sloped,fill=white,inner sep=.5pt]
        {$ab\novote$} (3.35,.8);
    \fill (3.5,.65) circle (1.7pt);
    \node[below,fill=white,inner sep=.5pt] at (3.5,.65)
      {$\mathsf{input}(ab\novote)$};
    \fill (5.25,.65) circle (1.7pt);
    \node[above,fill=white,inner sep=.5pt] at (5.25,.65)
      {$\mathsf{send}(m)$};
    \fill (6.0,-.95) circle (1.7pt);
    \node[below,fill=white,inner sep=.5pt] at (6.0,-.95)
      {$\mathsf{receive}(m)$};
    \draw[->] (5.25,.53) -- node[midway,right,fill=white,inner sep=.5pt]
      {$m=(ab\novote,d_m)$} (5.95,-.8);
    \fill (7.55,-.95) circle (1.7pt);
    \node[below,fill=white,inner sep=.5pt] at (7.55,-.95)
      {$\mathsf{input}(ab\novote)$};
    \draw[dashed] (8.2,-1.3) -- (8.2,1.0);
    \node[above] at (8.2,1.0) {$C_f$};
    \node[anchor=west] at (8.35,.65) {$\sigma_1$};
    \node[anchor=west] at (8.35,-.95) {$\sigma_1$};
  \end{tikzpicture}
  \caption{Two histories from $\sigma_a=\sigma_P(a\novote\novote)=(a\novote\novote,R_0)$ to
  $\sigma_1=(ab\novote,R_1)$, with $R_0$ and $R_1$ the tables of
  Figure~\ref{fig:election-frames}(a) and~(b).  In (a) each participant records
  $\mathsf{input}(ab\novote)$ directly; in (b) $q$ records it after receiving $m$.  Dashed lines are consistent cuts.}
  \label{fig:participant-runs}
\end{figure}

In Figure~\ref{fig:participant-runs}, the input $ab\novote$ reaches $q$
directly (a) or through a message from $p$ (b).  The histories differ, but at either
cut each participant run is $\mathsf{input}(ab\novote)$, so both endpoints are
$(\sigma_1,\sigma_1)$.

\subsection{Distributed consistency}
\label{sec:distributed-consistency}

A \emph{continuation} of a history at a consistent cut $C$ is a history $H'$
that extends the history events of $C$ with further history events after
them.  All of $H'$ is a consistent cut of $H'$, its \emph{final cut}.  Let $j$ be a common \indesc of the participants' current inputs at a consistent cut $C$.  The cut $C$ is \emph{confluent at} $j$ when some continuation at $C$ has a final cut with endpoint $\bigl((j,R)\bigr)_{p\in\mathcal N}$ for some $R$, so that every participant ends in the same state $(j,R)$.  History events at different participants are independent, so this holds exactly when the participant runs at $C$ are
confluent at $j$ (Definition~\ref{def:consistency}).

\begin{definition}[Consistency of distributed histories]
\label{def:distributed-consistency}
The distributed histories of an \spectype are \emph{consistent} when, for every
finite participant set, every consistent cut of every history is confluent
at every common \indesc of the participants' current inputs.
\end{definition}

\begin{proposition}[Distributed consistency]
\label{prop:distributed-consistency}
The runs of an \spectype are consistent if and only if its distributed
histories are consistent.
\end{proposition}

\begin{proof}[Proof sketch]
First assume the runs are consistent.  The participant runs at a cut are a
finite family of runs, so every cut is confluent at every common \indesc, and the histories are
consistent.  Now assume the histories are consistent.  Any finite family of
runs is the set of participant runs at the final cut of a history whose
participants exchange no messages, so the runs are consistent.  Full proof in
Appendix~\ref{app:proofs-sec4}.
\end{proof}

By Complete CALM, which reduces to Functional CALM with one output per input,
the distributed histories of a spec are consistent exactly when the spec
satisfies Lifted Monotonicity.  For specs with Lifted Monotonicity, the Survivor test
(Proposition~\ref{prop:confluence}) gives a row $(j,u)$ that every participant's
state still contains, where $j$ is a common \indesc of all participants' current
inputs; by Lemma~\ref{lem:run-residual}, $u$ \outrel everything they have
emitted.  Each participant can then perform $\mathsf{input}(j)$ and
$\mathsf{output}(u)$ and reach the same state without \commitment
(Proposition~\ref{prop:exposure-boundary}).  The spec guarantees only that such
a $u$ exists; choosing it, communicating it, and scheduling the events are left
to the system.

For functional specs,
Corollary~\ref{cor:state-by-input} also gives convergence:

\begin{corollary}[Convergence]
\label{cor:convergence}
In any history of a monotone functional specification, every
participant whose current input at a consistent cut is $j$ is in state
$\sigma_P(j)$.  Participants with the same current input therefore agree,
regardless of message paths, arrival order, or the outputs they emitted.
\end{corollary}

\begin{proof}
Each participant run is a run, so Corollary~\ref{cor:state-by-input} applies.
\end{proof}

With several outputs per input, convergence can fail even when the spec
satisfies Lifted Monotonicity.  In Figure~\ref{fig:ambiguity-frames}(g), a run at input $\{\texttt{s}\}$ that emitted
$\{\texttt{s}{\to}\texttt{s1}\}$ has subsumed the route sets that leave it out, while a run at the same input that
emitted nothing keeps them.  Consistency then guarantees only that such participants can be reconciled, not
that they agree.

\paragraph{Message data.}
The data $d$ of a message is any application state other than input that a
participant chooses to send, such as timestamps, version vectors, gossip, or votes.  Data
can influence which history events occur next at the receiver, but can
never directly change the receiver's specification state, and nothing in the
definitions above depends on it.
Participants may also share emitted outputs in $d$; this does not change
at which inputs a cut is confluent (Appendix~\ref{app:exchanging-outputs}).

\subsection{Systems and canonical runs}
\label{sec:systems}

A spec says which behaviors are permitted, not how a system chooses among them.
We now model those choices, and show that a system's inputs and outputs always
match a canonical run's.

For example, serializability can be enforced by any of a number of concurrency
control protocols (Figure~\ref{fig:ambiguity-frames}(f), formalized in
Appendix~\ref{app:cc-bases}).  Two-phase locking (2PL) orders transactions
pairwise through lock requests: when $B$ waits for a lock $A$ holds, the 2PL \commitment $\mathsf{order}(A,B)$ fixes $A$ before $B$ and nothing else.  Multiversion
timestamp ordering (MVTO)~\cite{bernstein1981concurrency} orders them totally
through timestamp assignment: the MVTO \commitment $\mathsf{stamp}(B)$ fixes $B$ before any
transaction not yet stamped.  The
two narrow the spec differently (Figure~\ref{fig:cc-commitments}).

The registrar of Section~\ref{sec:election-story} is a similar example.  When
it emits $\bot$, it also closes the polls (Figure~\ref{fig:election-frames}(d)),
although emitting $\bot$ does not require that \commitment.  Closing the polls is the registrar's
design decision, as fixing orderings is 2PL's.

Sections~\ref{sec:functional-calm}--\ref{sec:complete-calm} used only the
canonical \commitment $x_o$ that an output requires.  Systems may commit
earlier or remove more alternatives, as these examples do, by applying other
\commitments (Definition~\ref{def:commitment}).  The ordering \commitment
$\mathsf{order}(A,B)$ is one: it removes the serial orders with $B$ before $A$,
and no surviving order with $A$ before $B$ \outrel any of them.

\begin{definition}[\Commitment event]
\label{def:commitment-event}
A \emph{\commitment event} $\mathsf{apply}(\gamma)$ applies a nonempty
sequence $\gamma=e_1\cdots e_n$ of \commitments, each defined at the state
its predecessors leave, and emits nothing.
\end{definition}

The sequence as a whole is itself a \commitment.  Each row it removes was
removed by some step, and that step left no row at the same input that \outrel
the removed one.  Every row that remains at the end also remained after that
step, so the sequence as a whole removes no row that a remaining row \outrel.
Results stated for a single \commitment therefore apply to \commitment events
as well.

A system applies its own \commitments only in \commitment events.  A \commitment event may precede an output event, which then
has less to remove.  The environment observes nothing from a \commitment event.
\Commitment events are a third kind of specification event, so a distributed
history may contain them; Definitions~\ref{def:distributed-history} and~\ref{def:participant-run} apply
unchanged.

\begin{definition}[\Commitment basis, run under a basis]
\label{def:commitment-basis}\label{def:run-under-basis}
A \emph{\commitment basis} is a set $\mathcal B$ of \commitments.  A \emph{run
under $\mathcal B$} is a finite transition sequence from a canonical state
$\sigma_P(i)$ of input, output, and \commitment events in which every
\commitment event applies only \commitments in $\mathcal B$, and every output
event $\mathsf{output}(o)$ whose output $o$ is not already legal applies its
canonical \commitment $x_o\in\mathcal B$.
\end{definition}

A canonical run is then a run under the basis of all canonical \commitments that
contains no \commitment event.

\paragraph{Two bases on three transactions.}
We return to the concurrency protocols from the start of this subsection.
Each protocol's basis contains its ordering \commitments, $\mathsf{order}$ for
$\mathcal B_{\mathrm{2PL}}$ and $\mathsf{stamp}$ for $\mathcal B_{\mathrm{MVTO}}$,
together with \commitments for commits, aborts, and reads
(Appendix~\ref{app:cc-bases}).  Figure~\ref{fig:cc-commitments} shows each at
work on three transactions, a larger version of the ordering fiber of
Figure~\ref{fig:ambiguity-frames}(f).  2PL applies $\mathsf{order}(A,B)$ when
$B$ waits for a lock that $A$ holds, which leaves $C$'s position open
(panel~(h)).  MVTO applies $\mathsf{stamp}(B)$ when $B$ receives the first
timestamp, which places $B$ first (panel~(i)).  Both are \commitments that
leave some rows at the input.  Each is applied in a \commitment event
$\mathsf{apply}(\gamma)$, and neither is the canonical \commitment $x_o$ of
an output.

\begin{figure}[t]
  \centering
  \small
  \setlength{\tabcolsep}{5pt}
  \begin{tabular}[t]{@{}l l@{}}
    \multicolumn{2}{@{}l}{(h) 2PL: $A$ locks ahead of $B$}\\
    \toprule
    $\mathrm{in}$ & $\mathrm{out}$ \\ \midrule
    $\{A,B,C\}$ & \sh{$[A,B,C]$} \\
    $\{A,B,C\}$ & \sh{$[A,C,B]$} \\
    $\{A,B,C\}$ & \sh{$[C,A,B]$} \\
    \grayrow{$\{A,B,C\}$} & \grayrow{$[B,A,C]$} \\ \strikerow
    \grayrow{$\{A,B,C\}$} & \grayrow{$[B,C,A]$} \\ \strikerow
    \grayrow{$\{A,B,C\}$} & \grayrow{$[C,B,A]$} \\ \strikerow
    \bottomrule
  \end{tabular}
  \hfill
  \begin{tabular}[t]{@{}l l@{}}
    \multicolumn{2}{@{}l}{(i) MVTO: $B$ receives the first timestamp}\\
    \toprule
    $\mathrm{in}$ & $\mathrm{out}$ \\ \midrule
    $\{A,B,C\}$ & \sh{$[B,A,C]$} \\
    $\{A,B,C\}$ & \sh{$[B,C,A]$} \\
    \grayrow{$\{A,B,C\}$} & \grayrow{$[A,B,C]$} \\ \strikerow
    \grayrow{$\{A,B,C\}$} & \grayrow{$[A,C,B]$} \\ \strikerow
    \grayrow{$\{A,B,C\}$} & \grayrow{$[C,A,B]$} \\ \strikerow
    \grayrow{$\{A,B,C\}$} & \grayrow{$[C,B,A]$} \\ \strikerow
    \bottomrule
  \end{tabular}
  \caption{Serial orders of three transactions after one \commitment:
  (h) 2PL orders $A$ before $B$; (i) MVTO places $B$ first.  Only complete rows without aborts or reads
  are shown.}
  \label{fig:cc-commitments}
\end{figure}

\begin{definition}[Local history]
\label{def:local-history}
A \emph{local history} of participant $p$ is a finite sequence of $p$'s
history events, including the messages it has received.
\end{definition}

\begin{definition}[Strategy]
\label{def:strategy}
A \emph{strategy} for a participant $p$ is a nondeterministic transition
function: it maps each local history $\lambda$ of $p$ to the set
$\theta_p(\lambda)$ of output, \commitment, and send events that $p$ may
perform next.  Input and receive events are not chosen by the strategy: inputs
come from the environment and receives from matching sends.
\end{definition}

\begin{definition}[System]
\label{def:system}
A \emph{system} consists of a \commitment basis $\mathcal B$ and a strategy
$\theta_p$ for each participant $p$, whose output and \commitment events apply
only \commitments from $\mathcal B$.
\end{definition}

\begin{definition}[History of a system]
\label{def:produced-history}
A distributed history is a \emph{history of} a system, or is \emph{produced} by
it, when its input events come from the environment, its receives match sends,
and its output, \commitment, and send events are chosen by the participants'
strategies.
\end{definition}

Coordination mechanisms are part of the system, not of the spec: its strategies
describe how participants use them, and its basis records the \commitments they
can enforce.  A system's strategies select the events of an actual execution;
they need not realize every run the spec permits.  Participants may behave
differently; a replicated (SPMD) system is the special case in which they share
one strategy (Appendix~\ref{app:cap-response}).

Our model places no bound on the set of participants, and a strategy has no way
to enumerate them.  In general there is no list of all participants, which
Ameloot et al.\ call the $\mathsf{All}$ relation~\cite{ameloot2013relational}.
Such a list lets a participant wait for every other participant, for example to
collect a vote from each at some step of a protocol, and know when it has them
all.  Ameloot et al.\ show that this enlarges what a system can compute.  A
spec can include this knowledge as a modeling decision, by fixing the list of
all participants as part of the input in every row.  The list is then explicit
common knowledge, available to each participant's strategy from its initial
input.

\paragraph{Observable equivalence.}
Locks, timestamps, quorums, and consensus differ in how and when they decide.
The next lemma relates runs under any basis to canonical runs.

\begin{lemma}[Observable equivalence]
\label{lem:canonical-equivalent}
Let $\rho$ be a run of a system for a spec $\Spec$, under any \commitment
basis.  Removing the \commitment events from $\rho$ gives a canonical run of
$\Spec$ (Definition~\ref{def:run}) that is \emph{observably equivalent} to
$\rho$: the two accept the same inputs and emit the same outputs, and after each
event the canonical run keeps every row that $\rho$ keeps.
\end{lemma}

\begin{proof}[Proof sketch]
\Commitment events only remove rows, so dropping them can only keep more.  An
input or output event applied to a larger table keeps a larger table, and the
canonical run can perform it because the rows of $\rho$ are among its
own.  Full proof in
Appendix~\ref{app:bases}.
\end{proof}

To illustrate the lemma, consider 2PL on three transactions $A$, $B$, $C$ that
have each issued their operations and requested commit, where $B$ waits for a
lock that $A$ holds.  When $B$ begins to wait, 2PL applies
$\mathsf{order}(A,B)$ in a \commitment event, before any output, and removes
every row that places $B$ before $A$.  The canonical run drops that event and
keeps those rows for now.  Suppose 2PL fixes the orders $(A,C)$ and $(C,B)$ in
later \commitment events and emits $[A,C,B]$.  The canonical run emits
$[A,C,B]$ directly, and its output event applies the canonical \commitment of
$[A,C,B]$, which removes every row whose output does not \outrelbase
$[A,C,B]$.  Those include every row 2PL removed earlier, since $[A,C,B]$
contains each order 2PL fixed.  After the output both runs keep the same rows.

So the canonical \commitment of an output makes, all at once and at output
time, the decisions a system may have spread over earlier \commitment events.
A system can decide earlier than its outputs require, but the \commitments
those outputs require are the same.

The registrar shows that the lemma has no converse.  After closing the polls
at $ab\novote$ (Figure~\ref{fig:election-frames}(d)), it cannot accept the late
ballot $aba$.  But a canonical run can emit $\bot$, accept $aba$, and emit $a$.
In general, systems may constrain their runs arbitrarily
(Appendix~\ref{app:bases}), so their remainder tables may lack rows that their
canonical runs keep.  Hence a spec's runs can be consistent while a system's
runs are not.

The CALM and CAP criteria are about specs and
their canonical runs.  A system with no \commitment events inherits their
confluence directly, since its runs are canonical runs.  A system with
\commitments of its own keeps at most the rows of its canonical runs, and its
confluence depends on which rows those \commitments remove.

\subsection{Confluence of systems}
\label{sec:commitment-conflict}

The registrar shows how a system's own \commitment can block a confluence the
spec allows: once it closes the polls at $ab\novote$, it can never reach a common
state with a canonical run that accepts $aba$ and emits $a$.  To see exactly
which rows a system gives up, we compare a run with its canonical run
(Lemma~\ref{lem:canonical-equivalent}).  An input or output event tests each
row the same way in both runs, so only the system's \commitment events remove
more.  For a run $\rho$ under a basis whose canonical run ends at
$(i,\widehat R)$, let $\Delta(\rho)$ be the rows of $\widehat R$ that a
\commitment event of $\rho$ removed.

\begin{proposition}[System confluence]
\label{prop:system-confluence}
Let $\rho_1,\ldots,\rho_n$ be runs under any \commitment bases, where $\rho_r$
ends at $(i_r,R_r)$ and its canonical run ends at $(i_r,\widehat R_r)$.
\begin{enumerate}
  \item $R_r=\widehat R_r\setminus\Delta(\rho_r)$.
  \item Let $k$ be a common extension of $i_1,\ldots,i_n$.  The runs are
        confluent at $k$, by continuations whose output and \commitment events
        may apply any \commitments, if and only if some row at $k$ lies in
        every $\widehat R_r$ and in no $\Delta(\rho_r)$.
\end{enumerate}
\end{proposition}

\begin{proof}[Proof sketch]
(1) A row of $\widehat R_r$ passes every input and output event of $\rho_r$,
so only a \commitment event can remove it.  (2) By part~(1), the condition says
some row at $k$ lies in every $R_r$.  First assume the runs are confluent at
$k$.  A common state at $k$ has a row at $k$, and since rows are never added,
that row lies in every $R_r$.  Now assume $(k,u)$ lies in every $R_r$.  Then
$\mathsf{input}(k)$, $\mathsf{output}(u)$, and a \commitment event that removes
every row except $(k,u)$ and the rows at $k$ with outputs equivalent to $u$
bring each run to the same state, so the runs are confluent at $k$.  Full proof in Appendix~\ref{app:proofs-sec4}.
\end{proof}

The same test applies to a consistent cut of a history, with one run per
participant.  A continuation that adds events at each participant separately,
with no messages, brings the cut to a common state exactly when those runs
reach it (Section~\ref{sec:distributed-consistency}).

For the registrar, closing the polls at $ab\novote$ puts in $\Delta$ every row
with a late ballot, such as $(aba,a)$ and $(abb,b)$.  Its table then has no row
at $aba$, so no family of runs that includes it is confluent at $aba$.

Proposition~\ref{prop:system-confluence} asks whether a common state can be
reached by any continuations, not whether the system's own strategy would reach
it.  The two agree when, from each final state, the strategy can accept $k$ and
then emit any output whose row at $k$ remains in its table.

A strategy that refuses such an output breaks this agreement.  Suppose the
registrar, instead of applying the poll-closing \commitment when it emits
$\bot$ at $ab\novote$, never announces a winner after $ab\novote$, although the
spec allows it.  Its table still has the row $(aba,a)$ and its $\Delta$ is
empty, so the test finds a row it shares with a canonical run that accepts
$aba$ and emits $a$.  But the registrar never emits $a$, so its own
continuations cannot reach that state.

Proposition~\ref{prop:system-confluence} has a simple special case: two runs,
one of which removes a row by a \commitment while the other emits that row's
output.

\begin{proposition}[Commitment conflict]
\label{prop:commitment-conflict}
Let a run under some \commitment basis from $\sigma_P(i)$ apply a \commitment
that removes the row $(k,x)$, and let a run under any \commitment basis from
$\sigma_P(i)$ issue $\mathsf{input}(k)$ followed by $\mathsf{output}(x)$.  These two runs
\emph{conflict}: they are not confluent at $k$.
\end{proposition}

\begin{proof}[Proof sketch]
After $\mathsf{output}(x)$, every row at $k$ in the second run \outrel $x$.
The \commitment removed $(k,x)$ and left no row at $k$ in the first run's table
that \outrel $x$ (Definition~\ref{def:commitment}), and later events only
remove rows, so no row at $k$ survives in both runs.  By
Proposition~\ref{prop:system-confluence} they are not confluent at $k$.  Full
proof in Appendix~\ref{app:proofs-sec4}.
\end{proof}

Within one participant's run, the \commitment already prevents the conflict:
once $(k,x)$ is removed, that run can no longer emit $x$ at $k$.  The conflict
matters when the other run belongs to a different participant of the same
history.  A system then has two ways to avoid it: its participants coordinate
before one of them makes the choice, or only one participant is ever allowed to
make it.

\paragraph{Deployments.}
A \emph{deployment} records the second option: it fixes, for each participant,
which outputs it may emit at which inputs.  If every participant may emit every
output, two conflicting runs can occur at participants that have not heard
from each other.  If each choice has a single owner, they cannot, and the owner
can make the choice without coordinating with anyone
(Appendix~\ref{app:deployment}).

\paragraph{Locks in a shared-nothing database.}
Consider a database partitioned by key, with each key owned and locked by one
participant~\cite{stonebraker1986sharednothing}.  Choosing an order of two
transactions on one key requires a \commitment, but only the key's owner can
make that choice, so the two possible orders can never be chosen by different
participants.  A transaction spanning two keys is different: each owner orders
its own part, and the two owners can order the same pair oppositely.  They must
agree on one order.

Two-phase commit adds no further \commitment.  Each data owner's response to a \texttt{prepare} message
from the coordinator reports the
\commitment it has already made.  The coordinator then makes the monotone
decision whether all owners prepared; it detects whether their \commitments
agree and aborts when they do not.  It does not make the owners agree.  If the
coordinator instead aborts on timeout before the last vote arrives, that abort
applies another \commitment, like closing the polls.  One coordinator can make
that choice locally.  Replicated coordinators need distributed agreement to
make it, as in Paxos Commit~\cite{gray2006consensus}.

\subsection{The Complete CAP theorem}
\label{sec:complete-cap}

Brewer conjectured that a distributed system cannot be consistent, available,
and partition tolerant at once~\cite{brewer2000towards}; Gilbert and Lynch
proved it as an impossibility result for a linearizable register~\cite{gilbert2002cap}.
Complete CAP addresses the conjecture more broadly:
which specifications are subject to the CAP
conjecture, and which are not?

To pose the question, we separate the outputs that participants emit from
the conclusions an observer, such as a client, may draw by combining them.  The spec permits an observer's combined conclusions, but no participant emits them, and
availability does not need them.

\begin{definition}[Explicit possibility relation]
\label{def:explicit-relation}
An \emph{explicit possibility relation}, or \emph{explicit relation}, is a set
$E\subseteq P$ of rows whose outputs participants may actually emit.  Write $E(i)$
for its outputs at input $i$.
\end{definition}

Like Gilbert and Lynch, we examine a replicated register: a key-value store
with a single key, which is enough to reason about one object.  It starts at
$v_0$.  Participant $p$ receives write $w(v)$ and emits its acknowledgment;
concurrently, participant $q$ receives the read $r$ and emits $v_0$.  These two
outputs are in $E$.  A client that collects both responses sees the pair,
which is in $P\setminus E$: a linearizable register spec permits it, since a
linearization may place the read before the write, and the pair \outrel both local outputs.  Had the write finished before the read began, the
spec would not permit the pair.

Both availability and consistency are judged on histories in which a partition
separates the participants.

\begin{definition}[Partitioned history]
\label{def:partitioned-history}
A \emph{partitioned history} divides the participants into two or more
sides and contains no receive whose matching send lies on another
side.  Same-side messages are allowed, and participants still receive environment input
and emit outputs.
\end{definition}

We consider a system whose participants each follow a strategy that depends
only on their local history, and an environment that may send any legal input
to any participant.  Availability asks that, during a partition, each
participant can still emit every explicit output its state permits.

\begin{definition}[Same-side continuation]
\label{def:local-continuation}
Consider a partitioned history that the system produces, in which participant
$p$ ends in state $(i,R)$.  A \emph{same-side continuation} for $p$ extends that
history with events the system can produce on $p$'s side alone, with no new
input and no message from another side.  Participant $p$ can emit $o$ when some same-side continuation for $p$ emits
$o$.
\end{definition}

\begin{definition}[Availability]
\label{def:partition-notions}
The system has \emph{Availability for $E$} when, in every partitioned history
it produces, each participant $p$ can emit every explicit output its state
still permits at its current input, that is, every output in $R(i)\cap E(i)$,
where $(i,R)$ is $p$'s final state.
\end{definition}

Partition Consistency asks that, whatever participants emit on separate sides,
their states can still be brought to a common state.

\begin{definition}[Partition Consistency]
\label{def:partition-consistency}
The system has \emph{Partition Consistency} when the final cut of every
partitioned history it produces is confluent at every common \indesc of the
participants' current inputs.
\end{definition}

The \emph{canonical system for $E$} is the system whose participants follow
canonical runs and may emit any output of $E$ their state permits.

\begin{definition}[Escaping CAP]
\label{def:escapes-cap}
A spec \emph{escapes CAP} for $E$ when its canonical system for $E$ has
Availability for $E$ and Partition Consistency.
\end{definition}

Whether the canonical system's answers stay consistent across sides depends on
one condition.

\begin{definition}[Lifted Monotonicity for $E$]
\label{def:lm-for-A}
$\Spec$ satisfies \emph{Lifted Monotonicity for $E$} when every finite set of
rows of $E$ with a common \indesc $j$ has a common \outnoun in $\fiber{j}$.
\end{definition}

\begin{theorem}[Complete CAP]
\label{thm:complete-cap}
Let $E\subseteq P$ be an explicit relation for a spec $\Spec$.  Then $\Spec$
escapes CAP for $E$ if and only if it satisfies Lifted Monotonicity for $E$.
\end{theorem}

\begin{proof}[Proof sketch]
First assume the spec escapes CAP for $E$.  Take finitely many rows
$(i_1,o_1),\ldots,(i_k,o_k)$ of $E$ and an input $j$ that \inrel every $i_r$.  In a partitioned
history, put each participant $r$ alone on a side with input $i_r$, so
Availability lets it emit $o_r$.  Partition Consistency makes the final cut
confluent at $j$, so some row $(j,u)$ survives in every participant's state
(Proposition~\ref{prop:confluence}).  By Lemma~\ref{lem:run-residual}, $u$ \outrel every $o_r$, so the rows are confluent at $j$, which is Lifted Monotonicity for $E$.

Now assume Lifted Monotonicity for $E$.  The canonical system's participants
never wait: each can emit any output of $E$ that its state allows, so the
system has Availability for $E$.  Let $j$ be an input that \inrel every
participant's current input.  The rows emitted on all sides are rows of $E$, so they are confluent at $j$, through some row $(j,u)$.  The canonical system's participants follow canonical runs, so by
Lemma~\ref{lem:run-residual} each keeps every row whose input \inrel its own
and whose output \outrel everything it emitted, including $(j,u)$.  All participants can therefore
continue to one state at $j$, which is Partition Consistency.  Full proof in
Appendix~\ref{app:proofs-sec4}.
\end{proof}

For $E=P$, a spec escapes CAP exactly when Lifted Monotonicity holds; the canonical system then needs no \commitment. \Commitments and availability can still coexist when each \commitment is made by a single participant, as we see next.

Return to the election of Section~\ref{sec:election-story}, now with the spec
for the voters, which the registrar observes.  Each voter's ballot is a row of $E$, while the sets of ballots the
registrar may observe across voters are implicit outputs in $P\setminus E$.
Each voter decides between $a$ and $b$ and emits its ballot.  That decision is
a \commitment: a voter that votes $a$ rules out every set of ballots in which
it votes $b$.  But it is a single-voter \commitment involving no
communication, so the voter can make it and emit its ballot during a partition.
Each ballot is cast by a different voter, so ballots emitted on different
sides of a partition never conflict (Appendix~\ref{app:deployment} formalizes this
placement as a deployment), and every set of ballots an
observer may collect is consistent.

A partition does not change the intrinsic \commitments a spec needs; it
prevents separated participants from jointly enforcing them through
coordination.  The voters are unaffected, since each ballot is a \commitment
its voter makes alone.  When Lifted Monotonicity for $E$ fails, some output of
$E$ instead waits for another side, sacrificing Availability, or is emitted
anyway, sacrificing Partition Consistency.

\paragraph{Gilbert and Lynch.}
Gilbert and Lynch's theorem is a case of Complete CAP in which Lifted Monotonicity for $E$ fails. We show the failure for the register example, then extend the conclusion from the canonical system to every system.  Linearizability
respects real-time order: if one operation finishes before another begins,
every linearization puts the first before the
second~\cite{herlihy1990linearizability}.  In the register example, suppose the
write finishes before the read begins.  Every linearization then puts the write
first, so the read must return $v$.  But each of $p$ and $q$ sees only its own
request, and Availability requires $p$ to acknowledge the write and $q$ to
answer $v_0$.  These two rows are not confluent at the input with both requests, so
the spec does not satisfy Lifted Monotonicity for them.  Their theorem
concerns every system, not only the canonical one, so it needs one more step.

\begin{corollary}
\label{cor:complete-cap-any-system}
If some system for a spec has Availability for $E$ and Partition Consistency,
then so does the spec's canonical system for $E$.
\end{corollary}

\begin{proof}[Proof sketch]
By Theorem~\ref{thm:complete-cap}, it suffices to show that if Lifted
Monotonicity for $E$ fails, the system loses Partition Consistency.  Then some
rows of $E$ are not confluent at a common extension $j$ of their inputs.  We place each row's participant on
its own side of a partition, and Availability lets each one emit its row's
output.  By Observable equivalence, a participant keeps no more rows than its
canonical run, so each row it keeps at $j$ \outrel its own output.  A common
state at $j$ would need a row that \outrel every participant's output, and
there is none.  Full proof in Appendix~\ref{app:proofs-sec4}.
\end{proof}

A system might hope to keep more outputs available, without losing Partition
Consistency, by issuing \commitments that coordinate within each side of a
partition.  It cannot: Corollary~\ref{cor:complete-cap-any-system} holds for
every system, whether or not it issues \commitments.

By Corollary~\ref{cor:complete-cap-any-system}, no system can keep both outputs
available and stay consistent across a partition, which is Gilbert and Lynch's
theorem.

Two differences from Gilbert and Lynch's setting do not affect the conclusion.
In general our Availability and Partition Consistency are stronger than their
availability and consistency, but in the register example they coincide.  Each
fiber holds a single output, so emitting every permitted output is the same as
responding, and a common extension of the two answers is a single linearization
that explains both.  We also put the real-time order of the write and the read
in the input, while their model takes it from the execution; the theorem holds
in their model as well (Appendix~\ref{app:cap-response}).

\section{Composition and Tight Coordination}
\label{sec:composition}

A replicated key-value store poses a puzzle.  With mutable values, it is
non-monotone, and its replicas can diverge without coordination~\cite{alvaro2011bloom}. The pragmatic solution is
to impose a consensus-driven log upstream of the replicas, \`a la
Multi-Paxos~\cite{lamport2001paxos}, so all replicas execute the same commands
in the same order. But consensus is itself non-monotone: when two commands
arrive concurrently, either may be appended next, and the two alternative orders
have no common extension.  In essence, the standard solution to non-monotonicity
is more non-monotonicity! Why does it work?

That puzzle is explained by our framework.  As we'll see, two specs compose by equijoin,
and the composition of two monotone specs is also monotone. Moreover, for non-monotone specs,
as in our puzzle, we
can reason about which stage actually makes \commitments under composition.  When the
downstream needs no \commitment on the inputs the upstream offers, the
upstream is a \emph{sufficient guard}, and its \commitments then suffice for
the whole pipeline.  When, in addition, the upstream makes no \commitment the
composition does not need, it is a \emph{tight guard}, and the pipeline's
coordination is tight (\S\ref{sec:tight-coordination}).  Proofs are in
Appendix~\ref{app:composition}.
\label{sec:sequential-composition}

Formally, consider two \spectype[s]
${\Spec_1=(I,\InputOrd,M,\preceq_M,P_1)}$ and
${\Spec_2=(M,\preceq_M,O,\OutputOrd,P_2)}$
in which the output domain of $\Spec_1$ is the input domain of $\Spec_2$.  Their
\emph{sequential composition} is the \spectype from $I$ to $O$ with table
\begin{equation}
\begin{split}
  P_2\circ P_1
  &=\pi_{I,O}(P_1\Join_M P_2)\\
  &=\{(i,o):\exists m.\ (i,m)\in P_1\land(m,o)\in P_2\}
\end{split}
  \label{eq:spec-composition}
\end{equation}
where $\pi$ and $\Join$ are relational projection and equijoin.  Every $i\in I$
has a nonempty fiber in $P_2\circ P_1$, since $P_1$ is total on $I$ and $P_2$ is
total on $M$ (it has a row at every $m\in M$).  We run the two specs as a pipeline: $\Spec_2$ takes the outputs of $\Spec_1$ as
its inputs, and we say that $\Spec_1$ \emph{guards} $\Spec_2$.
Appendix~\ref{sec:parameterized-composition} allows each spec to have an input
of its own as well, as when a store serves reads that bypass the log.

\begin{lemma}[Composition]
\label{lem:composition}
If $\Spec_1$ and $\Spec_2$ satisfy Lifted Monotonicity, then so does the \spectype with table $P_2\circ P_1$.
\end{lemma}

\begin{proof}[Proof sketch]
By Proposition~\ref{prop:common-expansion-form}, it suffices to check the
row-consistency condition.  Take finitely many rows $(i_r,o_r)$ of $P_2\circ P_1$ and a common extension $j$ of their inputs.  Each row comes from an intermediate value $m_r$ with $(i_r,m_r)\in P_1$ and $(m_r,o_r)\in P_2$.  The rows $(i_r,m_r)$ are confluent at $j$ in $P_1$, through some $m\in P_1(j)$.  So $m$ extends every $m_r$, the rows $(m_r,o_r)$ are confluent at $m$ in $P_2$, through some $o\in P_2(m)$, and $(j,o)$ is a row of $P_2\circ P_1$ through which the rows $(i_r,o_r)$ are confluent at $j$.  Full proof
in Appendix~\ref{app:composition}.
\end{proof}

By associativity the lemma extends to finite pipelines.  For a function $f$,
its table $\{(i,f(i))\}$ satisfies Lifted Monotonicity exactly when $f$ is
monotone (\S\ref{sec:lifted-monotonicity}).

\subsection{Tight Coordination}
\label{sec:tight-coordination}

We return to the replicated store of the opening.  The key-value store is not
monotone on its own, but it is monotone on what the log offers: a log fixes
the order of commands, and a longer log only extends the responses.  The log
issues a \commitment to offer only such inputs, so the store need make no
\commitments.

More generally, let $\Spec_1$, which guards $\Spec_2$, be run by a system with any basis, and
let $\Spec_2$ run under a system of its own.  Write $\Spec_2|_{\Spec_1}$ for $\Spec_2$ restricted to the outputs of
$\Spec_1$.  We say that $P_2$ satisfies Lifted Monotonicity \emph{on} a set
$D\subseteq M$ when its restriction to inputs in $D$ does.

\begin{definition}[Sufficient guard]
\label{def:sufficient-guard}
$\Spec_1$ is a \emph{sufficient guard} for $\Spec_2$ when
$\Spec_2|_{\Spec_1}$ needs no \commitment: every output at every state it
reaches without \commitment is legal.
\end{definition}

Applying the No-\commitment test (Proposition~\ref{prop:exposure-boundary})
to $\Spec_2|_{\Spec_1}$ gives a static test.

\begin{proposition}[Guard test]
\label{prop:guard-test}
$\Spec_1$ is a sufficient guard for $\Spec_2$ exactly when
$\Spec_2$ satisfies Lifted Monotonicity on the outputs of $\Spec_1$.
\end{proposition}

\begin{proof}[Proof sketch]
Apply the No-\commitment test to $\Spec_2|_{\Spec_1}$, whose inputs are the
outputs of $\Spec_1$.  Full proof in Appendix~\ref{app:composition}.
\end{proof}

A sufficient guard leaves no \commitment to the downstream stage.

\begin{theorem}[Sufficient guards]
\label{thm:sufficient-guards}
Suppose $\Spec_1$ is a sufficient guard for $\Spec_2$, and fix any
system for $\Spec_1$.  Then $\Spec_2|_{\Spec_1}$ answers each output of a run
$\rho$ of that system without \commitment, and the pipeline's runs are
consistent whenever those of the system for $\Spec_1$ are.  Specifically,
replacing each output of $\rho$ by its answers, and dropping the \commitment
events of $\rho$, gives a canonical run of $P_2\circ P_1$, and such runs are
confluent at an input whenever the runs of the system for $\Spec_1$ are.
\end{theorem}

\begin{proof}[Proof sketch]
The outputs of a run of $\Spec_1$ form a chain, since each output event keeps
only rows above its output, so they are a valid input sequence for
$\Spec_2|_{\Spec_1}$, which answers each without \commitment.  Each row
$(i,m)$ that the run of $\Spec_1$ passes through is a row of $P_1$.  It joins
with a row $(m,o)$ of $P_2$ to yield the projected row $(i,o)$ of
$P_2\circ P_1$, and these rows form a run of $P_2\circ P_1$.  Full proof in
Appendix~\ref{app:composition}.
\end{proof}

Conversely, the upstream cannot do without \commitments when the composite is
not monotone.

\begin{corollary}[Upstream \commitments]
\label{cor:upstream-commitments}
If $\Spec_1$ is a sufficient guard for $\Spec_2$ and the composite
$P_2\circ P_1$ is not lifted monotone, some output of $\Spec_1$ needs a
\commitment.
\end{corollary}

\begin{proof}
Otherwise $\Spec_1$ satisfies Lifted Monotonicity by the No-\commitment test,
and $\Spec_2|_{\Spec_1}$ does by Proposition~\ref{prop:guard-test}.  The join
in $P_2\circ P_1$ uses only outputs of $\Spec_1$, so $P_2\circ P_1$ is also the
composite of $\Spec_1$ and $\Spec_2|_{\Spec_1}$, and it satisfies Lifted
Monotonicity by Lemma~\ref{lem:composition}, a contradiction.
\end{proof}

This answers the opening puzzle.  By Proposition~\ref{prop:guard-test}, the
store restricted to logs needs no \commitment, so locks in the store would be
pure overhead, and by Theorem~\ref{thm:sufficient-guards} the pipeline is consistent if the log is.  The store as a whole is not monotone, so by
Corollary~\ref{cor:upstream-commitments} some output of the log needs a
\commitment: the upstream log is where that belongs, and consensus provides
it.

Longer pipelines nest: when $\Spec_1$ is a sufficient guard for $\Spec_2$,
the composite $P_2\circ P_1$ can itself guard a third stage, and the theorem applies again
(Appendix~\ref{app:composition}).

A downstream implementation may still take \commitments the specs do not
require, as a store that takes locks behind a log would.  An implementation of
$\Spec_2$ alone, rather than of $\Spec_2|_{\Spec_1}$, may also fail to answer
some upstream outputs (Appendix~\ref{app:components}).

\paragraph{Necessary and tight guards.}
A sufficient guard may still make more \commitments than the composition
needs.  Suppose two commands $A$ and $B$ commute,
and the store's outputs do not record their order, as when it reports its
state.  Emitting the log $[A]$ then rules out the log $[B,A]$ by \commitment,
yet the composite needs no \commitment there, since both logs lead to the same
output.  That \commitment is \emph{unnecessary}: it prevents a
failure of Lifted Monotonicity in the upstream that leads to no failure in the
composite.

\begin{definition}[Necessary guard, tight guard]
\label{def:necessary-guard}
\label{def:tight-guard}
$\Spec_1$ is a \emph{necessary guard} for $\Spec_2$ when it makes no
unnecessary \commitment: whenever rows $(i,m)$ and $(k,x)$
of $P_1$ with $i\InputOrd k$ are not confluent at $k$ in $P_1$, some rows
$(i,o)$ and $(k,y)$ of $P_2\circ P_1$, with $o\in P_2(m)$ and $y\in P_2(x)$,
are not confluent at $k$ in $P_2\circ P_1$.

$\Spec_1$ is a \emph{tight guard} for $\Spec_2$ when it is both a
sufficient and a necessary guard, and the pipeline's coordination is then
\emph{tight}.
\end{definition}

The order $\preceq_M$ on the upstream's outputs determines which \commitments
it needs.  The prefix order makes $[A]$ and $[B,A]$ conflict, though the store
treats them alike, so it demands an unnecessary \commitment.  We therefore
compare
$\preceq_M$ with a relation derived from $P_2$: $m'$ \emph{covers} $m$,
written $m\sqsubseteq m'$, when every finite set of outputs in $P_2(m)$ has a
common extension in $P_2(m')$.  Confluence of rows and canonical
\commitments \emph{under} $\sqsubseteq$ use $\sqsubseteq$ in place of
$\preceq_M$.
For a functional downstream $f$, $m\sqsubseteq m'$ exactly when
$f(m)\OutputOrd f(m')$.

\begin{theorem}[Guard Tightening]
\label{thm:guard-tightening}
Let $\Spec_1$ guard $\Spec_2$.
\begin{enumerate}
    \item An order on $M$ makes $\Spec_1$ a sufficient guard for $\Spec_2$
        exactly when it is contained in $\sqsubseteq$ on the outputs of
        $\Spec_1$.
    \item If $\Spec_1$ is a sufficient guard for $\Spec_2$, every \commitment
        the composite needs is a canonical \commitment under $\sqsubseteq$.
        If each $P_2(m)$ has a greatest element, as when $\Spec_2$ is
        functional, this holds without that assumption, and the converse
        holds too.  So ordering $M$ by $\sqsubseteq$ makes $\Spec_1$ a tight
        guard.
\end{enumerate}
\end{theorem}

\begin{proof}[Proof sketch]
(1) is Proposition~\ref{prop:guard-test}, through the
row-consistency condition.  (2) Suppose
$m,x\sqsubseteq m'$ for some $m'\in P_1(k)$.  Then answers $o\in P_2(m)$ and
$y\in P_2(x)$ have extensions in $P_2(m')$, which is directed, so some
$u\in P_2(m')$ \outrel both, and $(i,o)$ and $(k,y)$ are confluent at $k$ in
the composite.  With
greatest answers $t\in P_2(m)$ and $s\in P_2(x)$, a common extension of $t$
and $s$ in the composite comes from some $m'$ with $m,x\sqsubseteq m'$.  With
greatest answers, $\sqsubseteq$ is reflexive and transitive, so (1) makes
$\Spec_1$ a sufficient guard under it, and its canonical \commitments are then
exactly the composite's.  Full proof in
Appendix~\ref{app:composition}.
\end{proof}

A greatest element of $P_2(m)$ is one answer to $m$ that extends every
other.  A functional downstream has one trivially, and the gap example below,
where $P_2(z)=\mathbb N$, has none.  Since every order under which $\Spec_1$ is
a sufficient guard is contained in $\sqsubseteq$, and a smaller order leaves at least as many pairs not confluent, each such order requires at least the
\commitments that $\sqsubseteq$ does.  At the other extreme, equality on $M$
requires the most, since every order on $M$ contains it.

Guard Tightening is the route to tight coordination: when the
downstream's answers have greatest elements, an upstream whose order requires
no more \commitments than $\sqsubseteq$ makes exactly the \commitments the
composition needs.  When the store answers each command in turn, as in
state-machine replication, $\sqsubseteq$ is the prefix order, and the
consensus log is a tight guard: execution is monotone and reflects the prefix
order, so every \commitment the log makes is one the store's clients can
observe (Appendix~\ref{app:tight}).  When the store reports only its state and
$A$ and $B$ commute, $[A,B]$ and $[B,A]$ cover each other.  The prefix order
keeps them apart: $(\{A\},[A])$ and
$(\{A,B\},[B,A])$ are not confluent at $\{A,B\}$ under the prefix order but
confluent at $\{A,B\}$ under $\sqsubseteq$.  So the commuting log is a
sufficient guard but not a necessary one.  Ordering the log by $\sqsubseteq$
treats such values as equivalent and, since the store is functional,
makes the log a tight guard by Guard Tightening; a log that orders only
non-commuting commands does just that (Generalized
Paxos~\cite{lamport2005generalized}, EPaxos~\cite{moraru2013epaxos}).  Reversing the order
on $M$ in the boolean toggle below is reordering by $\sqsubseteq$.  Without
greatest answers, a gap can remain.  Let $\Spec_1$ emit $z$ at input
$\mathit{lo}$ and any $n_k$ at a larger input $\mathit{hi}$, and let $\Spec_2$
answer $z$ with any number and $n_k$ with $k$.  No $n_k$ covers $z$, since no single
answer $k$ extends every number the downstream may answer to $z$, so $(\mathit{lo},z)$ and each row at $\mathit{hi}$ are not confluent at
$\mathit{hi}$ under $\sqsubseteq$, and so under every order on $M$ under which
$\Spec_1$ is a sufficient guard.  Yet the composite permits every number at both
inputs and needs no \commitment: a single implementation could answer at
$\mathit{lo}$ and later give the same answer at $\mathit{hi}$.  The upstream
must commit before it sees the downstream's answer, and the remaining gap is
the price of splitting the work into stages.

\subsection{Composite CALM}
\label{sec:composite-calm}

Sufficient guards show how an upstream
stage can discharge a pipeline's \commitment obligations.
\emph{Composite CALM} asks a stronger question: when does a pipeline need no
\commitment at any stage?  If every stage is monotone, the answer is immediate:
by Lemma~\ref{lem:composition} the composite is monotone too, and no stage needs
a \commitment.  The interesting case has a non-monotone stage, and it raises a
simpler puzzle than our opening one.

Consider the spec for an incrementing counter that is capped at $9$, and
compose it with a downstream spec for a double-digit threshold test, which
outputs \textsf{false} for values up through $9$ and \textsf{true} from $10$ on.
Under the usual boolean order, as in LVars~\cite{kuper2013lvars}, \textsf{false}
is only a lower bound, and the test is monotone.  Suppose instead that every
answer must be final: unchanging in the face of more input.  Then the output order is equality, and the test cannot
emit \textsf{false} without a \commitment, since a later count of $10$ would
contradict it.  Yet when composed after the capped counter, the pair of
stages needs no \commitment.  When do composed specs erase the need for
\commitment entirely?

The composite is itself an \spectype.  By Complete CALM, its runs are
consistent exactly when it satisfies Lifted Monotonicity, and then no output needs a
\commitment (Theorem~\ref{thm:complete-calm},
Proposition~\ref{prop:exposure-boundary}).  The orders on $I$ and $O$ are part
of the composite's contract with its environment.  The order on $M$ is not:
the composite's table never mentions it, and some order on $M$ may make both
stages monotone.  A composite can be jointly monotone, with no stage needing a
\commitment, in three ways.

First, every stage is monotone, as noted above.  The CRDT observations of Section~\ref{sec:crdt-instance} have this shape.  A
client query $q$ over the replica state $s$ is the composition $q\circ s$, and
compositions of monotone functions are monotone, so monotonicity of $q$ on the
whole state domain suffices.

Second, the upstream can filter out the inputs on which the downstream is not
monotone.  In the example above, the test fails monotonicity only when the
count passes from single to double digits, and the capped counter never does
so.  Restricted to the counts it does produce, the test always answers
\textsf{false}; it is constant, and so monotone under any output order,
including equality.  By Proposition~\ref{prop:guard-test}, the capped counter
is a sufficient guard, and since it is itself monotone, neither stage needs
a \commitment (Proposition~\ref{prop:joint-monotone}(2)).  This is tight
coordination in which the upstream itself needs no \commitment.

Third, a different order on $M$ can make two non-monotone stages monotone.
Consider two stages that each negate a boolean, so that their composite is the
identity.  We order the booleans with false below true for $I$, $O$, and $M$.
Neither stage is monotone, since each maps true, the larger value, to false.
Reversing the order on $M$, so that true lies below false, makes both stages
monotone, and the composite is the identity either way
(Appendix~\ref{app:composition}).

\begin{proposition}[Jointly monotone composites]
\label{prop:joint-monotone}
The composite $P_2\circ P_1$ satisfies Lifted Monotonicity, so its runs are
consistent and it needs no \commitment, in each of the following cases.
\begin{enumerate}
  \item $\Spec_1$ and $\Spec_2$ satisfy Lifted Monotonicity.
  \item $\Spec_1$ satisfies Lifted Monotonicity; $\Spec_2$ satisfies it on
        $\Spec_1$'s outputs.
  \item $\Spec_1$ and $\Spec_2$ satisfy Lifted Monotonicity via a different order
        on $M$.
\end{enumerate}
\end{proposition}

\begin{proof}[Proof sketch]
(1) is Lemma~\ref{lem:composition}.  (2) The proof of the lemma uses only
intermediate values in fibers of $\Spec_1$.  (3) The composite joins the two
tables on the intermediate value and never compares two intermediate values, so
its table does not depend on the order on $M$.  The lemma then applies under
the different order.  Full proof in Appendix~\ref{app:composition}.
\end{proof}

With greatest answers, the three cases cover every monotone composite: if each
$P_2(m)$ has a greatest element, ordering $M$ by $\sqsubseteq$ puts it in
case~(3) (Appendix~\ref{app:composition}).  Without greatest answers, the
cases can fail to cover a monotone composite: in the gap example of
Section~\ref{sec:tight-coordination}, no order on $M$ makes both stages
monotone.

\subsection{Intersection}
\label{sec:intersection}

Sequential composition feeds the outputs of one spec to the inputs of the
next.  Specs can also be combined on a shared interface.  An application spec,
such as a replicated store that answers reads by a fixed rule, may promise
nothing about isolation, while a consistency property, such as an isolation
level, constrains the same reads.  A system that meets both may produce only
the rows that both permit.

\begin{definition}[Intersection]
\label{def:intersection}
Let $\Spec_1=(I,\InputOrd,O,\OutputOrd,P_1)$ and
$\Spec_2=(I,\InputOrd,O,\OutputOrd,P_2)$ share their ordered domains.  Their
\emph{intersection} $\Spec_1\cap\Spec_2$ has table $P_1\cap P_2$.  It is an
\spectype when $P_1\cap P_2$ is total on $I$: every input has an output that
both permit.
\end{definition}

The shared domains matter.  For example, isolation levels are defined on
histories of reads, writes, and commits, so an application spec must describe
its inputs and outputs in those terms before it can be intersected with one.

Intersection does not preserve Lifted Monotonicity.  Let $a$ and $b$ be
incomparable outputs, each below both $c$ and $d$, with $c$ and $d$
incomparable, and let both specs have a single input.  The fiber $\{a,b,c\}$ is
directed, since $c$ \outrel each of its members, and so is $\{a,b,d\}$.  The
intersection of the two fibers, $\{a,b\}$, is not: no output in it \outrel
both $a$ and $b$.  Two specs that need no \commitment can thus have an
intersection that does.  A
monotone bound rules this out.

\begin{proposition}[Intersection under a monotone bound]
\label{prop:intersection-bound}
Let $g:I\to O$ be monotone.  Suppose that every output $\Spec_1$ permits at an
input $j$ lies below $g(j)$, and that $\Spec_1$ and $\Spec_2$ both permit $g(j)$
at every $j$.  Then $\Spec_1\cap\Spec_2$ is an \spectype and satisfies Lifted
Monotonicity.
\end{proposition}

\begin{proof}[Proof sketch]
The intersection permits $g(j)$ at every $j$, so it is total.  An output it
permits at $i\InputOrd j$ lies below $g(i)$, and $g(i)\OutputOrd g(j)$, so
$g(j)$ \outrel every output it permits at inputs below $j$.  Full proof in
Appendix~\ref{app:intersection}.
\end{proof}

A store whose answers follow a monotone rule therefore stays monotone when
intersected with any property that permits the rule's answers.
Section~\ref{sec:hat-instance} uses this construction for isolation levels and
a key-value store.

\section{Related Work and Mappings}
\label{sec:instances}
\label{sec:related}

Complete CALM provides a unified lens for many results across literatures.
For the key related ideas, we provide mappings here to our specification model.
Each setting declares an input, an output, and their orders before
applying either CALM theorem.

Three patterns of non-monotonicity recur in the mappings below and in
coordinated systems generally.  In each, an output settles something the input
has not: outputs that order concurrent events, as serial schedules,
total-order broadcast, and Lost Update do; outputs that choose a proper subset
of competing values, as consensus, leader election, and unique IDs do; and
outputs that block an input before it arrives, as closing the polls and a
final tally do (proofs in Appendix~\ref{app:classical-patterns}).  A system either enforces such
an output with a \commitment or leaves it unenforced, risking observable
inconsistency.  The mappings below include both.

\subsection{Query monotonicity}
\label{sec:query-instance}

Several CALM results for deterministic queries share one form: a function can
be computed without coordination exactly when it is monotone under some input
order.  Functional CALM states that form at the level of a spec, for any input
and output orders.  The CALM Theorem of Ameloot et al.\ (\S\ref{sec:functional-calm}) is stated in
a particular operational model, relational languages on transducer networks: a query can be evaluated without
coordination exactly when it is monotone, under that model's assumptions about
data placement and execution~\cite{ameloot2013relational,ameloot2015weaker}.
Functional CALM restates the result as a property of a spec, so that it can be
compared with specs outside Datalog; it does not replace the operational details
of that theorem.

Subsequent extensions to the original CALM Theorem map into our framework as
changes to the input domain and its order, or to the notion of coordination.
Zinn et al.\ gave nodes partial knowledge of how data is placed, which makes
more queries coordination-free, among them the non-monotone win-move
query~\cite{zinn2012winmove}.  Ameloot et al.\ characterized these queries by
weaker forms of monotonicity~\cite{ameloot2015weaker}.  A query is
\emph{domain-distinct-monotone} when $Q(I)\subseteq Q(I\cup J)$ whenever each
new fact in $J$ contains a value not in $I$, and
\emph{domain-disjoint-monotone} when this holds whenever each new fact uses
only such values.

We model each data placement by an input order.  For the placement that
Ameloot et al.\ call ``policy-aware'', $I\cup J$ \inrel $I$ only when $J$ is domain-distinct from $I$; for ``domain-guided''
placement, only when $J$ is domain-disjoint from $I$.  Under these orders,
functional monotonicity is exactly domain-distinct- and
domain-disjoint-monotonicity, so Functional CALM gives their
characterizations.

Li and Lee prove CALM for replicated objects with write, merge, and query
operations~\cite{li2025coordinationfree}.  In our terms, they prove a
counterpart of Functional CALM for input domains of sets ordered by inclusion
and arbitrary partially ordered output domains, with the existence of a
coordination-free implementation in place of consistency.

Baccaert and Ketsman extended CALM beyond deterministic
queries~\cite{baccaert2023dcbq,baccaert2026spectrum}.  Their running example
is leader election, which permits many outputs for one input and chooses one
candidate among competing ones, the second pattern above.  They also
changed how data placement is modeled.  Ameloot et al.\ require a query to be
computed correctly however its input is distributed across nodes.  Baccaert
and Ketsman make placement part of the problem instead.  Each fact names the
node that holds it, as in Dedalus and declarative
networking~\cite{alvaro2010dedalus,loo2009declarative}, and a
\emph{behavior} maps such located inputs to sets of permitted located outputs.

Their model maps directly to ours.  A behavior is an \spectype whose outputs
are ordered by inclusion.  Each fact names the participant that holds it.  One
input extends another by adding the facts of participants not yet heard from,
never by adding facts of a participant already heard from.  A network of
transducers is a system in which every participant runs the same strategy,
sees only its own facts and the contents of the messages it receives, and
broadcasts its messages.  Since its own facts include its name, a participant
knows its own name and the order on names, which is what their coordination-free
transducers know.  It does not know the list of all participants.  Under their
placement, a participant's input events contain only the facts that name it, at
the start; other facts reach it only in messages, and a participant passes on
only facts it holds.  A
transducer that computes a behavior translates into our framework as a system
whose runs stay confluent.  Every output that some fair run of the transducer can produce is
emitted by some run of that system.

Two sides of a partition cannot tell that they are disconnected, so each must
emit to be available, and when the partition heals, their outputs must have a
common \outnoun at any input that extends what both sides have received.
Complete CAP requires this condition across every partition.  Baccaert and
Ketsman's \emph{node monotonicity} is the same condition, restricted to sides
that have heard from disjoint sets of participants.  Since other facts reach a
participant only in messages and no message crosses a partition, the inputs of
the two sides never share a fact.  In our model, an input event at any
participant may contain any fact, so input events on both sides may contain
the same fact, and for nondeterministic behaviors node monotonicity is strictly
weaker than Lifted Monotonicity.

Unconstrained matching shows the gap.  Each item is a fact that names one
participant, and any set of disjoint pairs of items is a permitted output.
Under their placement, the inputs of the two sides contain disjoint sets of
items, so their pairings always combine.  In our model, input events on both
sides of a partition can contain the same item $i$.  One side can pair $i$ with
$j$ and the other with $k$, and no matching contains both pairs.  For
deterministic behaviors the two conditions coincide: a function's output on
each side's input must lie below its one output on the union, and node
monotonicity already requires this, since the union extends each side's input
by participants that side has not heard from.

Baccaert and Ketsman show that a behavior can be computed without the list of
all nodes exactly when it contains a node monotone behavior, and this
characterization holds in our formalism.  Its necessary half is the partition
argument above, applied to histories that respect their placement.  The
sufficient half does not follow from Complete CAP: participants on the same
side can hear from overlapping sets of participants, and their construction,
in which a group's smallest node decides for the group, makes \commitments that keep such rows consistent.  Our version does not need their requirement that
programs treat data values only as interchangeable symbols.

Their aims differ from ours.  They ask which behaviors become computable when
nodes share some piece of system knowledge, such as identifiers, an order on
them, or the list of all nodes.  For them, a behavior needs coordination when
it cannot be computed without the list of all nodes.  Correctness is judged by
the final output of a run, and the input is fixed when the run starts.  We ask
instead when a spec requires a \commitment, judge consistency at every point of
a run, and let input arrive over time.  A behavior can therefore be
coordination-free in their sense while still needing a \commitment in ours.
For them, coordination means needing the list of all nodes, so a \commitment
that one participant makes locally, as when a group's smallest node decides
for the group, does not count.  Our criterion identifies such local
\commitments as well.  Each refinement of CALM in this subsection
thus corresponds to a choice our model makes explicit: the input order, for the
placements of Ameloot et al.; sets ordered by inclusion as inputs, for Li and
Lee; and the facts that an input event at each participant may contain,
together with nondeterministic outputs, for Baccaert and Ketsman.

\subsection{CRDTs: monotone queries over convergent state}
\label{sec:crdt-instance}

A CRDT uses join semilattices to guarantee that replicas that have received the
same updates converge to equivalent states, whatever the order of delivery and
merge~\cite{shapiro2011crdt}.  That guarantee, Strong Eventual Consistency, is
the convergence of Corollary~\ref{cor:convergence}.  The corollary assumes only
that the spec is monotone, not that its input or output order is a join
semilattice.

Strong Eventual Consistency says nothing about what a client observes when it
queries a replica that has not yet received every update.  Such a query risks
Lamport's ``anomalous behavior''~\cite{lamport1978time}
(\S\ref{sec:functional-events}).  Laddad et al.\ proposed monotonicity as the
right test for safe observation from CRDT state~\cite{laddad2022keep}.  Their
example of an unsafe query uses a two-phase set, whose state
$(\mathit{Add},\mathit{Del})$ merges by componentwise union.  The query
$\mathit{Add}\setminus\mathit{Del}$ can expose added items that are later
deleted.  We make the proposed test precise with Functional CALM.

Let $(D,\sqsubseteq,\sqcup)$ be the state join-semilattice.  The input is the
finite set $U$ of immutable payloads received so far, ordered by inclusion, and
the replica's state is $s(U)=s_0\sqcup\bigsqcup U$.  Then $U\subseteq V$
implies $s(U)\sqsubseteq s(V)$, and $s(U_1)\sqcup s(U_2)=s(U_1\cup U_2)$ is the
least state \outrelpart both~\cite{shapiro2011crdt}.  A client observation
$q:D\to(Y,\preceq_Y)$ gives the client-visible spec $q\circ s$.  By
Functional CALM, the runs of $q\circ s$ are consistent exactly when $U\subseteq V$ implies
$q(s(U))\preceq_Y q(s(V))$.  Monotonicity of $q$ on all of $D$ suffices, but
$q$ need only be monotone between reachable states.  The query $\mathit{Add}\setminus\mathit{Del}$ over a two-phase set fails the test: a deleted item grows $U$ and shrinks the answer.  An answer that exposes an item as
present blocks its deletion before it arrives, the third pattern above.  The
two-phase set does not enforce the block: it accepts the deletion, and the
answer is contradicted.
Operation-based CRDTs fit too, with the causally closed set of
delivered operations as
input~\cite{shapiro2011crdt,burckhardt2014replicated,baquero2014operation,hellerstein2025turtles}.

Semilattice join laws are neither sufficient nor necessary for consistent
observation.  CRDT states form a
join-semilattice, but $q$ needs less:
Complete CALM requires only \emph{some} common \outnoun of query outputs
(\S\ref{sec:lifted-monotonicity}), not a least one.

\subsection{Highly available transactions}
\label{sec:hat-instance}

Bailis et al.\ define isolation levels by Adya's phenomena, adding sessions and
per-item dependencies, and give highly available algorithms for the weaker
levels~\cite{bailis2013hat,adya1999weakconsistency}.  We take their
assumptions as given.  Items are partially replicated, and a client sends each
operation to some replica of its item.  A transaction's writes carry a unique
client timestamp, and versions are ordered by timestamp, last writer wins.
The system never aborts a transaction on its own.  A replica answers a read
with the highest-timestamped committed write of the item that it has received;
we call this the store's \emph{read rule}.

The HAT results illustrate the intersection reasoning of
Section~\ref{sec:intersection}.  An isolation level and a store are separate
specs on the same histories: the level constrains which read answers may
appear together, and the store answers each read by its read rule.  The store
held to a level is their intersection, and we ask which intersections satisfy
Lifted Monotonicity.

\begin{definition}[HAT specification]
\label{def:hat-spec}
An input is a history ordered by happens-before.  It contains reads, each
tagged with its replica; commits, each with its client timestamp and write set; and arrivals
of committed writes at replicas.  The history preserves
each client's program order and each replica's arrival order, and a
transaction's reads precede its commit, which precedes the arrivals of its
writes at replicas.
One input \inrel another when it adds events, none of which happens before an
event already present.  An output is a set of read
answers, ordered by inclusion.  A history is \emph{complete} when it answers
every read and orders versions by timestamp.  For an isolation level $L$, the
table $P_L$ of $\Spec_L$ contains a row $(i,o)$ when some complete history
extends $i$, includes $o$, and exhibits none of the phenomena that $L$
proscribes.  The explicit relation $E$ contains, at each input, the output that
the read rule gives.
\end{definition}

\begin{definition}[Store specification]
\label{def:store-spec}
The \emph{store specification} $\Spec_{\mathrm{st}}$ has the inputs and outputs
of Definition~\ref{def:hat-spec}.  At each input, it permits any subset of the
answers that the read rule gives.
\end{definition}

The store spec says only that answers follow the read rule, and nothing about
isolation.  On its own it satisfies Lifted Monotonicity: a read's answer under the read rule depends only on the writes its replica received
before the read, so later input never changes it.  The store's rows that answer
every read are exactly the rows of $E$.

Holding the store to an isolation level $L$ means taking the intersection
$\Spec_{\mathrm{st}}\cap\Spec_L$ (Definition~\ref{def:intersection}), which
is total, since both specs permit
answering no reads.  By Proposition~\ref{prop:intersection-bound}, the
intersection keeps Lifted Monotonicity when $\Spec_L$ permits the read rule's
answers.

\begin{theorem}[HAT levels]
\label{thm:hat-levels}
Suppose each transaction's timestamp exceeds that of every version it reads.
For read uncommitted and read committed, $E\subseteq P_L$ and
$\Spec_{\mathrm{st}}\cap\Spec_L$ satisfies Lifted Monotonicity, so $\Spec_L$
satisfies Lifted Monotonicity for $E$, so some system for $\Spec_L$ has
Availability for $E$ and Partition Consistency.  For any level that proscribes
Lost Update, including snapshot isolation, repeatable read, and
serializability, $\Spec_{\mathrm{st}}\cap\Spec_L$ does not satisfy Lifted
Monotonicity, and $\Spec_L$ does not satisfy Lifted Monotonicity for any
explicit relation $E$ that answers every read; hence no system for $\Spec_L$
has Availability for such an $E$ and Partition Consistency.
\end{theorem}

\begin{proof}[Proof sketch]
The theorem has two parts.  First, for read uncommitted and read committed,
the timestamp condition rules out the dependency cycles of G0 and G1c, so
$\Spec_L$ permits the read rule's answers.
Proposition~\ref{prop:intersection-bound}, with the read rule as the bound,
then gives Lifted Monotonicity of $\Spec_{\mathrm{st}}\cap\Spec_L$, whose rows
include those of $E$.  Second, for a level that proscribes Lost Update, the
intersection loses Lifted Monotonicity, because the read rule's own answers can
exhibit Lost Update.  Let two transactions at different replicas each read
$x_0$ and write $x$.  Under the read rule, each replica can answer only $x_0$,
since neither has received the other's write.  Together the two answers contain
a cycle on $x$ with an anti-dependency, which is Lost Update: each answer
orders its own transaction before the other's write, the first pattern of
Section~\ref{sec:instances}, and the read rule enforces no single order.  So
the store's
answer at one replica has no extension in $\Spec_{\mathrm{st}}\cap\Spec_L$ once
both transactions are in the input.  The store needs no \commitment on its own,
so this need comes from the intersection with $\Spec_L$.  In both parts, the
conclusions about systems follow from Complete CAP
(Theorem~\ref{thm:complete-cap}) and
Corollary~\ref{cor:complete-cap-any-system}.
\end{proof}

The timestamp condition in the theorem is required, correcting an error in the
extended HAT paper~\cite{bailis2013hatextended}.  That paper's pseudocode
fixes each transaction's timestamp when it begins.  A transaction can then read
a version with a higher timestamp than its own and overwrite it, so its write
is ordered before the version it read: a G1c cycle, which read committed
proscribes.  The algorithm as written therefore does not guarantee read
committed.  Our framework exposed the problem directly as we were writing this
up, since the read rule's output at such an input is not in $P_{\mathrm{RC}}$.
Taking each timestamp at commit, above every version read, repairs the
algorithm.

The extended paper's algorithm for monotonic atomic view has the same
error~\cite[App.~B]{bailis2013hatextended}.  Its clients run the read committed
algorithm and also fix each timestamp when the transaction begins.  When a
transaction reads a version with a higher timestamp than its own and then
overwrites that item, the algorithm's answers contain the same G1c cycle.  So
the algorithm does not guarantee read committed, and therefore not monotonic
atomic view either.  With timestamps taken at commit, above every version
read, it provides monotonic atomic view (Appendix~\ref{app:hat-taxonomy}).

Bailis et al.\ also classify session guarantees and recency alongside further
isolation levels, calling them all \emph{guarantees}.  Each guarantee $G$
becomes a spec $\Spec_G$: a store that returns only versions the reading
replica has received, held to $G$.  Their high availability becomes
Availability for an explicit relation inside $\Spec_G$ that answers every read,
together with Partition Consistency, so by Theorem~\ref{thm:complete-cap} it
holds exactly when $\Spec_G$ satisfies Lifted Monotonicity for such a relation.
Their sticky availability, in which each client is served by a replica that
already holds its earlier writes, becomes Deployment Availability for a
deployment that assigns clients to replicas
(Appendix~\ref{app:deployment}).  Their algorithms supply the
explicit relations.  Our mapping recovers their classification of every
guarantee, each as an instance of Complete CAP or its deployment form
(Appendix~\ref{app:hat-taxonomy}).

Stating the classification precisely also exposed a mismatch in the extended
paper.  Its formal definition of the anomaly that read your writes rules out
lets the forbidden cycle pass through other clients' transactions on other
items~\cite[Def.~34]{bailis2013hatextended}.  Under that definition, read your
writes is not sticky available.  Consider two clients at partitioned replicas.
Client~1 writes $x$, and then a second transaction of client~1 reads $x$ and
$y$ at its own replica, which has its write of $x$ but no write of $y$.
Client~2 reads $x$ at the other replica, which has no write of $x$, and writes
$y$.  Every read is sticky, and each replica can return only the initial
version of the item it lacks.  Together these answers form the forbidden
cycle, so no system can keep both clients available and stay consistent, even
with each client assigned to its own replica.  The definition's own example
uses a single item, and in that form read your writes is sticky available, as
Bailis et al.\ state.

Mahajan et al.\ showed that real-time causal consistency can be enforced in a
system that is always available and one-way convergent, and that no stronger
consistency semantics can~\cite{mahajan2011consistency}.  Their availability,
that every request is eventually answered, is a response obligation in our
terms (Appendix~\ref{app:cap-response}).  Their one-way convergence, that a
node able to receive from another eventually reflects the updates the other
knows, needs no separate assumption, because our input records which writes
each replica has received.

We recover both parts of their result, with real-time precedence recorded in
the input.  The spec in which each replica answers its own reads under
real-time causal consistency can be implemented without a coordination
mechanism, and no strict strengthening of it can, provided the strengthening,
like Mahajan et al.'s implementations, does not consult real time
(Appendix~\ref{app:rtc}).

\paragraph{Application-aware isolation level selection.}
In addition to being coordination-free, weak isolation levels can provide
sufficient isolation for certain mixes of transactional applications.  Kaki et
al.~\cite{kaki2018alone} give a program logic for transactions under weak
isolation, with compositional proof rules for showing that a program maintains
its invariants.  Their inference procedure automatically finds, for each
transaction, the weakest isolation level under which its invariants are still
maintained, using an encoding of database computations that off-the-shelf SMT
solvers can check.  Mapping their results into our framework is an open
problem.

\subsection{Invariant confluence}
\label{sec:iconfluence-instance}

An invariant is a predicate that every state must satisfy.  We follow the model
of Bailis et al.~\cite{bailis2014coordination}.  A database state is a set of
versions, and replicas merge states by union.  Each transaction $t\in T$ runs
on the state of the replica that executes it and adds the versions it writes,
so $t(s)\supseteq s$.  An initial state $s_0$ satisfies the invariant
$\mathcal I$.

A state is \emph{reachable} when it can be derived from $s_0$ by applying
transactions and merging states, with every state along the way satisfying
$\mathcal I$.  Bailis et al.\ call $T$ \emph{invariant confluent} when any two
reachable states with a common ancestor merge to a state that satisfies
$\mathcal I$.  Every reachable state descends from $s_0$, so $T$ is invariant
confluent exactly when $\mathcal I(s_1\cup s_2)$ for all
reachable $s_1$ and $s_2$.

From an instance $(s_0,T,\mathcal I)$ of their model we construct a
spec $\Spec_{\mathcal I}$ in ours as follows (Appendix~\ref{app:iconfluence}).
Inputs are \emph{transaction histories}: partial orders of events, each
labeled by a transaction in $T$, where one event precedes another when the
second ran on a state that included the first's writes.  One history \inrel
another when it adds events, none of which precedes an event already present.
Legal inputs include only the \emph{locally valid} histories, in which each
transaction preserves $\mathcal I$ on the state it ran on.  Outputs are
database states, ordered by set-containment of versioned data.  The
\emph{state} of a history is $s_0$ together with every version its
transactions write.  Call a sub-history \emph{closed} when it contains every
event that precedes one of its own, and \emph{valid} when the state of each of
its closed sub-histories satisfies $\mathcal I$.  At input $h$, each valid
closed sub-history of $h$ contributes its state as a permitted output.  An
output may therefore leave out some of the input's transactions, but what it
includes must be closed and valid.

\begin{proposition}[Invariant confluence as Lifted Monotonicity]
\label{prop:iconfluence-lm}
$T$ is invariant confluent if and only if $\Spec_{\mathcal I}$ satisfies Lifted
Monotonicity, that is, if and only if the runs of $\Spec_{\mathcal I}$ are consistent.
\end{proposition}

\begin{proof}
See Appendix~\ref{app:iconfluence}.
\end{proof}

The uniqueness constraint in Bailis et al.'s payroll example is not invariant
confluent, and in $\Spec_{\mathcal I}$ the failure is a fiber that is not
directed.  Let $c_1$ and $c_2$ be creations at two replicas that each assign
ID~7, each running on a state in which ID~7 is unused, and let $h$ be the
history in which neither ran on a state that included the other's writes.
Each creation preserves $\mathcal I$ on the state it ran on, so $h$ is locally
valid, and it \inrel both one-creation histories.  At $h$ the closed
sub-histories are $\emptyset$, $\{c_1\}$, $\{c_2\}$, and $h$ itself.  The first
three are valid, but $h$ is not, since its state has two employees with ID~7.
So the outputs permitted at $h$ are $s_0$ and the states of $\{c_1\}$ and
$\{c_2\}$.  None of them contains the employees created by both $c_1$ and
$c_2$, so none \outrel both one-creation states.  The fiber at $h$ is therefore
not directed, and $\Spec_{\mathcal I}$ fails Lifted Monotonicity at $h$.  Unique IDs
ask each assignment of ID~7 to choose one employee among those competing for
it, the second pattern of Section~\ref{sec:instances}.  Without a
\commitment, each replica chooses its own, and the merged state violates the
invariant.  A
system must reserve IDs, reject one creation, or weaken the invariant.

\subsection{General operational specification models}
\label{sec:operational-models}
Many systems are specified as state machines.  This subsection shows how to
project a state machine written as an I/O automaton~\cite{lynch1987io} onto a
spec, and how to analyze the automaton with our results.  Each row of the
projected spec is a pair the automaton can be observed in: the requests it has
received and the replies it has sent, as they occur together in some trace.
Every execution of the automaton is a run of this spec.  The projection drops
the order in which an execution passes through these rows, however, so a run of
the spec may combine rows that no single execution reaches in that order.

An automaton may make \commitments along the way, like the registrar that closes
the polls early (Section~\ref{sec:systems}).  After such a \commitment, it can
no longer reach some rows that other executions do reach.  We therefore also
map the automaton into our system model, with its choices as \commitment
events.  Complete CALM tells whether the projected spec's runs are consistent.
Comparing the automaton's runs with the spec's canonical runs tells whether the
automaton itself is: it is confluent exactly when its \commitments never remove
every row its canonical runs could still agree on.  We close with two ways to
check these properties in practice: testing on bounded inputs, and proof.

\paragraph{The automaton model.}
We use the standard I/O-automaton
model~\cite{lynch1987io,lynch1996distributed}.  An automaton has input, output,
and internal actions, a set of states with a nonempty set of start states, and
a transition relation in which every input action is enabled in every state.
A finite execution takes a sequence of steps from a start state.  We read the
automaton as a server.  A client sends a request, the request arrives at the
server, and the server later hands it to its computation.  The arrival is the
automaton's input action.  Buffering requests, handing them to the computation,
and deciding when a reply fires are the automaton's own internal behavior, and
its replies are its output actions.  The \emph{receive sequence} of an
execution lists its input actions in order, and its \emph{reply sequence} lists
its output actions in order.  We order both kinds of sequence by prefix.  This
subsection treats a single automaton.  A server built from several
communicating automata is a distributed system in the sense of
Section~\ref{sec:distributed-execution}.

\begin{definition}[Projection of an automaton]
\label{def:projection}
The \emph{projection} of an automaton $M$ is the \spectype $P_M$ whose inputs
are receive sequences and whose outputs are reply sequences, and whose rows
are the pairs $(w,t)$ for which some finite execution of $M$ has receive
sequence $w$ and reply sequence $t$.  We call $P_M$ the \emph{projected spec}
of $M$.
\end{definition}

Followed step by step, an execution is a run of $P_M$.  An arrival is an input
event carrying the receive sequence so far, a reply is an output event carrying
the reply sequence so far, and an internal step is no event.  What grows along
the execution is the reply sequence, whatever the replies say: a counter's
reports may go up, down, and up again, and its reply sequence still only
grows.

\begin{proposition}[Projection]
\label{prop:projection}
Let $M$ be an automaton.
\begin{enumerate}
  \item $P_M$ is an \spectype, and every finite execution of $M$ is a run of
        $P_M$ from the canonical state at the empty receive sequence, with
        one input event per arrival and one output event per reply.
  \item $P_M$ satisfies Relational Monotonicity.  Its runs are consistent
        exactly when, at each receive sequence, every reply sequence in its
        fiber is a prefix of every longer one.
\end{enumerate}
\end{proposition}

\begin{proof}[Proof sketch]
(1) Every receive sequence has a row, since every input action is always
enabled.  Each event of the run is legal because the execution itself reaches
the row the event needs.  (2) An execution can receive any further requests
before it replies again, which gives Relational Monotonicity.  Finitely many
reply sequences have a common extension exactly when each is a prefix of every
longer one, and then the longest of them is one.  Complete CALM
(Theorem~\ref{thm:complete-calm}) then gives the stated condition.
\end{proof}

Since $P_M$ is an \spectype, the results of
Section~\ref{sec:distributed-execution} apply to it as to any other spec.

The converse of item~(1) fails.  Consider a registrar that receives ballots
$p$.  After the first ballot it may reply $x$ at once and ignore later
ballots, which closes the polls.  Or it may wait for a second ballot, reply
$x$, and later reply $z$.  Its projected spec has a run that receives $p$,
replies $x$, receives a second $p$, and replies $z$: each row the run passes
through occurs in some execution.  No execution has that trace, because a
registrar that replied early never replies $z$.  Every reply sequence of the
registrar is empty, $x$, or $xz$, so the runs of its projected spec are
consistent.

\paragraph{The automaton as a system.}
To recover what the projection drops, we map the automaton into our system
model (Definition~\ref{def:system}) with a single participant.  The rows of
$P_M$ that a finite execution can still reach, by continuing from its final
state, are its \emph{reachable rows}.  Removing every other row from a table
that contains them is a \commitment whenever it removes a row, because no
removed row is \outrelpast by a kept row at the same input.  To see this, take
a kept row $(k,w)$ and a row $(k,u)$ of the table with $w$ extending $u$.
Since $(k,u)$ is in the table, $u$ extends the execution's reply sequence
(Lemma~\ref{lem:run-residual}).  The execution can reach $(k,w)$, so it can
reply up to $u$, stop replying, and receive the rest of $k$.  It then reaches
$(k,u)$, so $(k,u)$ is reachable and is kept as well.

\begin{definition}[System of an automaton]
\label{def:automaton-system}
An execution of $M$ \emph{matches} a local history $\lambda$ when the input
and output events of $\lambda$ carry the execution's arrivals and replies in
order, and the table after $\lambda$ keeps all of the execution's reachable
rows.  An input event may carry several consecutive arrivals, or none.  The
\emph{system of} $M$ has one participant and runs over $P_M$.  Its basis
contains the canonical \commitments and, for each finite execution, the
\commitment that keeps only that execution's reachable rows.  After a local
history $\lambda$, its strategy allows the events that some execution matching
$\lambda$ allows next: the output event that adds a reply the execution's
final state can send, and the \commitment event that keeps only the
execution's reachable rows.
\end{definition}

\begin{proposition}[System of an automaton]
\label{prop:automaton-system}
Let $M$ be an automaton.
\begin{enumerate}
  \item Every history of the system of $M$ is matched by some execution of
        $M$.
  \item Every execution of $M$ is matched by a history of the system whose
        input and output events are its run of $P_M$
        (Proposition~\ref{prop:projection}) and whose table is exactly its
        reachable rows.
  \item After any history, one \commitment event brings the table down to
        exactly the reachable rows of a matching execution, unless the table
        is already there.
\end{enumerate}
\end{proposition}

\begin{proof}[Proof sketch]
(1) By induction on the history.  An input event extends the current receive
sequence, and the matching execution receives the added requests.  An output
or \commitment event is chosen by the strategy for some matching execution,
which then takes the corresponding step.  (2) By induction on the execution.
After each step, internal steps included, the history adds the step's event,
if any, and then, if the step lost reachable rows, the \commitment event
of~(3).  (3) The \commitment of the basis for the matching execution applies
by the argument above.
\end{proof}

Item~(3) reflects the one difference between the two models.  In the
automaton, a reply and the loss of reachable rows it causes happen in one
step.  In the system, the loss is a separate \commitment event, so an input
may arrive in between, and the \commitment then applies after that input.

The \commitment events record choices that traces do not show.  Consider an
automaton that, after the first request, takes an internal step to one of two
states.  Both reply $x$, but only the second can then reply $z$.  A second
automaton has only the second state.  The two automata have the same traces,
and therefore the same projected spec.  In the first, the internal step to the
first state loses reachable rows, so its system removes the rows with reply
sequence $xz$ in a \commitment event.  The system of the second automaton never
does.

\paragraph{When the automaton itself is confluent.}
\begin{definition}[Confluent automaton]
\label{def:confluent-automaton}
An automaton is \emph{confluent} when, for any finitely many finite executions
and any receive sequence $j$ that \inrel the receive sequence of each, there
is a reply sequence $u$ such that each of the executions can continue to an
execution with receive sequence $j$ and reply sequence $u$.
\end{definition}

\begin{proposition}[Confluent automata]
\label{prop:automaton-confluence}
Let $M$ be an automaton.
\begin{enumerate}
  \item $M$ is confluent if and only if, for every finite family of histories
        of its system and every common \indesc $j$ of their current inputs,
        the histories can be continued to a common state at $j$ by events
        that may apply any \commitments.
  \item The confluence of $M$ can be established by a row at a common
        \indesc $j$ that survives in all of the executions.  For an execution
        of $M$ with receive sequence $w$ and reply sequence $t$, two sets of
        rows determine which rows survive.  One is the remainder table
        $\widehat R$ at the end of its run of $P_M$
        (Proposition~\ref{prop:projection}), which consists of the rows
        $(k,u)$ of $P_M$ with $k$ extending $w$ and $u$ extending $t$.  The
        other is the set $\Delta$ of rows that the \commitment events of its
        history in Proposition~\ref{prop:automaton-system}(2) removed.  That
        history ends with the table $\widehat R\setminus\Delta$, so $M$ is
        confluent exactly when, for any finitely many executions and any
        common \indesc $j$ of their receive sequences, some row at $j$ lies in
        every $\widehat R$ and in no $\Delta$.
  \item If $M$ is confluent, the runs of $P_M$ are consistent.  The converse
        does not hold in general.
\end{enumerate}
\end{proposition}

\begin{proof}[Proof sketch]
(1) By Proposition~\ref{prop:automaton-system}, each history's table contains
the reachable rows of a matching execution, and each execution has a history
whose table is exactly its reachable rows.  By
Proposition~\ref{prop:system-confluence}(2), histories are confluent at $j$
exactly when some row at $j$ lies in all of their tables.  First assume $M$ is
confluent.  For any finitely many histories, a row $(j,u)$ reachable from all
of their matching executions lies in every table, so the histories are
confluent.  Now assume every finite family of histories is confluent, and take
finitely many executions.  Their histories from
Proposition~\ref{prop:automaton-system}(2) have a common row at $j$, which is
reachable from every one of the executions, so $M$ is confluent
(Definition~\ref{def:confluent-automaton}).  (2) The description of
$\widehat R$ is Lemma~\ref{lem:run-residual}, the table $\widehat R\setminus\Delta$
is Proposition~\ref{prop:system-confluence}(1), and the condition is then~(1)
restated with Proposition~\ref{prop:system-confluence}(2).  (3) Suppose $M$ is
confluent, and take executions with rows $(w_1,t_1),\ldots,(w_n,t_n)$ and a
common \indesc $j$ of the $w_r$.  Confluence gives a row $(j,u)$ reachable from
each, and $u$ extends each $t_r$.  This is the row-consistency condition
(Proposition~\ref{prop:common-expansion-form}), so by Complete CALM the runs of
$P_M$ are consistent.  The registrar is a counterexample to the converse.
\end{proof}

For the registrar, replying early loses every reachable row with reply
sequence $xz$, so its system removes them in a \commitment event.  An execution
that replied early and one that waited and replied $xz$ then have no common
reachable row at two ballots, and the registrar is not confluent, although the
runs of its projected spec are consistent.  Closing the polls is a \commitment the projected
spec does not require: a canonical run can emit $x$ after one ballot and still
emit $xz$ after the second.

A failure can also come from the projected spec itself.  Consider a counter
whose report carries the count at the moment the report fires.  After two
arrivals, one execution has reported $1$, before the second arrival, and
another has reported $2$.  Neither reply sequence is a prefix of the other, so
the projected spec does not satisfy Lifted Monotonicity.  A counter that records the count when each
request arrives, and later reports from that record, sends the reports
$1,2,\ldots,m$ in order.  Its projected spec satisfies Lifted Monotonicity, and the counter
itself is confluent.

\paragraph{Testing and proof.}
The projected spec and the system of an automaton are defined by every finite
execution of the automaton.  There are infinitely many such executions, so
both are semantic objects: we can reason about them, but not list them.  Two
standard approaches make that reasoning practical.

The first is testing.  A simulator or model checker runs the automaton on
bounded inputs~\cite{misra1986des,flanagan2005dpor}; the pairs of receive and
reply sequences it collects form a finite spec in our model, which can be
checked for violations of Lifted Monotonicity.  This is routine engineering
practice and has found serious bugs in production
systems~\cite{newcombe2015aws,zhou2021foundationdb}.  Testing suits
refutation.  If the projected spec does not satisfy Lifted Monotonicity, two
executions show it:
they have the same receive sequence, and neither reply sequence is a prefix of
the other.  A violation of the automaton's own confluence is shown by finitely
many executions and a receive sequence at which they have no common reachable
row, and confirming that they have none means exploring all of their
continuations to that receive sequence.  Testing cannot show that Lifted
Monotonicity holds, since that is a claim about every input.

The second approach is proof.  Verification work routinely proves that an
automaton refines a spec: every trace of the automaton is allowed by the spec.
Such proofs use invariants, properties of the automaton's state that hold
initially and are preserved by every step, so they are checked one step at a
time and never by enumerating executions.  They appear as forward
simulations~\cite{lynch1995forward} and refinement
mappings~\cite{abadi1991refinement}, and are mechanized for practical
distributed systems~\cite{hawblitzel2015ironfleet}.

A \emph{declared} spec $D$ gives, for each receive sequence, a set of reply
sequences by a rule; it does not enumerate executions.  To check Lifted
Monotonicity on $D$ instead of on the projected spec, $D$ must satisfy two
properties, and both are needed: if $D$ omits a reply sequence the automaton
can produce, it may hide an inconsistency, and if $D$ adds one the automaton cannot produce, it may create an inconsistency the automaton does not have.
\begin{itemize}
  \item \emph{Realizability.}  Every reply sequence that $D$ allows at a
        receive sequence is the reply sequence of some execution with that
        receive sequence.
  \item \emph{Domination.}  Every execution's reply sequence is a prefix of
        one that $D$ allows at its receive sequence.
\end{itemize}
Domination is the counterpart, for receive and reply sequences, of the
refinement obligation above.  Realizability is the obligation our use adds:
$D$ allows nothing the automaton cannot produce.  It too is proved once for all
inputs, typically by induction on the input.  $D$ need not match the automaton
exactly: it may leave out reply sequences that others extend.  For the
recording counter, the rule ``after $n$ arrivals, the reports $1,\ldots,n$''
satisfies both obligations.

\begin{proposition}[Declared specifications]
\label{prop:declared-specification}
Let $D$ be an \spectype whose outputs are reply sequences and whose inputs
correspond, by an order isomorphism, to receive sequences.  If $D$ satisfies
realizability and domination, then $D$ satisfies Lifted Monotonicity exactly
when the projected spec does.
\end{proposition}

\begin{proof}[Proof sketch]
We use the row-consistency form.  Given rows of the projected spec, we
dominate each by a row of $D$, take their common \outnoun in $D$, and realize
it.  Conversely, rows of $D$ are rows of the projected spec, and their common
\outnoun there is dominated back into $D$.
\end{proof}

Together, the two obligations let $D$ stand in for the projected spec:
Lifted Monotonicity is checked on $D$ exactly as everywhere else in this paper.
Whether the automaton itself is confluent is a question about its executions
(Proposition~\ref{prop:automaton-confluence}).

\section{Discussion}
\label{sec:discussion}

Several observations that motivated CALM can now be explained
in context. We revisit them briefly, because the
trajectory of the work points to next steps.

The computing literature is rich with coordination mechanisms. Prior to the CALM
Conjecture, the justifications for these mechanisms were often situational. In
the best cases---serializability being a prime example---there is a crisp,
whole-system invariant, and coordination mechanisms provably enforce system
correctness. But mechanisms like locks, barriers and total orders are also used
liberally for guaranteeing local operational properties of concurrent data
structures or algorithms, and have often been used poorly in both systems and
academic literature.

MapReduce~\cite{dean2004mapreduce} conflated barriers with communication. In database systems, query
engines typically pipeline data between operators in a query plan, and wait only
where a non-monotone operator, such as aggregation, needs all of its input to
produce an output~\cite{wilschut1991dataflow,graefe1990volcano}. MapReduce
instead tightly coupled barrier coordination into every communication phase,
so even monotone parallel operators like join had to wait for coordination.
This is not merely a software engineering decision: parallel complexity models
that count rounds, including Valiant's BSP model, also conflate communication and
coordination~\cite{valiant1990bsp,karloff2010mrc}. It was clear to the database
community that this was not intrinsically
required~\cite{dewitt2008mapreduce,condie2010online}.

Rather than repairing a specific model, the CALM Conjecture asked the underlying
question: when is coordination inherent, and why? The question was
first addressed in the context of declarative
networking~\cite{loo2009declarative}, another query processing setting. That
research agenda took Codd's tradition of declarative
specification~\cite{codd1970relational} and extended it to a broader domain,
demonstrating that full distributed systems could be implemented in declarative
languages, with practical
payoffs~\cite{loo2005overlays,loo2009declarative,chu2007declarative,alvaro2010boom,alvaro2011bloom}.
The generality of declarative networking applications, combined with query
processing intuition, suggested the answer to the CALM question: coordination
in any system is required if and only if the
system's declarative specification is non-monotone. Ameloot et al.\ proved it
for relational languages~\cite{ameloot2013relational}.

This paper raises the abstraction level yet further: from logic programs to the
\spectype[s] of any system. All that matters about a spec is which outputs it
permits for each input; the analysis needs no language or programming model,
declarative or otherwise. By going a level higher to specs, we not only set
aside the lack of mainstream adoption of declarative languages; we also capture
families of guarantees, like CAP and isolation levels, that have no application
or protocol in mind.

Raising the level to specs also enables deeper analysis.
Bloom~\cite{alvaro2011bloom} flagged syntactic non-monotonicity in programs and
conservatively injected coordination code in front of it. Coordination code is
itself non-monotone, so this only adds more non-monotone syntax, leaving the
puzzle of Section~\ref{sec:composition}. Later tools moved from syntax to
declared contracts and invariants in specific
contexts~\cite{sivaramakrishnan2015quelea,kaki2018alone}, picking for each
operation the weakest consistency or isolation level a store offers that still
satisfies them. We ask what a spec itself requires, independent of any menu of
levels, and by Observable equivalence (Lemma~\ref{lem:canonical-equivalent}) the
answer holds for every mechanism. Our work is also constructive on both sides of the
boundary (\S\ref{sec:composition}).

Underlying these issues is a reassessment of Codd's declarative tradition: how
``declarative'' can a specification be? How little structure does it need to be
useful? The question is timely. Coding agents today translate from natural
language to imperative code, which is often rife with bugs. This has renewed
interest in interposing formal specification languages like
Lean~\cite{demoura2021lean4} as a narrow waist in the agentic coding pipeline~\cite{hellerstein2023codd}.
Formal languages couple heuristic
AI-driven inference to deterministic proofs of correctness: neurosymbolic
computing~\cite{an2025neurosymbolic}. This presents a renewed
opportunity for practical declarative programming languages, including logic,
dataflow, and functional languages. These languages still specify protocols and
programs in a relatively high degree of detail, as in declarative networking.
Our \spectype[] model and its composition point to the possibility of
specification at a higher level still, suggesting we could declare far less detail
and still get the neurosymbolic advantage of heuristic search for
provably correct implementations.

\section{Conclusion}
\label{sec:conclusion}

The simplicity of \spectype[s] provides a unified foundation for the
literature on coordination. It also clarifies the intrinsic purpose of
coordination: removing possibilities via \commitments. The question of
coordination-free consistency becomes clear in this model: when must a system
make a fateful decision---a \commitment---in order to emit an output? Functional CALM answers
the question for functional specs:  runs stay confluent as input arrives exactly
when specs are monotone; otherwise \commitment is required.  This generalizes CALM to any function between ordered domains.

Complete CALM extends the answer to specs that offer several outputs
for one input.  The test is Lifted Monotonicity, which combines Relational
Monotonicity with directed fibers.  When the test fails, some reachable output
requires a \commitment to rule out a possible input/output pair.
Because the model makes \commitments explicit, we can also show where a
composition coordinates, and when that coordination is tight (\S\ref{sec:composition}).  For example, a totally ordered log
and a key-value store are each non-monotone on their own, yet when the log
guards the store, the log's \commitments are all the composition needs, and
the store needs none.  Multi-Paxos is one way to distribute those
\commitments.
Complete CAP is the partition form of the same equivalence: a spec escapes
the CAP tradeoff exactly when it satisfies Lifted Monotonicity for its explicit possibility relation.

Our model of \spectype[s] and \commitments opens a range of new questions.
One direction
extends the boolean criterion of this paper to a quantity: a Determination
Complexity theory that assesses \emph{how many} \commitments are intrinsically
needed to determine a spec. Another is Determination Provenance: a
given output---or its absence---is determined not only by inputs, but also by
the \commitments made in a run, which brings formal models of provenance to
concurrent systems. A third direction is to provide general-purpose static
analysis and synthesis of tightly coordinated systems in practical programming frameworks.

\appendix
\renewcommand{\theHsection}{appendix.\Alph{section}}
\section{Notation and Terms}
\label{app:notation}

The tables collect the paper's notation and defined terms, one per row, in the
order the sections introduce them, with a pointer to each definition.

\paragraph{Relational notation.}
To describe how execution narrows $P$, we use a compact relational notation.
For any relation
$R\subseteq I\times O$, we select rows by writing a predicate on its
columns $\mathrm{in}$ and $\mathrm{out}$ in square brackets: $R[\mathrm{in}\InputGe j]$ is the table of rows
whose input \inrel $j$, $R[\mathrm{out}\OutputGe o]$ is the table of rows
whose output \outrel $o$, and $R[\mathrm{in}\in K]$ is
the table of rows with inputs in $K$.  Commas separate conjuncts:
$R[\mathrm{in}\InputGe j,\ \mathrm{out}\OutputGe o]$ is the table of rows
satisfying both predicates.
To project onto a single column, we write $\pi_{\mathrm{in}}(R)$ and
$\pi_{\mathrm{out}}(R)$ for the sets of inputs and outputs that appear in $R$.
We say that a row $(k,u)$ is \emph{at input} $k$; $R[\mathrm{in}=k]$ is the table of
rows at $k$.

{\footnotesize
\setlength{\tabcolsep}{3pt}
\renewcommand{\arraystretch}{1.15}

\noindent
\begin{tabular}{@{}>{\raggedright\arraybackslash}p{.33\columnwidth}>{\raggedright\arraybackslash}p{.5\columnwidth}>{\raggedright\arraybackslash}p{.11\columnwidth}@{}}
  \toprule
  \multicolumn{3}{@{}l}{\textbf{Functional CALM (\S\ref{sec:functional-calm})}}\\
  \midrule
  functional \spectype & $\Spec=(I,\InputOrd,O,\OutputOrd,P)$ with $P$ total on $I$ and one output per input; the induced function $f$ & Def.~\ref{def:functional-specification}\\
  possibility relation $P$ & the spec's table of permitted rows $(i,o)$ & Def.~\ref{def:functional-specification}\\
  $i\InputOrd j$; $j$ \inrel $i$ & $j$ contains at least the input information in $i$; $j$ is an \indesc of $i$ & Sec.~\ref{sec:functional-specifications}\\
  $o\OutputOrd u$; $u$ \outrel $o$ & $u$ subsumes every externally visible claim in $o$; $u$ is an \outnoun of $o$ & Sec.~\ref{sec:functional-specifications}\\
  common \indesc & an input that \inrel each given input & Sec.~\ref{sec:functional-specifications}\\
  common \outnoun & an output that \outrel each given output & Sec.~\ref{sec:functional-specifications}\\
  confluent at $j$ (rows) & the rows have a common \outnoun at $j$, a common \indesc of their inputs & Sec.~\ref{sec:functional-specifications}\\
  row-consistent & confluent at every common \indesc of their inputs & Sec.~\ref{sec:functional-specifications}\\
  $(k,u)$ \outrelpast by $(k,w)$ & $w$ \outrel $u$, for two rows at the same input & Sec.~\ref{sec:functional-specifications}\\
  monotone (functional) & $i\InputOrd j\Rightarrow f(i)\OutputOrd f(j)$ & Def.~\ref{def:functional-monotonicity}\\
  selection $R[\ldots]$ & the rows of $R$ that satisfy a predicate on $\mathrm{in}$ and $\mathrm{out}$; commas separate conjuncts & Sec.~\ref{sec:functional-specifications}\\
  projection $\pi_{\mathrm{in}}(R)$ & the set of inputs that appear in $R$ & Sec.~\ref{sec:functional-specifications}\\
  projection $\pi_{\mathrm{out}}(R)$ & the set of outputs that appear in $R$ & Sec.~\ref{sec:functional-specifications}\\
  row at $k$ & a row $(k,u)$ & Sec.~\ref{sec:functional-specifications}\\
  specification state $(i,R)$ & current input $i$ and remainder table $R\subseteq P$, with a row at $i$ and only rows at \indescs of $i$ & Def.~\ref{def:residual-state}\\
  legal input & an input $k$ at which $R$ has a row & Def.~\ref{def:legal-input}\\
  canonical state $\sigma_P(i)$ & $(i,P[\mathrm{in}\InputGe i])$ & Eq.~\ref{eq:canonical-state}\\
  fiber $R(k)$ & $\pi_{\mathrm{out}}(R[\mathrm{in}=k])$, the outputs paired with $k$ & Def.~\ref{def:fiber}\\
  upper fiber at $k$ & $\pi_{\mathrm{out}}(R[\mathrm{in}\InputGe k])$ & Def.~\ref{def:fiber}\\
  input event $\mathsf{input}(j)$ & keeps $R[\mathrm{in}\InputGe j]$ and makes $j$ the current input & Eq.~\ref{eq:input-event}\\
  input loss & the rows an input event removes & Sec.~\ref{sec:functional-events}\\
  legal output & an output $o$ at $(i,R)$ that every output in the upper fiber \outrel & Def.~\ref{def:exposable}\\
  \commitment & a partial function on states that removes a nonempty set of rows, none \outrelpast by a remaining row at the same input & Def.~\ref{def:commitment}\\
  output event $\mathsf{output}(o)$ & keeps $R[\mathrm{out}\OutputGe o]$ and emits $o$ & Def.~\ref{def:output-event}\\
  canonical \commitment $x_o$ & the removal an output event makes when $o$ is not legal & Def.~\ref{def:canonical-commitment-functional}\\
  determined state & a state whose upper fiber has one element & Def.~\ref{def:determined}\\
  specification event (event) & an input or output event; from Section~\ref{sec:systems} on, also a \commitment event & Def.~\ref{def:specification-event}\\
  canonical run (run) & a finite sequence of specification events from a canonical state & Def.~\ref{def:run}\\
  \commitment-free run & a run none of whose output events applies a \commitment & Sec.~\ref{sec:executions}\\
  reachable without \commitment & reached by some \commitment-free run & Sec.~\ref{sec:executions}\\
  confluent at $j$ (runs) & each run of a finite family can be continued to one common state with current input $j$ & Def.~\ref{def:confluence}\\
  consistent (runs) & every finite family of runs is confluent at every common \indesc of their current inputs & Def.~\ref{def:consistency}\\
  \bottomrule
\end{tabular}

\medskip
\noindent
\begin{tabular}{@{}>{\raggedright\arraybackslash}p{.33\columnwidth}>{\raggedright\arraybackslash}p{.5\columnwidth}>{\raggedright\arraybackslash}p{.11\columnwidth}@{}}
  \toprule
  \multicolumn{3}{@{}l}{\textbf{Complete CALM (\S\ref{sec:complete-calm})}}\\
  \midrule
  \spectype & $\Spec=(I,\InputOrd,O,\OutputOrd,P)$ with $P\subseteq I\times O$ total on $I$ & Def.~\ref{def:relational-specification}\\
  nondeterministic & some fiber has more than one output & Def.~\ref{def:nondeterministic}\\
  Relational Monotonicity & every output keeps an \outnoun as input \ingrows & Def.~\ref{def:relational-monotonicity}\\
  directed fiber & every nonempty finite subset has a common extension in the fiber & Def.~\ref{def:fiber-directedness}\\
  directed fibers & every fiber of the spec is directed & Def.~\ref{def:fiber-directedness}\\
  joint refinement lift $A\LiftedOrd B$ & every nonempty finite $X\subseteq A$ has a common extension in $B$ & Eq.~\ref{eq:joint-output-order}\\
  Lifted Monotonicity (monotone) & $i\InputOrd j\Rightarrow\fiber{i}\LiftedOrd\fiber{j}$ & Def.~\ref{def:lifted-monotonicity}\\
  row-consistency condition & rows with a common \indesc $j$ have a common \outnoun in $\fiber{j}$ & Def.~\ref{def:common-expansion-condition}\\
  ruled out by $o$ & a row whose output does not \outrelbase $o$ & App.~\ref{app:lifted-output-events}\\
  $\Subsumed_o$ & removed rows $(k,x)$ for which some output at $k$ \outrel both $x$ and $o$ & Eq.~\ref{eq:output-subsumptions}\\
  $\Excluded_o$ & all other removed rows & Eq.~\ref{eq:commitment-removals}\\
  output subsumption & removing the rows of $\Subsumed_o$, which is not a \commitment & Def.~\ref{def:lifted-output-event}\\
  canonical \commitment $x_o$ (general) & removes the rows of $\Excluded_o$ & Def.~\ref{def:canonical-commitment}\\
  legal output (general) & an output $o$ with $\Excluded_o$ empty & Sec.~\ref{sec:lifted-output-events}\\
  \bottomrule
\end{tabular}

\medskip
\noindent
\begin{tabular}{@{}>{\raggedright\arraybackslash}p{.33\columnwidth}>{\raggedright\arraybackslash}p{.5\columnwidth}>{\raggedright\arraybackslash}p{.11\columnwidth}@{}}
  \toprule
  \multicolumn{3}{@{}l}{\textbf{Distributed execution (\S\ref{sec:distributed-execution})}}\\
  \midrule
  participant & holds a specification state and exchanges messages & Sec.~\ref{sec:histories}\\
  message $m=(j,d)$ & a cumulative input $j$, or $\varnothing$, and other data $d$ & Sec.~\ref{sec:histories}\\
  distributed history & history events with a total local order at each participant and a matching send for each receive & Def.~\ref{def:distributed-history}\\
  history event & a specification, send, or receive event at one participant & Def.~\ref{def:distributed-history}\\
  happens-before & transitive closure of the local orders and of send before receive & Def.~\ref{def:happens-before}\\
  consistent cut & a set of history events closed under happens-before & Def.~\ref{def:consistent-cut}\\
  participant run & $p$'s specification events at a cut $C$, in local order, ending at $\sigma_p(C)$ & Def.~\ref{def:participant-run}\\
  endpoint & $(\sigma_p(C))_{p\in\mathcal N}$ & Def.~\ref{def:endpoint}\\
  continuation & a history that extends the events of a cut with later events & Sec.~\ref{sec:distributed-consistency}\\
  final cut & all of a history, as a consistent cut & Sec.~\ref{sec:distributed-consistency}\\
  cut confluent at $j$ & some continuation brings every participant to one state $(j,R)$ & Sec.~\ref{sec:distributed-consistency}\\
  consistent (distributed histories) & every consistent cut of every history is confluent at every common \indesc & Def.~\ref{def:distributed-consistency}\\
  \commitment event $\mathsf{apply}(\gamma)$ & applies a nonempty sequence $\gamma$ of \commitments and emits nothing & Def.~\ref{def:commitment-event}\\
  \commitment basis $\mathcal B$ & a set of \commitments & Def.~\ref{def:commitment-basis}\\
  run under $\mathcal B$ & input, output, and \commitment events from a canonical state, applying only \commitments in $\mathcal B$ & Def.~\ref{def:run-under-basis}\\
  local history & $p$'s history events in order & Def.~\ref{def:local-history}\\
  strategy $\theta_p$ & maps a local history to the output, \commitment, and send events $p$ may perform next & Def.~\ref{def:strategy}\\
  system & a \commitment basis and a strategy per participant & Def.~\ref{def:system}\\
  history of a system & a history whose output, \commitment, and send events its strategies choose & Def.~\ref{def:produced-history}\\
  observably equivalent & same inputs and outputs, and the canonical run keeps every row the run keeps & Lem.~\ref{lem:canonical-equivalent}\\
  $\Delta(\rho)$ & rows of the final table of $\rho$'s canonical run that a \commitment event of $\rho$ removed & Sec.~\ref{sec:commitment-conflict}\\
  explicit possibility relation $E$ & $E\subseteq P$, the rows whose outputs participants may actually emit & Def.~\ref{def:explicit-relation}\\
  partitioned history & participants divided into two or more sides, with no message across sides & Def.~\ref{def:partitioned-history}\\
  side & one of the two or more classes into which a partitioned history divides the participants & Def.~\ref{def:partitioned-history}\\
  same-side continuation for $p$ & extends a partitioned history with events on $p$'s side alone, with no new input and no message from another side & Def.~\ref{def:local-continuation}\\
  Availability for $E$ & each participant can emit every output in $R(i)\cap E(i)$ in a same-side continuation & Def.~\ref{def:partition-notions}\\
  Partition Consistency & the final cut of every partitioned history is confluent at every common \indesc & Def.~\ref{def:partition-consistency}\\
  canonical system for $E$ & participants follow canonical runs and may emit any output of $E$ their state permits & Sec.~\ref{sec:complete-cap}\\
  escapes CAP for $E$ & the canonical system for $E$ has Availability for $E$ and Partition Consistency & Def.~\ref{def:escapes-cap}\\
  Lifted Monotonicity for $E$ & every finite set of rows of $E$ with a common \indesc $j$ has a common \outnoun in $\fiber{j}$ & Def.~\ref{def:lm-for-A}\\
  emitted rows & the pairs $(i,o)$ such that a run emits $o$ at current input $i$ & Def.~\ref{def:emitted-rows}\\
  deployment & one explicit relation $E_p$ per participant & Def.~\ref{def:deployment}\\
  response obligation & per participant, the outputs that count as answers, upward closed within fibers & Def.~\ref{def:response-obligation}\\
  response interface & a total, self-continuing explicit relation $T\subseteq Q$ & Def.~\ref{def:response-interface}\\
  \bottomrule
\end{tabular}

\medskip
\noindent
\begin{tabular}{@{}>{\raggedright\arraybackslash}p{.33\columnwidth}>{\raggedright\arraybackslash}p{.5\columnwidth}>{\raggedright\arraybackslash}p{.11\columnwidth}@{}}
  \toprule
  \multicolumn{3}{@{}l}{\textbf{Composition (\S\ref{sec:composition})}}\\
  \midrule
  sequential composition $P_2\circ P_1$ & $\{(i,o):\exists m.\ (i,m)\in P_1\land(m,o)\in P_2\}$, for stages sharing $(M,\preceq_M)$ & Eq.~\ref{eq:spec-composition}\\
  $\Spec_1$ guards $\Spec_2$ & $\Spec_2$ takes the outputs of $\Spec_1$ as its inputs & Sec.~\ref{sec:sequential-composition}\\
  restriction $\Spec_2|_{\Spec_1}$ & $\Spec_2$ restricted to the outputs of $\Spec_1$ & Sec.~\ref{sec:tight-coordination}\\
  Lifted Monotonicity on $D$ & the restriction of $P_2$ to inputs in $D$ satisfies Lifted Monotonicity & Sec.~\ref{sec:tight-coordination}\\
  sufficient guard & $\Spec_2|_{\Spec_1}$ needs no \commitment: every output at every state it reaches without \commitment is legal & Def.~\ref{def:sufficient-guard}\\
    $m\sqsubseteq m'$ ($m'$ covers $m$) & every finite set of outputs in $P_2(m)$ has a common \outnoun in $P_2(m')$ & Sec.~\ref{sec:tight-coordination}\\
  necessary guard & rows of $P_1$ not confluent under $\preceq_M$ have answers not confluent in $P_2\circ P_1$ & Def.~\ref{def:necessary-guard}\\
  tight guard & a sufficient and necessary guard & Def.~\ref{def:tight-guard}\\
  intersection $\Spec_1\cap\Spec_2$ & table $P_1\cap P_2$, for specs sharing their ordered domains & Def.~\ref{def:intersection}\\
  component & a \commitment basis and a set of admitted runs & Def.~\ref{def:component}\\
  admitted run & a run under the basis from $\sigma(i)$, with an output $o$ it emits and $(i,o)\in P$ & Def.~\ref{def:component}\\
  typed pipeline & two components joined by feeds & Def.~\ref{def:typed-pipeline}\\
  feed & an admitted run of $C_1$ with output $m$, paired with an admitted run of $C_2$ from $\sigma(m)$ & Def.~\ref{def:typed-pipeline}\\
  external observation & the pair $(i,o)$ that a feed shows & Def.~\ref{def:typed-pipeline}\\
  contract-sufficient & $\Spec_1$ is a sufficient guard for $\Spec_2$ & Def.~\ref{def:sufficient-pipeline}\\
  implementation-sufficient & no downstream run of a feed applies a \commitment & Def.~\ref{def:sufficient-pipeline}\\
  \bottomrule
\end{tabular}

\medskip
\noindent
\begin{tabular}{@{}>{\raggedright\arraybackslash}p{.33\columnwidth}>{\raggedright\arraybackslash}p{.5\columnwidth}>{\raggedright\arraybackslash}p{.11\columnwidth}@{}}
  \toprule
  \multicolumn{3}{@{}l}{\textbf{Related work and mappings (\S\ref{sec:instances})}}\\
  \midrule
  read rule & a replica answers a read with the highest-timestamped committed write of the item that it has received & Sec.~\ref{sec:hat-instance}\\
  HAT specification $\Spec_L$ & table $P_L$: rows $(i,o)$ for which some complete history extends $i$, includes $o$, and exhibits none of the phenomena that $L$ proscribes & Def.~\ref{def:hat-spec}\\
  store specification $\Spec_{\mathrm{st}}$ & any subset of the read rule's answers & Def.~\ref{def:store-spec}\\
  invariant confluent & the merge of any two reachable states satisfies $\mathcal I$ & Sec.~\ref{sec:iconfluence-instance}\\
  $\Spec_{\mathcal I}$ & transaction histories as inputs; the states of valid closed sub-histories as outputs & Sec.~\ref{sec:iconfluence-instance}\\
  receive sequence & an execution's input actions in order, ordered by prefix & Sec.~\ref{sec:operational-models}\\
  reply sequence & an execution's output actions in order, ordered by prefix & Sec.~\ref{sec:operational-models}\\
  projection $P_M$ & rows $(w,t)$ for which some finite execution of $M$ has receive sequence $w$ and reply sequence $t$ & Def.~\ref{def:projection}\\
  reachable rows & the rows of $P_M$ that a finite execution can still reach & Sec.~\ref{sec:operational-models}\\
  system of $M$ & one participant over $P_M$; its basis adds, for each execution, the \commitment that keeps only that execution's reachable rows & Def.~\ref{def:automaton-system}\\
  confluent automaton & any finitely many executions can continue to executions with one common receive sequence $j$ and one common reply sequence & Def.~\ref{def:confluent-automaton}\\
  declared spec $D$ & reply sequences given by a rule for each receive sequence & Sec.~\ref{sec:operational-models}\\
  realizability & every reply sequence $D$ allows at a receive sequence is that of some execution with that receive sequence & Sec.~\ref{sec:operational-models}\\
  domination & every execution's reply sequence is a prefix of one that $D$ allows at its receive sequence & Sec.~\ref{sec:operational-models}\\
  \bottomrule
\end{tabular}
}

\section{Functional CALM: Proofs}
\label{app:proofs-sec2}

This appendix proves the results of Section~\ref{sec:functional-calm}.

\begin{proof}[Proof of Lemma~\ref{lem:run-residual}]
The run starts at the canonical state $\sigma_P(i_0)$, whose table is
$P[\mathrm{in}\InputGe i_0]$. Each input event $\mathsf{input}(j)$ keeps only
the rows with $\mathrm{in}\InputGe j$ (Equation~\ref{eq:input-event}), and each
output event $\mathsf{output}(o)$ keeps only the rows with
$\mathrm{out}\OutputGe o$ (Equation~\ref{eq:output-action}). The input
conditions reduce to the last one, $\mathrm{in}\InputGe i$, because each later
input is an extension of the earlier ones. One output condition remains for each
emitted output.
\end{proof}

\begin{proof}[Proof of Proposition~\ref{prop:confluence}]
First assume the runs are confluent at $j$.  The common state has a row at
$j$, and since rows are never added, that row belongs to every $R_r$.  Now
assume $(j,u)$ belongs to every $R_r$.  Then $\mathsf{input}(j)$ is defined at
every $(i_r,R_r)$, and $\mathsf{output}(u)$ is defined at
$(j,R_r[\mathrm{in}\InputGe j])$, since $u\in R_r(j)$.  By
Lemma~\ref{lem:run-residual}, $(j,u)\in R_r$ says that $u$ \outrel every output
emitted on run $r$, so after $\mathsf{input}(j)$ and $\mathsf{output}(u)$ the
same lemma leaves exactly
$P[\mathrm{in}\InputGe j,\ \mathrm{out}\OutputGe u]$, independent of the
starting run.  The runs are confluent at $j$.
\end{proof}

\begin{proof}[Proof of Lemma~\ref{lem:confluence-legality}]
We use legality in a form that covers both functional and nondeterministic
specs: $o$ is legal at $(i,R)$ when, for every row $(k,x)\in R$, some
$u\in R(k)$ \outrel both $o$ and $x$.  In a functional table the only row at
$k$ is $(k,x)$, so this says $x$ \outrel $o$
(Definition~\ref{def:exposable}); with several outputs per input it says that
$\Excluded_o$ is empty (Section~\ref{sec:lifted-output-events}).  An output
event $\mathsf{output}(o)$ from $(i,R)$ applies a \commitment exactly when $o$
is not legal at $(i,R)$, so the lemma's condition says that every output a run
can emit is legal at the state where it is emitted.

First assume the runs are consistent, and suppose, for contradiction, that an
output event of some run applies a \commitment.  The run ends at $(i,R)$
before that event, and $o\in R(i)$ is not legal there, so some $(k,x)\in R$
has no $u\in R(k)$ \outrelpart both $o$ and $x$.  We extend the run in two ways: by $\mathsf{output}(o)$, and by
$\mathsf{input}(k)$ followed by $\mathsf{output}(x)$, which is defined since
$(k,x)$ remains after $\mathsf{input}(k)$.  The input $k$ is a common \indesc
of the two current inputs.  A row at $k$ in both resulting tables would be a
row $(k,u)\in R$ with $u$ \outrelpart both $o$ and $x$, and there is none.  By
Proposition~\ref{prop:confluence}, the two runs are not confluent at $k$,
contradicting consistency.

Now assume no output event applies a \commitment, so every output that a run
can emit is legal, and let runs end at $(i_r,R_r)$, $r=1,\ldots,m$, with
common \indesc $j$.  Each $R_r$ keeps a
row at its current input, say $(i_r,u_r)$, and by Lemma~\ref{lem:run-residual}
$u_r$ \outrel every output emitted on run $r$.  The canonical state
$\sigma_P(i_r)$ is reached by the empty run and has $u_r$ as an output, so $u_r$
is legal there.  By totality there is some $v_0\in P(j)$.  For $r=1,\ldots,m$
in turn, $(j,v_{r-1})$ is a row of $\sigma_P(i_r)$, since $j\InputGe i_r$, so
legality of $u_r$ gives $v_r\in P(j)$ \outrelpart both $u_r$ and $v_{r-1}$.
Then $v_m$ \outrel every $u_r$.  By Lemma~\ref{lem:run-residual},
$(j,v_m)\in R_r$ for every $r$, since $j\InputGe i_r$ and $v_m$ \outrel every
output emitted on run $r$.  Proposition~\ref{prop:confluence} gives
confluence at $j$.
\end{proof}

\begin{proof}[Proof of Corollary~\ref{cor:state-by-input}]
By Lemma~\ref{lem:run-residual}, the run's remainder table is
$P[\mathrm{in}\InputGe i,\ \mathrm{out}\OutputGe o_1,\ \ldots,\
\mathrm{out}\OutputGe o_m]$, where each $o_\ell=f(i_\ell)$ for some
$i_\ell\InputOrd i$.  Every row whose input \inrel $i$ has the form $(k,f(k))$ with
$i_\ell\InputOrd k$, so monotonicity gives $o_\ell\OutputOrd f(k)$, and the
output conditions remove nothing.
\end{proof}

\begin{proof}[Proof of Proposition~\ref{prop:functional-exposure-boundary}]
First assume $f$ is monotone, and let a \commitment-free run reach $(i,R)$.
By Corollary~\ref{cor:state-by-input}, $(i,R)=\sigma_P(i)$, so the legal
$j\InputGe i$ are exactly the \indescs of $i$.  Monotonicity gives
$f(i)\OutputOrd f(j)$ at each, so $f(i)$ is a legal output
(Definition~\ref{def:exposable}).  Now assume every output at every state
reachable without \commitment is legal.  At $\sigma_P(i)$, which the empty run
reaches, legality of $f(i)$ gives $f(i)\OutputOrd f(j)$ for every \indesc $j$
of $i$, so $f$ is monotone.
\end{proof}

\begin{proof}[Proof of Theorem~\ref{thm:functional-calm}]
By Lemma~\ref{lem:confluence-legality}, the runs are consistent exactly when
no output event of a run applies a \commitment, that is, when every output that
a run can emit is legal.  This condition is equivalent to legality at every
state reachable without \commitment.  One direction is immediate, since such
states are reached by runs.  For the other, assume every output at every
state reachable without \commitment is legal.  By induction along any run,
each of its output events is then legal, so the run is \commitment-free, every
state it reaches is reachable without \commitment, and every output it can
emit is legal.  By
Proposition~\ref{prop:functional-exposure-boundary}, legality at every state
reachable without \commitment holds exactly when $f$ is monotone.
\end{proof}

\section{Complete CALM: Proofs and Details}
\label{app:proofs-sec3}

This appendix proves the results of Section~\ref{sec:complete-calm} and gives the formal
definition of the lifted output events of Section~\ref{sec:lifted-output-events}.

\subsection{Proofs}

\begin{proof}[Proof of Proposition~\ref{prop:common-expansion-form}]
First assume Lifted Monotonicity, and let $(i_1,o_1),\ldots,(i_k,o_k)$ be
rows with common \indesc $j$.  If $k=0$, any $u\in\fiber{j}$ will do, and
$\fiber{j}$ is nonempty since $\Spec$ is total.  Otherwise, since
$i_r\InputOrd j$, Relational Monotonicity gives, for each $r$, a
$u_r\in\fiber{j}$ with $o_r\OutputOrd u_r$.  The $u_r$ all lie in the one
fiber $\fiber{j}$, which is directed, so some $u\in\fiber{j}$ \outrel every
$u_r$, hence every $o_r$.  Now assume the row-consistency condition, and let
$i\InputOrd j$ and $X\subseteq\fiber{i}$ be nonempty and finite.  The rows $\{(i,x):x\in X\}$
have common \indesc $j$, so some $u\in\fiber{j}$ \outrel all of $X$; that is
$\fiber{i}\LiftedOrd\fiber{j}$.
\end{proof}

The Survivor test (Proposition~\ref{prop:confluence}) and
Lemma~\ref{lem:run-residual} combine into the form that the later arguments
use.

\begin{lemma}[Reconciliation]
\label{lem:reconciliation}
Let runs end at $(i_r,R_r)$, $r=1,\ldots,m$, let $j$ be a common \indesc of
the $i_r$, and let $X$ be the set of outputs emitted on any of them.  The runs
are confluent at $j$ if and only if some $u\in\fiber{j}$ \outrel every output
in $X$.
\end{lemma}

\begin{proof}
By Lemma~\ref{lem:run-residual}, $(j,u)\in R_r$ exactly when $u\in\fiber{j}$
\outrel every output emitted on run $r$, and Proposition~\ref{prop:confluence}
gives the claim.
\end{proof}

\begin{proof}[Proof of Proposition~\ref{prop:exposure-boundary}]
By Proposition~\ref{prop:common-expansion-form} it suffices to treat the
row-consistency condition.  First assume it holds.  Let $(i,R)$ be reached by
a run (\commitment-free or not), and let $o\in R(i)$, $j\InputGe i$ legal, and
$x\in R(j)$.  The rows $(i,o)$ and $(j,x)$ of $P$ have common \indesc $j$, so
the condition supplies $u\in P(j)$ \outrelpart both $o$ and $x$.  By
Lemma~\ref{lem:run-residual}, $(j,u)\in R$, since $j\InputGe i$ and $u$
\outrel $o$, which \outrel every output emitted on the run.  So $u\in R(j)$
\outrel $o$ and $x$, and the legal-output test
(Proposition~\ref{prop:exposability}) makes $o$ a legal output.  Now assume
every output at every state reachable without \commitment is legal, and let
rows
$(i_1,o_1),\ldots,(i_k,o_k)$ have common \indesc $j$.  By totality there is some
$u_0\in P(j)$.  For $r=1,\ldots,k$ in turn, the legal-output test for $o_r$ at
$\sigma_P(i_r)$, applied with $x=u_{r-1}$, gives $u_r\in P(j)$ \outrelpart both
$o_r$ and $u_{r-1}$.  Then $u_k$ \outrel every $o_r$, and the
row-consistency condition follows.
\end{proof}

\begin{proof}[Proof of Theorem~\ref{thm:complete-calm}]
The proof of Lemma~\ref{lem:confluence-legality} covers nondeterministic specs,
so the runs are consistent exactly when no output event of a run applies a
\commitment, that is, when every output that a run can emit is legal.  As in
the proof of Theorem~\ref{thm:functional-calm}, this holds exactly when every output at every state reachable without \commitment is legal.  By
Proposition~\ref{prop:exposure-boundary}, that holds exactly when the spec
satisfies Lifted Monotonicity.
\end{proof}

\paragraph{Pairs suffice for the row-consistency condition.}
\label{app:pairs-suffice}
Suppose every two rows with a common \indesc $j$ have a common \outnoun in
$\fiber{j}$.  Then the condition holds for every finite family of rows with
common \indesc $j$.  We replace the first two rows of the family by a row
$(j,u)$ whose output $u$ \outrel both, and repeat until one row remains; a
single row, paired with any row at $j$, has an \outnoun in $\fiber{j}$.  The
pairwise hypothesis must cover all rows of $P$, including rows at $j$ itself,
since the replacement rows are at $j$.  Checking only rows at single-proposal
inputs does not suffice, even with families of up to $k$ rows.  In the
$k$-set agreement example of Appendix~\ref{app:additional-applications},
every family of at most $k$ rows at such inputs is compatible at each common
\indesc, but the $k+1$ proposal rows are not.

\subsection{Lifted output events}
\label{app:lifted-output-events}

Fix a state $(i,R)$ and an output $o\in R(i)$.  A row is \emph{ruled out by}
$o$ when its output does not \outrelbase $o$.  The ruled-out rows fall into two
disjoint sets:
\begin{align}
  \Subsumed_o&=\{(k,x)\in R\setminus R[\mathrm{out}\OutputGe o]:\notag\\
    &\qquad\quad x\OutputOrd u\text{ for some }u\in R[\mathrm{out}\OutputGe o](k)\},
      \label{eq:output-subsumptions}\\
  \Excluded_o&=(R\setminus R[\mathrm{out}\OutputGe o])\setminus\Subsumed_o.
      \label{eq:commitment-removals}
\end{align}

\begin{definition}[Lifted output event, output subsumption]
\label{def:lifted-output-event}
Let $(i,R)$ be a state and $o\in R(i)$.  The \emph{output event}
$\mathsf{output}(o)$ keeps exactly the rows whose outputs \outrelbase $o$ and
emits $o$, atomically, ending at $(i,R[\mathrm{out}\OutputGe o])$ as in
Equation~\ref{eq:output-action}.  It removes each row $(k,x)$ in
$\Subsumed_o$, for which some $u\in R(k)$ \outrel both $x$ and $o$, by
\emph{output subsumption}, which is not a \commitment.
\end{definition}

\begin{definition}[Canonical \commitment]
\label{def:canonical-commitment}
When $\Excluded_o\neq\varnothing$, the output event applies the \emph{canonical
\commitment} $x_o$:
\begin{equation}
  x_o(i,R)=\bigl(i,\ R\setminus\Excluded_o\bigr).
  \label{eq:canonical-commitment-relational}
\end{equation}
  Thus $x_o$ removes exactly the rows in $\Excluded_o$.  At an input $k$, if
  $\Excluded_o(k)=R(k)$, it rules out $k$; otherwise it rules out only the
  outputs in $\Excluded_o(k)$.  A row in $\Subsumed_o(k)$ has a kept \outnoun
  at $k$, so removing it by output subsumption does not rule out $k$.
\end{definition}

In Figure~\ref{fig:ambiguity-frames}(f), emitting $[A,B]$ places the row for
$[B,A]$ in $\Excluded_o$, so the event rules out that output at $\{A,B\}$.
In panel~(g), the removed route sets, at both inputs, are in $\Subsumed_o$,
each \outrelpast by a retained route set at its input, so the event subsumes them without
\commitment.  With one row per input, $\Subsumed_o$ is empty, and $x_o$ rules
out each input whose row is removed; this is the functional output event of
Definition~\ref{def:output-event}.  Apart from output subsumption and
\commitments, a row leaves the remainder table only by input loss
(\S\ref{sec:functional-events}).

Legality lifts with the event: $o$ is a legal output when $\Excluded_o$ is
empty, and then $\mathsf{output}(o)$ applies no \commitment, though it may
still subsume the rows in $\Subsumed_o$.  The legal-output test lifts as well:

\begin{proposition}[Legal-output test]
\label{prop:exposability}
An output $o\in R(i)$ is legal at $(i,R)$ if and only if, for every legal
$j\InputGe i$ and $x\in R(j)$, some $u\in R(j)$ \outrel both:
\begin{equation}
  o\OutputOrd u
  \qquad\text{and}\qquad
  x\OutputOrd u.
  \label{eq:safe-exposure}
\end{equation}
\end{proposition}

\begin{proof}
The output $o$ is legal exactly when $\Excluded_o$ is empty, that is, when
every ruled-out row $(j,x)$ is \outrelpast by some row $(j,u)$ \outrelpart
$o$.  Rows that already \outrelbase $o$ satisfy the test with $u=x$, so this
is the stated test.
\end{proof}

\section{Distributed Execution: Proofs and Details}
\label{app:proofs-sec4}

This appendix proves the results of Section~\ref{sec:distributed-execution}, supplies the
\commitment bases behind Section~\ref{sec:systems}, works out the bases of two-phase
locking and multiversion timestamp ordering, and extends Complete CAP to deployments and
to explicit relations that the system chooses.

\subsection{Proofs}

\begin{proof}[Proof of Proposition~\ref{prop:distributed-consistency}]
A cut is confluent at $j$ exactly when its participant runs are
(\S\ref{sec:distributed-consistency}).  First assume runs are consistent.  The
participant runs at any cut form a finite family of runs, so histories are
consistent.  Now assume histories are consistent.  Every finite family of runs
is the family of participant runs, at the final cut, of a message-free history
with one participant per run, so runs are consistent.
\end{proof}

\paragraph{Exchanging outputs.}
\label{app:exchanging-outputs}
Suppose participants share their \emph{outputs} in the data $d$ of messages
(\S\ref{sec:distributed-consistency}).  A participant that receives an
output $o$ its sender has emitted can do either or both of two things.  If
$o\in R(i)$ at its own current input $i$, it can emit $o$ itself, an ordinary
output event of Definition~\ref{def:distributed-history}.  It can also
silently honor $o$ by restricting its table to $R[\mathrm{out}\OutputGe o]$.
The restricted table may have no row at the participant's current input, so it
is not always a specification state.  Silent honoring is therefore not a
specification event, and we reason about it on tables.

Silent honoring never changes whether a cut is confluent at a common \indesc
$j$.  It cannot make a cut confluent at $j$, because it only removes rows, and
confluence asks for a row that survives everywhere
(Proposition~\ref{prop:confluence}).  It cannot make a cut non-confluent at
$j$ either.  A row that survives at $j$ already
\outrel every output emitted in the cut (Lemma~\ref{lem:run-residual}),
including $o$, so the restriction keeps the row.  The confluence of a cut
therefore depends on which outputs were emitted, not on who has heard about
them.

\begin{proof}[Proof of Proposition~\ref{prop:system-confluence}]
(1) By Lemma~\ref{lem:canonical-equivalent}, $R_r\subseteq\widehat R_r$.
Rows are never added, so a row that a \commitment event of $\rho_r$ removed is
not in $R_r$; hence $R_r\subseteq\widehat R_r\setminus\Delta(\rho_r)$.
Conversely, let $(k,u)\in\widehat R_r\setminus\Delta(\rho_r)$.  Both runs
start at the same canonical state, which contains $(k,u)$.  Consider an event of
$\rho_r$ at a state containing $(k,u)$.  If it is $\mathsf{input}(j)$, the
canonical run also takes $\mathsf{input}(j)$ and keeps $(k,u)$, so
$j\InputOrd k$ and the event keeps $(k,u)$.  If it is $\mathsf{output}(o)$,
then likewise $o\OutputOrd u$ and the event keeps $(k,u)$.  If it is a
\commitment event, it keeps $(k,u)$ because $(k,u)\notin\Delta(\rho_r)$.  By
induction on the events, $(k,u)\in R_r$.

(2) By~(1), $(k,u)\in R_r$ exactly when $(k,u)\in\widehat R_r$ and
$(k,u)\notin\Delta(\rho_r)$, so it suffices to show that the runs are confluent
at $k$ exactly when some row $(k,u)$ lies in every $R_r$.  First assume
continuations bring every run to a common state $(k,R)$.  That state has a row
$(k,u)$, and events never add rows, so $(k,u)\in R_r$ for every $r$.  Now
assume $(k,u)\in R_r$ for every $r$.  Since $i_r\InputOrd k$, the event
$\mathsf{input}(k)$ is legal at $(i_r,R_r)$ and keeps $(k,u)$, so
$\mathsf{output}(u)$ is then enabled.  It ends at $(k,R'_r)$ with
$R'_r=R_r[\mathrm{in}\InputGe k,\ \mathrm{out}\OutputGe u]$.  If $R'_r$ has no
row other than $(k,u)$, the run has reached $(k,\{(k,u)\})$.  Otherwise
consider the partial function that maps each state with current input $k$
whose table contains $(k,u)$ and some other row, and whose rows at $k$ all
\outrelbase $u$, to $(k,\{(k,u)\})$.  It keeps the current input and removes a
nonempty set of rows.  A removed row at $k$ is some $(k,x)$ with
$u\OutputOrd x$ and $x\neq u$, so $x\not\OutputOrd u$; a removed row at
another input has no remaining row at its input.  So the function is a
\commitment (Definition~\ref{def:commitment}), and a \commitment event that
applies it ends at $(k,\{(k,u)\})$.  Every run therefore reaches
$(k,\{(k,u)\})$.
\end{proof}

\begin{proof}[Proof of Proposition~\ref{prop:commitment-conflict}]
Events only remove rows, so $(k,x)$ is present at $\sigma_P(i)$.  Thus $k$
\inrel $i$ and $x\in\fiber{k}$, so $\mathsf{input}(k)$ and then
$\mathsf{output}(x)$ are enabled from $\sigma_P(i)$, and a canonical second
run exists.  Under any basis, after $\mathsf{output}(x)$ every row at $k$
\outrel $x$.  The \commitment leaves no such row at $k$
(Equation~\ref{eq:commitment-unsubsumed}), and if it removed every row at $k$,
there is none at all.  Later events only remove rows, so no row at $k$
survives in both runs, and by Proposition~\ref{prop:system-confluence} the
runs are not confluent at $k$.
\end{proof}

\begin{proof}[Proof of Theorem~\ref{thm:complete-cap}]
The canonical system for $E$ has the canonical basis and no \commitment
events.  Its strategies emit any output $o$ of $E$ in the current fiber by
the canonical output event $\mathsf{output}(o)$ of
Section~\ref{sec:complete-calm}, which applies the \commitment $x_o$ from the
canonical basis (Definition~\ref{def:run-under-basis}).  So the system has
Availability for $E$ whatever the spec.

First assume $\Spec$ satisfies Lifted Monotonicity for $E$.  Every output the
canonical system emits is the output of a row of $E$ whose input lies below
the current one.  So at a common \indesc $j$ of the current inputs, the
emitted outputs have a common \outnoun in $\fiber{j}$.  The system emits only
canonical output events, so Lemma~\ref{lem:reconciliation} makes this
confluence at $j$, and Partition Consistency follows.

Now assume some system, such as the canonical system for $E$, has
Availability for $E$ and Partition Consistency.  Let
$(i_1,o_1),\ldots,(i_n,o_n)$ be rows of $E$ with a common \indesc $j$.  We
find some $u\in\fiber{j}$ that \outrel every $o_r$.  If $n=0$, then $j$
appears in some row $(j,u)$, since $\Spec$ is total.  Otherwise, take one
participant $p_r$ per row, starting at input $i_r$, each on its own side and
running without messages.  A partitioned history has at least two sides, so if
$n=1$ we add a second participant for that row on another side.  By
Availability for $E$, each $p_r$ can emit $o_r$.  By Partition Consistency,
the participants then reach one state $(j,R)$, which has some row $(j,u)$
(Definition~\ref{def:residual-state}).  Rows are never added, so $(j,u)$ was
already in each $p_r$'s table just after it emitted $o_r$.  Emitting an output
keeps only rows whose output \outrel it
(Definition~\ref{def:lifted-output-event}), so $u$ \outrel every $o_r$, and
Lifted Monotonicity for $E$ follows.
\end{proof}

\begin{proof}[Proof of Corollary~\ref{cor:complete-cap-any-system}]
The converse argument in the proof of Theorem~\ref{thm:complete-cap} applies
to any system with Availability for $E$ and Partition Consistency, so $\Spec$
satisfies Lifted Monotonicity for $E$.  By Theorem~\ref{thm:complete-cap},
$\Spec$ escapes CAP for $E$.
\end{proof}

\subsection{\Commitment bases}
\label{app:bases}\label{app:commitment-bases}

This subsection shows that the canonical \commitments of
Sections~\ref{sec:functional-calm}--\ref{sec:complete-calm} are \commitments
(Definition~\ref{def:commitment}) and form a canonical basis, and proves
Observable equivalence (Lemma~\ref{lem:canonical-equivalent}).  Because
\commitments are partial functions, a basis may offer a \commitment only at the
states where its own precondition holds.

\begin{definition}[Canonical basis]
\label{def:canonical-commitment-basis}
The \emph{canonical \commitment basis} collects the canonical \commitments of
Definition~\ref{def:canonical-commitment},
\begin{equation}
  \mathcal B_P
    =\{x_o:o\in O,\ \operatorname{dom}(x_o)\neq\varnothing\},
  \label{eq:canonical-commitment-basis-relational}
\end{equation}
where $x_o$ is defined at the states with $o\in R(i)$ and $\Excluded_o\neq\varnothing$.
\end{definition}

\paragraph{The canonical \commitment is a \commitment.}
Definition~\ref{def:commitment} asks that $x_o$ keep the current input, which
it does, remove a nonempty set of rows, which $\Excluded_o\neq\varnothing$
ensures, and remove no row that a kept row at the same input \outrel
(Equation~\ref{eq:commitment-unsubsumed}).  Only the last needs an argument.
Suppose, for contradiction, that a removed row $(k,x)\in\Excluded_o$ is
\outrelpast by a row $(k,y)$ that $x_o$ leaves.  Either $y$ \outrel $o$, or
$y\in\Subsumed_o$ and is \outrelpast by a row \outrelpart $o$.  By
transitivity, $x$ is then \outrelpast by a row \outrelpart $o$, so
$x\in\Subsumed_o$, a contradiction.

Canonical runs (Definition~\ref{def:run}) are exactly the runs under the
canonical basis $\mathcal B_P$ of Definition~\ref{def:canonical-commitment-basis}
with no \commitment events.

\begin{proof}[Proof of Lemma~\ref{lem:canonical-equivalent}]
We drop the \commitment events from a run under $\mathcal B$.  By induction on
the events, the result is a canonical run whose table $\widehat R$ contains
the original's table $R$ after each event; both start at the same canonical
state.  A \commitment event keeps the current input and only removes rows, so
dropping it preserves $R\subseteq\widehat R$.  An input event
$\mathsf{input}(j)$ legal at $R$ is legal at $\widehat R$, and keeps
$R[\mathrm{in}\InputGe j]\subseteq\widehat R[\mathrm{in}\InputGe j]$.  An
output event $\mathsf{output}(o)$ at $R$ can be issued at $\widehat R$, since
$o\in R(i)\subseteq\widehat R(i)$, with $x_o$ from the canonical basis if it
applies one, and keeps
$R[\mathrm{out}\OutputGe o]\subseteq\widehat R[\mathrm{out}\OutputGe o]$.
\end{proof}

\begin{corollary}[Reconciliation under any basis]
\label{cor:reconciliation-any-basis}
Let runs under any \commitment bases end at $(i_r,R_r)$, and let $j$ be a
common \indesc of the $i_r$.  If some continuations, whose output and \commitment events
may apply any \commitments, bring them to a common state with current input $j$,
then some $u\in\fiber{j}$ \outrel every output emitted on every run.
\end{corollary}

\begin{proof}
The common state has a row $(j,u)$, and events never add rows, so
$(j,u)\in R_r$.  By Lemma~\ref{lem:canonical-equivalent}, $R_r$ is contained in
the table of a canonical run with the same inputs and outputs, and
Lemma~\ref{lem:run-residual} makes $u$ \outrel the outputs of run $r$.
\end{proof}

The converse of Lemma~\ref{lem:canonical-equivalent} fails: a canonical run
need not be observably equivalent to any run under a particular system's
basis.  That basis may lack the \commitment that an output event requires,
which disables that output.  Looking only at the runs that one system produces
could therefore hide non-monotonicity, and hide a canonical run that conflicts
with a \commitment.  We instead judge consistency, and the effect of
\commitments, against the canonical runs of the declared spec.  A behavioral
restriction narrows that comparison only when it is part of the declared
contract.

\subsection{\Commitment bases for two-phase locking and multiversion timestamp ordering}
\label{app:cc-bases}

This subsection gives in full the serializability spec of
Sections~\ref{sec:ambiguity-story} and~\ref{sec:systems}, and the \commitment
bases of two-phase locking (2PL) and multiversion timestamp ordering (MVTO).
Schedulers report results while transactions are still running, such as a
read result before its reader commits, so an output may leave transactions
undecided.

\paragraph{Specification.}
An input is a finite set of transactions, each given by the sequence of
operations that its client has issued so far: reads and writes of keys,
possibly followed by a request to commit.  One input precedes another when
every transaction of the first is in the second and its sequence in the first
is a prefix of its sequence in the second; once a transaction has requested
commit, its sequence no longer grows.  An input does not order the operations
of different transactions.  Their arrival order is the order of input events
in a run, which the environment controls.

An output has three parts: the committed transactions, in a serial order; the
aborted transactions; and the read results reported so far.  A read result
says which version the $k$-th operation of a transaction, a read of a key $x$,
returned: the write of $x$ by some transaction (its last write of $x$), or the
initial version.  A transaction that is
neither committed nor aborted is \emph{undecided}, and its reads may already be
reported.  One output \outrel another when it extends all three parts.

An output is \emph{complete} at an input when it commits or aborts every
transaction of the input, and not both, and commits only transactions whose
commit request is in the input.  A complete output is permitted when it meets two
conditions.  First, each reported read result is for a read in the input and
returns the initial version or the write of its key by a transaction of the
input.  Second, each read of a key $x$ by a committed transaction $T$ returns
$T$'s own write of $x$ if $T$ writes $x$ earlier in its sequence.  Otherwise it
returns the write of $x$ by the last transaction before $T$ in the serial order
that writes $x$, or the initial version if none does.  Reads by aborted
transactions are otherwise unconstrained.  Any output is permitted when some permitted complete output
\outrel it, that is, when it can still be completed.  The figures of
Sections~\ref{sec:ambiguity-story}--\ref{sec:systems} show only complete
outputs without aborts or read results, as serial orders, at inputs in which
every transaction has requested commit.

We write an output as $(L,X,\rho)$, where $L$ is the serial order, a total
order on the set $\mathrm{dom}\,L$ of committed transactions, $X$ is the set of
aborted ones, and $\rho$ is the set of reported read results.  Then
\[
  \begin{aligned}
  (L,X,\rho)\OutputOrd(L',X',\rho')
  \iff{}& X\subseteq X',\ \ \mathrm{dom}\,L\subseteq\mathrm{dom}\,L',\\
         & L'|_{\mathrm{dom}\,L}=L,\ \ \rho\subseteq\rho'.
  \end{aligned}
\]
Every input has a permitted row: the one that aborts every transaction and
reports nothing.  Let $\{A,B\}$ be an input in which $A$ and $B$ have each
requested commit.  The rows $(\{A,B\},([A,B],\varnothing,\varnothing))$ and
$(\{A,B\},([B,A],\varnothing,\varnothing))$ of
Figure~\ref{fig:ambiguity-frames}(f) have no common \outnoun, since it would
order $A$ both before and after $B$.  So the spec is not monotone.  The
\emph{abort-free} spec permits a row when it can be completed with
$X=\varnothing$.  For this spec, a complete output instead commits exactly the
transactions whose commit request is in the input and leaves the others
undecided.  The same two rows show that it is not monotone either.

\paragraph{\Commitments.}
Each \commitment is a condition on complete outputs.  Applied at a state
$(i,R)$, it keeps the rows of $R$ that can still be completed, within $R$, by a
complete row satisfying the condition, and removes the rest.  It is defined
when it keeps some row at $i$ and removes some row.
\begin{itemize}
  \item $\mathsf{order}(A,B)$, for distinct $A,B\in i$: $B$ is not committed
        before $A$, that is, not $B<_L A$.
  \item $\mathsf{TxCommit}(A)$, for $A$ whose commit request is in $i$:
        $A\notin X$.\quad
        $\mathsf{abort}(A)$, for $A\in i$: $A\notin\mathrm{dom}\,L$.
  \item $\mathsf{stamp}(A)$, for $A\in i\setminus Z$, given the set $Z\subseteq
        i$ of transactions stamped before $A$: no committed transaction outside
        $Z\cup\{A\}$ precedes $A$ in $L$, and no committed transaction of $Z$
        follows $A$.
  \item $\mathsf{read}(A,k,v)$, for $A\in i$ whose $k$-th operation in $i$ is a
        read: the output reports that this read returned the version $v$.
\end{itemize}
Each is a \commitment (Definition~\ref{def:commitment}).  In particular, no
removed row is \outrelpast by a kept one, as
Equation~\ref{eq:commitment-unsubsumed} requires: if a kept row \outrel another
row at the same input, the complete row that keeps the first also \outrel the
second, so the second is kept too.  The first four conditions concern only $L$
and $X$, and for them a complete row is kept exactly when it satisfies the
condition, since two permitted complete outputs at one input, one
\outrelpart the other, have the same serial order and abort set.

The bases are
$\mathcal B_{\mathrm{2PL}}=\{\mathsf{order},\mathsf{TxCommit},\mathsf{abort},\mathsf{read}\}$
and
$\mathcal B_{\mathrm{MVTO}}=\{\mathsf{stamp},\mathsf{TxCommit},\mathsf{abort},\mathsf{read}\}$,
over all transactions, positions, and versions.  They share
$\mathsf{TxCommit}$, $\mathsf{abort}$, and $\mathsf{read}$, and differ in how
they order.  Let $\{A,B,C\}$ be an input in which $A$, $B$, and $C$ have each
requested commit, after any operations.  At its canonical state, no \commitment
of $\mathcal B_{\mathrm{MVTO}}$ has the effect of $\mathsf{order}(A,B)$, and no
\commitment of $\mathcal B_{\mathrm{2PL}}$ has the effect of $\mathsf{stamp}(B)$
with $Z=\varnothing$.

A scheduler applies its \commitments in \commitment events
$\mathsf{apply}(\gamma)$ (Definition~\ref{def:commitment-event}) and reports
results by output events $\mathsf{output}(o)$
(Definition~\ref{def:run-under-basis}).  Each protocol step below responds to
the arrival of an operation, an input event, and says which \commitments it
applies.  Every state that a run under either basis reaches has two
properties: each row has a complete \outnoun in the table at its input, and
any output between two rows of the table at one input is itself a row of the
table.  So a scheduler can report a read by first applying
$\mathsf{read}(A,k,v)$ in a \commitment event, after which emitting the result
is a legal output.

\paragraph{2PL.}
When an operation of $B$ arrives, 2PL requests a lock on its key, shared for a
read and exclusive for a write.  If $A$ holds that lock in a conflicting mode,
$B$ waits, and the scheduler applies $\mathsf{order}(A,B)$ in a \commitment
event and reports nothing.  If every remaining complete row that commits both
already places $A$ before $B$, then $\mathsf{order}(A,B)$ would remove
nothing, so it is undefined and the scheduler applies nothing.  This happens
when an earlier wait has applied $\mathsf{order}(A,B)$, or when a reported
read result fixes the order, as when $A$ has read the initial version of the
key that $B$ writes.  On complete rows without aborts, $\mathsf{order}(A,B)$
keeps exactly those with $A$ before $B$ (Figure~\ref{fig:cc-commitments}(h)).
At the canonical state of $\{A,B,C\}$ it removes $[B,A,C]$ and keeps
$[A,B,C]$ and $[C,A,B]$.  So it is a \commitment, and it is not the canonical
\commitment $x_o$ of any output $o$.  When a read is granted its lock, the
scheduler applies $\mathsf{read}$ and reports the result.  When a commit
request arrives, the scheduler applies $\mathsf{TxCommit}$ and releases the
transaction's locks.

A deadlock between $A$ and $B$ arises when $B$ waits for a lock that $A$ holds
and an operation of $A$ then waits for a lock that $B$ holds.  The scheduler
has then applied $\mathsf{order}(A,B)$ and later $\mathsf{order}(B,A)$, so no
remaining row, at any input, commits both $A$ and $B$.  Every remaining
complete row aborts one of them, and resolving the deadlock means aborting
one.  Once the scheduler applies $\mathsf{TxCommit}(A)$, every remaining
complete row aborts $B$, and $\mathsf{TxCommit}(B)$ is never again defined.
Under the abort-free spec, once both orderings have been applied, no remaining
row is at an input in which both $A$ and $B$ have requested commit.  The
system never accepts both requests, so at least one of $A$ and $B$ waits
forever.  Whether a deadlock can be resolved thus depends on which of the two
specs the system implements.

\paragraph{MVTO.}
We follow Bernstein and Goodman~\cite{bernstein1981concurrency}.  When the
first operation of a transaction $A$ arrives, $A$ receives a timestamp above
all issued ones, and the scheduler applies $\mathsf{stamp}(A)$, with $Z$ the
transactions stamped before $A$, in a \commitment event.  On complete rows
without aborts, $\mathsf{stamp}(B)$ with $Z=\varnothing$ keeps exactly those
with $B$ first (Figure~\ref{fig:cc-commitments}(i)).  At the canonical state
of $\{A,B,C\}$ it is a \commitment, and it is not the canonical \commitment
$x_o$ of any output.  A read of $x$ by $A$ returns $A$'s own write of $x$ if
$A$ has written $x$, and otherwise the committed version of $x$ with the
largest timestamp below $A$'s.  The scheduler applies $\mathsf{read}$ and
reports the result.  A write of $x$ by $A$ is rejected, and $A$ aborted, if a
transaction stamped after $A$ has already read a version of $x$ stamped before
$A$.  The scheduler then applies $\mathsf{abort}(A)$ and reports the abort.
When the commit request of $A$ arrives, the scheduler applies
$\mathsf{TxCommit}(A)$.

\emph{What reporting a read excludes.}  Suppose $B$ is stamped after $A$, and
$B$'s read of $x$ returned the initial version or the write of $x$ by a
transaction stamped before $A$.  Consider an input in which $A$ writes $x$.
There, a complete row that satisfies both stamps and commits both $A$ and $B$
places $A$ between that version and $B$.  In that row the read must return
$A$'s write or a later one, so the row does not permit the reported read.
Hence, after the two stamps, no row at an input in which $A$ writes $x$ both
reports the read and commits $A$ and $B$.  Once $\mathsf{read}$ has been
applied, no remaining row at such an input commits both, whether it reports
the read or not.  Reporting the read therefore rules out one possibility: that
$A$ writes $x$ and both $A$ and $B$ commit.  MVTO reports the read when it
arrives, whether or not $A$'s write has arrived.

The condition on the version, that $B$ read a version older than $A$, is
needed.  Consider an input of three transactions $A$, $W$, $B$, stamped in
that order, each having requested commit, with $A$ and $W$ writing $x$ and $B$
only reading it.  The complete row that commits all three in stamp order and
reports that $B$ read $W$'s write is permitted and satisfies every stamp.
There $B$ read a version written after $A$, and all three commit.

\emph{When the abort is entailed.}  Which of $A$ and $B$ aborts depends on
whether $B$ is decided.  Suppose $A$'s write of $x$ has arrived.  If every
remaining complete row commits $B$, for instance after $\mathsf{TxCommit}(B)$,
then every remaining complete row aborts $A$.  An output that reports the read
and that every remaining row \outrel, with $A$'s abort added, is then a legal
output.  So the spec entails the abort.  While $B$ is undecided, the spec does
not entail it.

Consider a run in which $A$'s first operation, a read of $x$, arrives, then
$B$'s read of $x$, and then $A$'s write of $x$, with further transactions free
to arrive.  The scheduler applies $\mathsf{stamp}(A)$ and $\mathsf{stamp}(B)$
in \commitment events and reports that $B$ read the initial version of $x$.
Reporting the read is not a legal output on its own, since at an input in
which $A$ writes $x$, the row that commits $A$ and then $B$ cannot be
completed with it.  Once the scheduler applies $\mathsf{read}$ in a
\commitment event, reporting the read is a legal output.  After $A$'s write
arrives, two completions remain at the input in which both then request
commit: one commits $B$ and aborts $A$, and the other commits $A$ and aborts
$B$.  Emitting $A$'s abort is then not a legal output, and
$\mathsf{abort}(A)$ and $\mathsf{abort}(B)$ are both defined and keep
different completions.  MVTO's write rule aborts $A$ at once: the scheduler
applies $\mathsf{abort}(A)$ in a \commitment event, after which reporting the
abort is a legal output.  That abort is the scheduler's choice between two
outcomes the spec still permits.  If the scheduler instead waits until $B$'s
commit request arrives and applies $\mathsf{TxCommit}(B)$, $A$'s abort becomes
a legal output.

\paragraph{Comparison.}
As operations arrive, 2PL applies \commitments that fix the relative order of
conflicting transactions only.  MVTO, when a transaction's first operation
arrives, applies a \commitment that fixes its position against all present and
future transactions.  Both report reads through $\mathsf{read}$.  Under MVTO,
reporting a read can itself rule out that two transactions both commit, and
the write rule settles this by an abort that the spec forces only after the
reader commits.

\subsection{Deployments: one explicit relation per participant}
\label{app:deployment}

Theorem~\ref{thm:complete-cap} fixes one explicit relation $E$ for all
participants.  A deployment fixes one for each participant: the rows that
participant may emit.  It models a system that assigns authority, such as a
key owned by one participant or a replica that answers only its own reads.
Assigning rows to participants removes no row from the spec, and
\commitments are still compared with all of the spec's canonical runs.

\begin{definition}[Deployment, emitted rows, partition]
\label{def:deployment}\label{def:emitted-rows}
A \emph{deployment} over a finite participant set $\mathcal N$ assigns to each $p\in\mathcal N$
an explicit relation $E_p\subseteq P$, the rows $p$ may emit: $p$ performs an output event
$\mathsf{output}(o)$ (Definition~\ref{def:run-under-basis}) at current input
$i$ only if $(i,o)\in E_p$.  The \emph{emitted rows} of a run are the
pairs $(i,o)$ such that the run emits $o$ at current input $i$; the emitted
rows of a system are those of its participants' runs in the histories it
produces.  A \emph{partition} $\Pi$ of $\mathcal N$ divides $\mathcal N$ into
two or more nonempty sides, as in a partitioned history
(Definition~\ref{def:partitioned-history}).
\end{definition}

\begin{definition}[Participant run family, silent participant, separable run family]
\label{def:participant-run-family}
A \emph{participant run family} assigns to each of finitely many participants
$p\in\mathcal N$ an initial input and a canonical run whose emitted rows lie in $E_p$.
A participant whose run emits no output is \emph{silent}.  The family is
\emph{separable under $\Pi$} when either participants on at least two
sides emit, or exactly one distinct row is emitted and some silent
participant lies on another side.
\end{definition}

\begin{definition}[Monotone across partitions]
\label{def:monotone-across-partitions}
Relative to a deployment, $\Spec$ is \emph{monotone across $\Pi$} when, for
every participant run family separable under $\Pi$ and every common \indesc $j$
of its participants' current inputs, some $u\in\fiber{j}$ \outrel every
emitted output.  It is \emph{monotone across partitions} when it is monotone
across every partition of $\mathcal N$.
\end{definition}

Monotonicity across partitions plays the role of Lifted Monotonicity for $E$
(Definition~\ref{def:lm-for-A}), restricted to the families that a partition
can separate.  If $\Spec$ satisfies Lifted Monotonicity for a relation that
contains every $E_p$, it is monotone across partitions, since the emitted rows
of a family are rows of that relation at inputs below $j$.  So Lifted
Monotonicity for $E$ suffices for the uniform deployment $E_p=E$, and Lifted
Monotonicity suffices for every deployment.

The converse can fail, because a deployment can keep conflicting rows at one
participant.  Rows authorized only to $p$ can be emitted only by $p$'s run,
and the outputs that run emits form a chain, since each output event keeps
only rows whose outputs extend it.  So a family whose emitted rows are all
authorized only to $p$ is separable only when it emits a single row.  For
example, let $\mathcal N=\{p,q\}$, let every fiber be $\{a,b\}$ ordered by
equality, and let $E_p=P$ and $E_q=\varnothing$, so that $p$ alone names a
leader.  A run of $p$ emits at most one distinct output, and $q$ never emits.
Every separable family therefore consists of one row $(i,a)$ or $(i,b)$ beside
a silent $q$.  That output lies in $\fiber{j}$ for every \indesc $j$ of $i$.
So $\Spec$ is monotone across partitions.  Yet $(i,a)$ and $(i,b)$ have no
common \outnoun, and Lifted Monotonicity fails.

The same happens whenever two conditions hold.  First, each authorized row has
an \outnoun at every \indesc of its input.  Second, every set of authorized
rows that lacks a common \outnoun at some common \indesc is authorized only to
a single participant.  The causal register of Appendix~\ref{app:registers} is
monotone across partitions without Lifted Monotonicity and has no single
authority.  It need not satisfy these conditions, since a set of rows lacking
a common \outnoun may also contain rows of other participants.

The converse does hold when every participant may emit every row.  Suppose
$|\mathcal N|\ge2$ and $E_p=P$ for every $p$.  Then every pair of rows can be
split across a partition, and every single row can be paired with a silent
participant at the common \indesc.  Pairs of rows suffice for the
row-consistency condition (Appendix~\ref{app:pairs-suffice}), so
monotonicity across partitions is then Lifted Monotonicity.

\begin{definition}[Deployment Availability]
\label{def:deployment-availability}
Relative to a deployment, a system has \emph{Deployment Availability} when, in
every partitioned history over $\mathcal N$ it produces, each participant $p$
can emit every output in $R(i)\cap E_p(i)$, where $(i,R)$ is $p$'s final state
(Definition~\ref{def:local-continuation}).
\end{definition}

This is Availability for $E_p$, one relation per participant.  For the uniform
deployment $E_p=E$ it is Availability for $E$ on histories over $\mathcal N$.

A deployment is \emph{single-authority} when $E_q=\varnothing$ for every
participant $q$ but one, as in the leader example above.  Every separable
family is then one row emitted by the single authority beside a silent
participant, and monotonicity across partitions says only that each such row
has an \outnoun at every \indesc of its input.  Placement-sensitive models
such as relational-transducer CALM~\cite{ameloot2013relational,ameloot2015weaker}
also cover this case.

A deployment keeps the conflicting choice from being made at another
participant, but it does not remove the semantic conflict: a canonical run
still conflicts with the authority's \commitment
(Proposition~\ref{prop:commitment-conflict}).  For example, another canonical
run of the spec conflicts with a lock manager's choice even when only the
manager is authorized to emit the result.  A deployment also does not restrict
which inputs a participant receives: the environment may deliver any input to
any participant.

\begin{proposition}[Partition compatibility under a deployment]
\label{thm:complete-cap-deployment}
Fix a deployment.  The following are equivalent.
\begin{enumerate}
  \item For every partition $\Pi$, every participant run family separable
        under $\Pi$ is confluent at every common \indesc $j$ of its
        participants' current inputs.
  \item $\Spec$ is monotone across partitions.
\end{enumerate}
\end{proposition}

\begin{proof}
Fix $\Pi$, a separable family, and $j$, and let $X$ be the set of emitted
outputs.  By Lemma~\ref{lem:reconciliation}, the family is confluent at $j$ if
and only if some $u\in\fiber{j}$ \outrel every output in $X$.  The two
conditions thus agree family by family.
\end{proof}

The proposition concerns canonical runs and does not use Deployment
Availability.  It also applies to the runs of a system, since by
Lemma~\ref{lem:canonical-equivalent} each such run is observably equivalent to
a canonical run with the same emitted rows, which therefore lie in $E_p$.  So
if $\Spec$ is monotone across partitions, a separable family of system runs in
a partitioned history, emitting as Deployment Availability permits, has
emitted outputs with a common \outnoun $u\in\fiber{j}$ at every common
\indesc $j$.  The deployment need not authorize any participant to emit $u$.
Confluence through authorized output events would also require that
authority, a condition we do not impose.

With $|\mathcal N|\ge2$ and $E_p=P$ for every $p$, item~(2) is Lifted
Monotonicity.  The proposition is then the step behind
Theorem~\ref{thm:complete-cap} with $E=P$: separable families are confluent at
every common \indesc exactly when Lifted Monotonicity holds.

For systems, a deployment has its own form of Complete CAP, which rests on one
more fact about families: with two or more participants, monotonicity across
partitions already covers every family, separable or not.

\begin{lemma}[Compatibility of participant run families]
\label{lem:deployment-compatible}
Fix a deployment with $|\mathcal N|\ge2$.  Then $\Spec$ is monotone across
partitions if and only if, for every participant run family and every common
\indesc $j$ of its participants' current inputs, some $u\in\fiber{j}$
\outrel every emitted output.
\end{lemma}

\begin{proof}[Proof sketch]
The condition covers every family, separable or not, so it implies
monotonicity across partitions.  For the other direction, assume $\Spec$ is
monotone across partitions.  If participants $p$ and $q$ both emit, a
partition with $p$ on one side and $q$ on the other makes the family
separable, and the separable case applies.  If no participant emits, any
$u\in\fiber{j}$ will do.  Otherwise exactly one participant $p$ emits.  Its
outputs form a chain, so the last of them \outrel the rest.  Consider $p$'s
last row beside one other participant, placed on another side and kept
silent.  This is a separable family at the same $j$, and the common \outnoun
for that row \outrel every output of $p$.
\end{proof}

The \emph{canonical system} of a deployment is the system whose participants
follow canonical runs, each participant $p$ emitting any output of $E_p$ at
its current input that extends its earlier outputs.

\begin{proposition}[Complete CAP for a deployment]
\label{prop:deployment-complete-cap}
Fix a deployment with $|\mathcal N|\ge2$.  Some system has Deployment
Availability and Partition Consistency on histories over $\mathcal N$ if and
only if $\Spec$ is monotone across partitions.  When it is, the deployment's
canonical system is such a system.
\end{proposition}

\begin{proof}[Proof sketch]
First assume $\Spec$ is monotone across partitions.  The canonical system has
Deployment Availability, since each participant can emit any authorized output
that its state allows.  Its participants' runs in a partitioned history form a
participant run family, so by Lemma~\ref{lem:deployment-compatible} their
emitted outputs have a common \outnoun in $\fiber{j}$.  By
Lemma~\ref{lem:reconciliation} the runs are then confluent at $j$, so
Partition Consistency follows.

Now assume some system has both properties.  Consider a participant run
family and a common \indesc $j$.  Each participant's emitted outputs form a
chain, so it suffices to reach each participant's last emitted row.  We start
each participant alone on its own side, at the input of that row.  Deployment
Availability lets it emit the row's output.  Partition Consistency then gives
the emitted outputs a common \outnoun in $\fiber{j}$, so $\Spec$ is monotone
across partitions.
\end{proof}

\subsection{Choosing the explicit relation}
\label{app:cap-response}

Theorem~\ref{thm:complete-cap} fixes the explicit relation $E$ in advance.
Gilbert and Lynch ask only that every request receive some
response~\cite{gilbert2002cap}, and leave the system to choose which.  Here
the system chooses its explicit relation among declared acceptable responses.
The result is Theorem~\ref{thm:complete-cap} applied to the chosen relation
(Theorem~\ref{thm:complete-cap-response}), and it yields the Gilbert--Lynch
theorem with precedence taken from real time (Corollary~\ref{cor:gilbert-lynch}).
Throughout, we fix an \spectype $\Spec$ and the partitioned histories of
Section~\ref{sec:complete-cap}; the environment supplies any input to any
participant and resolves the strategies' nondeterministic choices.  Emitted
rows are as in Definition~\ref{def:emitted-rows}.

\subsubsection{Obligations, availability, and locality}

\begin{definition}[Response obligation, answering]
\label{def:response-obligation}
A \emph{response obligation} assigns to each participant $p$ a relation
$Q_p\subseteq P$ that is total on $I$ and upward closed within fibers: if
$(i,o)\in Q_p$, $(i,u)\in P$, and $o\OutputOrd u$, then $(i,u)\in Q_p$.  An
output in $Q_p(i)$ \emph{answers} at $i$ for $p$.
\end{definition}

\begin{definition}[Uniform obligation]
\label{def:uniform-obligation}
A response obligation is \emph{uniform} when every $Q_p$ is the same
relation $Q$.
\end{definition}

The obligation is part of the declared spec.  Response Availability is only
as strong as $Q_p$: if a single output answers at every input, every spec
meets condition~(1) of Theorem~\ref{thm:complete-cap-response}.

A continuation of a partitioned history is \emph{within $p$'s side} when
the system produces it, it adds only history events of $p$'s side, and none of its
receives takes a message sent from another side.  A same-side continuation
for $p$ (Definition~\ref{def:local-continuation}) is one within $p$'s side
that adds no input event.

\begin{definition}[Response Availability]
\label{def:response-availability}
A system has \emph{Response Availability} for $(Q_p)_{p\in\mathcal N}$ when, in every
partitioned history it produces, each participant $p$, ending with current
input $i$, satisfies two conditions.
\begin{enumerate}
  \item \emph{Answering.}  In some continuation within $p$'s side that
        adds no input event at $p$, $p$ emits some output in $Q_p(i)$.
  \item \emph{Input acceptance.}  For every $j\InputGe i$, in some
        continuation within $p$'s side, $p$ records $\mathsf{input}(j)$.
\end{enumerate}
\end{definition}

Input acceptance says that $p$ never rules out an \indesc of its current
input.  It excludes specs whose answers require ruling out such an input, such
as the election under the equality order (Figure~\ref{fig:nonmonotone-frame}),
and it matches Gilbert and Lynch's requirement that nodes accept every
request.

\begin{definition}[Response Partition Consistency]
\label{def:response-partition-consistency}
A system has \emph{Response Partition Consistency} when, in every partitioned
history it produces and at every common \indesc $j$ of the participants'
current inputs, some $u\in\fiber{j}$ \outrel every output emitted in the
history.
\end{definition}

The converse of Theorem~\ref{thm:complete-cap} and
Corollary~\ref{cor:complete-cap-any-system} use only this part of Partition
Consistency, so both also hold with Response Partition Consistency in place of
Partition Consistency.  Confluence at $j$ implies a common \outnoun at $j$
(Corollary~\ref{cor:reconciliation-any-basis}).  The two coincide for a system
with basis $\mathcal B_P$ and no \commitment events, whose participant runs
are canonical runs, to which Lemma~\ref{lem:reconciliation} applies
(\S\ref{sec:distributed-consistency}).  With other \commitments they can
differ: participants may prune their states so that they cannot be
reconciled, while their emitted outputs still have a common \outnoun.

A system is \emph{replicated} (SPMD) when its obligation is uniform and every
participant runs the same strategy, which depends only on the participant's
initial input and local history, not on its identity or the number of
participants.

\begin{lemma}[Locality]
\label{lem:locality}
For each participant $p$ in some set of participants, let $\ell_p$ be a local
history that $p$'s
strategy can produce from initial input $i_p$ and that contains no receive.
Placing each $\ell_p$ on its own side, with at least two sides and
any further participants
idle, gives a partitioned history that the system produces.  If the system is
replicated, a receive-free local history that the strategy can produce for one participant
it can produce for any participant with the same initial input, in any
participant set.
\end{lemma}

\begin{proof}
Whether a history is produced is checked one history event at a time against each
participant's own local prefix, as fixed before
Definition~\ref{def:local-continuation},
and there are no receives to match.
Replication makes that check independent of the participant.
\end{proof}

\subsubsection{Response interfaces}

\begin{definition}[Response interface, self-continuing]
\label{def:response-interface}
For a uniform obligation $Q$, a \emph{response interface} is an explicit
relation $T\subseteq Q$ that is total on $I$ and \emph{self-continuing}: if
$(i,o)\in T$ and $i\InputOrd j$, then some $u\in T(j)$ \outrel $o$.
\end{definition}

Lifted Monotonicity for $T$ is that of Definition~\ref{def:lm-for-A}, and
the reconciling output may lie outside $T$.  With $T=P$ it is the
row-consistency condition of Proposition~\ref{prop:common-expansion-form}, and
by Lemma~\ref{lem:reconciliation} it is confluence of the canonical runs whose
emitted rows lie in $T$.

A system can always serve a response interface.  Consider the canonical system
for $T$ (Section~\ref{sec:complete-cap}).  It emits, by the canonical output
event, any output of $T$ in the current fiber.  Each participant's emitted
outputs form a chain $o_1\OutputOrd\cdots\OutputOrd o_m$, so by
Lemma~\ref{lem:run-residual} its table is
$P[\mathrm{in}\InputGe i,\ \mathrm{out}\OutputGe o_m]$, or
$P[\mathrm{in}\InputGe i]$ before any emission.  Self-continuation applied to
the row of $o_m$, or totality of $T$, gives some $u\in T(i)$ in that table.
So every participant's fiber intersects $T$, and Availability for $T$ gives
Answering.  The same argument at an input $k\InputGe i$ leaves a row at $k$,
which gives Input acceptance.

\begin{theorem}[Complete CAP with chosen responses]
\label{thm:complete-cap-response}
For an \spectype $\Spec$ and a uniform obligation $Q$,
the following are equivalent:
\begin{enumerate}
  \item some replicated system has Response Availability for $Q$ and
        Response Partition Consistency;
  \item $\Spec$ satisfies Lifted Monotonicity for some response interface $T\subseteq Q$;
  \item for some response interface $T\subseteq Q$, some system has
        Availability for $T$ and Partition Consistency.
\end{enumerate}
\end{theorem}

\begin{proof}
(2$\Leftrightarrow$3) is Theorem~\ref{thm:complete-cap} and
Corollary~\ref{cor:complete-cap-any-system} with $E=T$.

(2$\Rightarrow$1.)  We take the canonical system for $T$.  It is replicated, and by the
paragraph above it has Answering and Input acceptance.  The outputs emitted in
a history are rows of $T$ at inputs that every common \indesc $j$ of the
current inputs \inrel, so Lifted Monotonicity for $T$ gives Response Partition
Consistency.

(1$\Rightarrow$2.)  We build $T$ from the outputs that participants emit
before receiving any message, which we call \emph{isolated} emissions.  Let
$T$ be the set of rows $(i,o)$ with $o\in Q(i)$ emitted in an isolated
emission in some partitioned history of the system; such a history and
participant are a \emph{witness} for $(i,o)$.  Then $T\subseteq Q$.  The
proof uses two facts about witnesses.  First, consider a participant alone on
its side.  Every continuation that Response Availability supplies adds only
its own history events, which in a partitioned history include no receive, so
the emission they yield is isolated.  Second, by Lemma~\ref{lem:locality}, a
witness's local history up to the emission can be replayed by a participant
alone on its side, with an idle participant forming a second side.

\emph{Totality:} for each $i$, Answering in the empty history, for a
participant with initial input $i$ alone on its side, gives a row of $T$ at
$i$.  \emph{Self-continuation:} let $(i,o)\in T$ and $j\InputGe i$.  We replay
the witness and record $\mathsf{input}(j)$ by Input acceptance, truncating
just after it.  Answering then gives an emission of some $u\in Q(j)$.  By
Lemmas~\ref{lem:canonical-equivalent} and~\ref{lem:run-residual}, the table
before that emission is contained in $P[\mathrm{in}\InputGe j,\
\mathrm{out}\OutputGe o]$.  So $o\OutputOrd u$ and $(j,u)\in T$.
\emph{Lifted Monotonicity for $T$:} let rows $(i_r,o_r)$ of $T$ have a common
\indesc $j$.  We replay each witness on its own participant, truncated after
the emission, and add a participant with initial input $j$.  By
Lemma~\ref{lem:locality} this is a partitioned history of the system.
Response Partition Consistency, applied to that history, gives
$u\in\fiber{j}$ \outrelpart every $o_r$.  If there are no rows, totality of
$P$ gives some $u\in\fiber{j}$.
\end{proof}

\begin{remark}[Without replication]
For a system that is not replicated, fix the participant set.  Let $T_p$ be
the rows $p$ emits in histories where it is alone on its side.  The same proof
gives one response interface $T_p\subseteq Q_p$ per participant.  Each
strategy sees only its own participant's initial input, and each participant
is replayed once.  So consistency is required only for families that draw on
each $T_p$ at most once.
\end{remark}

\subsubsection{Gilbert--Lynch with real-time precedence}

Section~\ref{sec:complete-cap} derives the Gilbert--Lynch theorem from
Corollary~\ref{cor:complete-cap-any-system}, with real-time precedence supplied
as input.  Here precedence comes from real time, and consistency is required
only at the precedence an execution realizes.

\begin{lemma}[Forced answer]
\label{lem:forced-answer}
Let $i\InputOrd j$, and suppose no output in $Q_q(i)$ has an \outnoun in
$\fiber{j}$.  In a system with Answering for $(Q_p)$
(Definition~\ref{def:response-availability}(1)), consider a partitioned
history that the system produces, in which $q$ is alone on its side,
$q$'s current input is $i$, and $j$ \inrel every participant's current input.
Then some continuation within $q$'s side that adds only history events of $q$,
and no input event, has emitted outputs with no common \outnoun in
$\fiber{j}$.
\end{lemma}

\begin{proof}
Answering gives such a continuation in which $q$ emits some $o\in Q_q(i)$.  No
current input changes, so $j$ still \inrel all of them, and no output of
$\fiber{j}$ \outrel $o$.
\end{proof}

\paragraph{The register.}
An input is a finite operation history: a set of read and write invocations at
replicas with a strict partial order of real-time precedence.  One history
\inrel another when it contains all of its operations and precedence edges.
An output is a response map assigning values to reads and acknowledgments to
writes, ordered by extension.  A linearization of $h$ is a total order of its
operations extending precedence~\cite{herlihy1990linearizability} in which
each answered read returns the last preceding write, or the initial value
$v_0$.  The fiber $\fiber{h}$ contains the maps explained by some
linearization of $h$.  The obligation $Q_p(h)$ consists of the maps in
$\fiber{h}$ that answer every read and acknowledge every write invoked at $p$.
Every finite history has a linearization, so $Q_p$ is total, and it is upward
closed within fibers.  Real-time precedence belongs to the input: the
environment invokes $b$ after seeing $a$'s response, so it supplies the edge
from $a$ to $b$ (\S\ref{sec:election-story}).

Real time is the order of history events in an execution.  An
\emph{interleaving} $\lambda$ of a distributed history is a linear order of
its history events consistent with happens-before.  No strategy observes it.
The \emph{invocation} of an operation is the first input event at its replica
whose input contains it.  Its \emph{response} is the first output event at
that replica whose map answers it.  For a prefix $C$ of $\lambda$ (a
consistent cut), the \emph{realized history} $h_C$ contains the operations
invoked in $C$, with an edge from $a$ to $b$ when $a$'s response precedes
$b$'s invocation in $\lambda$.  If $C'\subseteq C$ is also a prefix, $h_C$
\inrel $h_{C'}$.  The realized history need not be a common \indesc of the
participants' current inputs: an input event may introduce two operations with
an edge between them before either responds.

A \emph{client history} has two participants with empty initial inputs, and
each of its input events adds one operation, invoked at its participant after
all of that participant's earlier operations are answered.

\begin{corollary}[Gilbert--Lynch]
\label{cor:gilbert-lynch}
No system for the linearizable register has Answering for $(Q_p)$ together
with the following, both required only of client histories: for every
partitioned client history it produces, interleaving $\lambda$, and prefix $C$
of $\lambda$ such that $h_C$ is a common \indesc of the current inputs at $C$,
the outputs emitted in $C$ have a common \outnoun in $\fiber{h_C}$.
\end{corollary}

\begin{proof}
Let replicas $p$ and $q$ start with empty inputs on separate sides, and let
the environment invoke a write of $\dog\neq v_0$ at $p$.  Answering gives a
continuation in which $p$ acknowledges it.  The environment then invokes a read
$r_q$ at $q$, after the acknowledgment in $\lambda$.  The history is a client
history.  No message crosses the partition, so $q$'s current input is
$i_q=\{r_q\}$.  The realized history $j=h_C$ contains the write, $r_q$, and an
edge from the write to $r_q$, and it is a common \indesc of both current
inputs.  Every map in $Q_q(i_q)$ answers $r_q$ with $v_0$, and every
linearization of $j$ places the write first.  So no output in $Q_q(i_q)$ has
an \outnoun in $\fiber{j}$.  Lemma~\ref{lem:forced-answer} gives a
continuation adding no input event.  The realized history stays $j$, and the
emitted outputs have no common \outnoun in $\fiber{j}$.
\end{proof}

A system with Response Availability, and with a common \outnoun at $h_C$ for
every prefix $C$ of every partitioned history, satisfies these hypotheses, so
the corollary excludes it too.  The register's obligation differs by replica, so
Theorem~\ref{thm:complete-cap-response} does not apply as stated.  The
corollary yields Gilbert and Lynch's theorem in their asynchronous network
model~\cite{gilbert2002cap}.

\begin{proposition}[Gilbert--Lynch, asynchronous model]
\label{prop:gl-async}
In the asynchronous network model with at least two nodes, no algorithm for a
read/write register is both available (a request pending at a node is
answered in some extension of the execution that adds no request and loses
every message) and atomic (every execution with well-formed clients has a
linearization).
\end{proposition}

\begin{proof}[Proof sketch]
We translate the algorithm into a system with the canonical basis.  Participant
$p$ runs node $p$.  An invocation becomes the input event whose input is $p$'s
view: the operations invoked at $p$, with no precedence.  A response becomes an
output event emitting $p$'s response map, and $p$ may also re-emit its current
map.  With every message lost, $p$'s history is node $p$'s run in the
execution where it is alone.  Atomicity of that execution puts the map in the
fiber at the view, so the system can perform its output events under the
canonical basis.  Availability, applied to each unanswered operation of the
view, gives Answering.  An interleaving of a client history is an execution of
the algorithm whose real-time order is the precedence of $h_C$.  So atomicity
gives a common \outnoun in $\fiber{h_C}$, and the corollary excludes the
algorithm.
\end{proof}

Gilbert and Lynch require availability and atomicity in all fair executions,
including those in which messages are lost, in Lynch's asynchronous network
model~\cite[Theorem~1]{gilbert2002cap}\cite[Chapter~8]{lynch1996distributed}.
Both imply the forms in the proposition.  In their model any message may be
lost and client requests are inputs, so a finite execution extends to a fair
execution with no new request in which every message is lost, as in their
own proofs.  Availability answers the pending request at some finite point of
that extension, and the prefix up to that response is the extension the
proposition requires.  Atomicity of the fair extension gives atomicity of the
finite execution, since atomicity is a safety property and so holds of every
prefix of an atomic execution~\cite[Chapter~13]{lynch1996distributed}.

\section{Composition: Proofs and Details}
\label{app:composition}

This appendix proves the results of Section~\ref{sec:composition}, defines
composition with independent inputs and its curried evaluation, and works out
the replicated store with bare reads.

\subsection{Sequential composition}

\begin{proof}[Proof of Lemma~\ref{lem:composition}]
We verify the row-consistency condition
(Proposition~\ref{prop:common-expansion-form}).  Let rows $(i_r,o_r)$ of
$P_2\circ P_1$, finitely many, have common \indesc $j$, and choose $m_r$ with
$(i_r,m_r)\in P_1$ and $(m_r,o_r)\in P_2$.  The condition for $\Spec_1$ gives
$m\in P_1(j)$ \outrelpart every $m_r$.  Then $m$ is a common \indesc of the
inputs of the rows $(m_r,o_r)$, and the condition for $\Spec_2$ gives
$u\in P_2(m)$ \outrelpart every $o_r$.  So $(j,u)\in P_2\circ P_1$ is the
required common \outnoun.
\end{proof}

Tables of functions are closed under $\circ$, so Lemma~\ref{lem:composition}
specializes to functions.

\begin{corollary}[Functional composition]
\label{cor:functional-composition}
If $f:(I,\InputOrd)\to(M,\preceq_M)$ and
$g:(M,\preceq_M)\to(O,\OutputOrd)$ are monotone, so is $g\circ f$.
\end{corollary}

\begin{proof}
The tables of $f$ and $g$ satisfy Lifted Monotonicity and compose to the
table of $g\circ f$, so Lemma~\ref{lem:composition} applies.
\end{proof}

\subsection{Composition with independent inputs}
\label{sec:parameterized-composition}

Sequential composition cannot express a downstream stage with input of its
own, such as a replica that serves bare reads (\S\ref{sec:composition}).  We
therefore give each stage a separate input.  Let
\begin{align*}
  \Spec_F&=\bigl(I\times A,\preceq_I\times\preceq_A,B,\preceq_B,F\bigr),\\
  \Spec_G&=\bigl(J\times B,\preceq_J\times\preceq_B,C,\preceq_C,G\bigr)
\end{align*}
be \spectype[s] in which the output domain $(B,\preceq_B)$ of $\Spec_F$ is
part of the input domain of $\Spec_G$.  Input $i\in I$ goes only to $F$,
input $j\in J$ goes only to $G$, and $B$ carries values from $F$ to $G$.
Splitting $F$'s input as $(i,a)$ makes $F$ at a fixed $i$ a table from $A$ to
$B$, which the curried evaluation below uses.  The \emph{independently
parameterized composition} of $\Spec_F$ and $\Spec_G$ is
\begin{align}
  G\star F
  &=\pi_{I,J,A,C}\bigl(F\Join_{B}G\bigr) \notag\\
  &=\{((i,j,a),c):\exists b.\ ((i,a),b)\in F
                         \land ((j,b),c)\in G\},
  \label{eq:parameterized-composition}
\end{align}
with the product order on $I\times J\times A$.  Since $F$ is total on
$I\times A$ and $G$ is total on $J\times B$, every composed input has a
nonempty fiber.

\begin{lemma}[Independently parameterized composition]
\label{lem:parameterized-composition}
If $\Spec_F$ and $\Spec_G$ satisfy Lifted Monotonicity, then the specification
with table $G\star F$ satisfies Lifted Monotonicity.
\end{lemma}

\begin{proof}
Let rows $((i_r,j_r,a_r),c_r)$, finitely many, have common \indesc $(i,j,a)$,
with witnesses $b_r$ from Equation~\ref{eq:parameterized-composition}.  The
row-consistency condition for $\Spec_F$ at $(i,a)$ gives $b\in F(i,a)$
\outrelpart every $b_r$.  Then $(j,b)$ is a common \indesc of the
$(j_r,b_r)$, and the condition for $\Spec_G$ gives $c\in G(j,b)$ \outrelpart
every $c_r$.  So $((i,j,a),c)\in G\star F$ is the required common \outnoun.
\end{proof}

\paragraph{Curried evaluation.}
For fixed parameters let
\begin{equation}
  F_i=\pi_{A,B}(F[I=i]),
  \qquad
  G_j=\pi_{B,C}(G[J=j]),
  \label{eq:curried-fibers}
\end{equation}
with the selection of Section~\ref{sec:functional-specifications}.

\begin{proposition}[Curried equijoin plan]
\label{prop:curried-query-plan}
The table in Equation~\ref{eq:parameterized-composition} equals
\begin{equation}
  \bigcup_{(i,j)\in I\times J}
  \left\{((i,j,a),c):
    (a,c)\in\pi_{A,C}\bigl(F_i\Join_B G_j\bigr)\right\}.
  \label{eq:curried-query-plan}
\end{equation}
\end{proposition}

\begin{proof}
By Equation~\ref{eq:curried-fibers}, $((i,a),b)\in F$ and $((j,b),c)\in G$
iff $(a,b)\in F_i$ and $(b,c)\in G_j$.  Quantifying over $b$ gives the same
condition on $((i,j,a),c)$ in both expressions.
\end{proof}

For monotone $f:I\times A\to B$ and $g:J\times B\to C$, the instance of
Lemma~\ref{lem:parameterized-composition} for their tables says that
\begin{equation}
  h(i,j,a)=g\bigl(j,f(i,a)\bigr)
  \label{eq:curried-composition}
\end{equation}
is monotone.  Write $f^\sharp(i)$ for the function $a\mapsto f(i,a)$, and
similarly $g^\sharp(j)$ and $h^\sharp(i,j)$.  In this curried form
$h^\sharp(i,j)=g^\sharp(j)\circ f^\sharp(i)$, which is the plan of
Proposition~\ref{prop:curried-query-plan} on tables of functions.

\subsection{Jointly monotone composites}

Suppose every stage satisfies Lifted Monotonicity and $\Spec_1$ runs under
the empty basis.  Then the pipeline emits only what $P_2\circ P_1$ emits under
the empty basis.  Since $P_2\circ P_1$ satisfies Lifted Monotonicity, every
output at the states of these runs is legal
(Proposition~\ref{prop:exposure-boundary}).  So the next stage again has an
empty-basis guard, and by induction no stage of the pipeline issues a
\commitment.

\begin{proof}[Proof of Proposition~\ref{prop:joint-monotone}]
(1) is Lemma~\ref{lem:composition}, and also the case $D=M$ of (2).

(2) Let $D$ contain the outputs of $\Spec_1$.  Take rows $(i_r,o_r)$ of $P_2\circ
P_1$ with a common \indesc $j$, and intermediate values $m_r$ with
$(i_r,m_r)\in P_1$ and $(m_r,o_r)\in P_2$.  The row-consistency condition for
$\Spec_1$ gives $m\in P_1(j)$ that \outrel every $m_r$.  All of $m_r$ and $m$
lie in $D$, so the condition for $\Spec_2$ on $D$ gives $o\in P_2(m)$ that
\outrel every $o_r$, and $(j,o)\in P_2\circ P_1$.  No closure property of $D$
is used.

(3) The table $P_2\circ P_1$ of Equation~\ref{eq:spec-composition} mentions
no order on $M$, and Lifted Monotonicity of the composite uses only the
orders on $I$ and $O$.  So the composite under $\preceq_M$ is the composite
under the other order, which satisfies Lifted Monotonicity by (1).  The same
holds after re-encoding the intermediate values by a map $e:M\to M'$ with a
decoding $e':M'\to M$ such that $e'(e(m))=m$.  Replacing $P_1$ by
$\{(i,e(m)):(i,m)\in P_1\}$ and $P_2$ by $\{(m',o):(e'(m'),o)\in P_2\}$
leaves the composite table unchanged, whatever the orders on $M$ and $M'$,
since a join through $e(m)$ is a join through $e'(e(m))=m$.
\end{proof}

\paragraph{The reach of case (3).}  Suppose each $P_2(m)$ has a greatest
element $t(m)$ and $P_2\circ P_1$ satisfies Lifted Monotonicity.  Order $M$ by
$\sqsubseteq$ (Section~\ref{sec:tight-coordination}), which is reflexive,
since each $P_2(m)$ is directed, and transitive.  $\Spec_2$ satisfies Lifted
Monotonicity under it by the definition of $\sqsubseteq$.  For $\Spec_1$, let
$i\InputOrd j$ and let $X\subseteq P_1(i)$ be finite.  The answers $t(x)$ for
$x\in X$ are outputs of the composite at $i$, so Lifted Monotonicity of the
composite gives a row $(j,u)$ of $P_2\circ P_1$ with every $t(x)\OutputOrd u$,
and $u\in P_2(m')$ for some $m'\in P_1(j)$.  Every finite set of outputs in
$P_2(x)$ lies below $t(x)$, and so below $u$, so $x\sqsubseteq m'$ for every
$x\in X$: $\Spec_1$ satisfies Lifted Monotonicity, and case (3) applies.  For
functions $f:I\to M$ and $g:M\to O$ with $g(f(\cdot))$ monotone, $\sqsubseteq$
places $a$ below $b$ when $g(a)\OutputOrd g(b)$.

Without greatest answers, case (3) can fail.  In the gap example of
Section~\ref{sec:tight-coordination} the composite satisfies Lifted
Monotonicity.  Under any order on $M$ for which $\Spec_2$ satisfies Lifted
Monotonicity, $\Spec_1$ is a sufficient guard
(Proposition~\ref{prop:guard-test}), so the rows $(\mathit{lo},z)$ and
$(\mathit{hi},n_0)$ are not confluent at $\mathit{hi}$
(Appendix~\ref{app:tight}), and $\Spec_1$ does not satisfy Lifted
Monotonicity.

The capped counter of Section~\ref{sec:composite-calm} is an instance of case
(2).  Let $\mathcal S_1$ be the functional spec induced by $c(n)=\min(n,9)$ on
$(\mathbb N,\le)$, and $\mathcal S_2$ that induced by the test $t(n)=[10\le
n]$, whose two outputs are ordered by equality.  $\mathcal S_2$ violates
Lifted Monotonicity, since $[10\le 9]$ and $[10\le 10]$ differ: at
$\sigma_{\mathcal S_2}(m)$ with $m\le 9$, emitting $[10\le m]$ must remove the
row $(10,\mathit{true})$ by \commitment.  But $\mathcal S_2$ satisfies Lifted
Monotonicity on $\{n\le 9\}$, which contains the outputs of $\mathcal S_1$.
So the composite satisfies Lifted Monotonicity by case (2), $\mathcal S_1$ is a
sufficient guard for $\mathcal S_2$ (Proposition~\ref{prop:guard-test}), and under the empty basis
$\mathcal S_2|_{\mathcal S_1}$ answers every count that $\mathcal S_1$ emits.

The boolean toggle is an instance of case (3).  Let $I$, $M$, and $O$ be the
booleans with $\mathit{false}$ below $\mathit{true}$, and let both stages be
negation.  For a function $f$, the output $f(i)$ is a legal output at
$\sigma(i)$ exactly when $f(i)\preceq_M f(k)$ for every $k\InputGe i$.  At
$i=\mathit{false}$ the first stage would emit $\mathit{true}$, which is not
below its output $\mathit{false}$ at the larger input, so it cannot emit there
under the empty basis.  Under the
reversed order on $M$, negation is monotone in both stages, and by case (3)
the composite is the identity in either order.

A monotone composite can also leave a non-monotone upstream unable to emit.
Let $\mathcal S_1$ have one input $*$ and possibility relation
$\{(*,x),(*,y)\}$ with $x,y$ incomparable, and let $\mathcal S_2$ be the
functional spec induced by a function $f$ that maps both to one value.  The
composite has one row, so it satisfies Lifted Monotonicity.  At the canonical
state of $\mathcal S_1$ neither $x$ nor $y$ is a legal output.  Under the
empty basis an output event emits only a legal output, and no \commitment
event is possible (Definition~\ref{def:run-under-basis}).  There is no other
input to receive.  So every run of $\mathcal S_1$ under the empty basis stays at the canonical
state and emits nothing, and $\mathcal S_2$ receives nothing.  Ordering $x$
and $y$ as equivalent, as $\sqsubseteq$ does, makes $\mathcal S_1$ satisfy
Lifted Monotonicity, and by case (3) the composite is unchanged.

\subsection{Sufficient guards}

Throughout, $\Spec_1$ guards $\Spec_2$, and $\mathcal B$ is the basis of the
system for $\Spec_1$.  The inputs of $\Spec_2|_{\Spec_1}$ are the outputs of
$\Spec_1$, ordered by $\preceq_M$, and its rows are the rows of $P_2$ at those
inputs.  For tables $R$ of $\Spec_1$ and $T$ of $\Spec_2$ we write $R\Join
T=\{(k,x):\exists m.\ (k,m)\in R\land(m,x)\in T\}$.  The \emph{image} of a
state $\sigma=(i,R)$ of $\Spec_1$ is the state $(i,R\Join P_2)$ of
$P_2\circ P_1$.  The image of $\sigma_{\Spec_1}(i)$ is
$\sigma_{P_2\circ P_1}(i)$, since $P_1[\mathrm{in}\InputGe i]\Join P_2$ is the
composite table restricted to the \indescs of $i$.

A run of the pipeline is a run of $\Spec_1$ under $\mathcal B$, together with
a run of $\Spec_2|_{\Spec_1}$ from $\sigma(m_1)$, where $m_1$ is the first
output of $\Spec_1$.  Each later output of $\Spec_1$ is an input event of
$\Spec_2|_{\Spec_1}$, and between these inputs $\Spec_2|_{\Spec_1}$ emits any
number of outputs, possibly none.  The environment observes
$\mathsf{input}(j)$ and the outputs of $\Spec_2|_{\Spec_1}$.

The proofs below rest on three facts that hold when $\Spec_1$ is a sufficient
guard, so that $\Spec_2|_{\Spec_1}$ satisfies the row-consistency condition
(Propositions~\ref{prop:exposure-boundary}
and~\ref{prop:common-expansion-form}).  Facts (a) and (b) concern a state $\tau$ that
$\Spec_2|_{\Spec_1}$ reaches without \commitment.
\begin{enumerate}[label=(\alph*)]
\item Every input $j\succeq_M\tau.\mathrm{cur}$ has a row in $\tau$.  The
  row-consistency condition gives $u\in P_2(j)$ that \outrel an output $v$ at
  the current input.  The row $(j,u)$ survives, since $v$ \outrel every output
  emitted so far.
\item Every output at the current input of $\tau$ is a legal output, by
  Definition~\ref{def:sufficient-guard}.
\item If an output $m$ at the current input of a state $\sigma$ of $\Spec_1$
  is a legal output, then every $o\in P_2(m)$ is a legal output at the image
  of $\sigma$.  Each row $(k,x)$ of the image comes from $(k,m')\in R_\sigma$
  and $(m',x)\in P_2$.  If $m\preceq_M m'$, the row-consistency condition for
  the rows $(m,o)$ and $(m',x)$ at $m'$ gives $u\in P_2(m')$ with
  $o,x\OutputOrd u$.  Otherwise legality of $m$ gives $(k,w)\in R_\sigma$ with
  $m,m'\preceq_M w$, and the condition at $w$ gives such a $u\in P_2(w)$.
  Either way $(k,u)$ is a row of the image that \outrel both $o$ and $x$, so
  $o$ is legal (Proposition~\ref{prop:exposability}).
\end{enumerate}

\begin{proof}[Proof of Proposition~\ref{prop:guard-test}]
By Proposition~\ref{prop:exposure-boundary}, $\Spec_2|_{\Spec_1}$ needs
no \commitment exactly when it satisfies Lifted Monotonicity.  Its inputs are
the outputs of $\Spec_1$, so this is Lifted Monotonicity of $\Spec_2$ on the
outputs of $\Spec_1$.
\end{proof}

\begin{proof}[Proof of Theorem~\ref{thm:sufficient-guards}]
After $\mathsf{output}(m)$ every row of $\Spec_1$ \outrel $m$
(Definition~\ref{def:lifted-output-event}), and later events only remove
rows, so the outputs $m_1\preceq_M m_2\preceq_M\cdots$ of a run of $\Spec_1$
form a chain.  $\Spec_2|_{\Spec_1}$ starts at $\sigma(m_1)$, which it reaches
without \commitment.  By facts (a) and (b) it can receive each $m_{r+1}$ and
answer it with a legal output, so it follows the run under the empty basis.

Now consider any run of the pipeline in which $\Spec_2|_{\Spec_1}$ runs under
the empty basis.  We replace each output $m_r$ of $\Spec_1$ by the answers to
it, drop the \commitment events, and show that the result is a run of
$P_2\circ P_1$.  Let $\sigma$ be the state of $\Spec_1$, $T$ the table of
$\Spec_2|_{\Spec_1}$, and $\rho$ the composite's state.  Before the first
output of $\Spec_1$ we maintain $\rho.\mathrm{cur}=\sigma.\mathrm{cur}$ and
$R_\sigma\Join P_2\subseteq R_\rho$.  Afterwards we maintain $R_\sigma\Join
T\subseteq R_\rho$, and every row $(k,m)$ of $R_\sigma$ has $m$ above the
current input of $\Spec_2|_{\Spec_1}$.  Both hold at the start, where the
image of $\sigma_{\Spec_1}(i)$ is $\sigma_{P_2\circ P_1}(i)$.

We check each event of $\Spec_1$.  For $\mathsf{input}(j)$, a row $(j,m)$ of
$R_\sigma$ and, by fact (a), a row $(m,x)$ of $T$ give $(j,x)\in R_\rho$.  So
the composite can take $\mathsf{input}(j)$ as well, and both events keep the
rows at \indescs of $j$.  A \commitment event of $\Spec_1$ keeps the current
input and only removes rows of $R_\sigma$, so both invariants survive it.  For
$\mathsf{output}(m)$ at input $i$, output containment gives $(i,m)\in
R_\sigma$, and $\Spec_1$ keeps $R_\sigma[\mathrm{out}\succeq_M m]$, while
$\Spec_2|_{\Spec_1}$ receives $m$ and keeps $T[\mathrm{in}\succeq_M m]$.  Each
answer $o$ has $(m,o)\in T$, so $(i,o)\in R_\rho$.  The composite's
$\mathsf{output}(o)$ keeps $R_\rho[\mathrm{out}\OutputGe o]$, and
$\Spec_2|_{\Spec_1}$ keeps $T[\mathrm{out}\OutputGe o]$, so the join stays
inside the composite's table.

For confluence, consider $n$ runs of the pipeline whose runs of $\Spec_1$ end
at $\sigma_1,\ldots,\sigma_n$, whose runs of $\Spec_2|_{\Spec_1}$ end at
$\tau_1,\ldots,\tau_n$, and whose composite runs end at
$\rho_1,\ldots,\rho_n$.  Suppose each $\sigma_r$ continues under some basis to
a common state with a row $(j,m)$.  Events never add rows, so $(j,m)\in
R_{\sigma_r}$, and by the invariant above $m$ is above the current input of
each $\tau_r$.  Let $v_r$ be an output at the current input of $\tau_r$.  The
row-consistency condition gives $u\in P_2(m)$ that \outrel every $v_r$, and
$(m,u)$ survives in every $\tau_r$.  So $(j,u)\in R_{\rho_r}$ for every $r$,
and the composite runs are confluent at $j$
(Proposition~\ref{prop:confluence}).
\end{proof}

Corollary~\ref{cor:upstream-commitments} also holds in a stronger form: the
output of $\Spec_1$ that needs a \commitment can be taken at a canonical state.

\begin{proof}[Proof of the stronger form]
The proof of Proposition~\ref{prop:exposure-boundary} derives the
row-consistency condition from legality at canonical states alone.  So if
Lifted Monotonicity fails for $P_2\circ P_1$, some $o$ with $(i,m)\in P_1$ and
$o\in P_2(m)$ is not a legal output at $\sigma_{P_2\circ P_1}(i)$, the image
of $\sigma_{\Spec_1}(i)$.  By fact (c), no $m'\in P_1(i)$ with $o\in P_2(m')$
is a legal output at $\sigma_{\Spec_1}(i)$.  In particular, emitting $m$
there applies a \commitment.
\end{proof}

The sufficient-guard hypothesis is needed in Corollary~\ref{cor:upstream-commitments}.  With $\Spec_1$ the identity on the booleans and
$\Spec_2$ negation, the composite is negation, which does not satisfy Lifted
Monotonicity, yet $\Spec_1$ needs no \commitment.

The \emph{\commitment-erased replacement} of a system has the empty basis and
the same strategy, with its \commitment events dropped.

\begin{corollary}[Empty store]
\label{cor:empty-store}
Let $\Spec_1$ be a sufficient guard for $\Spec_2$, and fix any system for $\Spec_2|_{\Spec_1}$.  For every run
of the pipeline, the replacement has one with the same run of $\Spec_1$ and
the same inputs and outputs of $\Spec_2|_{\Spec_1}$.  After each event, the
replacement's table contains the original's.
\end{corollary}

\begin{proof}
The replacement's run drops the original's \commitment events.  As in the
proof of Lemma~\ref{lem:canonical-equivalent}, it then has the same inputs and
outputs, and after each event its table contains the original's.  By fact
(b), each output the original emits is a legal output for the replacement, so
the empty basis suffices.  The replacement's strategy allows each event,
since the original's did.
\end{proof}

The converse fails, as locks in the store behind the log show.  Let
$\Spec_1$ be the log spec of Example~\ref{ex:committed-log-composition}, under
its canonical basis with a strategy that allows every event, and let
$\Spec_2$ be the table of any monotone function $f$ of the log.  Every log is
an output of $\Spec_1$, so $\Spec_2|_{\Spec_1}$ is $\Spec_2$.  For a set $Q$
of logs, removing every row whose input is in $Q$ empties the fibers at $Q$.
It is a \commitment of $\Spec_2$ at any state whose current input is not in
$Q$ and that has such a row.

Consider two writes $a$ and $b$ to one key, and a lock granted to $a$, which
empties the fibers of the logs that start with $b$.  A store that takes the
lock only once its current log is nonempty applies no \commitment: every log
it can still receive has the same first command, so the lock removes no row.
A store that takes the lock at the empty log empties logs that $\Spec_1$ can
still offer, although its output $f([\,])$ was a legal output without the
lock.  At $\sigma_{\Spec_1}([b])$, for instance, $\Spec_1$ can emit $[\,]$ and
then $[b]$.  The replacement emits $f([\,])$, receives $[b]$, and emits
$f([b])$, but the locking store has emptied $[b]$ and cannot receive it.

\begin{corollary}[Pipelines]
\label{cor:pipelines}
Let $\Spec_1$ be a sufficient guard for $\Spec_2$.  With any systems for $\Spec_1$ and
$\Spec_2|_{\Spec_1}$, the outputs of every run of the pipeline are emitted by
a run of $P_2\circ P_1$ under its canonical basis.  So the pipeline guards a
further stage $\Spec_3$, and Theorem~\ref{thm:sufficient-guards} applies
with $P_2\circ P_1$ in place of $\Spec_1$.
\end{corollary}

\begin{proof}
By Corollary~\ref{cor:empty-store}, every run of the pipeline has one with the
same outputs in which $\Spec_2|_{\Spec_1}$ runs under the empty basis.
Theorem~\ref{thm:sufficient-guards}, applied to it and a prefix of the
run of $\Spec_1$, gives a run of $P_2\circ P_1$ with these outputs.  The
canonical basis supplies the canonical \commitment of each output event that
needs one (Definition~\ref{def:run-under-basis}), so this is a run under the
canonical basis.  By induction the same holds along a pipeline
$\Spec_n\circ\cdots\circ\Spec_1$ in which each prefix guards the next
stage.  Composition is associative, so the bracketing does not matter.
\end{proof}

When the store is a function, some \commitments of $\Spec_1$ are invisible to
the composite.  Let $\Spec_2$ be the functional spec induced by a monotone
function $f$.  For a table $R$ of $\Spec_1$ we write
$f(R)=\{(k,f(m)):(k,m)\in R\}$, and we write $f(\sigma)$ for the image of a
state $\sigma$ of $\Spec_1$.

\begin{proposition}[Invisible coordination]
\label{prop:invisible-coordination}
Let $f$ be monotone, and let $\Spec_1$ apply $\gamma$ at $\sigma$ in a
\commitment event and then emit $m$ by $\mathsf{output}(m)$, a legal output
after $\gamma$.  If every row that $\gamma$ removes has a row with the same
input and the same image that $\gamma$ keeps, then $f(m)$ is a legal output at
$f(\sigma)$.
\end{proposition}

\begin{proof}
Let $\gamma$ take $\sigma$ to $\sigma_c$, at which $m$ is a legal output.  We
first show that the \commitments of $\Spec_1$ cover those of the composite.
Let $\rho$ be a composite state with $f(R_\sigma)\subseteq R_\rho$, and let
$(k,m')\in R_{\sigma_c}$.  If $m\preceq_M m'$, then $f(m)\OutputOrd f(m')$.
Otherwise legality gives $(k,u)\in R_{\sigma_c}$ with $m\preceq_M u$ and
$m'\preceq_M u$, and then $(k,f(u))\in R_\rho$ subsumes $(k,f(m'))$.  Either
way $(k,f(m'))$ is not in $\Excluded_{f(m)}$ at $\rho$.  So every row that the
composite's $\mathsf{output}(f(m))$ removes by \commitment has had every
preimage in $R_\sigma$ removed by $\gamma$.  The hypothesis gives
$f(R_{\sigma_c})=f(R_\sigma)$, so with $\rho=f(\sigma)$ no row of $f(\sigma)$
is in $\Excluded_{f(m)}$, and $f(m)$ is a legal output at $f(\sigma)$.
\end{proof}

Consider a store that answers each key from that key's commands, in log
order.  Let $\Spec_1$ have one input with fiber $\{[a,b],[b,a]\}$, for
commands $a$ and $b$ on different keys, and let $f$ send a log to the family
of its subsequences on each key.  The two logs are incomparable with no common
\outnoun, so $[a,b]$ is not a legal output, and $\Spec_1$ must order them.
But $f([a,b])=f([b,a])$, the composite has one row, and $f([a,b])$ is a legal
output.  The order of $a$ and $b$ is a \commitment inside $\Spec_1$ that
conflicts with a canonical run, and the composite does not see it.

The image under $f$ of a \commitment of $\Spec_1$ need not be a \commitment
of the composite.  Let $\Spec_1$ have one input with fiber $\{x,y\}$, $x$ and
$y$ incomparable, and $f(x)=\mathit{lo}\OutputOrd\mathit{hi}=f(y)$.  The
\commitment $e_y$ removes $x$ and keeps $y$.  Its image would remove
$\mathit{lo}$ and keep $\mathit{hi}$, violating
Equation~\ref{eq:commitment-unsubsumed}, so no \commitment of the composite
has that effect.  In the composite, $\mathit{hi}$ is a legal output at the
canonical state, since it subsumes $\mathit{lo}$.

Monotonicity of $f$ is needed.  Let $\Spec_1$ be the functional spec induced
by the identity on $M$, which needs no \commitment.  For $m\preceq_M m'$, its
run $\mathsf{output}(m)$, $\mathsf{input}(m')$, $\mathsf{output}(m')$ maps to
a run of the composite exactly when $f(m)\OutputOrd f(m')$.

Nor can a downstream stage that needs a \commitment of its own delegate it
upstream.  If $\Spec_1$ is the identity on one input, no \commitment of
$\Spec_1$ exists.  If $\Spec_2$ relates that input to two incomparable
outputs, every output event of the composite at its canonical state applies a
\commitment, which $\Spec_1$ cannot make.

\begin{example}[Replicated store]
\label{ex:committed-log-composition}
Consider again state-machine replication, with every put and get passing
through the log.  Let $I$ be the proposal streams and $M$ the logs, both
ordered by prefix, and $O$ the cumulative response histories, ordered by
extension.  Let $\Spec_1$ relate each proposal stream to the logs of its
commands, and let the store $\Spec_2$ be the functional spec induced by the
function $f$ that executes a log on a deterministic state machine, which is
monotone in the prefix order.  Then $\Spec_1$ is non-monotone, since distinct next entries are inconsistent rows in one fiber.  So is the key-value store
$P_2\circ P_1$, since concurrent commands ordered oppositely give inconsistent rows.  Multi-Paxos runs $\Spec_1$ under its own basis, deciding one log
slot at a time~\cite{schneider1990smr,lamport2001paxos}.  Since $f$ is
monotone, the log is a sufficient guard.  By Theorem~\ref{thm:sufficient-guards} the store under the empty basis follows every run of the log, and by Corollary~\ref{cor:upstream-commitments} the log's \commitments cannot be
dropped.  By Corollary~\ref{cor:empty-store}, a store that takes locks of its
own can run with no \commitments at all.
\end{example}

\begin{example}[Bare reads]
\label{ex:bare-reads}
We add bare reads to Example~\ref{ex:committed-log-composition}.  Let $I$
and $B$ be the proposal streams and logs of that example, $A$ a one-point set,
and $J$ the histories of bare reads received by a replica, ordered by prefix.
Let $C$ pair a response history for logged commands, ordered by extension,
with a set of bare-read responses, ordered by inclusion.  $F$ is the log spec
$\Spec_1$ of that example, with the one-point $A$ added to its input, and $G$
relates a bare-read history $j$ and a log $b$ to the responses
obtained by executing $b$, together with at least one response to each bare
read in $j$.  Responses are facts, so a set of them may answer one read more
than once.  As in Example~\ref{ex:committed-log-composition}, $F$ and
$G\star F$ are non-monotone.

Suppose first that each bare-read response records the log position it read
from, as in ``$x=5$ as of entry 17.''  A longer log never contradicts such a
response, so $G$ is monotone.  By the no-\commitment test, every output of
$G$ at every state reachable without \commitment is then legal, so bare reads
need no \commitment of their own.  Suppose instead that the spec
demands the latest value, so that a read of $x$ must return the value of $x$
after all of $b$.  Then a bare read can be contradicted by an entry it has not
seen, and $G$ is not monotone.  By the no-\commitment test some output of $G$
then needs a \commitment of the store's own.  Implementations typically
enforce it by routing reads through the log.
\end{example}

\subsection{Necessary and tight guards}
\label{app:tight}

Throughout, $\Spec_1$ guards $\Spec_2$.  Confluence of rows at an input is
as in Section~\ref{sec:functional-calm}, read under $\preceq_M$, under
$\sqsubseteq$, or under any other relation on $M$ put in its place; rows of
$P_2\circ P_1$ are read under $\OutputOrd$.  At a canonical state, canonical
\commitments remove exactly the rows that are not confluent with the emitted row: for
$(i,m)\in P_1$, emitting $m$ at $\sigma_{P_1}(i)$ puts $(k,x)$ in
$\Excluded_m$ exactly when $(i,m)$ and $(k,x)$ are not confluent at $k$.
Indeed $(k,x)$ is subsumed exactly when it is confluent with $(i,m)$ at $k$
in the state's table, and at $\sigma_{P_1}(i)$ that table has the fiber
$P_1(k)$ at $k$.  So a spec satisfies Lifted Monotonicity exactly when all rows $(i,m)$ and $(k,x)$ with $i\InputOrd k$ are confluent at $k$
(Proposition~\ref{prop:common-expansion-form}).

\begin{proof}[Proof of Theorem~\ref{thm:guard-tightening}]
(1) By Proposition~\ref{prop:guard-test}, $\Spec_1$ is a sufficient guard exactly when
$\Spec_2$ satisfies the row-consistency condition on the outputs of
$\Spec_1$.  For outputs $m\preceq_M m'$, that condition at $m'$ for finitely
many rows of $P_2(m)$ is the statement $m\sqsubseteq m'$, and conversely,
finitely many rows $(m_r,o_r)$ with a common extension $m'$ of their inputs
are confluent at $m'$ by $m_r\sqsubseteq m'$ together with directedness of
$P_2(m')$, which is $m'\sqsubseteq m'$.  So the orders on $M$ under which
$\Spec_1$ is a sufficient guard are those contained in $\sqsubseteq$ on the
outputs of $\Spec_1$.

(2) Let $(i,o)$ and $(k,y)$ be rows of $P_2\circ P_1$ through $(i,m)$ and
$(k,x)$ of $P_1$, and suppose $(i,m)$ and $(k,x)$ are confluent at $k$ under
$\sqsubseteq$: some $m'\in P_1(k)$ has $m,x\sqsubseteq m'$.  Then $o$ and $y$
each have an extension in $P_2(m')$, and since $P_2(m')$ is directed by the above,
some $u\in P_2(m')$ \outrel both.  The row $(k,u)$ of $P_2\circ P_1$ shows
that $(i,o)$ and $(k,y)$ are confluent at $k$.  Moreover, an
order under which $\Spec_1$ is a sufficient guard is contained in $\sqsubseteq$ on the
outputs of $\Spec_1$, and a smaller order has fewer common extensions,
so every pair not confluent at $k$ under $\sqsubseteq$ is not confluent at $k$ under it.

Now let $t$ and $s$ be the greatest elements of $P_2(m)$ and $P_2(x)$.  If
$(i,t)$ and $(k,s)$ are confluent at $k$ in the composite, some row $(k,u)$
of $P_2\circ P_1$ has $t,s\OutputOrd u$, and $u\in P_2(m')$ for some
$m'\in P_1(k)$.  Every finite set of outputs in $P_2(m)$ lies below $t$, and
so below $u$, so $m\sqsubseteq m'$, and likewise $x\sqsubseteq m'$.  So
if $(i,m)$ and $(k,x)$ are not confluent at $k$ under $\sqsubseteq$, neither are
$(i,t)$ and $(k,s)$, and the first part gives the converse for any answers.
Greatest elements also make each $P_2(m')$ directed, so the first part holds
without assuming a sufficient guard.
For the last claim, let $\sqsubseteq$ be an order and order $M$ by it.  By
(1), $\Spec_1$ is a sufficient guard, and its pairs not confluent under its
order are those not confluent under $\sqsubseteq$, which are exactly those
with answers not confluent in the composite.  So $\Spec_1$ is also a
necessary guard.
\end{proof}

For a functional downstream $f$, $m\sqsubseteq m'$ exactly when $f(m)\OutputOrd
f(m')$, and the converse in (2) holds without assuming that $\Spec_1$ is a
sufficient guard.

The commuting log of Section~\ref{sec:tight-coordination} is formalized with
two increments $A$ and $B$, inputs the sets of commands proposed so far, logs
ordered by prefix, and the store answering a log with its length.  The
log is a sufficient guard, since length is monotone in the prefix order.  The rows
$(\{A\},[A])$ and $(\{A,B\},[B,A])$ are not confluent at $\{A,B\}$ under the
prefix order, so emitting $[A]$ at $\{A\}$ removes $(\{A,B\},[B,A])$ by
\commitment.  But the composite answers both logs with $1$ and $2$, which have
the common extension $2$, and under $\sqsubseteq$, which compares lengths, the
rows are confluent at $\{A,B\}$.  So the log is a sufficient guard but not a necessary one.  The log does
not satisfy Lifted Monotonicity, while the composite does.

\begin{proposition}[Tight guards]
\label{prop:tight-guards}
Let $\Spec_1$ guard $\Spec_2$.
\begin{enumerate}
  \item If $\Spec_1$ satisfies Lifted Monotonicity, it is a necessary guard
        for $\Spec_2$.
  \item If $\Spec_2$ is induced by a monotone function $f$ that reflects
        $\preceq_M$ on the outputs of $\Spec_1$, so that
        $f(m)\OutputOrd f(m')$ implies $m\preceq_M m'$, then $\Spec_1$ is a
        tight guard for $\Spec_2$.
\end{enumerate}
\end{proposition}

\begin{proof}
(1) All rows $(i,m)$ and $(k,x)$ of $P_1$ with $i\InputOrd k$ are
confluent at $k$ (Proposition~\ref{prop:common-expansion-form}), so the
condition of Definition~\ref{def:necessary-guard} holds vacuously.

(2) Since $f$ is monotone, $\Spec_2$ satisfies Lifted Monotonicity, so
$\Spec_1$ is a sufficient guard by Proposition~\ref{prop:guard-test}.  For a
functional downstream, $m\sqsubseteq m'$ exactly when $f(m)\OutputOrd f(m')$,
so on the outputs of $\Spec_1$ reflection makes $\sqsubseteq$ contained in
$\preceq_M$.  Rows of $P_1$ not confluent under $\preceq_M$ are then
not confluent under $\sqsubseteq$, and by
Theorem~\ref{thm:guard-tightening}(2) their answers are not confluent in
$P_2\circ P_1$.
\end{proof}

Both parts apply to logs.  In Example~\ref{ex:committed-log-composition},
suppose the store answers each command in turn, so that $f$ maps a log to its
commands paired with their responses.  The commands of
$f(b)$ are then $b$ itself, so $f(b)\OutputOrd f(b')$ implies $b\preceq_M b'$, and
by (2) the log is a tight guard for the store.  For two commuting increments,
a partial-order log that records only which commands it has logged, as a set
ordered by inclusion, may emit at each input any subset of the commands
received, and the store answers a set with its size.  This log satisfies
Lifted Monotonicity, since the set of all commands received at a larger input
extends every earlier output.  By (1) it is a necessary guard, and since size
is monotone it is a tight guard.

In the gap example, $M=\{z\}\cup\{n_k:k\in\mathbb N\}$ under equality, with
$P_1(\mathit{lo})=\{z\}$, $P_1(\mathit{hi})=\{n_k\}$, $P_2(z)=\mathbb N$, and
$P_2(n_k)=\{k\}$, where $\mathbb N$ has its usual order.  $\Spec_1$ is a
sufficient guard, since $\mathbb N$ is directed and singletons are.  No $n_k$ has
$z\sqsubseteq n_k$, since $\{k+1\}\subseteq P_2(z)$ has no extension in
$\{k\}$, so $(\mathit{lo},z)$ is not confluent at $\mathit{hi}$ with any row at $\mathit{hi}$
under $\sqsubseteq$, and so, by the proof of (2), under every order on $M$ under
which $\Spec_1$ is a sufficient guard.  The composite has the fiber $\mathbb N$
at both inputs and satisfies Lifted Monotonicity, so $\Spec_1$ is not a
necessary guard under any such order.

\subsection{Intersection}
\label{app:intersection}

\begin{proof}[Proof of Proposition~\ref{prop:intersection-bound}]
Since $g(j)\in P_1(j)\cap P_2(j)$ for every $j$, the table $P_1\cap P_2$ is
total.  Let $i\InputOrd j$ and let $X\subseteq P_1(i)\cap P_2(i)$ be finite and
nonempty.  Every $x\in X$ is in $P_1(i)$, so $x\OutputOrd g(i)$, and
$g(i)\OutputOrd g(j)$ since $g$ is monotone.  So $g(j)\in P_1(j)\cap P_2(j)$
\outrel every member of $X$, which is Lifted Monotonicity.
\end{proof}

In the example of Section~\ref{sec:intersection}, both specs satisfy Lifted
Monotonicity: with a single input, Lifted Monotonicity is directedness of the
one fiber, and $c$ and $d$ bound their fibers.  The intersection is total, since
$a$ is in both fibers, and its fiber $\{a,b\}$ is not directed.  By
Proposition~\ref{prop:exposure-boundary}, some output of the intersection at a
state reachable without \commitment is not legal.

\subsection{Components and typed pipelines}
\label{app:components}

Proposition~\ref{prop:guard-test} and Theorem~\ref{thm:sufficient-guards} concern specs, with a system fixed only
for the upstream stage.  Here each stage is seen from outside, through the
runs it admits, so that particular implementations of both stages, each with
its own basis and participants, can be examined.

\begin{definition}[Component, admitted run]
\label{def:component}
A \emph{component} for a specification $\Spec$ consists of a \commitment basis
$\mathcal B$ and a set of \emph{admitted runs}.  Each admitted run is a run
under $\mathcal B$ from a canonical state $\sigma(i)$
(Definition~\ref{def:run-under-basis}), together with an output $o$ that it
emits, such that $(i,o)\in P$.
\end{definition}

\begin{definition}[Implementation form]
\label{def:implementation-form}
A component may carry an \emph{implementation form}, $\mathsf{Process}$,
$\mathsf{Cluster}$, or $\mathsf{Composite}$, which records how it is
realized: the internal system model and coordination mechanisms that justify
its admitted runs.  Typed composition does not use it.
\end{definition}

A system for $\Spec$ (Definition~\ref{def:system}) induces a component.  Its
admitted runs are the participant runs of produced histories, each paired
with an output $o$ that the run emits such that $(i,o)\in P$, where $i$ is the
run's initial input.  The system's messages, partitions, and \commitments stay
inside the component.  Restricted to one-participant histories, the induced
component is a process, and a replicated system induces a cluster.  The
\emph{maximally permissive} component under $\mathcal B$ admits every run
under $\mathcal B$ that Definition~\ref{def:component} allows.

The two-phase commit decision discussed after
Proposition~\ref{prop:commitment-conflict} can take either form.  Run by the
coordinator alone, it is a process, and no conflicting choice is placed at
another participant.  Replicated, it is a cluster, and one can be.  Either way
the timeout applies a \commitment that conflicts with a canonical run.

Let $C_1$ be a component for $\Spec_1$ and $C_2$ a component for $\Spec_2$.

\begin{definition}[Feed, typed pipeline, external observation]
\label{def:typed-pipeline}
A \emph{feed} pairs an admitted run of $C_1$, with input $i$ and output $m$,
with an admitted run of $C_2$ from $\sigma(m)$, with output $o$.  The
\emph{typed pipeline} of $C_1$ and $C_2$ shows the \emph{external
observation} $(i,o)$ of each feed; the value $m$, both runs, and their
\commitments stay hidden.
\end{definition}

\begin{definition}[Contract-sufficient, implementation-sufficient]
\label{def:sufficient-pipeline}
A typed pipeline is \emph{contract-sufficient} when $\Spec_1$ is a sufficient
guard for $\Spec_2$ (Definition~\ref{def:sufficient-guard}), and
\emph{implementation-sufficient} when no downstream run of a feed applies a
\commitment.
\end{definition}

The value $m$ is the downstream run's initial input, so receiving it is not
a \commitment of $C_2$.  Glue stages that copy, merge, route, or adapt values
are components with specs of their own, and they may issue \commitments of
their own.

\begin{proposition}[Typed pipelines]
\label{prop:typed-pipelines}
\leavevmode
\begin{enumerate}
\item Every external observation admitted by a typed pipeline is a row of
      $P_2\circ P_1$.
\item If a pipeline is implementation-sufficient, the downstream run of every
      feed, with its \commitments erased, is a \commitment-free run of
      $\Spec_2$, in which every output is legal.
\item Contract sufficiency does not imply implementation sufficiency: behind a
      log, the early-lock store is contract-sufficient and still issues
      \commitments.
\item Contract sufficiency does not guarantee feeds when the downstream
      component is for $\Spec_2$ rather than $\Spec_2|_{\Spec_1}$: behind the
      capped counter, the maximally permissive empty-basis component for the
      threshold test admits no feed at an offered count.
\item If, in a contract-sufficient pipeline, the downstream is instead the
      maximally permissive empty-basis component for $\Spec_2|_{\Spec_1}$,
      every admitted upstream boundary run has a feed, and the pipeline is
      implementation-sufficient.
\end{enumerate}
\end{proposition}

\begin{proof}
(1) The two runs of a feed certify $(i,m)\in P_1$ and $(m,o)\in P_2$.

(2) A run that applies no \commitment has no \commitment event, and each of its
output events $\mathsf{output}(o)$ applies no \commitment, which holds exactly
when $o$ is a legal output.

(3) The early-lock store above, fed by a log, is contract-sufficient, since the
store's table satisfies Lifted Monotonicity on logs.  From the empty log it
admits a run whose \commitment event takes the lock.

(4) Consider the capped counter of Section~\ref{sec:composite-calm} feeding
the double-digit threshold test.  The pipeline is contract-sufficient, since the
test satisfies Lifted Monotonicity on the counts the counter offers.  An
admitted run of the test from an offered count $m\le 9$ must emit
$\mathit{false}$, its answer at $m$.  While its current count is below $10$,
its table keeps the row $(10,\mathit{true})$, so $\mathit{false}$ is not a
legal output.  Once the count reaches $10$, only $\mathit{true}$ remains.  So
no feed exists.

(5) Let $m$ be the output of an admitted run of $C_1$.  By Proposition~\ref{prop:guard-test}, $\Spec_2|_{\Spec_1}$ satisfies
Lifted Monotonicity, so every $o\in P_2(m)$ is a legal output at $\sigma(m)$
(Proposition~\ref{prop:exposure-boundary}).  The one-event run
$\mathsf{output}(o)$ from $\sigma(m)$ is then an admitted run of $C_2$.
Under the empty basis no run applies a \commitment.
\end{proof}

Item (5) covers one upstream output, while
Theorem~\ref{thm:sufficient-guards} follows the entire emitted stream.
Item (4) is why (5) uses a component for $\Spec_2|_{\Spec_1}$ rather than for
$\Spec_2$: its canonical states keep only rows at inputs the upstream offers.

The $\mathsf{Process}$ and $\mathsf{Cluster}$ forms follow Hydro's location
types~\cite{hydro}, which connect to the surrounding dataflow through typed
edges.  The analogy concerns interfaces only: the results do not model
Hydro's scheduler, placement, or failure semantics.  A typed pipeline is
binary, and finite acyclic graphs follow by associativity of composition and
explicit glue components.  We do not treat cyclic dataflows, scheduling,
physical placement, co-location, or failure domains, and we do not merge the
participants of different components into one history.

\section{Related Work and Mappings: Details}
\label{app:related}

This appendix supplies the details behind the mappings of Section~\ref{sec:instances}.

\subsection{Invariant confluence (\S\ref{sec:iconfluence-instance})}
\label{app:iconfluence}

A \emph{transaction history} $h$ (here simply \emph{history}) is a finite
partial order of transaction events, each labeled by a transaction in $T$,
where $e'<e$ means that $e$ ran on a state that included the writes of $e'$.
A sub-history is \emph{closed} when it is downward closed: it contains every
event that precedes one of its own.  History
$h'$ \inrel $h$ when $h$ is a closed sub-history of $h'$.  Then the events
that $h'$ adds precede no event of $h$, and each event of $h$ keeps exactly
the predecessors it had.

The state of a history is computed event by event.  An event $e$, labeled by
the transaction $t_e$, runs on
$\mathrm{pre}(e)=s_0\cup\bigcup_{e'<e}\mathrm{post}(e')$ and produces
$\mathrm{post}(e)=t_e(\mathrm{pre}(e))$.  The state of $h$ is
$\mathrm{eval}(h)=s_0\cup\bigcup_{e\in h}\mathrm{post}(e)$.  Because $h$ is a
partial order and transactions are inflationary ($t(s)\supseteq s$),
$\mathrm{pre}(e)$ is the state of the sub-history strictly below $e$, and
$\mathrm{post}(e)$ is the state of the sub-history at or below $e$.  A
closed sub-history keeps each event's predecessors, so it also keeps each
event's $\mathrm{pre}$ and $\mathrm{post}$.  Hence the state of a union of
closed sub-histories is the union of their states.

A history is \emph{locally valid} when every $\mathrm{pre}(e)$ and
$\mathrm{post}(e)$ satisfies $\mathcal I$.  This is the case when replicas
abort transactions that would violate the invariant on their own state.  A
history is \emph{valid} when the state of each of its closed sub-histories
satisfies $\mathcal I$, that is, $\mathcal I(\mathrm{eval}(c))$ for every
closed sub-history $c$.

The spec $\Spec_{\mathcal I}$ has the locally valid histories as inputs;
a closed sub-history of a locally valid history is locally valid, so the
inputs are closed under the order.
An output is a state, and $u$ \outrel $s$ when $s\subseteq u$.  At input $h$,
each valid closed sub-history of $h$ contributes its state as a permitted
output:
\[
  P_{\mathcal I}(h)=\{\mathrm{eval}(g):g\text{ a valid closed
  sub-history of }h\}.
\]
Taking $g$ empty gives $s_0$, so the spec is total.

\begin{proof}[Proof of Proposition~\ref{prop:iconfluence-lm}]
We use the row-consistency form of Lifted Monotonicity
(Proposition~\ref{prop:common-expansion-form}).  We first show, by induction
on $g$, that the state of every valid history $g$ is reachable.  Let $e$ be
maximal in $g$.  Then $g\setminus\{e\}$ and the sub-history strictly below $e$
are closed, hence valid.  By induction their states
$\mathrm{eval}(g\setminus\{e\})$ and $\mathrm{pre}(e)$ are reachable.  The
state $\mathrm{post}(e)$ is a valid transaction step from $\mathrm{pre}(e)$,
and $\mathrm{eval}(g)=\mathrm{eval}(g\setminus\{e\})\cup\mathrm{post}(e)$ is a
valid merge.

First assume that $T$ is invariant confluent.  Let rows
$(h_r,\mathrm{eval}(g_r))$, $r=1,\ldots,k$, have a common \indesc $j$.  Then each $g_r$ is a closed sub-history of $j$.  Let
$g=g_1\cup\cdots\cup g_k$; we show that $g$ is valid.  Every closed
sub-history $c$ of $g$ is the union of the pieces $c\cap g_r$.  Each piece is
a closed sub-history of the valid $g_r$, so its state is reachable.  Thus
$\mathrm{eval}(c)$ is a union of reachable states.
Invariant confluence makes each successive union valid, and a valid merge of
reachable states is reachable, so $\mathcal I(\mathrm{eval}(c))$.  Hence $g$
is valid, and $\mathrm{eval}(g)\in P_{\mathcal I}(j)$ \outrel every
$\mathrm{eval}(g_r)$.  For $k=0$, take $g$ empty.  Lifted Monotonicity
follows.

Now assume that $\Spec_{\mathcal I}$ satisfies Lifted Monotonicity.  Only its
directed fibers are used.  We first show that valid histories are closed under
disjoint union.  Let $g_1$ and $g_2$ be valid histories, and let $j$ be their
disjoint union, with events renamed apart and no order between the two
parts.  Each event keeps its
predecessors, so its $\mathrm{pre}$ and $\mathrm{post}$ are states of closed
sub-histories of a valid $g_r$.  So $j$ is locally valid, and
$\mathrm{eval}(g_1),\mathrm{eval}(g_2)\in P_{\mathcal I}(j)$.  Directed
fibers give a valid sub-history $g$ of $j$ with
$\mathrm{eval}(g)\supseteq\mathrm{eval}(g_1)\cup\mathrm{eval}(g_2)=\mathrm{eval}(j)\supseteq\mathrm{eval}(g)$.
So $\mathcal I(\mathrm{eval}(g_1)\cup\mathrm{eval}(g_2))$.  Every closed
$c\subseteq j$ is the disjoint union of $c\cap g_1$ and $c\cap g_2$, both
valid, and the same argument gives $\mathcal I(\mathrm{eval}(c))$.  So $j$ is
valid.  Valid histories are also closed under adding an event above all
others when its transaction step is valid.  By induction on derivations, every
reachable state is therefore the state of a valid history.  The pairwise
argument then gives $\mathcal I(s_1\cup s_2)$ for all reachable $s_1,s_2$.
Invariant confluence follows.
\end{proof}

\subsection{The HAT taxonomy (\S\ref{sec:hat-instance})}
\label{app:hat-taxonomy}

Bailis et al.\ classify transactional guarantees by what a system can provide
while every client that can reach a replica keeps getting
answers~\cite{bailis2013hat}.  The extended paper defines each guarantee by
the anomaly it rules out and gives algorithms for the guarantees that can be
provided~\cite{bailis2013hatextended}.  We use those definitions and
algorithms as written, except that where a definition disagrees with its own
example, we follow the example.

\paragraph{Inputs and outputs.}
The inputs extend the histories of Definition~\ref{def:hat-spec}.  Each item
has a fixed set of replicas, and a read is served by a replica of its item.  A
\emph{predicate read} reads a fixed range of items at once and returns a
version of each.  A transaction's reads are ordered.  A client runs one
transaction at a time.  Its transactions form its \emph{session}, and its
commit timestamps increase along the session.  Outputs are sets of read
answers, as before.

\paragraph{Guarantees.}
Each guarantee rules out one anomaly of the extended paper: a pattern in a
graph whose nodes are transactions.  The edges are the write-, read-, and
anti-dependencies of Adya on each item~\cite{adya1999weakconsistency}, and
session edges from each transaction to its client's later ones.  For example,
read your writes rules out a transaction that reads a version of $x$ older
than one that an earlier transaction of the same client wrote.  The extended
paper gives no formal definition of recency.  We take it to require that a
read return a version at least as new as every write by another transaction
whose commit happens before the read.  Because a replica can return only
versions it holds, the \emph{visibility store} permits an answer only if its
version is the initial one or arrived at the reading replica before the
read.

\begin{definition}[Guarantee spec, high availability, sticky read, sticky availability]
\label{def:hat-availability}
For a guarantee $G$, the table of $\Spec_G$ contains an output at an input when
the visibility store permits the output and the input extends to a complete
history that includes the output and has no anomaly that $G$ rules out.  The
guarantee $G$ is \emph{highly available} when some explicit relation
$E\subseteq P_G$ answers every read and some system for $\Spec_G$ has
Availability for $E$ and Partition Consistency.  A read is \emph{sticky} when
its replica has received every write of the client's earlier transactions to
the items that replica holds.  The guarantee $G$ is \emph{sticky available}
when the same holds for an $E$ that answers every sticky read.
\end{definition}

Sticky availability follows Bailis et al.: a client gets answers whenever its
transaction runs against a copy of the database that reflects all of its
earlier operations.  By Theorem~\ref{thm:complete-cap} and
Corollary~\ref{cor:complete-cap-any-system}, $G$ is highly available exactly
when $\Spec_G$ satisfies Lifted Monotonicity for some $E\subseteq P_G$ that
answers every read, and likewise for sticky availability.

\paragraph{Reveal rules.}
A transaction is \emph{pending stable} before a read when every replica of
every item it wrote has received its write before the read.  Bailis et al.'s
algorithm for monotonic atomic view reveals a write once its transaction is
pending stable.  For the session guarantees they also ``wait to reveal new
writes \ldots\ until each write's respective dependencies are visible on all
replicas.''  For cut isolation, a client caches the versions its transaction
has read.  The \emph{reveal rule} combines the three.  A write is revealed at
a read once its transaction and every transaction whose commit happens before
its commit are pending stable.  A transaction reads, for each item, the
highest version revealed at every one of its reads so far.  The \emph{sticky
reveal rule} answers only sticky reads, and it also reveals the client's own
earlier writes, which a sticky read's replica holds.

\begin{theorem}[HAT taxonomy]
\label{thm:hat-taxonomy}
\leavevmode
\begin{enumerate}
  \item Under the timestamp condition of Theorem~\ref{thm:hat-levels}, read
        committed and monotonic atomic view are highly available.  Item cut
        isolation, predicate cut isolation, monotonic reads, monotonic
        writes, and writes follow reads are highly available.
  \item Read your writes, PRAM, and causal consistency are sticky available
        and not highly available.  Given at least two participants and any
        assignment of clients to them, consider the deployment that lets each
        client's participant give the sticky reveal rule's answers to that
        client's reads.  Some system has Deployment Availability and Partition
        Consistency for it.
  \item Recency is neither highly available nor sticky available.  For any
        level that rules out Lost Update or Write Skew, no system has
        Availability for an explicit relation that answers every read and
        Partition Consistency.
\end{enumerate}
\end{theorem}

\begin{proof}[Proof sketch]
(1) The reveal rule answers every read.  Each answer depends only on the read's
past, so later input never changes it, and the rule is monotone.  A monotone
relation inside the spec satisfies Lifted Monotonicity for itself: rows below
$j$ have the relation's own answer at $j$ as a common \outnoun.  The rule
returns only versions the replica holds, and its answers have none of the
anomalies these guarantees rule out.  Consider monotonic reads.  If an earlier
transaction of a client read a version of $x$, that version was revealed.  It
had then reached every replica of $x$, so it stays revealed at every later
read.  The client's later transaction reads after the earlier one commits, so
it reads that version of $x$ or a newer one.
For read committed and monotonic atomic view, the timestamp condition rules
out G1c, as in Theorem~\ref{thm:hat-levels}.  By
Theorem~\ref{thm:complete-cap}, each guarantee is highly available.

(2) The sticky reveal rule does the same for every sticky read.  The deployment
authorizes only answers of this one monotone relation inside the spec, so it
is monotone across partitions, and Proposition~\ref{prop:deployment-complete-cap}
gives the system.  Without stickiness, consider a client that writes $x$ at
one replica and then reads $x$ at another replica that has not received the
write.  That replica holds only the initial version of $x$, which read your
writes rules out.  No explicit relation inside the spec answers the read.
PRAM and causal consistency include read your writes.

(3) For recency, another client's write of $x$ commits before a read but never
reaches the reading replica, which again holds only the initial version.  The
read is sticky, so stickiness does not help.  For Write Skew, two transactions
at partitioned replicas each read the item the other writes.  Each replica can
return only the initial version, and together the answers form the cycle that
Write Skew names, as for Lost Update in Theorem~\ref{thm:hat-levels}.
\end{proof}

\paragraph{The algorithm for monotonic atomic view.}
The extended paper also gives pseudocode for monotonic atomic
view~\cite[App.~B]{bailis2013hatextended}.  A replica keeps a write in a
\emph{pending} set until it has received acknowledgments from every replica of
the items its transaction wrote, and then moves it to a \emph{good} set.  A
read of $x$ may carry a required timestamp.  Without one, the replica returns
its highest good version of $x$.  With one, it returns a good version at least
as new, or else the pending write with exactly that timestamp.  After each
read, the client raises the required timestamp of every item that the
returned write's transaction also wrote, and it clears these timestamps when
the transaction ends.  The input does not fix which acknowledgments a
replica has received, so each replica's good set may be any set of pending
stable transactions that depends only on the read's past.  The pending case
returns only a write whose transaction an earlier read of the same transaction
found good, so every answer is a version the replica holds.

\begin{proposition}[The algorithm provides monotonic atomic view]
\label{prop:hat-mav-algorithm}
For every choice of the replicas' good sets, and under the timestamp condition
of Theorem~\ref{thm:hat-levels}, the answers of the algorithm for monotonic
atomic view form an explicit relation inside $P_{\mathrm{MAV}}$ that answers
every read and is monotone.  So monotonic atomic view is highly available.
\end{proposition}

\begin{proof}[Proof sketch]
Each answer depends only on the read's past, so the relation is monotone, and
it answers every read.  The timestamp condition rules out G1c, as for read
committed.  For the anomaly of monotonic atomic view, suppose a transaction
reads $x$ from $T_j$ and later reads $y$, which $T_j$ also wrote.  The first
read raised the required timestamp of $y$ to $T_j$'s, so the second read
returns $T_j$'s version of $y$ or a newer one.  Theorem~\ref{thm:complete-cap}
gives the system.
\end{proof}

\paragraph{Differences in form.}
Our classification matches theirs and differs in form in three ways.
First, our framework separates two reasons a guarantee can be unavailable.
The unavailable isolation levels fail Lifted Monotonicity: answers that the two
sides give have no common extension.  Read your writes without stickiness, and
recency in every case, fail earlier.  Some read has no permitted answer at all
(items (2) and (3) of the proof of Theorem~\ref{thm:hat-taxonomy}), so the
question of Lifted Monotonicity never arises.  Second, monotonic writes needs
the version order to follow each session's order.  Bailis et al.\ obtain this
from client timestamps, and our inputs assume it.  Third, we model a predicate
as a fixed range of items, and general predicates, whose items depend on the
data, are not modeled.  Bailis et al.\ also list several recency guarantees
without formal definitions, among them bounds on staleness and safe, regular,
and linearizable register semantics.  We model only the recency guarantee
defined above.

\subsection{Registers}
\label{app:registers}

Section~\ref{sec:complete-cap} derives the Gilbert--Lynch theorem from a
replica that is cut off from a completed write and must still answer a read.
The example here differs: two concurrent writes reach both replicas before the
partition.

\paragraph{The causal register.}
The causal register differs from the register of
Appendix~\ref{app:cap-response} in two ways.  Its input records causal
precedence instead of real-time precedence, and its output records the write
that supplies each read instead of the value returned.  An input is a finite operation history $h$: a set of reads and writes of one
register $x$ at replicas, with a strict partial order $\to$ of causal
precedence.  Causal precedence is program order at each replica, plus an edge
from each write to the operations a replica invokes after receiving it,
closed transitively.  A set of operations of $h'$ is \emph{causally closed}
when it contains every operation that causally precedes one of its members.
We set $h\InputOrd h'$ when the operations of $h$ form a causally closed
subset of $h'$ and $\to$ restricted to them is the order of $h$.  Thus an
\indesc never adds an operation to the causal past of an existing one.

An output is a \emph{source map} $\rho$: a partial map assigning to some reads
of $h$ a write of $h$ or the initial value $v_0$, ordered by extension.  A
\emph{serialization} for replica $p$ is a total order of $p$'s operations and
all writes of $h$ that extends $\to$.  It explains a read of $p$ with source
$w$ when $w$ is the last write before the read, or $v_0$ if there is none.
Under causal consistency~\cite{lamport1978time,ahamad1995causal},
$\rho\in\fiber{h}$ when each replica has a serialization explaining every read
of that replica in the domain of $\rho$.  Under linearizability, a single
total order of all operations of $h$ extending $\to$ must explain every read
in the domain~\cite{herlihy1990linearizability}.  Real-time edges beyond $\to$
play no role in the example below.

Two replicas start empty.  Participant $p$ completes $w_p(x,\dog)$ while $q$
concurrently completes $w_q(x,\catv)$.  The replicas exchange both writes.
Then a partition separates them, and each serves a local read, $r_p$ at $p$
and $r_q$ at $q$.  Both writes causally precede both reads, and the writes are
concurrent.

Under linearizability, a source map that sends $r_p$ to $w_q(x,\catv)$ needs a
linearization with $w_p(x,\dog)$ before $w_q(x,\catv)$.  One that sends $r_q$
to $w_p(x,\dog)$ needs the opposite order.  No source map \outrel both, so the
register is not monotone across the partition $\{\{p\},\{q\}\}$ under any
deployment authorizing these two reads.

Under causal consistency the two replicas may serialize the concurrent writes
differently, so one source map sends $r_p$ to $w_q(x,\catv)$ and $r_q$ to
$w_p(x,\dog)$.  More generally, consider the deployment in which each replica
answers only its own reads ($E_p$ holds the rows whose source map is defined
only on reads of $p$).  Under it the causal register is monotone across
partitions, for every history.  For each participant $p$ whose run emits, let
$\rho_p$ be the last source map it emits, at input $i_p$; then $\rho_p$
\outrel every output $p$ emitted.  Let $j$ be a common \indesc.  We extend a
serialization $L_p$ of $i_p$ explaining $\rho_p$ to $j$ by appending the new
operations of $p$ and the new writes in any order extending $\to$.  No new
operation causally precedes an old one, so the result extends $\to$ and
explains the same reads.  With these serializations for the replicas that
emitted, and arbitrary ones for the others, the union of the $\rho_p$ lies in
$\fiber{j}$ and \outrel every emitted output.
Separability is not used.  Two answers to the same read remain incompatible,
but only the reading replica may emit either.

\subsection{Real-time causal consistency}
\label{app:rtc}

Mahajan et al.\ showed that a system can provide real-time causal consistency,
but no stronger consistency, while staying always available and one-way
convergent~\cite{mahajan2011consistency}.  We recover both parts with real-time
precedence recorded in the input.  The argument rests on a property of
products of independent views.

\paragraph{Independent products.}
In the real-time causal register below, each replica independently chooses
a view of its own reads.  The facts we need hold for any such product of
independent views.  Fix participants
$\mathcal N$ with $|\mathcal N|\ge2$.  For each $p$ let $(W_p,\le_p)$ be a poset
and, for each input $i$, let $V_p(i)\subseteq W_p$ be a nonempty antichain, the
\emph{views} $p$ may expose at $i$.  A \emph{partial tuple} $x$ assigns to each
$p$ a value $x_p$ in $W_p$ or $\bot$, ordered coordinate-wise with $\bot$ below
every value; $x$ is a \emph{restriction} of a tuple $y$ when each coordinate
$x_p$ is $\bot$ or $y_p$.  The \emph{product interface} has a row
$(i,x)$ for every restriction $x$ of every $y\in\prod_pV_p(i)$.  A
\emph{strengthening} $T$ selects at each input a set of tuples in
$\prod_pV_p(i)$, and its rows are their restrictions.  It is \emph{strict} when
it omits a row of the product interface, and \emph{locally conservative} when
every $v\in V_p(i)$ is the $p$-coordinate of a tuple it selects at $i$.  Under
the \emph{owner deployment}, $p$ may emit exactly the rows whose tuple is
$\bot$ outside coordinate $p$.

\begin{theorem}[Independent products]
\label{thm:product-frontier}
\leavevmode
\begin{enumerate}
  \item If views continue (for $i\InputOrd j$ and $v\in V_p(i)$ some $w\in
        V_p(j)$ has $v\le_p w$), the product interface is monotone across
        partitions under the owner deployment.
  \item No strict, locally conservative strengthening is monotone across
        partitions under its owner deployment.
\end{enumerate}
\end{theorem}

\begin{proof}
(1) Fix a participant run family and a common \indesc $j$.  If $p$ belongs to
the family, the remainder table of $p$'s run has a row $(k_p,u)$, and $u$ \outrel every output $p$ emitted
(Lemma~\ref{lem:run-residual}).  As $u$ is a restriction of some
$y\in\prod_qV_q(k_p)$, continuation gives $w_p\in V_p(j)$ with $y_p\le_p w_p$.
Otherwise pick any $w_p\in V_p(j)$.  Each output of $p$ is $\bot$ outside
coordinate $p$, so $w$ \outrel every emitted output.
(2) Strictness gives $y\in\prod_pV_p(i)$ that $T$ does not select at $i$.  By
local conservativity, the tuple $e_p$ that is $y_p$ at $p$ and $\bot$
elsewhere is a row of $T$ at $i$, authorized for $p$.  Let every participant
emit $e_p$ at $i$; every partition separates two of them.  A common \outnoun
at $i$ is a restriction of a selected $z$ with $y_p\le_p z_p$ for every $p$, so
$y=z$ by the antichain property, a contradiction.
\end{proof}

\paragraph{Real-time causal views.}
Consider the histories of Appendix~\ref{app:registers} with a second strict
partial order $\prec$: $a\prec b$ when $a$ ended before $b$ began.  We assume
that $a\to b$ implies $a\prec b$, so replicas propagate only completed writes.
We also assume that $\to$ totally orders the operations of each replica.  One
input \inrel another when it does so as a history and $\prec$ agrees on the
old operations.  A write $w$ is \emph{visible} to a read $r$ when $w\to r$.  A
\emph{real-time causal view} of replica $p$ assigns to every read $r$ of $p$ a
source $s(r)$ such that
\begin{enumerate}
  \item $s(r)$ is a write visible to $r$, or $v_0$ if none is;
  \item $s(r)\not\prec w$ for every write $w$ visible to $r$;
  \item if $r'\to r$ is a read of $p$ to which $s(r)$ is visible, then
        $s(r')=s(r)$.
\end{enumerate}
Intuitively, these items are Mahajan et al.'s checks for the order that $p$'s
reads expose, in which each visible write precedes the read and each
overwritten write precedes the source.  Item~(1) says a replica returns only
writes it has received.  Item~(2) is their real-time check ``time doesn't
travel backward.''  Item~(3) keeps one replica's overwrites consistent.  Their
reads return all concurrent latest writes, while ours return one, so a view
matches their check when that write is the only latest one.

Views exist.  As each read's source we can take the visible write that comes
latest in a linear extension of $\prec$.  Views also continue.  An \indesc
changes neither a read's causal past nor $\prec$ on it, so an old read keeps
its source.  A new read takes as source its latest visible write in a linear
extension of $\prec$ together with the edges $x\prec s(r)$, one for each old
read $r$ and each write $x\ne s(r)$ visible to $r$.  All views of $p$ at an
input are defined on the same reads.  So no view extends a different one, and
the views form an antichain.  In the real-time causal product interface, each
replica answers its own reads.  By Theorem~\ref{thm:product-frontier}(1), this
interface is monotone across partitions under the owner deployment.

\paragraph{Maximality.}
A strengthening is \emph{real-time blind} when a tuple it selects at an input
stays selected at the same history with fewer real-time edges, provided its
views are still real-time causal there.  Mahajan et al.'s implementations do
not observe real time, so they return the same answers at both inputs.

\begin{proposition}[Real-time causal maximality]
\label{prop:rtc-maximal}
No strict, real-time blind strengthening of the real-time causal product
interface is monotone across partitions under its owner deployment.
\end{proposition}

\begin{proof}
By Theorem~\ref{thm:product-frontier}(2) it suffices to show local
conservativity.  Let $v$ be a view of $p$ at $(h,\prec)$.  We add real-time
edges that leave $v$ as $p$'s only view: for each read $r$ of $p$ and each
write $x\ne s(r)$ visible to $r$, we place $x$ before $s(r)$.  Formally, let
$a\prec_v b$ when $a\prec b$, or when $a=x$ or $a\prec x$, and $s(r)=b$ or
$s(r)\prec b$, for such $r$ and $x$.  Items~(2) and~(3) and program order make
$\prec_v$ a strict partial order, and it contains $\to$.  At $(h,\prec_v)$,
$v$ is $p$'s only view.  By~(1) a source is a visible write, and a visible
$x\ne s(r)$ has $x\prec_v s(r)$, which~(2) forbids.  Every input has a row of
the strengthening, and rows are restrictions of selected tuples.  So the
strengthening selects some tuple at $(h,\prec_v)$, and that tuple's
$p$-coordinate is $v$.  Real-time edges only exclude views, so the tuple's
views are real-time causal at $(h,\prec)$.  Blindness then selects the tuple
there.
\end{proof}

Each hypothesis has a counterpart in Mahajan et al.'s theorem.  Their always
availability corresponds to two conditions here: $T$ selects a tuple at every
input, so every replica can answer every read, and the system has Deployment
Availability.  One-way convergence needs no separate hypothesis: the input
records which writes each replica has received, the environment supplies that
input, and views must respect $\to$.  Real-time blindness is their assumption
that an implementation does not observe real time, which their proof uses to
swap the real-time order of concurrent writes.  The proof of
Proposition~\ref{prop:rtc-maximal} needs it too, since a strengthening may read
real time from the input, for instance to break ties between writes that
overlap in real time.

\section{Classical Problems in Distributed Computing}
\label{app:additional-applications}

This appendix applies Lifted Monotonicity to classical problems of
distributed computing.  For each problem we give an \spectype: its inputs and
outputs, their orders, and its table.  We then find where Lifted Monotonicity
holds or fails, usually through the row-consistency condition
(Proposition~\ref{prop:common-expansion-form}), and read off what Complete
CALM and Complete CAP say about it.  For queues and search structures we also
ask how much a monotone spec can promise, and identify the strongest one under
stated hypotheses.  Classical results about failures, computability, and
termination depend on a model of computation, and none of them follows from
these specs.  We cite them where they bear on an example.

\subsection{Patterns and problems}
\label{app:classical-patterns}

Many specs follow one of four patterns, and each pattern settles Lifted
Monotonicity directly.  In the first, outputs grow in a join semilattice, as
the route sets of Figure~\ref{fig:ambiguity-frames} do under union.  In the
second, an output orders two events that the input leaves unordered, as the
serial orders of the same figure do.  In the third, outputs choose a proper
subset of competing values, as consensus chooses one proposal.  In the fourth,
an output blocks an input before it arrives, as closing the polls blocks every
ballot not yet received.  The first pattern gives Lifted Monotonicity, and
each of the other three, the patterns of non-monotonicity of
Section~\ref{sec:instances}, gives a witness to its failure.

\begin{lemma}[Growth in a semilattice]
\label{lem:accumulating}
Suppose the outputs of an \spectype form a join semilattice under
$\OutputOrd$.  If every output permitted at an input is still permitted at each
\indesc of that input, and the join of two outputs permitted at one input is
permitted there, then the spec is monotone.
\end{lemma}

\begin{proof}
Take $i\InputOrd j$ and a nonempty finite $X\subseteq\fiber{i}$.  Each member
of $X$ is permitted at $j$, so by induction on the size of $X$ their join is
permitted at $j$, and it \outrel every member of $X$.
\end{proof}

The common case is sets of facts ordered by inclusion, with union as the join;
counters under $\leq$, with the maximum as the join, are another.  The proof
uses only that each fiber contains, with any two outputs, a common \outnoun of
both, so the join need not be least.

Two events are \emph{concurrent} in an input when the input does not order
them, as for transactions $A$ and $B$ at input $\{A,B\}$.  If the spec lets an
output order them either way, two runs can choose opposite orders, and no
later output reconciles the two.

\begin{lemma}[Ordering concurrent events]
\label{lem:order-commitment}
Fix an input $j$.  Suppose that each output $o$ of an \spectype has a relation
$<_o$ on events, the pairs that $o$ places in order, such that
\begin{enumerate}
  \item $o\OutputOrd u$ implies ${<_o}\subseteq{<_u}$, and
  \item no $u\in\fiber{j}$ has both $a<_u b$ and $b<_u a$.
\end{enumerate}
If rows $(i_1,o_1)$ and $(i_2,o_2)$ with $i_1\InputOrd j$ and $i_2\InputOrd j$
have $a<_{o_1}b$ and $b<_{o_2}a$, then the spec is not monotone.  If
$i_1=i_2=j$, the fiber at $j$ is not directed.
\end{lemma}

\begin{proof}
A common \outnoun $u$ of $o_1$ and $o_2$ in $\fiber{j}$ would have $a<_u b$
and $b<_u a$ by~(1), which contradicts~(2).  So the row-consistency condition
fails, and by Proposition~\ref{prop:common-expansion-form} so does Lifted
Monotonicity.  When $i_1=i_2=j$, the same argument shows that $o_1$ and $o_2$
are incompatible in $\fiber{j}$.
\end{proof}

\begin{lemma}[Bounded number of values]
\label{lem:bounded-values}
Fix an input $j$ and a number $k$.  Suppose that each output $o$ of an
\spectype has a set $D(o)$ of values it decides, such that
\begin{enumerate}
  \item $o\OutputOrd u$ implies $D(o)\subseteq D(u)$, and
  \item every $u\in\fiber{j}$ has at most $k$ values in $D(u)$.
\end{enumerate}
If $k+1$ rows $(i_r,o_r)$ with $i_r\InputOrd j$ each decide a different value
$v_r\in D(o_r)$, then they cannot all grow to one output at $j$, so the spec is
not monotone.
\end{lemma}

\begin{proof}
A common \outnoun $u$ of the $k+1$ outputs in $\fiber{j}$ would have all
$k+1$ values in $D(u)$ by~(1), which contradicts~(2).  So the row-consistency
condition fails, and Proposition~\ref{prop:common-expansion-form} applies.
\end{proof}

\begin{lemma}[Blocking before arrival]
\label{lem:blocking}
Suppose the inputs of an \spectype are sets ordered by $\subseteq$, and its
outputs are pairs $(A,B)$ of disjoint sets, the accepted and the blocked
facts, ordered componentwise by $\subseteq$.  If rows $(i,(A,B))$ and
$(j,(A',B'))$ with $i\subseteq j$ have a fact $f\in B\cap A'$, the spec is not
monotone.
\end{lemma}

\begin{proof}
A common \outnoun $(A'',B'')$ of the two outputs in $\fiber{j}$ has
$f\in B\subseteq B''$ and $f\in A'\subseteq A''$, but $A''$ and $B''$ are
disjoint.  So the row-consistency condition fails, and
Proposition~\ref{prop:common-expansion-form} applies.
\end{proof}

Consider a variant of the election whose outputs record which ballots the
registrar has accepted and which it has blocked. An input is the set of ballots
received, so $ab\novote$ stands for $\{(v_1,a),(v_2,b)\}$, and the order on
ballot vectors is $\subseteq$. An output is a pair $(A,B)$ of disjoint sets of
ballots with $A\cup B\supseteq i$: every received ballot is accepted or
blocked. The spec includes the rows \begin{itemize} \item $(ab\novote,\
(ab\novote,\ \varnothing))$, with no decision about $v_3$; \item $(ab\novote,\
(ab\novote,\ \{(v_3,a),(v_3,b)\}))$, with $v_3$'s ballot blocked before it
arrives; and \item $(aba,\ (aba,\ \varnothing))$, with $v_3$'s ballot accepted
when it arrives. \end{itemize} The second and third rows have no common
\outnoun at $aba$, since $(v_3,a)$ would be both accepted and blocked, so by
Lemma~\ref{lem:blocking} the spec is not monotone.  Accepting a ballot when it
arrives reports evidence.  Blocking $v_3$'s ballot before it arrives is a claim
about an input not yet seen: that the ballot will not count.  Blocking every
ballot not yet received is closing the polls
(Figure~\ref{fig:election-frames}(d)), and it is a \commitment.  Because the
election has a fixed list of voters, once every voter has voted no ballot
remains to arrive, and closing the polls blocks nothing that could later be
accepted.

When Lemma~\ref{lem:order-commitment}, \ref{lem:bounded-values}, or
\ref{lem:blocking} applies, Complete CALM (Theorem~\ref{thm:complete-calm})
says that the runs of the spec are not consistent.  The No-\commitment test
(Proposition~\ref{prop:exposure-boundary}) adds that some output, at some state
reachable without \commitment, is not legal, so emitting it applies a
\commitment.  With $E=P$, Theorem~\ref{thm:complete-cap} and
Corollary~\ref{cor:complete-cap-any-system} add that no system for the spec has
Availability for $P$ and Partition Consistency.  The witness rows, each placed
on its own side of a partition, show why.  When Lemma~\ref{lem:accumulating}
applies, the runs are consistent, and the canonical system for $P$ needs no
\commitment and escapes CAP for $P$.

Every \commitment that the other three patterns require can be made by one
participant, as the leader of an epoch makes each per-slot \commitment in
Multi-Paxos (Section~\ref{app:consensus}), and a deployment records that
placement (Section~\ref{sec:commitment-conflict}).  Placing a \commitment at
one participant changes who makes it, not whether the spec requires it.  The
rest of this appendix applies the patterns to broadcast, consensus, leader
election, snapshots, queues, and search structures.

\begin{table}[t]
  \centering
  \footnotesize
  \setlength{\tabcolsep}{3pt}
  \begin{tabular}{@{}>{\raggedright\arraybackslash}p{.36\columnwidth}>{\raggedright\arraybackslash}p{.44\columnwidth}l@{}}
    \toprule
    Problem & Output (order) & Monotone \\
    \midrule
    Reliable broadcast & delivered messages ($\subseteq$) & yes \\
    Causal broadcast & causally closed set of delivered messages ($\subseteq$) & yes \\
    Grow-only counter & count ($\leq$) & yes \\
    Total-order broadcast & delivery sequence (prefix) & no \\
    Consensus, $k$-set agreement & decision map (extension) & no \\
    Leader election, participants listed & leader or $\bot$ ($\bot$ lowest) & yes \\
    Leader election, participants unlisted & leader or $\bot$ ($\bot$ lowest) & no \\
    Global snapshot & a cut per scan (extension) & yes$^*$ \\
    Atomic snapshot & a view per scan (extension) & no$^*$ \\
    Causal queue & dequeue history (prefix) & yes \\
    Forward-reachability lookup & valid claims ($\subseteq$) & yes \\
    \bottomrule
  \end{tabular}
  \caption{Classical problems in this appendix.  Each row names the spec's
  outputs with their order and says whether the spec is monotone.
  $^*$Monotone across partitions when each process answers its own scans
  (Section~\ref{app:snapshots}).}
  \label{tab:classical-problems}
\end{table}

Table~\ref{tab:classical-problems} lists the problems treated below.

\subsection{Broadcast}
\label{app:broadcast}

In \emph{reliable broadcast}, an input is a finite set of broadcast messages,
and an output is a set of delivered messages.  An output is permitted when it
contains only broadcast messages, so Lemma~\ref{lem:accumulating} applies.
Gossip and anti-entropy protocols~\cite{demers1987epidemic} propagate sets of
updates in the same way.  Reliable broadcast also requires that a message
delivered by one correct process is eventually delivered by every correct
process~\cite{hadzilacos1994modular}.  That requirement concerns termination,
and the spec does not express it.

In \emph{causal broadcast}, a message is delivered only after every message
causally before it.  An input is a finite causal history of broadcasts: a
finite set of events with a strict partial order of causal precedence.  A
history $h'$ \inrel $h$, $h\InputOrd h'$, when the events of $h$ form a
causally closed subset of $h'$ that carries the restricted order, so an
\indesc never adds an event to the causal past of an old one.  An output is a
set of delivered messages, permitted when it is causally closed in the input.
Such a set stays causally closed as the history grows, and a union of causally
closed sets is causally closed, so Lemma~\ref{lem:accumulating} applies.

A \emph{grow-only counter}~\cite{shapiro2011crdt} takes as input a finite set
of increments and outputs a count, ordered by $\leq$.  A count is permitted
when it is at most the number of increments.  A permitted count stays
permitted as increments arrive, and the larger of two permitted counts is
permitted, so Lemma~\ref{lem:accumulating} applies with the maximum as the join.

\emph{Total-order broadcast}~\cite{defago2004total} has the inputs of reliable
broadcast, but an output is a delivery sequence of distinct broadcast
messages, ordered by prefix.  Two delivery sequences have a common \outnoun
only when one is a prefix of the other, which is the agreement that
total-order broadcast requires.  A sequence places $a$ before $b$ when it lists
$a$, and lists $b$, if at all, after $a$.  An extension of the sequence keeps
this placement, and no sequence places two messages in both orders.  The rows
$(\{a\},[a])$ and $(\{b\},[b])$ have the common \indesc $\{a,b\}$, so by
Lemma~\ref{lem:order-commitment} the spec is not monotone: two processes that
each deliver their own message first cannot be reconciled.

The specs thus separate the broadcast problems as the classical results do.
In asynchronous systems with crash failures, reliable and causal broadcast can
be implemented without agreement~\cite{hadzilacos1994modular}, while
total-order broadcast and consensus are equivalent
problems~\cite{chandra1996unreliable}.  Total-order broadcast needs a
\commitment, and Section~\ref{app:consensus} shows where consensus protocols
make it.

\subsection{Consensus: committed authority, then monotone evidence}
\label{app:consensus}

In consensus, inputs are finite sets of proposals, at most one per
participant, ordered by inclusion.  Outputs are partial decision maps from
participants to values, ordered by extension, and every permitted output
decides at most one value, which some participant proposed.  Two participants
that propose different values and each decide their own give two rows with a
common \indesc, so by Lemma~\ref{lem:bounded-values} with $k=1$ the spec is not
monotone.  The same holds for $k$-set agreement~\cite{chaudhuri1993more}, which
allows up to $k$ decided values: $k+1$ participants that each decide their own
proposal give the witness of Lemma~\ref{lem:bounded-values}.  The
impossibility result of Fischer, Lynch, and
Paterson~\cite{fischer1985impossibility} is a different statement: no
deterministic asynchronous protocol solves consensus if one process may crash.
That result concerns termination, which the spec does not express.

Practical consensus protocols such as Paxos~\cite{lamport2001paxos} and
Raft~\cite{ongaro2014raft} still make a \commitment for every slot, and their
structure determines who makes it.  Each epoch, a ballot in Paxos or a term in
Raft, has an authority, such as a leader, and a rule fixes which epoch succeeds
which.
Within an epoch, the authority chooses the value proposed for each slot of a
replicated log, and voters certify it.  Once the authority and the succession
rule are fixed, what remains to report is evidence: certified epochs, votes,
and certified decisions.  We show that reporting this evidence needs no
further \commitment.

Inputs are finite sets of protocol messages (proposals, votes,
certificates), ordered by inclusion.  An output is a triple $(E,V,D)$: a
certified epoch chain $E$, which is a finite sequence of epochs, each
certified under the authority of its predecessor; a set $V$ of votes or
certificates; and a set $D$ of decisions, at most one value per epoch and slot,
each certified for that epoch and slot.  Outputs are ordered by
\begin{equation}
  (E,V,D)\OutputOrd(E',V',D')
  \quad\Longleftrightarrow\quad
  E\preceq_{\rm pref}E',\quad V\subseteq V',\quad D\subseteq D',
  \label{eq:authority-refinement}
\end{equation}
where $E\preceq_{\rm pref}E'$ means that $E$ is a prefix of $E'$.  Write $c(i)$
for the certified evidence at input $i$: the epoch chain that the messages of
$i$ certify under the succession rule, the votes and certificates in $i$, and
the decisions certified in $i$.  The \emph{evidence spec} permits at $i$
exactly the outputs that $c(i)$ \outrel.  Protocol safety is the assumption
that $c(i)$ is defined at every input and grows with the input:
$i\InputOrd j$ implies $c(i)\OutputOrd c(j)$.  For Paxos, this is the
property that quorum intersection provides: once a value is certified for a
slot, no later certificate for that slot names another
value~\cite{lamport2001paxos}.

\begin{proposition}[Monotone evidence]
\label{prop:monotone-evidence}
If $i\InputOrd j$ implies $c(i)\OutputOrd c(j)$, the evidence spec is
monotone, and every output at every state reachable without \commitment is
legal.
\end{proposition}

\begin{proof}
Take rows $(i_r,o_r)$ with a common \indesc $j$.  Each
$o_r\OutputOrd c(i_r)\OutputOrd c(j)$, and $c(j)$ is permitted at $j$.  So the
row-consistency condition holds, and
Proposition~\ref{prop:common-expansion-form} gives Lifted Monotonicity.  The
second claim follows by Proposition~\ref{prop:exposure-boundary}.
\end{proof}

The same reasoning applies at a single state.  If every row remaining at a
state \outrel an output $o$ of the current input, emitting $o$ removes no row,
so $o$ is a legal output.  Under protocol safety, this is the case for an
authority change or a certified decision once it is part of the evidence, so
emitting it needs no further \commitment.

The \commitments lie in the choices that the evidence spec takes as given.
Competing proposals for an undecided slot have incompatible extensions of $D$,
since a permitted output decides at most one value per epoch and slot.  So
choosing the value for a slot is a \commitment, and every slot needs one.  Once
the authority of an epoch is fixed, only that authority may choose values for
the epoch's slots.  Each choice is then a \commitment that one participant
makes locally, like a voter's ballot in Section~\ref{sec:election-story}, and
it needs no coordination with the other participants.  If the succession rule
were not fixed, choosing among competing next epochs would require a
\commitment as well.  Protocols place the per-slot \commitments differently.
In Multi-Paxos and Raft, the leader of the current epoch makes them.
Mencius~\cite{mao2008mencius} instead assigns each slot to a fixed replica in
advance, and that replica makes the slot's \commitment, possibly choosing a
no-op.

\paragraph{Shared memory.}
Herlihy's consensus hierarchy~\cite{herlihy1991waitfree} ranks shared objects
by the largest number of processes for which they solve consensus wait-free,
together with read/write registers.  Our framework treats the two halves of
that result differently.  Consider a shared object as a single participant
with a queue of requests, an automaton in the sense of
Section~\ref{sec:operational-models}.  The queue orders concurrent requests.
From sets of requests to request sequences, it is the total-order broadcast
spec of Section~\ref{app:broadcast}, so ordering them is a \commitment, and
the object makes it locally, as the authority of an epoch does above.

The positive half of the hierarchy follows from a monotonicity test.  With
compare-and-swap, each process applies compare-and-swap with its proposal to
a location that starts empty, and its response names the value of the first
compare-and-swap, the same value in every response.  Each process decides the
value in its response.  That decision stage is a monotone function of the
object's responses, so it needs no \commitment on them (Proposition~\ref{prop:guard-test}), and the object's ordering is the only \commitment the pipeline needs (Theorem~\ref{thm:sufficient-guards}).  Each process decides after one operation of its own, so the
algorithm is wait-free for any number of processes.  Test-and-set gives the
same for two processes: each process first writes its proposal to its own
register, the winner of test-and-set decides its own proposal, and the loser
reads and decides the winner's.

The impossibility half, for example that read/write registers do not solve
consensus wait-free for two processes, is a termination property, and safety
alone does not separate the objects.  Over registers, every process can wait
until one designated process has written its proposal and then decide it.
That decision stage is also monotone, so it needs no \commitment and never
decides two values, but the other processes wait forever if the designated
one crashes.  The specs can represent the termination requirement, by asking
that each process decide after finitely many of its own steps whatever the
others do, and Herlihy's valency argument can be restated in that
representation.  The restatement is a translation of his proof, and the CALM
theorems give no alternative proof.  The same holds for the bounds on
renaming~\cite{attiya1990renaming,herlihy1999topological}.

\subsection{Leader election}
\label{app:leader}

Leader election asks that exactly one participant become the leader.  In
asynchronous networks without failures, rings of processes with unique
identifiers elect the largest identifier deterministically, as in the
algorithms of Le Lann~\cite{lelann1977distributed}, Chang and
Roberts~\cite{chang1979improved}, and Hirschberg and
Sinclair~\cite{hirschberg1980decentralized}.  In a ring of anonymous
processes no deterministic algorithm elects a leader~\cite{angluin1980local}.
With crash failures and a majority of correct processes, $\Omega$, which
eventually names the same correct leader at every correct process, is the
weakest failure detector for solving consensus~\cite{chandra1996weakest}.

The spec one would write first gives the expected answer.  Its input is the
set of candidates heard so far, ordered by inclusion, and an output is $\bot$,
for no leader yet, or a candidate heard, ordered with $\bot$ below every
leader and leaders incomparable.  Two candidates, each elected at the input
where only it has been heard, have no common \outnoun at the input where both
have been heard.  So the spec is not monotone, and by
Theorem~\ref{thm:complete-cap} no system has Availability for $P$ and
Partition Consistency.  Electing a leader then needs a \commitment, such as
closing the list of candidates, as the registrar closes the polls in
Section~\ref{sec:election-story}.  Without one, the two sides of a partition
can each elect a leader.

In practice, though, leader election usually runs in a known, finite
population, such as the members of a replica group or the processes of a
ring.  Including the list of all participants in the input, the
$\mathsf{All}$ relation of Ameloot et al.\ (Section~\ref{sec:systems}),
changes the answer.  The input then also records, for some of the listed
participants, whether each is a candidate, and inputs are ordered by extension
of these records.  The spec is functional.  Its output is the largest
candidate once every listed participant has reported, or $\bot$ if there is
none, and $\bot$ before that.  The spec is monotone, since more input never
changes a report already received, so once every participant has reported the
output is fixed.  By Functional CALM (Theorem~\ref{thm:functional-calm}), its
runs are consistent, and electing the largest candidate needs no
\commitment: a participant that has not heard from everyone emits $\bot$ and
waits.  It waits for a report from each listed participant, which is
evidence, and not for the absence of further candidates, which would block
inputs before they arrive (Lemma~\ref{lem:blocking}).  In the ring algorithms
above, a message that returns to its sender has passed every process, which
tells the sender what the list of participants would.

So the leader election that systems usually run needs no \commitment.  This
follows directly from the result of Ameloot et al.\ that the list of all
participants enlarges what can be computed without coordination.  The cost is
availability: a participant cut off from any listed participant waits
forever.  Fault-tolerant protocols such as Raft elect by majority vote instead,
and each voter's ballot is a \commitment that the voter makes locally, as in
Section~\ref{app:consensus}.

Leader election is also the running example of Baccaert and
Ketsman~\cite{baccaert2026spectrum}, and their results for transducer networks
match both specs.  Without the list of all nodes, leader election is not
coordination-free: groups of nodes that have heard from different nodes can
each elect the smallest identifier they know, and a node with a smaller
identifier, heard later, contradicts both~\cite[Prop.~37]{baccaert2026spectrum}.
With the list and an order on identifiers, and every node a candidate, each
node elects the smallest identifier in the list without
communicating~\cite[Ex.~6]{baccaert2026spectrum}.  The list alone is not
enough: without an order on identifiers or another way to break symmetry, no
transducer network elects a leader~\cite[Thm.~1]{baccaert2026spectrum}, the
counterpart of Angluin's result for anonymous rings.  Our spec with the list
presumes such an order, since it elects the largest candidate.  The two
frameworks differ in what they call coordination.  For Baccaert and Ketsman, a
behavior that cannot be computed without the list of all nodes needs
coordination (Section~\ref{sec:query-instance}).  In our framework the list is
part of the input, and the question is whether, given that input, the spec
requires a \commitment.  For election by largest candidate it does not.

\subsection{Snapshots}
\label{app:snapshots}

Chandy and Lamport's algorithm records a consistent global state of a running
message-passing system without stopping it~\cite{chandy1985distributed}.  An
atomic snapshot object answers each scan with a view of all components, and a
single order of all scans and updates must explain every
answer~\cite{afek1993atomic}.  Both answer requests for a view of the system.
The specs below show that the first is monotone across partitions when each
process answers its own scans, and the second is not.

An input is a finite causal history of updates and scans, each event at one
participant, ordered as in Section~\ref{app:broadcast}.  For a scan $s$ of a
history $h$, a \emph{cut for $s$} is a set $C$ of updates of $h$ that is
causally closed among updates, contains every update causally before $s$, and
contains no update causally after $s$.  A global snapshot records a
consistent cut~\cite{chandy1985distributed}.  The two further conditions tie
the cut to the scan.  A \emph{view} of participant $p$ at $h$ assigns a cut to
each scan of $p$ in $h$, and views are ordered by extension.  The
\emph{global-snapshot spec} is the product interface of these views
(Appendix~\ref{app:rtc}): its outputs are partial tuples with one view per
participant.  Under the owner deployment, each participant answers only its
own scans.

\begin{proposition}[Global and atomic snapshots]
\label{prop:snapshots}
\leavevmode
\begin{enumerate}
  \item The global-snapshot spec is monotone across partitions under the owner
        deployment.
  \item The atomic snapshot is the strengthening of the global-snapshot spec
        that selects the tuples of views explained by one strict total order
        $L$ of all events of the input that extends causal precedence: the
        view of each scan is the set of updates before it in $L$.  With two
        participants that can each perform an update and then a scan, it is
        not monotone across partitions under the owner deployment.
\end{enumerate}
\end{proposition}

\begin{proof}
(1) The views of a participant at one input are defined on the same scans,
so they form an antichain, and the updates causally before a scan form a cut
for it, so views exist.  By Theorem~\ref{thm:product-frontier}(1), it
suffices that views continue.  An \indesc adds nothing to the causal past of
an old event, so the cut of an old scan is still a cut for it.  A new scan can take as its cut the
updates causally before it.

(2) The updates before a scan in such an $L$ form a cut for it, since $L$
extends causal precedence, so the selected tuples are tuples of views.  Let
participant $p$ perform update $u_p$ and then scan $s_p$, and let $q$
perform $u_q$ and then $s_q$, with no messages between them.  The answer
$\{u_p\}$ to $s_p$ is explained by the order $u_p\,s_p\,u_q\,s_q$, and the
answer $\{u_q\}$ to $s_q$ by $u_q\,s_q\,u_p\,s_p$.  An order explaining both
would place $u_p$ before $s_p$, $s_p$ before $u_q$, $u_q$ before $s_q$, and
$s_q$ before $u_p$.  So the two answers, emitted by $p$ and $q$ on two sides of
a partition, have no common \outnoun.
\end{proof}

By Proposition~\ref{prop:deployment-complete-cap}, the owner deployment's
canonical system for the global-snapshot spec has Deployment Availability and
Partition Consistency, and no system for the atomic snapshot has both.  The
global-snapshot spec permits the two answers of the proof together: neither
cut contains the other, which an atomic snapshot forbids.  Atomic snapshots
nonetheless have wait-free implementations from atomic
registers~\cite{afek1993atomic}.  In the model of
Section~\ref{app:consensus}, the memory that serves those registers is a
single participant, and its order of service makes the \commitments.  In the
proof, by contrast, each process answers its own scan.

\subsection{Queues}
\label{app:queue}

A replicated queue cannot report a total dequeue order of two concurrent
enqueues without a \commitment, by Lemma~\ref{lem:order-commitment}.  We
identify what a queue can report without one, and show that nothing more is
possible when each item can be dequeued as soon as the items causally before
it have been.

An input is a finite causal history of enqueues, ordered as in
Section~\ref{app:broadcast}.  Each enqueue adds a distinct item, which we
identify with the event.  An output is a \emph{dequeue history} $(D,<)$, a
partially ordered history of dequeue events: a set $D$ of the enqueued items
dequeued so far, with a strict order $<$ on $D$ that the queue guarantees
among them.  Dequeue histories are ordered by prefix:
$(D,<)\OutputOrd(D',<')$ when $D\subseteq D'$, the order $<'$ agrees with $<$
on $D$, and no item of $D'\setminus D$ comes before an item of $D$ under
$<'$.  On totally ordered histories this is the usual prefix order.  The
\emph{causal queue} permits at $h$ the dequeue histories $(D,\to)$ in which $D$
is a causally closed set of items of $h$ and $\to$ is causal precedence
restricted to $D$.

\begin{proposition}[Causal queue]
\label{prop:causal-queue}
\leavevmode
\begin{enumerate}
  \item The causal queue satisfies Lifted Monotonicity.
  \item Consider any spec with these inputs and outputs whose fiber at $h$ is
        directed, and which permits at $h$, for each item $b$ of $h$, a
        dequeue history whose set consists of $b$ and the items causally
        before it.  Every order that a dequeue history it permits at $h$
        guarantees between two distinct items of $h$ is causal precedence.
  \item Consider a spec that satisfies Lifted Monotonicity.  Suppose that its
        dequeue histories list only items of the input, contain every item
        causally before a listed one, order every causally ordered pair of
        listed items causally, and never place an item before itself, and that
        every causally closed set of items is the set of some dequeue history
        it permits.  Then it permits exactly the dequeue histories of the
        causal queue.
\end{enumerate}
\end{proposition}

\begin{proof}
(1) We check the row-consistency condition.  Given rows $(h_r,(D_r,\to))$
with a common \indesc $j$, the union $D$ of the $D_r$ is causally closed in
$j$, and $(D,\to)$ \outrel each $(D_r,\to)$: causal precedence on $D_r$ is the
same in $h_r$ and $j$, and an item causally before an item of $D_r$ lies in
$D_r$.

(2) Suppose a permitted dequeue history $(D,<)$ has $a<b$ with $a\neq b$, and
$a$ is not causally before $b$.  Let $(D_b,<_b)$ be a permitted dequeue history
whose set consists of $b$ and the items causally before it, so
$a\notin D_b$.  A common \outnoun $(D'',<'')$ in the fiber has $a<''b$ with
$b\in D_b$, so $a$ comes before an item of $D_b$, and the prefix order forces
$a\in D_b$, a contradiction.

(3) A spec that satisfies Lifted Monotonicity has directed fibers, and the set
consisting of an item and the items causally before it is causally closed.
So by~(2) and the remaining hypotheses, each permitted dequeue history orders
exactly the causally ordered pairs of its causally closed set, and is a dequeue
history of the causal queue.  Conversely, a dequeue history $(D,\to)$ of the
causal queue has the same set as some permitted dequeue history, which by the
same argument is $(D,\to)$.
\end{proof}

The hypothesis of~(2) says that an item can be dequeued as soon as the items
causally before it have been.  It is the counterpart of local conservativity
in Theorem~\ref{thm:product-frontier}(2).  Under it, a spec that guarantees any
order between two concurrent items has a fiber that is not directed, so it is
not monotone.  A queue that guarantees a total dequeue order is one example.

\subsection{Search structures}
\label{app:search}

B-link trees keep a search starting at a split node correct by following a
right link, so the parent can be repaired later~\cite{lehman1981efficient}.  We
show that the claims such searches support only accumulate, compare them with
the least monotone relaxation of exact lookup, and show that for B-link trees
the two coincide.

Inputs are finite legal restructuring histories $h$, ordered by prefix.  The
structure after $h$ is a finite directed graph whose nodes are
\emph{positions}; each edge is a \emph{child edge} or a \emph{link edge}, and
some positions hold keys.  For each key $k$, every position has at most one
outgoing \emph{search edge for $k$}, the edge a search for $k$ follows there;
the \emph{search for $k$ from $\pi$} follows search edges for $k$ from position
$\pi$.  An \emph{occurrence} of $k$ is a position holding $k$, and a search
for $k$ stops there: an occurrence of $k$ has no search edge for $k$.
\emph{Current} means present after $h$.  The \emph{claim} $(k,\pi)$ is
\emph{valid} at $h$ when the search for $k$ from $\pi$ reaches a current
occurrence of $k$ after finitely many edges, and \emph{exact} when it does so
by child edges alone.  Each restructuring step must satisfy:
\begin{enumerate}
  \item positions keep stable identities;
  \item if $\pi\to\pi'$ was the search edge for $k$ at $\pi$ before the step,
        the search for $k$ from $\pi$ after the step reaches some position
        that the search for $k$ from $\pi'$ reached before it; and
  \item from every position that held $k$ before the step, the search for $k$
        reaches a current occurrence of $k$ after it.
\end{enumerate}
Condition~(2) does not ask that the search from $\pi$ still reach $\pi'$.
Posting a new node's separator into its parent changes the parent's search
edge for some keys from the old child to the new node, which lies to the
right of the old child, so the search no longer reaches the old child.  It
reaches the new node, which the old search reached through the old child's
right link.  Deletion and reclamation can violate (2) and (3) and are
excluded.

\begin{proposition}[Persistent forwarding]
\label{prop:persistent-forwarding}
Under (1)--(3), every claim valid at $h$ is valid at every \indesc of $h$.
\end{proposition}

\begin{proof}
By induction on the history it suffices to treat one step.  We show, by
induction on $n$, that if the search for $k$ from $\pi$ reaches an occurrence
of $k$ after $n$ edges before the step, then after the step it reaches a
current occurrence of $k$.  If $n=0$, $\pi$ held $k$, and (3) applies.
Otherwise let $\pi\to\pi'$ be the first edge.  By (2), the search from $\pi$
after the step reaches a position $\rho$ that the search from $\pi'$ reached
before it.  Before the step, each position had at most one search edge for $k$
and the search stopped at the occurrence, so $\rho$ lies on the path from
$\pi'$ to the occurrence, fewer than $n$ edges from it.  By induction, the
search from $\rho$ after the step reaches a current occurrence, and so does the
search from $\pi$.
\end{proof}

A lookup spec reports sets of claims.  For a map $V$ that gives a set $V(h)$
of claims at each input, the spec \emph{bounded by $V$} permits at $h$ every
set of claims contained in $V(h)$, ordered by inclusion.  A spec bounded by a
smaller map promises more.  Exact lookup is bounded by the map $X$ of exact
claims, and forward-reachability lookup by the map of valid claims.

\begin{proposition}[Remembered claims]
\label{prop:remembered-claims}
\leavevmode
\begin{enumerate}
  \item The spec bounded by $V$ satisfies Lifted Monotonicity if and only if
        $h\InputOrd h'$ implies $V(h)\subseteq V(h')$.
  \item Let $\widehat X(h)$ be the set of claims exact at some input below
        $h$.  The spec bounded by $\widehat X$ satisfies Lifted Monotonicity,
        and every map $V\supseteq X$ for which the spec bounded by $V$ does
        has $\widehat X\subseteq V$.
  \item Let $\widetilde X(h)$ be the set of claims exact at some input below an
        input whose current structure is that of $h$.  The spec bounded by
        $\widetilde X$ satisfies Lifted Monotonicity, and every map
        $V\supseteq X$ that depends only on the current structure, and for
        which the spec bounded by $V$ satisfies Lifted Monotonicity, has
        $\widetilde X\subseteq V$.
  \item Valid claims contain $\widetilde X$, and they equal $\widetilde X$ at
        every input if and only if every valid claim at every input is in
        $\widetilde X$.
\end{enumerate}
\end{proposition}

\begin{proof}
(1) First assume Lifted Monotonicity.  A claim $c\in V(h)$ gives the row
$(h,\{c\})$, which has an \outnoun at $h'$, so $c\in V(h')$.  Now assume that
$V$ grows.  The union of finitely many sets permitted at inputs below $h'$ is
contained in $V(h')$.  (2) $\widehat X$ grows by construction.  A claim exact
at $h_0\InputOrd h$ lies in $V(h_0)$, and so in $V(h)$ by~(1).  (3) Whether a
history can be continued by a step depends only on its current structure.  So
if $h_0$ lies below an input $h'$ with the structure of $h$, and $h$ grows to
$k$, then $h'$ grows to an input with the structure of $k$, and $h_0$ lies
below it.  Hence $\widetilde X$ grows.  For $V$ as in~(3), a claim exact at
$h_0\InputOrd h'$ lies in $V(h')$ by~(1), and $V(h')=V(h)$.  (4) The valid
claims depend only on the current structure, contain $X$, and grow by
Proposition~\ref{prop:persistent-forwarding}, so~(3) applies.
\end{proof}

A lookup inspects the current structure, not the history, so~(3) gives the
relevant comparison.  By~(4), forward-reachability lookup is the strongest
lookup that depends only on the structure and satisfies Lifted Monotonicity
exactly when every valid claim was exact at some earlier point of some history
that reaches the same structure.  For a key that a search from $\pi$ reaches
through link edges, this asks that some history to the same structure first
place the key where a search from $\pi$ finds it by child edges.  We show that
this holds for B-link trees.

In our model of Lehman and Yao's trees, each node has a level and a low key.
The nodes of a level are ordered by low key, each has a link edge to its right
neighbor, and a node covers the keys above its low key up to the low key of
its right neighbor.  Leaves hold keys, and an internal node has a child edge to
each node of the level below that has been \emph{posted} to it.  A search for
$k$ at a node that does not cover $k$ follows the link edge; at an internal
node that covers $k$, it follows the child edge to the posted child with the
largest low key below $k$.  There are four steps: insert a key into the leaf
that covers it; half-split a node at a separator, creating a right sibling
that covers the keys above the separator; post a node to the level above; and
grow a new root above the top level, with child edges to the old root and to a
node split off from it.  Steps are atomic, split points are arbitrary, and
there are no deletions.  Each step satisfies (1)--(3), and searches stop at
leaves.

\begin{proposition}[B-link trees]
\label{prop:blink}
In the B-link tree model, every claim valid at $h$ lies in
$\widetilde X(h)$.  So the valid claims equal $\widetilde X$, and
forward-reachability lookup is the strongest lookup that depends only on the
current structure and satisfies Lifted Monotonicity.
\end{proposition}

\begin{proof}
Let $(k,\pi)$ be valid at $h$, with current tree $G$.  Then $\pi$ has a low
key below $k$, and $k$ has been inserted.  Remove from $G$ the levels above
the level of $\pi$, the nodes of $\pi$'s level whose low keys lie strictly
between those of $\pi$ and $k$, and, on each lower level, the nodes not yet
posted that lie between the last posted node below $k$ and $k$.  The result is
a tree $W$ in which the search for $k$ from $\pi$ uses child edges only, so
$(k,\pi)$ is exact in $W$.  Every tree whose nodes, keys, and posted nodes are
among those of a reachable tree $G$ can be completed to $G$ by steps: insert
the missing keys, then, level by level from the bottom, split off the missing
nodes, grow the root, and post the missing nodes.  So some history reaches $W$
and then $G$, and $(k,\pi)$ is exact at an input below one whose current
structure is $G$.  By Proposition~\ref{prop:remembered-claims}(4), the valid
claims equal $\widetilde X$, and~(3) gives the rest.
\end{proof}

\begin{acks}
Thanks to Conor Power, whose persistent questioning of the CALM canon prompted this work.
Much gratitude to Peter Alvaro, Natacha Crooks, Chris Douglas, Tyler Hou, Martin
Kleppmann, Paris Koutris, Edward Lee, Shulu Li, Vincent Liu, David Chu McElroy,
Sarah Morin, and Dan Suciu for helpful feedback. Thanks also to
anonymous reviewers of this and previous submissions for frank feedback that pushed
this work deeper.

Generative AI tools (Infinity, Kiro, Claude and ChatGPT) were used as interactive writing and
reviewing assistants during the preparation of this work, including drafting
and revising text, and synthesizing Lean proofs for all mathematical arguments.  All
content was reviewed, validated, and approved by the author.
\end{acks}

\bibliographystyle{ACM-Reference-Format}
\bibliography{references}


\begin{thebibliography}{82}


\ifx \showCODEN    \undefined \def \showCODEN     #1{\unskip}     \fi
\ifx \showISBNx    \undefined \def \showISBNx     #1{\unskip}     \fi
\ifx \showISBNxiii \undefined \def \showISBNxiii  #1{\unskip}     \fi
\ifx \showISSN     \undefined \def \showISSN      #1{\unskip}     \fi
\ifx \showLCCN     \undefined \def \showLCCN      #1{\unskip}     \fi
\ifx \shownote     \undefined \def \shownote      #1{#1}          \fi
\ifx \showarticletitle \undefined \def \showarticletitle #1{#1}   \fi
\ifx \showURL      \undefined \def \showURL       {\relax}        \fi
\providecommand\bibfield[2]{#2}
\providecommand\bibinfo[2]{#2}
\providecommand\natexlab[1]{#1}
\providecommand\showeprint[2][]{arXiv:#2}

\bibitem[Abadi and Lamport(1991)]%
        {abadi1991refinement}
\bibfield{author}{\bibinfo{person}{Mart{\'\i}n Abadi} {and}
  \bibinfo{person}{Leslie Lamport}.} \bibinfo{year}{1991}\natexlab{}.
\newblock \showarticletitle{The Existence of Refinement Mappings}.
\newblock \bibinfo{journal}{\emph{Theor. Comput. Sci.}} \bibinfo{volume}{82},
  \bibinfo{number}{2} (\bibinfo{year}{1991}), \bibinfo{pages}{253--284}.
\newblock
\href{https://doi.org/10.1016/0304-3975(91)90224-P}{doi:\nolinkurl{10.1016/0304-3975(91)90224-P}}


\bibitem[Abiteboul et~al\mbox{.}(1995)]%
        {abiteboul1995foundations}
\bibfield{author}{\bibinfo{person}{Serge Abiteboul}, \bibinfo{person}{Richard
  Hull}, {and} \bibinfo{person}{Victor Vianu}.}
  \bibinfo{year}{1995}\natexlab{}.
\newblock \bibinfo{booktitle}{\emph{Foundations of Databases}}.
\newblock \bibinfo{publisher}{Addison-Wesley}.
\newblock


\bibitem[Abramsky and Jung(1994)]%
        {abramsky1994domain}
\bibfield{author}{\bibinfo{person}{Samson Abramsky} {and}
  \bibinfo{person}{Achim Jung}.} \bibinfo{year}{1994}\natexlab{}.
\newblock \showarticletitle{Domain Theory}.
\newblock In \bibinfo{booktitle}{\emph{Handbook of Logic in Computer Science}}.
  Vol.~\bibinfo{volume}{3}. \bibinfo{publisher}{Oxford University Press},
  \bibinfo{pages}{1--168}.
\newblock


\bibitem[Adya(1999)]%
        {adya1999weakconsistency}
\bibfield{author}{\bibinfo{person}{Atul Adya}.}
  \bibinfo{year}{1999}\natexlab{}.
\newblock \emph{\bibinfo{title}{Weak Consistency: A Generalized Theory and
  Optimistic Implementations for Distributed Transactions}}.
\newblock \bibinfo{thesistype}{Ph.\,D. Dissertation}.
  \bibinfo{school}{Massachusetts Institute of Technology},
  \bibinfo{address}{Cambridge, MA, USA}.
\newblock
\urldef\tempurl%
\url{https://dspace.mit.edu/handle/1721.1/149899}
\showURL{%
\tempurl}


\bibitem[Afek et~al\mbox{.}(1993)]%
        {afek1993atomic}
\bibfield{author}{\bibinfo{person}{Yehuda Afek}, \bibinfo{person}{Hagit
  Attiya}, \bibinfo{person}{Danny Dolev}, \bibinfo{person}{Eli Gafni},
  \bibinfo{person}{Michael Merritt}, {and} \bibinfo{person}{Nir Shavit}.}
  \bibinfo{year}{1993}\natexlab{}.
\newblock \showarticletitle{Atomic Snapshots of Shared Memory}.
\newblock \bibinfo{journal}{\emph{J. {ACM}}} \bibinfo{volume}{40},
  \bibinfo{number}{4} (\bibinfo{year}{1993}), \bibinfo{pages}{873--890}.
\newblock
\href{https://doi.org/10.1145/153724.153741}{doi:\nolinkurl{10.1145/153724.153741}}


\bibitem[Ahamad et~al\mbox{.}(1995)]%
        {ahamad1995causal}
\bibfield{author}{\bibinfo{person}{Mustaque Ahamad}, \bibinfo{person}{Gil
  Neiger}, \bibinfo{person}{James~E. Burns}, \bibinfo{person}{Prince Kohli},
  {and} \bibinfo{person}{Phillip~W. Hutto}.} \bibinfo{year}{1995}\natexlab{}.
\newblock \showarticletitle{Causal Memory: Definitions, Implementation, and
  Programming}.
\newblock \bibinfo{journal}{\emph{Distributed Computing}} \bibinfo{volume}{9},
  \bibinfo{number}{1} (\bibinfo{year}{1995}), \bibinfo{pages}{37--49}.
\newblock
\href{https://doi.org/10.1007/BF01784241}{doi:\nolinkurl{10.1007/BF01784241}}


\bibitem[Alvaro et~al\mbox{.}(2010a)]%
        {alvaro2010boom}
\bibfield{author}{\bibinfo{person}{Peter Alvaro}, \bibinfo{person}{Tyson
  Condie}, \bibinfo{person}{Neil Conway}, \bibinfo{person}{Khaled Elmeleegy},
  \bibinfo{person}{Joseph~M. Hellerstein}, {and} \bibinfo{person}{Russell
  Sears}.} \bibinfo{year}{2010}\natexlab{a}.
\newblock \showarticletitle{{BOOM} Analytics: Exploring Data-Centric,
  Declarative Programming for the Cloud}. In
  \bibinfo{booktitle}{\emph{Proceedings of the 5th European Conference on
  Computer Systems}} (Paris, France) \emph{(\bibinfo{series}{EuroSys '10})}.
  \bibinfo{publisher}{Association for Computing Machinery},
  \bibinfo{address}{New York, NY, USA}, \bibinfo{pages}{223–236}.
\newblock
\showISBNx{9781605585772}
\href{https://doi.org/10.1145/1755913.1755937}{doi:\nolinkurl{10.1145/1755913.1755937}}


\bibitem[Alvaro et~al\mbox{.}(2011)]%
        {alvaro2011bloom}
\bibfield{author}{\bibinfo{person}{Peter Alvaro}, \bibinfo{person}{Neil
  Conway}, \bibinfo{person}{Joseph~M Hellerstein}, {and}
  \bibinfo{person}{William~R Marczak}.} \bibinfo{year}{2011}\natexlab{}.
\newblock \showarticletitle{Consistency Analysis in Bloom: a CALM and Collected
  Approach.}. In \bibinfo{booktitle}{\emph{CIDR}}. \bibinfo{pages}{249--260}.
\newblock


\bibitem[Alvaro et~al\mbox{.}(2010b)]%
        {alvaro2010dedalus}
\bibfield{author}{\bibinfo{person}{Peter Alvaro}, \bibinfo{person}{William~R
  Marczak}, \bibinfo{person}{Neil Conway}, \bibinfo{person}{Joseph~M
  Hellerstein}, \bibinfo{person}{David Maier}, {and} \bibinfo{person}{Russell
  Sears}.} \bibinfo{year}{2010}\natexlab{b}.
\newblock \showarticletitle{Dedalus: Datalog in time and space}. In
  \bibinfo{booktitle}{\emph{International Datalog 2.0 Workshop}}. Springer,
  \bibinfo{pages}{262--281}.
\newblock


\bibitem[Ameloot et~al\mbox{.}(2015)]%
        {ameloot2015weaker}
\bibfield{author}{\bibinfo{person}{Tom~J. Ameloot}, \bibinfo{person}{Bas
  Ketsman}, \bibinfo{person}{Frank Neven}, {and} \bibinfo{person}{Daniel
  Zinn}.} \bibinfo{year}{2015}\natexlab{}.
\newblock \showarticletitle{Weaker Forms of Monotonicity for Declarative
  Networking: A More Fine-Grained Answer to the {CALM}-Conjecture}.
\newblock \bibinfo{journal}{\emph{{ACM} Transactions on Database Systems}}
  \bibinfo{volume}{40}, \bibinfo{number}{4} (\bibinfo{year}{2015}),
  \bibinfo{pages}{1--45}.
\newblock
\href{https://doi.org/10.1145/2809784}{doi:\nolinkurl{10.1145/2809784}}


\bibitem[Ameloot et~al\mbox{.}(2013)]%
        {ameloot2013relational}
\bibfield{author}{\bibinfo{person}{Tom~J. Ameloot}, \bibinfo{person}{Frank
  Neven}, {and} \bibinfo{person}{Jan Van~den Bussche}.}
  \bibinfo{year}{2013}\natexlab{}.
\newblock \showarticletitle{Relational Transducers for Declarative Networking}.
\newblock \bibinfo{journal}{\emph{Journal of the {ACM}}} \bibinfo{volume}{60},
  \bibinfo{number}{2} (\bibinfo{year}{2013}), \bibinfo{pages}{1--38}.
\newblock
\href{https://doi.org/10.1145/2450142.2450151}{doi:\nolinkurl{10.1145/2450142.2450151}}


\bibitem[An et~al\mbox{.}(2025)]%
        {an2025neurosymbolic}
\bibfield{author}{\bibinfo{person}{Chenyang An}, \bibinfo{person}{Sam Bayless},
  \bibinfo{person}{Stefano Buliani}, \bibinfo{person}{Darion Cassel},
  \bibinfo{person}{Byron Cook}, \bibinfo{person}{Duncan Clough},
  \bibinfo{person}{R{\'e}mi Delmas}, \bibinfo{person}{Nafi Diallo},
  \bibinfo{person}{Ferhat Erata}, \bibinfo{person}{Nick Feng},
  \bibinfo{person}{Dimitra Giannakopoulou}, \bibinfo{person}{Aman Goel},
  \bibinfo{person}{Aditya Gokhale}, \bibinfo{person}{Joe Hendrix},
  \bibinfo{person}{Victor Heorhiadi}, \bibinfo{person}{Marc Hudak},
  \bibinfo{person}{Dejan Jovanovi{\'c}}, \bibinfo{person}{Andrew~M. Kent},
  \bibinfo{person}{Benjamin Kiesl-Reiter}, \bibinfo{person}{Jeffrey~J. Kuna},
  \bibinfo{person}{Nadia Labai}, \bibinfo{person}{Joseph Lilien},
  \bibinfo{person}{Divya Raghunathan}, \bibinfo{person}{Zvonimir
  Rakamari{\'c}}, \bibinfo{person}{Niloofar Razavi}, \bibinfo{person}{Michael
  Tautschnig}, \bibinfo{person}{Ali Torkamani}, \bibinfo{person}{Nathaniel
  Weir}, \bibinfo{person}{Michael~W. Whalen}, {and} \bibinfo{person}{Jianan
  Yao}.} \bibinfo{year}{2025}\natexlab{}.
\newblock \bibinfo{title}{A Neurosymbolic Approach to Natural Language
  Formalization and Verification}.
\newblock
\showeprint[arxiv]{2511.09008}~[cs.CL]


\bibitem[Angluin(1980)]%
        {angluin1980local}
\bibfield{author}{\bibinfo{person}{Dana Angluin}.}
  \bibinfo{year}{1980}\natexlab{}.
\newblock \showarticletitle{Local and Global Properties in Networks of
  Processors}. In \bibinfo{booktitle}{\emph{Proceedings of the 12th Annual ACM
  Symposium on Theory of Computing (STOC)}}. \bibinfo{pages}{82--93}.
\newblock
\href{https://doi.org/10.1145/800141.804655}{doi:\nolinkurl{10.1145/800141.804655}}


\bibitem[Attiya et~al\mbox{.}(1990)]%
        {attiya1990renaming}
\bibfield{author}{\bibinfo{person}{Hagit Attiya}, \bibinfo{person}{Amotz
  Bar-Noy}, \bibinfo{person}{Danny Dolev}, \bibinfo{person}{David Peleg}, {and}
  \bibinfo{person}{R{\"{u}}diger Reischuk}.} \bibinfo{year}{1990}\natexlab{}.
\newblock \showarticletitle{Renaming in an Asynchronous Environment}.
\newblock \bibinfo{journal}{\emph{J. {ACM}}} \bibinfo{volume}{37},
  \bibinfo{number}{3} (\bibinfo{year}{1990}), \bibinfo{pages}{524--548}.
\newblock
\href{https://doi.org/10.1145/79147.79158}{doi:\nolinkurl{10.1145/79147.79158}}


\bibitem[Baccaert and Ketsman(2023)]%
        {baccaert2023dcbq}
\bibfield{author}{\bibinfo{person}{Tim Baccaert} {and} \bibinfo{person}{Bas
  Ketsman}.} \bibinfo{year}{2023}\natexlab{}.
\newblock \showarticletitle{Distributed Consistency Beyond Queries}. In
  \bibinfo{booktitle}{\emph{Proceedings of the 42nd ACM SIGMOD-SIGACT-SIGAI
  Symposium on Principles of Database Systems (PODS)}}.
  \bibinfo{pages}{47--58}.
\newblock
\href{https://doi.org/10.1145/3584372.3588657}{doi:\nolinkurl{10.1145/3584372.3588657}}


\bibitem[Baccaert and Ketsman(2026)]%
        {baccaert2026spectrum}
\bibfield{author}{\bibinfo{person}{Tim Baccaert} {and} \bibinfo{person}{Bas
  Ketsman}.} \bibinfo{year}{2026}\natexlab{}.
\newblock \showarticletitle{A Generalized {CALM} Theorem for Non-Deterministic
  Computation in Asynchronous Distributed Systems}.
\newblock \bibinfo{journal}{\emph{Information Systems}}  \bibinfo{volume}{138}
  (\bibinfo{year}{2026}), \bibinfo{pages}{102691}.
\newblock
\href{https://doi.org/10.1016/j.is.2026.102691}{doi:\nolinkurl{10.1016/j.is.2026.102691}}


\bibitem[Bailis et~al\mbox{.}(2013a)]%
        {bailis2013hat}
\bibfield{author}{\bibinfo{person}{Peter Bailis}, \bibinfo{person}{Aaron
  Davidson}, \bibinfo{person}{Alan Fekete}, \bibinfo{person}{Ali Ghodsi},
  \bibinfo{person}{Joseph~M. Hellerstein}, {and} \bibinfo{person}{Ion Stoica}.}
  \bibinfo{year}{2013}\natexlab{a}.
\newblock \showarticletitle{Highly Available Transactions: Virtues and
  Limitations}.
\newblock \bibinfo{journal}{\emph{Proceedings of the VLDB Endowment}}
  \bibinfo{volume}{7}, \bibinfo{number}{3} (\bibinfo{year}{2013}),
  \bibinfo{pages}{181--192}.
\newblock
\href{https://doi.org/10.14778/2732232.2732237}{doi:\nolinkurl{10.14778/2732232.2732237}}


\bibitem[Bailis et~al\mbox{.}(2013b)]%
        {bailis2013hatextended}
\bibfield{author}{\bibinfo{person}{Peter Bailis}, \bibinfo{person}{Aaron
  Davidson}, \bibinfo{person}{Alan Fekete}, \bibinfo{person}{Ali Ghodsi},
  \bibinfo{person}{Joseph~M. Hellerstein}, {and} \bibinfo{person}{Ion Stoica}.}
  \bibinfo{year}{2013}\natexlab{b}.
\newblock \bibinfo{title}{Highly Available Transactions: Virtues and
  Limitations (Extended Version)}.
\newblock \bibinfo{howpublished}{arXiv:1302.0309}.
\newblock
\urldef\tempurl%
\url{https://arxiv.org/abs/1302.0309}
\showURL{%
\tempurl}


\bibitem[Bailis et~al\mbox{.}(2014)]%
        {bailis2014coordination}
\bibfield{author}{\bibinfo{person}{Peter Bailis}, \bibinfo{person}{Alan
  Fekete}, \bibinfo{person}{Michael~J Franklin}, \bibinfo{person}{Ali Ghodsi},
  \bibinfo{person}{Joseph~M Hellerstein}, {and} \bibinfo{person}{Ion Stoica}.}
  \bibinfo{year}{2014}\natexlab{}.
\newblock \showarticletitle{Coordination avoidance in database systems}.
\newblock \bibinfo{journal}{\emph{Proceedings of the VLDB Endowment}}
  \bibinfo{volume}{8}, \bibinfo{number}{3} (\bibinfo{year}{2014}),
  \bibinfo{pages}{185--196}.
\newblock
\href{https://doi.org/10.14778/2735508.2735509}{doi:\nolinkurl{10.14778/2735508.2735509}}


\bibitem[Baquero et~al\mbox{.}(2014)]%
        {baquero2014operation}
\bibfield{author}{\bibinfo{person}{Carlos Baquero},
  \bibinfo{person}{Paulo~S{\'e}rgio Almeida}, {and} \bibinfo{person}{Ali
  Shoker}.} \bibinfo{year}{2014}\natexlab{}.
\newblock \showarticletitle{Making Operation-Based {CRDT}s Operation-Based}. In
  \bibinfo{booktitle}{\emph{Distributed Applications and Interoperable Systems
  (DAIS)}} \emph{(\bibinfo{series}{Lecture Notes in Computer Science},
  Vol.~\bibinfo{volume}{8460})}. \bibinfo{publisher}{Springer},
  \bibinfo{pages}{126--140}.
\newblock
\href{https://doi.org/10.1007/978-3-662-43352-2_11}{doi:\nolinkurl{10.1007/978-3-662-43352-2_11}}


\bibitem[Bernstein and Goodman(1981)]%
        {bernstein1981concurrency}
\bibfield{author}{\bibinfo{person}{Philip~A. Bernstein} {and}
  \bibinfo{person}{Nathan Goodman}.} \bibinfo{year}{1981}\natexlab{}.
\newblock \showarticletitle{Concurrency Control in Distributed Database
  Systems}.
\newblock \bibinfo{journal}{\emph{Comput. Surveys}} \bibinfo{volume}{13},
  \bibinfo{number}{2} (\bibinfo{year}{1981}), \bibinfo{pages}{185--221}.
\newblock
\href{https://doi.org/10.1145/356842.356846}{doi:\nolinkurl{10.1145/356842.356846}}


\bibitem[Brewer(2012)]%
        {brewer2012cap}
\bibfield{author}{\bibinfo{person}{Eric Brewer}.}
  \bibinfo{year}{2012}\natexlab{}.
\newblock \showarticletitle{CAP twelve years later: How the "rules" have
  changed}.
\newblock \bibinfo{journal}{\emph{Computer}} \bibinfo{volume}{45},
  \bibinfo{number}{2} (\bibinfo{year}{2012}), \bibinfo{pages}{23--29}.
\newblock
\href{https://doi.org/10.1109/MC.2012.37}{doi:\nolinkurl{10.1109/MC.2012.37}}


\bibitem[Brewer(2000)]%
        {brewer2000towards}
\bibfield{author}{\bibinfo{person}{Eric~A. Brewer}.}
  \bibinfo{year}{2000}\natexlab{}.
\newblock \showarticletitle{Towards Robust Distributed Systems}. In
  \bibinfo{booktitle}{\emph{Proceedings of the Nineteenth Annual {ACM}
  Symposium on Principles of Distributed Computing ({PODC})}}.
  \bibinfo{publisher}{ACM}, \bibinfo{pages}{7}.
\newblock
\href{https://doi.org/10.1145/343477.343502}{doi:\nolinkurl{10.1145/343477.343502}}


\bibitem[Burckhardt et~al\mbox{.}(2014)]%
        {burckhardt2014replicated}
\bibfield{author}{\bibinfo{person}{Sebastian Burckhardt},
  \bibinfo{person}{Alexey Gotsman}, \bibinfo{person}{Hongseok Yang}, {and}
  \bibinfo{person}{Marek Zawirski}.} \bibinfo{year}{2014}\natexlab{}.
\newblock \showarticletitle{Replicated Data Types: Specification, Verification,
  Optimality}. In \bibinfo{booktitle}{\emph{Proceedings of the 41st ACM
  SIGPLAN-SIGACT Symposium on Principles of Programming Languages (POPL)}}.
  \bibinfo{pages}{271--284}.
\newblock
\href{https://doi.org/10.1145/2535838.2535848}{doi:\nolinkurl{10.1145/2535838.2535848}}


\bibitem[Chandra et~al\mbox{.}(1996)]%
        {chandra1996weakest}
\bibfield{author}{\bibinfo{person}{Tushar~Deepak Chandra},
  \bibinfo{person}{Vassos Hadzilacos}, {and} \bibinfo{person}{Sam Toueg}.}
  \bibinfo{year}{1996}\natexlab{}.
\newblock \showarticletitle{The Weakest Failure Detector for Solving
  Consensus}.
\newblock \bibinfo{journal}{\emph{J. ACM}} \bibinfo{volume}{43},
  \bibinfo{number}{4} (\bibinfo{year}{1996}), \bibinfo{pages}{685--722}.
\newblock
\href{https://doi.org/10.1145/234533.234549}{doi:\nolinkurl{10.1145/234533.234549}}


\bibitem[Chandra and Toueg(1996)]%
        {chandra1996unreliable}
\bibfield{author}{\bibinfo{person}{Tushar~D. Chandra} {and}
  \bibinfo{person}{Sam Toueg}.} \bibinfo{year}{1996}\natexlab{}.
\newblock \showarticletitle{Unreliable Failure Detectors for Reliable
  Distributed Systems}.
\newblock \bibinfo{journal}{\emph{J. {ACM}}} \bibinfo{volume}{43},
  \bibinfo{number}{2} (\bibinfo{year}{1996}), \bibinfo{pages}{225--267}.
\newblock
\href{https://doi.org/10.1145/226643.226647}{doi:\nolinkurl{10.1145/226643.226647}}


\bibitem[Chandy and Lamport(1985)]%
        {chandy1985distributed}
\bibfield{author}{\bibinfo{person}{K.~Mani Chandy} {and}
  \bibinfo{person}{Leslie Lamport}.} \bibinfo{year}{1985}\natexlab{}.
\newblock \showarticletitle{Distributed Snapshots: Determining Global States of
  Distributed Systems}.
\newblock \bibinfo{journal}{\emph{{ACM} Trans. Comput. Syst.}}
  \bibinfo{volume}{3}, \bibinfo{number}{1} (\bibinfo{year}{1985}),
  \bibinfo{pages}{63--75}.
\newblock
\href{https://doi.org/10.1145/214451.214456}{doi:\nolinkurl{10.1145/214451.214456}}


\bibitem[Chang and Roberts(1979)]%
        {chang1979improved}
\bibfield{author}{\bibinfo{person}{Ernest Chang} {and}
  \bibinfo{person}{Rosemary Roberts}.} \bibinfo{year}{1979}\natexlab{}.
\newblock \showarticletitle{An Improved Algorithm for Decentralized
  Extrema-Finding in Circular Configurations of Processes}.
\newblock \bibinfo{journal}{\emph{Commun. ACM}} \bibinfo{volume}{22},
  \bibinfo{number}{5} (\bibinfo{year}{1979}), \bibinfo{pages}{281--283}.
\newblock
\href{https://doi.org/10.1145/359104.359108}{doi:\nolinkurl{10.1145/359104.359108}}


\bibitem[Chaudhuri(1993)]%
        {chaudhuri1993more}
\bibfield{author}{\bibinfo{person}{Soma Chaudhuri}.}
  \bibinfo{year}{1993}\natexlab{}.
\newblock \showarticletitle{More Choices Allow More Faults: Set Consensus
  Problems in Totally Asynchronous Systems}.
\newblock \bibinfo{journal}{\emph{Information and Computation}}
  \bibinfo{volume}{105}, \bibinfo{number}{1} (\bibinfo{year}{1993}),
  \bibinfo{pages}{132--158}.
\newblock
\href{https://doi.org/10.1006/inco.1993.1043}{doi:\nolinkurl{10.1006/inco.1993.1043}}


\bibitem[Chu et~al\mbox{.}(2007)]%
        {chu2007declarative}
\bibfield{author}{\bibinfo{person}{David Chu}, \bibinfo{person}{Lucian Popa},
  \bibinfo{person}{Arsalan Tavakoli}, \bibinfo{person}{Joseph~M. Hellerstein},
  \bibinfo{person}{Philip Levis}, \bibinfo{person}{Scott Shenker}, {and}
  \bibinfo{person}{Ion Stoica}.} \bibinfo{year}{2007}\natexlab{}.
\newblock \showarticletitle{The Design and Implementation of a Declarative
  Sensor Network System}. In \bibinfo{booktitle}{\emph{Proceedings of the 5th
  International Conference on Embedded Networked Sensor Systems (SenSys)}}.
  \bibinfo{pages}{175--188}.
\newblock
\href{https://doi.org/10.1145/1322263.1322281}{doi:\nolinkurl{10.1145/1322263.1322281}}


\bibitem[Codd(1970)]%
        {codd1970relational}
\bibfield{author}{\bibinfo{person}{E.~F. Codd}.}
  \bibinfo{year}{1970}\natexlab{}.
\newblock \showarticletitle{A Relational Model of Data for Large Shared Data
  Banks}.
\newblock \bibinfo{journal}{\emph{Commun. ACM}} \bibinfo{volume}{13},
  \bibinfo{number}{6} (\bibinfo{year}{1970}), \bibinfo{pages}{377--387}.
\newblock
\href{https://doi.org/10.1145/362384.362685}{doi:\nolinkurl{10.1145/362384.362685}}


\bibitem[Condie et~al\mbox{.}(2010)]%
        {condie2010online}
\bibfield{author}{\bibinfo{person}{Tyson Condie}, \bibinfo{person}{Neil
  Conway}, \bibinfo{person}{Peter Alvaro}, \bibinfo{person}{Joseph~M.
  Hellerstein}, \bibinfo{person}{Khaled Elmeleegy}, {and}
  \bibinfo{person}{Russell Sears}.} \bibinfo{year}{2010}\natexlab{}.
\newblock \showarticletitle{MapReduce online}. In
  \bibinfo{booktitle}{\emph{Proceedings of the 7th USENIX Conference on
  Networked Systems Design and Implementation}} (San Jose, California)
  \emph{(\bibinfo{series}{NSDI'10})}. \bibinfo{publisher}{USENIX Association},
  \bibinfo{address}{USA}, \bibinfo{pages}{21}.
\newblock


\bibitem[de~Moura and Ullrich(2021)]%
        {demoura2021lean4}
\bibfield{author}{\bibinfo{person}{Leonardo de Moura} {and}
  \bibinfo{person}{Sebastian Ullrich}.} \bibinfo{year}{2021}\natexlab{}.
\newblock \showarticletitle{The {Lean} 4 Theorem Prover and Programming
  Language}. In \bibinfo{booktitle}{\emph{Automated Deduction -- CADE 28}}
  \emph{(\bibinfo{series}{Lecture Notes in Computer Science},
  Vol.~\bibinfo{volume}{12699})}. \bibinfo{publisher}{Springer},
  \bibinfo{pages}{625--635}.
\newblock
\href{https://doi.org/10.1007/978-3-030-79876-5_37}{doi:\nolinkurl{10.1007/978-3-030-79876-5_37}}


\bibitem[Dean and Ghemawat(2004)]%
        {dean2004mapreduce}
\bibfield{author}{\bibinfo{person}{Jeffrey Dean} {and} \bibinfo{person}{Sanjay
  Ghemawat}.} \bibinfo{year}{2004}\natexlab{}.
\newblock \showarticletitle{{MapReduce}: Simplified Data Processing on Large
  Clusters}. In \bibinfo{booktitle}{\emph{6th Symposium on Operating Systems
  Design and Implementation (OSDI)}}. \bibinfo{publisher}{USENIX Association},
  \bibinfo{pages}{137--150}.
\newblock


\bibitem[D{\'{e}}fago et~al\mbox{.}(2004)]%
        {defago2004total}
\bibfield{author}{\bibinfo{person}{Xavier D{\'{e}}fago},
  \bibinfo{person}{Andr{\'{e}} Schiper}, {and} \bibinfo{person}{P{\'{e}}ter
  Urb{\'{a}}n}.} \bibinfo{year}{2004}\natexlab{}.
\newblock \showarticletitle{Total Order Broadcast and Multicast Algorithms:
  Taxonomy and Survey}.
\newblock \bibinfo{journal}{\emph{Comput. Surveys}} \bibinfo{volume}{36},
  \bibinfo{number}{4} (\bibinfo{year}{2004}), \bibinfo{pages}{372--421}.
\newblock
\href{https://doi.org/10.1145/1041680.1041682}{doi:\nolinkurl{10.1145/1041680.1041682}}


\bibitem[Demers et~al\mbox{.}(1987)]%
        {demers1987epidemic}
\bibfield{author}{\bibinfo{person}{Alan Demers}, \bibinfo{person}{Dan Greene},
  \bibinfo{person}{Carl Hauser}, \bibinfo{person}{Wes Irish},
  \bibinfo{person}{John Larson}, \bibinfo{person}{Scott Shenker},
  \bibinfo{person}{Howard Sturgis}, \bibinfo{person}{Dan Swinehart}, {and}
  \bibinfo{person}{Doug Terry}.} \bibinfo{year}{1987}\natexlab{}.
\newblock \showarticletitle{Epidemic Algorithms for Replicated Database
  Maintenance}. In \bibinfo{booktitle}{\emph{Proceedings of the Sixth Annual
  {ACM} Symposium on Principles of Distributed Computing (PODC)}}.
  \bibinfo{pages}{1--12}.
\newblock
\href{https://doi.org/10.1145/41840.41841}{doi:\nolinkurl{10.1145/41840.41841}}


\bibitem[DeWitt and Stonebraker(2008)]%
        {dewitt2008mapreduce}
\bibfield{author}{\bibinfo{person}{David~J. DeWitt} {and}
  \bibinfo{person}{Michael Stonebraker}.} \bibinfo{year}{2008}\natexlab{}.
\newblock \bibinfo{title}{{MapReduce}: A Major Step Backwards}.
\newblock \bibinfo{howpublished}{The Database Column (blog)}.
\newblock
\newblock
\shownote{Accessed via archived copy}.


\bibitem[Fischer et~al\mbox{.}(1985)]%
        {fischer1985impossibility}
\bibfield{author}{\bibinfo{person}{Michael~J. Fischer},
  \bibinfo{person}{Nancy~A. Lynch}, {and} \bibinfo{person}{Mike Paterson}.}
  \bibinfo{year}{1985}\natexlab{}.
\newblock \showarticletitle{Impossibility of Distributed Consensus with One
  Faulty Process}.
\newblock \bibinfo{journal}{\emph{J. {ACM}}} \bibinfo{volume}{32},
  \bibinfo{number}{2} (\bibinfo{year}{1985}), \bibinfo{pages}{374--382}.
\newblock
\href{https://doi.org/10.1145/3149.214121}{doi:\nolinkurl{10.1145/3149.214121}}


\bibitem[Flanagan and Godefroid(2005)]%
        {flanagan2005dpor}
\bibfield{author}{\bibinfo{person}{Cormac Flanagan} {and}
  \bibinfo{person}{Patrice Godefroid}.} \bibinfo{year}{2005}\natexlab{}.
\newblock \showarticletitle{Dynamic Partial-Order Reduction for Model Checking
  Software}. In \bibinfo{booktitle}{\emph{Proceedings of the 32nd {ACM}
  {SIGPLAN-SIGACT} Symposium on Principles of Programming Languages (POPL)}}.
  \bibinfo{publisher}{ACM}, \bibinfo{pages}{110--121}.
\newblock
\href{https://doi.org/10.1145/1040305.1040315}{doi:\nolinkurl{10.1145/1040305.1040315}}


\bibitem[Gilbert and Lynch(2002)]%
        {gilbert2002cap}
\bibfield{author}{\bibinfo{person}{Seth Gilbert} {and}
  \bibinfo{person}{Nancy~A. Lynch}.} \bibinfo{year}{2002}\natexlab{}.
\newblock \showarticletitle{Brewer's conjecture and the feasibility of
  consistent, available, partition-tolerant web services}.
\newblock \bibinfo{journal}{\emph{{SIGACT} News}} \bibinfo{volume}{33},
  \bibinfo{number}{2} (\bibinfo{year}{2002}), \bibinfo{pages}{51--59}.
\newblock
\href{https://doi.org/10.1145/564585.564601}{doi:\nolinkurl{10.1145/564585.564601}}


\bibitem[Graefe(1990)]%
        {graefe1990volcano}
\bibfield{author}{\bibinfo{person}{Goetz Graefe}.}
  \bibinfo{year}{1990}\natexlab{}.
\newblock \showarticletitle{Encapsulation of Parallelism in the {Volcano} Query
  Processing System}. In \bibinfo{booktitle}{\emph{Proceedings of the 1990 ACM
  SIGMOD International Conference on Management of Data}}.
  \bibinfo{pages}{102--111}.
\newblock
\href{https://doi.org/10.1145/93597.98720}{doi:\nolinkurl{10.1145/93597.98720}}


\bibitem[Gray and Lamport(2006)]%
        {gray2006consensus}
\bibfield{author}{\bibinfo{person}{Jim Gray} {and} \bibinfo{person}{Leslie
  Lamport}.} \bibinfo{year}{2006}\natexlab{}.
\newblock \showarticletitle{Consensus on Transaction Commit}.
\newblock \bibinfo{journal}{\emph{ACM Transactions on Database Systems}}
  \bibinfo{volume}{31}, \bibinfo{number}{1} (\bibinfo{year}{2006}),
  \bibinfo{pages}{133--160}.
\newblock
\href{https://doi.org/10.1145/1132863.1132867}{doi:\nolinkurl{10.1145/1132863.1132867}}


\bibitem[Hadzilacos and Toueg(1994)]%
        {hadzilacos1994modular}
\bibfield{author}{\bibinfo{person}{Vassos Hadzilacos} {and}
  \bibinfo{person}{Sam Toueg}.} \bibinfo{year}{1994}\natexlab{}.
\newblock \bibinfo{booktitle}{\emph{A Modular Approach to Fault-Tolerant
  Broadcasts and Related Problems}}.
\newblock \bibinfo{type}{{T}echnical {R}eport} TR 94-1425.
  \bibinfo{institution}{Cornell University}.
\newblock


\bibitem[Hawblitzel et~al\mbox{.}(2015)]%
        {hawblitzel2015ironfleet}
\bibfield{author}{\bibinfo{person}{Chris Hawblitzel}, \bibinfo{person}{Jon
  Howell}, \bibinfo{person}{Manos Kapritsos}, \bibinfo{person}{Jacob~R. Lorch},
  \bibinfo{person}{Bryan Parno}, \bibinfo{person}{Michael~L. Roberts},
  \bibinfo{person}{Srinath Setty}, {and} \bibinfo{person}{Brian Zill}.}
  \bibinfo{year}{2015}\natexlab{}.
\newblock \showarticletitle{{IronFleet}: Proving Practical Distributed Systems
  Correct}. In \bibinfo{booktitle}{\emph{Proceedings of the 25th Symposium on
  Operating Systems Principles (SOSP)}}. \bibinfo{publisher}{ACM},
  \bibinfo{pages}{1--17}.
\newblock
\href{https://doi.org/10.1145/2815400.2815428}{doi:\nolinkurl{10.1145/2815400.2815428}}


\bibitem[Hellerstein(2010)]%
        {hellerstein2010declarative}
\bibfield{author}{\bibinfo{person}{Joseph~M Hellerstein}.}
  \bibinfo{year}{2010}\natexlab{}.
\newblock \showarticletitle{The declarative imperative: experiences and
  conjectures in distributed logic}.
\newblock \bibinfo{journal}{\emph{ACM SIGMOD Record}} \bibinfo{volume}{39},
  \bibinfo{number}{1} (\bibinfo{year}{2010}), \bibinfo{pages}{5--19}.
\newblock
\href{https://doi.org/10.1145/1860702.1860704}{doi:\nolinkurl{10.1145/1860702.1860704}}


\bibitem[Hellerstein(2023)]%
        {hellerstein2023codd}
\bibfield{author}{\bibinfo{person}{Joseph~M. Hellerstein}.}
  \bibinfo{year}{2023}\natexlab{}.
\newblock \bibinfo{title}{A Broader Perspective on Declarativity}.
\newblock \bibinfo{howpublished}{Edgar F. Codd Innovations Award Lecture, {ACM
  SIGMOD} International Conference on Management of Data}.
\newblock
\newblock
\shownote{Slides at
  \url{https://speakerdeck.com/jhellerstein/codd-award-lecture-2023-a-broader-perspective-on-declarativity}}.


\bibitem[Hellerstein(2025)]%
        {hellerstein2025turtles}
\bibfield{author}{\bibinfo{person}{Joseph~M. Hellerstein}.}
  \bibinfo{year}{2025}\natexlab{}.
\newblock \bibinfo{title}{{CRDT}s \#2: Turtles All the Way Down}.
\newblock \bibinfo{howpublished}{Blog post,
  \url{https://jhellerstein.github.io/blog/crdt-turtles/}}.
\newblock


\bibitem[Hellerstein and Alvaro(2020)]%
        {hellerstein2020keeping}
\bibfield{author}{\bibinfo{person}{Joseph~M. Hellerstein} {and}
  \bibinfo{person}{Peter Alvaro}.} \bibinfo{year}{2020}\natexlab{}.
\newblock \showarticletitle{Keeping {CALM:} when distributed consistency is
  easy}.
\newblock \bibinfo{journal}{\emph{Commun. {ACM}}} \bibinfo{volume}{63},
  \bibinfo{number}{9} (\bibinfo{year}{2020}), \bibinfo{pages}{72--81}.
\newblock
\href{https://doi.org/10.1145/3369736}{doi:\nolinkurl{10.1145/3369736}}


\bibitem[Herlihy(1991)]%
        {herlihy1991waitfree}
\bibfield{author}{\bibinfo{person}{Maurice Herlihy}.}
  \bibinfo{year}{1991}\natexlab{}.
\newblock \showarticletitle{Wait-Free Synchronization}.
\newblock \bibinfo{journal}{\emph{{ACM} Trans. Program. Lang. Syst.}}
  \bibinfo{volume}{13}, \bibinfo{number}{1} (\bibinfo{year}{1991}),
  \bibinfo{pages}{124--149}.
\newblock
\href{https://doi.org/10.1145/114005.102808}{doi:\nolinkurl{10.1145/114005.102808}}


\bibitem[Herlihy and Shavit(1999)]%
        {herlihy1999topological}
\bibfield{author}{\bibinfo{person}{Maurice Herlihy} {and} \bibinfo{person}{Nir
  Shavit}.} \bibinfo{year}{1999}\natexlab{}.
\newblock \showarticletitle{The topological structure of asynchronous
  computability}.
\newblock \bibinfo{journal}{\emph{J. {ACM}}} \bibinfo{volume}{46},
  \bibinfo{number}{6} (\bibinfo{year}{1999}), \bibinfo{pages}{858--923}.
\newblock
\href{https://doi.org/10.1145/331524.331529}{doi:\nolinkurl{10.1145/331524.331529}}


\bibitem[Herlihy and Wing(1990)]%
        {herlihy1990linearizability}
\bibfield{author}{\bibinfo{person}{Maurice Herlihy} {and}
  \bibinfo{person}{Jeannette~M. Wing}.} \bibinfo{year}{1990}\natexlab{}.
\newblock \showarticletitle{Linearizability: {A} Correctness Condition for
  Concurrent Objects}.
\newblock \bibinfo{journal}{\emph{{ACM} Trans. Program. Lang. Syst.}}
  \bibinfo{volume}{12}, \bibinfo{number}{3} (\bibinfo{year}{1990}),
  \bibinfo{pages}{463--492}.
\newblock
\href{https://doi.org/10.1145/78969.78972}{doi:\nolinkurl{10.1145/78969.78972}}


\bibitem[Hirschberg and Sinclair(1980)]%
        {hirschberg1980decentralized}
\bibfield{author}{\bibinfo{person}{Daniel~S. Hirschberg} {and}
  \bibinfo{person}{J.~B. Sinclair}.} \bibinfo{year}{1980}\natexlab{}.
\newblock \showarticletitle{Decentralized Extrema-Finding in Circular
  Configurations of Processors}.
\newblock \bibinfo{journal}{\emph{Commun. ACM}} \bibinfo{volume}{23},
  \bibinfo{number}{11} (\bibinfo{year}{1980}), \bibinfo{pages}{627--628}.
\newblock
\href{https://doi.org/10.1145/359024.359029}{doi:\nolinkurl{10.1145/359024.359029}}


\bibitem[Kaki et~al\mbox{.}(2018)]%
        {kaki2018alone}
\bibfield{author}{\bibinfo{person}{Gowtham Kaki}, \bibinfo{person}{Kartik
  Nagar}, \bibinfo{person}{Mahsa Najafzadeh}, {and} \bibinfo{person}{Suresh
  Jagannathan}.} \bibinfo{year}{2018}\natexlab{}.
\newblock \showarticletitle{Alone Together: Compositional Reasoning and
  Inference for Weak Isolation}.
\newblock \bibinfo{journal}{\emph{Proceedings of the ACM on Programming
  Languages}} \bibinfo{volume}{2}, \bibinfo{number}{POPL}
  (\bibinfo{year}{2018}), \bibinfo{pages}{27:1--27:34}.
\newblock
\href{https://doi.org/10.1145/3158115}{doi:\nolinkurl{10.1145/3158115}}


\bibitem[Karloff et~al\mbox{.}(2010)]%
        {karloff2010mrc}
\bibfield{author}{\bibinfo{person}{Howard Karloff}, \bibinfo{person}{Siddharth
  Suri}, {and} \bibinfo{person}{Sergei Vassilvitskii}.}
  \bibinfo{year}{2010}\natexlab{}.
\newblock \showarticletitle{A Model of Computation for {MapReduce}}. In
  \bibinfo{booktitle}{\emph{Proceedings of the Twenty-First Annual ACM-SIAM
  Symposium on Discrete Algorithms (SODA)}}. \bibinfo{pages}{938--948}.
\newblock
\href{https://doi.org/10.1137/1.9781611973075.76}{doi:\nolinkurl{10.1137/1.9781611973075.76}}


\bibitem[Kuper and Newton(2013)]%
        {kuper2013lvars}
\bibfield{author}{\bibinfo{person}{Lindsey Kuper} {and}
  \bibinfo{person}{Ryan~R. Newton}.} \bibinfo{year}{2013}\natexlab{}.
\newblock \showarticletitle{{LVars}: Lattice-Based Data Structures for
  Deterministic Parallelism}. In \bibinfo{booktitle}{\emph{Proc.\ 2nd ACM
  SIGPLAN Workshop on Functional High-Performance Computing (FHPC)}}.
  \bibinfo{pages}{71--84}.
\newblock
\href{https://doi.org/10.1145/2502323.2502326}{doi:\nolinkurl{10.1145/2502323.2502326}}


\bibitem[Laddad et~al\mbox{.}(2022)]%
        {laddad2022keep}
\bibfield{author}{\bibinfo{person}{Shadaj Laddad}, \bibinfo{person}{Conor
  Power}, \bibinfo{person}{Mae Milano}, \bibinfo{person}{Alvin Cheung},
  \bibinfo{person}{Natacha Crooks}, {and} \bibinfo{person}{Joseph~M.
  Hellerstein}.} \bibinfo{year}{2022}\natexlab{}.
\newblock \showarticletitle{Keep {CALM} and {CRDT} On}.
\newblock \bibinfo{journal}{\emph{Proceedings of the {VLDB} Endowment}}
  \bibinfo{volume}{16}, \bibinfo{number}{4} (\bibinfo{year}{2022}),
  \bibinfo{pages}{856--863}.
\newblock
\href{https://doi.org/10.14778/3574245.3574268}{doi:\nolinkurl{10.14778/3574245.3574268}}


\bibitem[Laddad et~al\mbox{.}(2024)]%
        {hydro}
\bibfield{author}{\bibinfo{person}{Shadaj Laddad}, \bibinfo{person}{Mingwei
  Samuel}, \bibinfo{person}{Conor Power}, \bibinfo{person}{Mae Milano},
  \bibinfo{person}{Alvin Cheung}, \bibinfo{person}{Joseph~M. Hellerstein},
  {et~al\mbox{.}}} \bibinfo{year}{2024}\natexlab{}.
\newblock \bibinfo{title}{Hydro: A Framework for Correct-by-Construction
  Distributed Programming}.
\newblock \bibinfo{howpublished}{\url{https://hydro.run}}.
\newblock
\newblock
\shownote{Open-source project, UC Berkeley}.


\bibitem[Lamport(1978)]%
        {lamport1978time}
\bibfield{author}{\bibinfo{person}{Leslie Lamport}.}
  \bibinfo{year}{1978}\natexlab{}.
\newblock \showarticletitle{Time, Clocks, and the Ordering of Events in a
  Distributed System}.
\newblock \bibinfo{journal}{\emph{Commun. {ACM}}} \bibinfo{volume}{21},
  \bibinfo{number}{7} (\bibinfo{year}{1978}), \bibinfo{pages}{558--565}.
\newblock
\href{https://doi.org/10.1145/359545.359563}{doi:\nolinkurl{10.1145/359545.359563}}


\bibitem[Lamport(2001)]%
        {lamport2001paxos}
\bibfield{author}{\bibinfo{person}{Leslie Lamport}.}
  \bibinfo{year}{2001}\natexlab{}.
\newblock \showarticletitle{Paxos Made Simple}.
\newblock \bibinfo{journal}{\emph{ACM SIGACT News}} \bibinfo{volume}{32},
  \bibinfo{number}{4} (\bibinfo{year}{2001}), \bibinfo{pages}{51--58}.
\newblock


\bibitem[Lamport(2005)]%
        {lamport2005generalized}
\bibfield{author}{\bibinfo{person}{Leslie Lamport}.}
  \bibinfo{year}{2005}\natexlab{}.
\newblock \bibinfo{booktitle}{\emph{Generalized Consensus and {Paxos}}}.
\newblock \bibinfo{type}{{T}echnical {R}eport} MSR-TR-2005-33.
  \bibinfo{institution}{Microsoft Research}.
\newblock


\bibitem[{Le Lann}(1977)]%
        {lelann1977distributed}
\bibfield{author}{\bibinfo{person}{G{\'e}rard {Le Lann}}.}
  \bibinfo{year}{1977}\natexlab{}.
\newblock \showarticletitle{Distributed Systems---Towards a Formal Approach}.
  In \bibinfo{booktitle}{\emph{Information Processing 77, Proceedings of the
  IFIP Congress}}. \bibinfo{publisher}{North-Holland},
  \bibinfo{pages}{155--160}.
\newblock


\bibitem[Lehman and Yao(1981)]%
        {lehman1981efficient}
\bibfield{author}{\bibinfo{person}{Philip~L. Lehman} {and}
  \bibinfo{person}{S.~Bing Yao}.} \bibinfo{year}{1981}\natexlab{}.
\newblock \showarticletitle{Efficient Locking for Concurrent Operations on
  {B}-Trees}.
\newblock \bibinfo{journal}{\emph{{ACM} Trans. Database Syst.}}
  \bibinfo{volume}{6}, \bibinfo{number}{4} (\bibinfo{year}{1981}),
  \bibinfo{pages}{650--670}.
\newblock


\bibitem[Li and Lee(2025)]%
        {li2025coordinationfree}
\bibfield{author}{\bibinfo{person}{Shulu Li} {and} \bibinfo{person}{Edward~A.
  Lee}.} \bibinfo{year}{2025}\natexlab{}.
\newblock \showarticletitle{A Preliminary Model of Coordination-free
  Consistency}.
\newblock \bibinfo{journal}{\emph{arXiv}}  \bibinfo{volume}{abs/2504.01141}
  (\bibinfo{year}{2025}).
\newblock
\href{https://doi.org/10.48550/arXiv.2504.01141}{doi:\nolinkurl{10.48550/arXiv.2504.01141}}


\bibitem[Loo et~al\mbox{.}(2009)]%
        {loo2009declarative}
\bibfield{author}{\bibinfo{person}{Boon~Thau Loo}, \bibinfo{person}{Tyson
  Condie}, \bibinfo{person}{Minos Garofalakis}, \bibinfo{person}{David~E. Gay},
  \bibinfo{person}{Joseph~M. Hellerstein}, \bibinfo{person}{Petros Maniatis},
  \bibinfo{person}{Raghu Ramakrishnan}, \bibinfo{person}{Timothy Roscoe}, {and}
  \bibinfo{person}{Ion Stoica}.} \bibinfo{year}{2009}\natexlab{}.
\newblock \showarticletitle{Declarative Networking}.
\newblock \bibinfo{journal}{\emph{Commun. ACM}} \bibinfo{volume}{52},
  \bibinfo{number}{11} (\bibinfo{year}{2009}), \bibinfo{pages}{87--95}.
\newblock
\href{https://doi.org/10.1145/1592761.1592785}{doi:\nolinkurl{10.1145/1592761.1592785}}


\bibitem[Loo et~al\mbox{.}(2005)]%
        {loo2005overlays}
\bibfield{author}{\bibinfo{person}{Boon~Thau Loo}, \bibinfo{person}{Tyson
  Condie}, \bibinfo{person}{Joseph~M. Hellerstein}, \bibinfo{person}{Petros
  Maniatis}, \bibinfo{person}{Timothy Roscoe}, {and} \bibinfo{person}{Ion
  Stoica}.} \bibinfo{year}{2005}\natexlab{}.
\newblock \showarticletitle{Implementing Declarative Overlays}. In
  \bibinfo{booktitle}{\emph{Proceedings of the 20th ACM Symposium on Operating
  Systems Principles (SOSP)}}. \bibinfo{pages}{75--90}.
\newblock
\href{https://doi.org/10.1145/1095810.1095818}{doi:\nolinkurl{10.1145/1095810.1095818}}


\bibitem[Lynch and Vaandrager(1995)]%
        {lynch1995forward}
\bibfield{author}{\bibinfo{person}{Nancy Lynch} {and} \bibinfo{person}{Frits
  Vaandrager}.} \bibinfo{year}{1995}\natexlab{}.
\newblock \showarticletitle{Forward and Backward Simulations}.
\newblock \bibinfo{journal}{\emph{Inf. Comput.}} \bibinfo{volume}{121},
  \bibinfo{number}{2} (\bibinfo{year}{1995}), \bibinfo{pages}{214--233}.
\newblock
\href{https://doi.org/10.1006/inco.1995.1134}{doi:\nolinkurl{10.1006/inco.1995.1134}}


\bibitem[Lynch(1996)]%
        {lynch1996distributed}
\bibfield{author}{\bibinfo{person}{Nancy~A. Lynch}.}
  \bibinfo{year}{1996}\natexlab{}.
\newblock \bibinfo{booktitle}{\emph{Distributed Algorithms}}.
\newblock \bibinfo{publisher}{Morgan Kaufmann}.
\newblock
\showISBNx{978-1-55860-348-6}


\bibitem[Lynch and Tuttle(1987)]%
        {lynch1987io}
\bibfield{author}{\bibinfo{person}{Nancy~A. Lynch} {and}
  \bibinfo{person}{Mark~R. Tuttle}.} \bibinfo{year}{1987}\natexlab{}.
\newblock \showarticletitle{Hierarchical Correctness Proofs for Distributed
  Algorithms}. In \bibinfo{booktitle}{\emph{Proceedings of the 6th Annual ACM
  Symposium on Principles of Distributed Computing (PODC)}}.
  \bibinfo{publisher}{ACM}, \bibinfo{pages}{137--151}.
\newblock
\href{https://doi.org/10.1145/41840.41852}{doi:\nolinkurl{10.1145/41840.41852}}


\bibitem[Mahajan et~al\mbox{.}(2011)]%
        {mahajan2011consistency}
\bibfield{author}{\bibinfo{person}{Prince Mahajan}, \bibinfo{person}{Lorenzo
  Alvisi}, {and} \bibinfo{person}{Mike Dahlin}.}
  \bibinfo{year}{2011}\natexlab{}.
\newblock \bibinfo{booktitle}{\emph{Consistency, Availability, and
  Convergence}}.
\newblock \bibinfo{type}{{T}echnical {R}eport} UTCS TR-11-22.
  \bibinfo{institution}{Department of Computer Science, The University of Texas
  at Austin}.
\newblock


\bibitem[Mao et~al\mbox{.}(2008)]%
        {mao2008mencius}
\bibfield{author}{\bibinfo{person}{Junfeng Mao}, \bibinfo{person}{Flavio~P.
  Junqueira}, {and} \bibinfo{person}{Keith Marzullo}.}
  \bibinfo{year}{2008}\natexlab{}.
\newblock \showarticletitle{Mencius: Building Efficient Replicated State
  Machines for WANs}. In \bibinfo{booktitle}{\emph{Proceedings of the 8th
  USENIX Symposium on Operating Systems Design and Implementation (OSDI)}}.
\newblock


\bibitem[Misra(1986)]%
        {misra1986des}
\bibfield{author}{\bibinfo{person}{Jayadev Misra}.}
  \bibinfo{year}{1986}\natexlab{}.
\newblock \showarticletitle{Distributed Discrete-Event Simulation}.
\newblock \bibinfo{journal}{\emph{{ACM} Comput. Surv.}} \bibinfo{volume}{18},
  \bibinfo{number}{1} (\bibinfo{year}{1986}), \bibinfo{pages}{39--65}.
\newblock
\href{https://doi.org/10.1145/6462.6485}{doi:\nolinkurl{10.1145/6462.6485}}


\bibitem[Moraru et~al\mbox{.}(2013)]%
        {moraru2013epaxos}
\bibfield{author}{\bibinfo{person}{Iulian Moraru}, \bibinfo{person}{David~G.
  Andersen}, {and} \bibinfo{person}{Michael Kaminsky}.}
  \bibinfo{year}{2013}\natexlab{}.
\newblock \showarticletitle{There Is More Consensus in Egalitarian
  Parliaments}. In \bibinfo{booktitle}{\emph{Proceedings of the 24th ACM
  Symposium on Operating Systems Principles (SOSP)}}.
  \bibinfo{pages}{358--372}.
\newblock
\href{https://doi.org/10.1145/2517349.2517350}{doi:\nolinkurl{10.1145/2517349.2517350}}


\bibitem[Newcombe et~al\mbox{.}(2015)]%
        {newcombe2015aws}
\bibfield{author}{\bibinfo{person}{Chris Newcombe}, \bibinfo{person}{Tim Rath},
  \bibinfo{person}{Fan Zhang}, \bibinfo{person}{Bogdan Munteanu},
  \bibinfo{person}{Marc Brooker}, {and} \bibinfo{person}{Michael Deardeuff}.}
  \bibinfo{year}{2015}\natexlab{}.
\newblock \showarticletitle{How {Amazon Web Services} Uses Formal Methods}.
\newblock \bibinfo{journal}{\emph{Commun. {ACM}}} \bibinfo{volume}{58},
  \bibinfo{number}{4} (\bibinfo{year}{2015}), \bibinfo{pages}{66--73}.
\newblock
\href{https://doi.org/10.1145/2699417}{doi:\nolinkurl{10.1145/2699417}}


\bibitem[Ongaro and Ousterhout(2014)]%
        {ongaro2014raft}
\bibfield{author}{\bibinfo{person}{Diego Ongaro} {and} \bibinfo{person}{John
  Ousterhout}.} \bibinfo{year}{2014}\natexlab{}.
\newblock \showarticletitle{In Search of an Understandable Consensus
  Algorithm}. In \bibinfo{booktitle}{\emph{Proceedings of the 2014 USENIX
  Annual Technical Conference (USENIX ATC)}}. \bibinfo{address}{Philadelphia,
  PA}, \bibinfo{pages}{305--319}.
\newblock


\bibitem[Schneider(1990)]%
        {schneider1990smr}
\bibfield{author}{\bibinfo{person}{Fred~B. Schneider}.}
  \bibinfo{year}{1990}\natexlab{}.
\newblock \showarticletitle{Implementing Fault-Tolerant Services Using the
  State Machine Approach: A Tutorial}.
\newblock \bibinfo{journal}{\emph{Comput. Surveys}} \bibinfo{volume}{22},
  \bibinfo{number}{4} (\bibinfo{year}{1990}), \bibinfo{pages}{299--319}.
\newblock
\href{https://doi.org/10.1145/98163.98167}{doi:\nolinkurl{10.1145/98163.98167}}


\bibitem[Shapiro et~al\mbox{.}(2011)]%
        {shapiro2011crdt}
\bibfield{author}{\bibinfo{person}{Marc Shapiro}, \bibinfo{person}{Nuno
  Pregui{\c{c}}a}, \bibinfo{person}{Carlos Baquero}, {and}
  \bibinfo{person}{Marek Zawirski}.} \bibinfo{year}{2011}\natexlab{}.
\newblock \showarticletitle{Conflict-Free Replicated Data Types}. In
  \bibinfo{booktitle}{\emph{Stabilization, Safety, and Security of Distributed
  Systems ({SSS} 2011)}} \emph{(\bibinfo{series}{Lecture Notes in Computer
  Science}, Vol.~\bibinfo{volume}{6976})}. \bibinfo{publisher}{Springer},
  \bibinfo{pages}{386--400}.
\newblock
\href{https://doi.org/10.1007/978-3-642-24550-3_29}{doi:\nolinkurl{10.1007/978-3-642-24550-3_29}}


\bibitem[Sivaramakrishnan et~al\mbox{.}(2015)]%
        {sivaramakrishnan2015quelea}
\bibfield{author}{\bibinfo{person}{KC Sivaramakrishnan},
  \bibinfo{person}{Gowtham Kaki}, {and} \bibinfo{person}{Suresh Jagannathan}.}
  \bibinfo{year}{2015}\natexlab{}.
\newblock \showarticletitle{Declarative Programming over Eventually Consistent
  Data Stores}. In \bibinfo{booktitle}{\emph{Proceedings of the 36th ACM
  SIGPLAN Conference on Programming Language Design and Implementation
  (PLDI)}}. \bibinfo{pages}{413--424}.
\newblock
\href{https://doi.org/10.1145/2737924.2737981}{doi:\nolinkurl{10.1145/2737924.2737981}}


\bibitem[Stonebraker(1986)]%
        {stonebraker1986sharednothing}
\bibfield{author}{\bibinfo{person}{Michael Stonebraker}.}
  \bibinfo{year}{1986}\natexlab{}.
\newblock \showarticletitle{The Case for Shared Nothing}.
\newblock \bibinfo{journal}{\emph{IEEE Database Engineering Bulletin}}
  \bibinfo{volume}{9}, \bibinfo{number}{1} (\bibinfo{year}{1986}),
  \bibinfo{pages}{4--9}.
\newblock


\bibitem[Valiant(1990)]%
        {valiant1990bsp}
\bibfield{author}{\bibinfo{person}{Leslie~G. Valiant}.}
  \bibinfo{year}{1990}\natexlab{}.
\newblock \showarticletitle{A Bridging Model for Parallel Computation}.
\newblock \bibinfo{journal}{\emph{Commun. ACM}} \bibinfo{volume}{33},
  \bibinfo{number}{8} (\bibinfo{year}{1990}), \bibinfo{pages}{103--111}.
\newblock
\href{https://doi.org/10.1145/79173.79181}{doi:\nolinkurl{10.1145/79173.79181}}


\bibitem[Wilschut and Apers(1991)]%
        {wilschut1991dataflow}
\bibfield{author}{\bibinfo{person}{Annita~N. Wilschut} {and}
  \bibinfo{person}{Peter M.~G. Apers}.} \bibinfo{year}{1991}\natexlab{}.
\newblock \showarticletitle{Dataflow Query Execution in a Parallel Main-Memory
  Environment}. In \bibinfo{booktitle}{\emph{Proceedings of the First
  International Conference on Parallel and Distributed Information Systems
  (PDIS)}}. \bibinfo{pages}{68--77}.
\newblock
\href{https://doi.org/10.1109/PDIS.1991.183069}{doi:\nolinkurl{10.1109/PDIS.1991.183069}}


\bibitem[Zhou et~al\mbox{.}(2021)]%
        {zhou2021foundationdb}
\bibfield{author}{\bibinfo{person}{Jingyu Zhou}, \bibinfo{person}{Meng Xu},
  \bibinfo{person}{Alexander Shraer}, \bibinfo{person}{Bala Namasivayam},
  \bibinfo{person}{Alex Miller}, \bibinfo{person}{Evan Tschannen},
  \bibinfo{person}{Steve Atherton}, \bibinfo{person}{Andrew~J. Beamon},
  \bibinfo{person}{Rusty Sears}, \bibinfo{person}{John Leach},
  \bibinfo{person}{Dave Rosenthal}, \bibinfo{person}{Xin Dong},
  \bibinfo{person}{Will Wilson}, \bibinfo{person}{Ben Collins},
  \bibinfo{person}{David Scherer}, \bibinfo{person}{Alec Grieser},
  \bibinfo{person}{Young Liu}, \bibinfo{person}{Alvin Moore},
  \bibinfo{person}{Bhaskar Muppana}, \bibinfo{person}{Xiaoge Su}, {and}
  \bibinfo{person}{Vishesh Yadav}.} \bibinfo{year}{2021}\natexlab{}.
\newblock \showarticletitle{{FoundationDB}: A Distributed Unbundled
  Transactional Key Value Store}. In \bibinfo{booktitle}{\emph{Proceedings of
  the 2021 International Conference on Management of Data (SIGMOD)}}.
  \bibinfo{publisher}{ACM}, \bibinfo{pages}{2653--2666}.
\newblock
\href{https://doi.org/10.1145/3448016.3457559}{doi:\nolinkurl{10.1145/3448016.3457559}}


\bibitem[Zinn et~al\mbox{.}(2012)]%
        {zinn2012winmove}
\bibfield{author}{\bibinfo{person}{Daniel Zinn}, \bibinfo{person}{Todd~J.
  Green}, {and} \bibinfo{person}{Bertram Lud{\"a}scher}.}
  \bibinfo{year}{2012}\natexlab{}.
\newblock \showarticletitle{Win-move is Coordination-Free (Sometimes)}. In
  \bibinfo{booktitle}{\emph{Proceedings of the 15th International Conference on
  Database Theory (ICDT)}}. \bibinfo{pages}{99--113}.
\newblock
\href{https://doi.org/10.1145/2274576.2274588}{doi:\nolinkurl{10.1145/2274576.2274588}}


\end{thebibliography}

\end{document}